\documentclass[runningheads]{llncs}
\usepackage[T1]{fontenc}
\usepackage{graphicx}
\usepackage{amssymb,amsmath}
\usepackage{mathrsfs}
\usepackage{xfrac}
\usepackage{tikz}
\usetikzlibrary{calc,arrows,petri,shapes.geometric}
\tikzstyle{place}=[circle,thick,draw=blue!75,fill=blue!20,minimum size=6mm]
\tikzstyle{red place}=[place,draw=red!75,fill=red!20]
\tikzstyle{transition}=[rectangle,thick,draw=black!75,fill=black!20,minimum size=5mm]
\pgfdeclarelayer{bg}
\pgfsetlayers{bg,main}
\usepackage{lmodern}
\usepackage{hyperref}
\usepackage[capitalise,nameinlink]{cleveref}
\usepackage{booktabs}

\usepackage{placeins}

\newcommand{\pre}[1]{{^{\bullet}#1}}
\newcommand{\post}[1]{#1^{\bullet}}

\newcommand{\magnitudename}{\text{mag}}
\newcommand{\magnitude}{C_{\magnitudename}}
\newcommand{\varietyname}{\text{var}}
\newcommand{\variety}{C_{\varietyname}}
\newcommand{\supportname}{\text{len}}
\newcommand{\support}{C_{\supportname}}
\newcommand{\tlminname}{\text{TL-min}}
\newcommand{\tlmin}{C_{\tlminname}}
\newcommand{\tlavgname}{\text{TL-avg}}
\newcommand{\tlavg}{C_{\tlavgname}}
\newcommand{\tlmaxname}{\text{TL-max}}
\newcommand{\tlmax}{C_{\tlmaxname}}
\newcommand{\levelofdetailname}{\text{LOD}}
\newcommand{\levelofdetail}{C_{\levelofdetailname}}
\newcommand{\numberoftiesname}{\text{t-comp}}
\newcommand{\numberofties}{C_{\numberoftiesname}}
\newcommand{\lempelzivname}{\text{LZ}}
\newcommand{\lempelziv}{C_{\lempelzivname}}
\newcommand{\numberuniquetracesname}{\text{DT-\#}}
\newcommand{\numberuniquetraces}{C_{\numberuniquetracesname}}
\newcommand{\percentageuniquetracesname}{\text{DT-\%}}
\newcommand{\percentageuniquetraces}{C_{\percentageuniquetracesname}}
\newcommand{\structurename}{\text{struct}}
\newcommand{\structure}{C_{\structurename}}
\newcommand{\affinityname}{\text{affinity}}
\newcommand{\affinity}{C_{\affinityname}}
\newcommand{\deviationfromrandomname}{\text{dev-R}}
\newcommand{\deviationfromrandom}{C_{\deviationfromrandomname}}
\newcommand{\avgdistname}{\text{avg-dist}}
\newcommand{\avgdist}{C_{\avgdistname}}
\newcommand{\varentropyname}{\text{var-e}}
\newcommand{\varentropy}{C_{\varentropyname}}
\newcommand{\normvarentropyname}{\text{nvar-e}}
\newcommand{\normvarentropy}{C_{\normvarentropyname}}
\newcommand{\seqentropyname}{\text{seq-e}}
\newcommand{\seqentropy}{C_{\seqentropyname}}
\newcommand{\normseqentropyname}{\text{nseq-e}}
\newcommand{\normseqentropy}{C_{\normseqentropyname}}

\newcommand{\sizename}{\text{size}}
\newcommand{\size}{C_{\sizename}}
\newcommand{\mismatchname}{\text{MM}}
\newcommand{\mismatch}{C_{\mismatchname}}
\newcommand{\connhetname}{\text{CH}}
\newcommand{\connhet}{C_{\connhetname}}
\newcommand{\crossconnname}{\text{CC}}
\newcommand{\crossconn}{C_{\crossconnname}}
\newcommand{\tokensplitname}{\text{ts}}
\newcommand{\tokensplit}{C_{\tokensplitname}}
\newcommand{\separabilityname}{\text{sep}}
\newcommand{\separability}{C_{\separabilityname}}
\newcommand{\controlflowname}{\text{CFC}}
\newcommand{\controlflow}{C_{\controlflowname}}
\newcommand{\avgconnname}{\text{acd}}
\newcommand{\avgconn}{C_{\avgconnname}}
\newcommand{\maxconnname}{\text{mcd}}
\newcommand{\maxconn}{C_{\maxconnname}}
\newcommand{\sequentialityname}{\text{seq}}
\newcommand{\sequentiality}{C_{\sequentialityname}}
\newcommand{\depthname}{\text{depth}}
\newcommand{\depth}{C_{\depthname}}
\newcommand{\diametername}{\text{diam}}
\newcommand{\diameter}{C_{\diametername}}
\newcommand{\cyclicityname}{\text{cyc}}
\newcommand{\cyclicity}{C_{\cyclicityname}}
\newcommand{\netconnname}{\text{CNC}}
\newcommand{\netconn}{C_{\netconnname}}
\newcommand{\densityname}{\text{dens}}
\newcommand{\density}{C_{\densityname}}
\newcommand{\duplicatename}{\text{dup}}
\newcommand{\duplicate}{C_{\duplicatename}}
\newcommand{\emptyseqname}{\emptyset}
\newcommand{\emptyseq}{C_{\emptyseqname}}

\newcommand{\mgeq}{\textcolor{green!45!blue}{\geq}}
\newcommand{\mgreater}{\textcolor{green!90!black}{>}}
\newcommand{\meq}{\textcolor{blue}{=}}
\newcommand{\mleq}{\textcolor{green!45!blue}{\leq}}
\newcommand{\mless}{\textcolor{green!90!black}{<}}
\newcommand{\norel}{\textcolor{red}{X}}

\newcommand{\moc}{\textsl{MoC}}
\newcommand{\loc}{\textsl{LoC}}

\def\scalefactor{0.85}
\def\pad{\hspace*{1.5mm}}

\makeatletter
\newcommand{\shiftleft}{\hspace*{-\@totalleftmargin}}

\begin{document}

\title{Mind the Gap: A Formal Investigation of the Relationship Between Log and Model Complexity -- Extended Version}
\titlerunning{The Relationship Between Log and Model Complexity}
%
\author{Patrizia Schalk\inst{1} \and 
Artem Polyvyanyy\inst{2}}
\authorrunning{P. Schalk et al.}
%
\institute{University of Augsburg, Universitätsstraße 6a, 86159 Augsburg, Germany\\
\email{patrizia.schalk@uni-a.de} \and
Melbourne Connect, The University of Melbourne, VIC, 3010, Australia\\
\email{artem.polyvyanyy@unimelb.edu.au } 
} 
\maketitle%
\begin{abstract}
Simple process models are key for effectively communicating the outcomes of process mining.
An important question in this context is whether the complexity of event logs used as inputs to process discovery algorithms can serve as a reliable indicator of the complexity of the resulting process models.
Although various complexity measures for both event logs and process models have been proposed in the literature, the relationship between input and output complexity remains largely unexplored.
In particular, there are no established guidelines or theoretical foundations that explain how the complexity of an event log influences the complexity of the discovered model.
This paper examines whether formal guarantees exist such that increasing the complexity of event logs leads to increased complexity in the discovered models.
We study 18 log complexity measures and 17 process model complexity measures across five process discovery algorithms.
Our findings reveal that only the complexity of the flower model can be established by an event log complexity measure.
For all other algorithms, we investigate which log complexity measures influence the complexity of the discovered models.
The results show that current log complexity measures are insufficient to decide which discovery algorithms to choose to construct simple models.
We propose that authors of process discovery algorithms provide insights into which log complexity measures predict the complexity of their results.
\end{abstract}
\section{Introduction}
\label{sec:intro}
Processes are everywhere in our daily lives.
Starting from handling orders in an online shop, ranging over the executions of treatments in hospitals, to things as mundane as following a recipe.
It comes to no surprise that organisations are eager to find and optimise such processes in a structured and automated fashion.
To aid organisations with this task is the goal of \emph{process mining}~\cite{Aal16}.
This relatively young research discipline essentially consists of three phases:
Techniques for \emph{process discovery} automatically find a process model for previously recorded data of the system.
Since there are many process discovery techniques to choose from, \emph{conformance checking} enables its users to decide which process model represents the data best without having to scan through the entire dataset~\cite{CarDSW18}.
Finally, during \emph{process enhancement}, the discovered and selected models give conclusions on how to adapt the real process to make it more efficient or rule-conformant.

Since the last phase depends on the specific process at hand, research in process mining is especially interested in the first two phases.
As such, the literature presents a vast amount of process discovery techniques that still regularly finds new additions.
The quality of the resulting models is checked within four quality dimensions:
\emph{Fitness} rewards models that can replay all behaviour in the data.
\emph{Precision}, on the other hand, rewards models that do not deviate from this behaviour.
The model $M$ of \cref{fig:tracenet-example} shows that fitness and precision alone are not enough to ensure good model-quality, since $M$ has perfect fitness and precision, but is merely another way to represent the raw data.
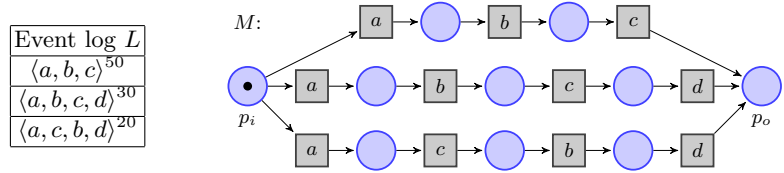
\begin{figure}[ht]
	\centering
	\begin{minipage}{0.2\textwidth}
		\centering
		\begin{tabular}{|c|}\hline
			Event log $L$ \\ \hline
			$\langle a,b,c \rangle^{50}$ \\ \hline
			$\langle a,b,c,d \rangle^{30}$ \\ \hline
			$\langle a,c,b,d \rangle^{20}$ \\ \hline
		\end{tabular}
	\end{minipage}
	\begin{minipage}{0.7\textwidth}
		\centering
		\scalebox{\scalefactor}{
		\begin{tikzpicture}[node distance = 1cm,>=stealth',bend angle=0,auto]
			\node[place,tokens=1,label=below:$p_i$] (start) {};
			\node[above of=start] {$M$:};
			\node[transition,right of=start] (a) {$a$}
			edge [pre] (start);
			\node[place,right of=a] (p1) {}
			edge [pre] (a);
			\node[transition,right of=p1] (b) {$b$}
			edge [pre] (p1);
			\node[place,right of=b] (p2) {}
			edge [pre] (b);
			\node[transition,right of=p2] (c) {$c$}
			edge [pre] (p2);
			\node[place,right of=c] (p3) {}
			edge [pre] (c);
			\node[transition,right of=p3] (d) {$d$}
			edge [pre] (p3);
			\node[place,right of=d,label=below:$p_o$] (end) {}
			edge [pre] (d);
			\node (dummy) at (p1) {};
			\node[transition,below of=a] (a) {$a$}
			edge [pre] (start);
			\node[place,right of=a] (p1) {}
			edge [pre] (a);
			\node[transition,right of=p1] (c) {$c$}
			edge [pre] (p1);
			\node[place,right of=c] (p2) {}
			edge [pre] (c);
			\node[transition,right of=p2] (b) {$b$}
			edge [pre] (p2);
			\node[place,right of=b] (p3) {}
			edge [pre] (b);
			\node[transition,right of=p3] (d) {$d$}
			edge [pre] (p3)
			edge [post] (end);
			\node[transition,above of=dummy] (a) {$a$}
			edge [pre] (start);
			\node[place,right of=a] (p1) {}
			edge [pre] (a);
			\node[transition,right of=p1] (b) {$b$}
			edge [pre] (p1);
			\node[place,right of=b] (p2) {}
			edge [pre] (b);
			\node[transition,right of=p2] (c) {$c$}
			edge [pre] (p2)
			edge [post] (end);
		\end{tikzpicture}}
	\end{minipage}
	\caption{An event log $L$ and its trace net $M$ with perfect fitness and precision.}
	\label{fig:tracenet-example}
\end{figure}
Thus, \emph{generalisation} rewards models that deviate from the recorded data, if these deviations are possible executions in the process.
The \emph{simplicity} dimension rewards models that are easy to read and understand.

High simplicity is crucial to analyse the model during the process enhancement phase, and to present the findings to stake-holders and decision-makers.
Furthermore, low simplicity in a process model indicates the existence of errors in the model~\cite{Men08}.
Due to its importance, multiple measures for this dimension emerged in the literature.
We call these measures \emph{simplicity measures}, if simpler models receive higher values, or \emph{complexity measures}, if simpler models receive lower values.
High complexity in process models is often the result of complex input data, rather than the fault of process discovery techniques~\cite{Aal12}.
In turn, complexity measures for recorded data are as important, as they aim to estimate the complexity of the model before process discovery~\cite{Guen09}. 

Yet, to this date, there is no proved theoretical connection between complexity measures for data and for models.
In this paper, we analyse whether complexity measures for data can predict the complexity of models mined with specific process discovery techniques.
We describe the state of the art in \cref{sec:related-work} and set the scene with the necessary definitions in \cref{sec:definitions}.
In \cref{sec:relationships}, we investigate how increasing complexity of the underlying data influences the complexity of automatically discovered models.
We investigate two baseline discovery algorithms and three more advanced mining techniques and discuss what types of complexity measures for data are currently missing.
Finally, in \cref{sec:conclusion}, we summarise and give suggestions for future research.

\section{Related Work}
\label{sec:related-work}
Complex process models come with several disadvantages.
Mendling~\cite{Men08} showed that complex models are more likely to contain errors and that complexity measures can predict these errors, highlighting the importance of the simplicity dimension.
To further emphasise this importance, Reijers~et~al.~\cite{ReiM11} investigated the influence of complex structures in process models to their understandability.
They found that measures that punish connectors in a model are best-suited to predict its understandability.
Yet, they found that personal factors like experience have the highest impact on understandability.
Lieben~et~al.~\cite{LieDJJ18} showed via a factor analysis that most of the complexity measures in the literature fall into four different dimensions.
Thereby, they considerably reduce the amount of complexity measures process analysts have to choose from when evaluating simplicity.
Schalk~et~al.~\cite{SchBL24} further deepened this analysis by comparing mathematical properties of complexity measures inside the same dimension.

On the side of complexity measures for data, Günther~\cite{Guen09} found that poor quality in data means poor quality in discovered process models.
They therefore defined multiple complexity measures for so-called event logs, which are typically used to store recorded data in business processes.
The goal of the defined complexity measures is to evaluate the structure of event logs, and to select a suitable process discovery algorithm for the analysis of these logs~\cite[p. 50]{Guen09}.
Furthermore, they propose to use these measures to estimate the computational complexity of process mining algorithms.
Yet, concrete guidelines for which process discovery algorithm to choose when certain log complexity scores are high are missing. 
Augusto~et~al.~\cite{AugMVW22} therefore analysed the influence of log complexity on the fitness, precision, size, and control flow complexity of three high-level discovery algorithms.
Using statistical analysis, they found that only the number of different event names in the event log (variety) and the average edit distance between two traces of the log are good predictors.
Furthermore, they defined four new graph-entropy-based complexity measures, out of which one is a good predictor for the fitness of the model returned by the split miner.

Surprised by these findings, in this paper, we investigate whether there is a theoretical connection between existing log complexity measures and the complexity of discovered process models.
To do so, we use the models of five simple process discovery techniques and research the effect of increasing log complexity on their model complexity.
We use the 18 log complexity measures collected by Augusto~et~al.~\cite{AugMVW22} and the 17 model complexity measures collected by Lieben~et~al.~\cite{LieDJJ18}.
Since only the model complexity scores of the flower model show a direct connection to existing log complexity measures, we continue the analysis by providing measures that are better-suited to predict model complexity of the discovered models.
This way, we enable users of log complexity measures to draw the right conclusions. 

\section{Basic Definitions}
\label{sec:definitions}
We define $\mathbb{N} := \{1,2,3,\dots\}$ as the set of natural numbers, $\mathbb{N}_0 := \mathbb{N} \cup \{0\}$ as the set of non-negative natural numbers, and $\mathbb{R}_0^+$ as the set of non-negative real numbers.
Let $A$ be an alphabet. 
A \emph{trace} over $A$ is is a sequence of elements drawn from $A$, i.e., $\sigma = \langle a_1, \dots, a_n \rangle$, where $a_1, \dots, a_n \in A$.
The \emph{length} of such a trace is denoted by $|\sigma| := n$.
The (unique) trace with length $0$ is denoted by $\epsilon$ and called the \emph{empty trace}.
For all $i \in \{1, \dots, n\}$, we write $\sigma(i) := a_i$ to address the element at the $i$-th position in the trace.
For two arbitrary traces $\sigma_1 = \langle a_1, \dots, a_k \rangle$ and $\sigma_2 = \langle b_1, \dots, b_l \rangle$ over $A$, we define their \emph{concatenation} as the trace $\sigma_1 \cdot \sigma_2 := \langle a_1, \dots, a_k, b_1, \dots, b_l \rangle$.
For a trace $\sigma$, the $n$-ary concatenation of $\sigma$ is defined inductively as $\sigma^0 = \varepsilon$ and $\sigma^{n+1} = \sigma \cdot \sigma^n$ for $n \geq 0$.

For any set $D$, we define a \emph{multiset} $m$ as a total function $m : D \rightarrow \mathbb{N}_0$, where for any $d \in D$, $m(d)$ is the number of occurences of the element $d$ in the multiset~$m$.
For two multisets $m_1, m_2$, we define $m_1 + m_2$ as the multiset $(m_1 + m_2)$ with $\forall d \in D: (m_1 + m_2)(d) = m_1(d) + m_2(d)$.
We write $m_1 \sqsubseteq m_2$ if $\forall d \in D: m_1(d) \leq m_2(d)$, and $m_1 \sqsubset m_2$ if $m_1 \sqsubseteq m_2$ and $m_1 \neq m_2$.
We define the \emph{support} of a multiset $m$ as the set $supp(m) := \{d \in D \mid m(d) > 0\}$.
An \emph{event log} $L$ is a multiset of traces. 
We represent event logs the way shown in the example of \cref{fig:tracenet-example}, by adding the frequency of each trace to its superset.

\begin{definition}[Petri nets and workflow nets]
A (unlabeled) Petri net is a triple $N = (P,T,F)$, where $P$ is the set of places, $T$ is the set of transitions, $P \cap T = \emptyset$, and $F \subseteq (P \times T) \cup (T \times P)$ is the flow relation.
For any place $p \in P$, we define its preset as $\pre{p} := \{t \in T \mid (t,p) \in F\}$ and its postset as $\post{p} := \{t \in T \mid (p,t) \in F\}$.
We define pre-and postsets of transitions accordingly.

A workflow net is a $7$-tuple $W = (P,T,F,\ell,A,p_i,p_o)$, where $(P,T,F)$ defines a Petri net, $\ell : T \rightarrow (A \cup \{\tau\})$ is a function assigning a label of $A$ or the special label $\tau \not\in A$ to the transitions in the net, and where $p_i, p_o \in P$ are places with:
\begin{itemize}
	\item $p_i$ is the only place without incoming arcs, i.e. $\pre{p_i} = \emptyset$,
	\item $p_o$ is the only place without outgoing arcs, i.e. $\post{p_o} = \emptyset$,
	\item every node lies on some path from $p_i$ to $p_o$.
\end{itemize}
Transitions $t \in T$ with $\ell(t) = \tau$ are called silent transitions.
\end{definition}

\cref{fig:tracenet-example} shows an example for a workflow net $M$.
To visually distinguish between places and transitions, we draw places as circles and transitions as rectangles.
As $M$ demonstrates, the labeling function enables us to assign the same label to multiple different transitions.
Furthermore, every arc in a Petri net has a place as start point and a transition as end point or a transition as start point and a place as end point.
In other words, there can never be arcs between two places or between two transitions.

It is possible that multiple arcs leave or enter a place or a transition.
If multiple arcs leave a place, the transitions in its postset compete for the tokens in the place.
Thus, such places initiate a choice between the transitions in its postset.
On the other hand, if multiple arcs leave a transition, then this transition initiates a parallel execution.
Most complexity measures are interested in these special types of nodes in a Petri net.
Thus, we next define the notion of \emph{connectors} in a workflow net.

\begin{definition}[Connectors in workflow nets]
\label{def:connectors}
Let $W = (P, T, F, \ell, p_i, p_o)$ be a workflow net, where $t \in T$ is a transition and $p \in P$ is a place.
\begin{itemize}
\item If $|\post{p}| > 1$, we call $p$ an \emph{\texttt{xor}-split}.
\item If $|\pre{p}| > 1$, we call $p$ an \emph{\texttt{xor}-join}.
\item If $|\post{t}| > 1$, we call $t$ an \emph{\texttt{and}-split}.
\item If $|\pre{t}| > 1$, we call $t$ an \emph{\texttt{and}-join}.
\end{itemize}
Accordingly, we define 
\begin{itemize}
\setlength\itemsep{0.25em}
\item the set of \texttt{xor}-splits in $W$ as $\mathcal{S}_{\texttt{xor}}^W := \{p \in P \mid |\post{p}| > 1\}$,
\item the set of \texttt{xor}-joins in $W$ as $\mathcal{J}_{\texttt{xor}}^W := \{p \in P \mid |\pre{p}| > 1\}$,
\item the set of \texttt{and}-splits in $W$ as $\mathcal{S}_{\texttt{and}}^W := \{t \in T \mid |\post{t}| > 1\}$,
\item the set of \texttt{and}-joins in $W$ as $\mathcal{J}_{\texttt{and}}^W := \{t \in T \mid |\pre{t}| > 1\}$.
\end{itemize}
Note that these sets are not necessarily disjoint.
The set of xor-connectors in $W$ is $\mathcal{C}_{\texttt{xor}}^W := \mathcal{S}_{\texttt{xor}}^W \cup \mathcal{J}_{\texttt{xor}}^W$, the set of and-connectors in $W$ is $\mathcal{C}_{\texttt{and}}^W := \mathcal{S}_{\texttt{and}}^W \cup \mathcal{J}_{\texttt{and}}^W$ and the set of all connectors is $\mathcal{C}^W := \mathcal{C}_{\texttt{xor}}^W \cup \mathcal{C}_{\texttt{and}}^W$.
\end{definition}

Most of the discovery techniques we investigate produce workflow nets.
Yet, we are aware that organisations often use directly follows graphs (DFG) and extend our analyses to this model-type.

\begin{definition}[Directly follows graph]
\label{def:directly-follows-graph}
Let $L$ be an event log over a set of activity names $A$.
For $x, y \in A$, we write $x >_L y$ if there is a trace $\sigma \in L$ with $\sigma(i) = x$ and $\sigma(i+1) = y$ for some $i \in \{1, \dots, |\sigma|\}$.
The directly follows graph for $L$ is the graph $DFG(L) = (V,E)$ with $V := A \cup \{\triangleright, \square\}$, where $\triangleright, \square \not\in A$, and with
\begin{align*}
E := &\{(\triangleright, x) \mid \exists \sigma \in L: \sigma(1) = x\} \\
\cup\, &\{(x, y) \mid x >_L y\} \\
\cup\, &\{(x, \square) \mid \exists \sigma \in L: \sigma(|\sigma|) = x\}.
\end{align*}
For an event log $L$ and its directly follows graph $DFG(L)$, we denote the set of vertices in $DFG(L)$ by $V(DFG(L))$ and the set of edges in $DFG(L)$ by $E(DFG(L))$.
\end{definition}

\cref{fig:dfg-example} shows an example of a directly follows graph for the example event log $L$ shown in \cref{fig:tracenet-example}.
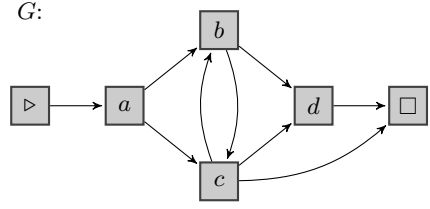
\begin{figure}
	\centering
	\begin{tikzpicture}[node distance = 1.25cm,>=stealth',bend angle=0,auto]
		\node[transition] (start) {$\triangleright$};
		\node[above of=start] {$G$:};
		\node[transition,right of=start] (a) {$a$}
		edge [pre] (start);
		\node[transition,right of=a,yshift=1cm] (b) {$b$}
		edge [pre] (a);
		\node[transition,right of=a,yshift=-1cm] (c) {$c$}
		edge [pre] (a)
		edge [pre,bend right=20] (b)
		edge [post,bend left=20] (b);
		\node[transition,right of=b,yshift=-1cm] (d) {$d$}
		edge [pre] (b)
		edge [pre] (c);
		\node[transition,right of=d] (end) {$\square$}
		edge [pre] (d)
		edge [pre,bend left=20] (c);
	\end{tikzpicture}
	\caption{The directly follows graph $G$ for the event log $L$ of \cref{fig:tracenet-example}.}
	\label{fig:dfg-example}
\end{figure}

\subsection{Complexity of Process Models}
\label{sec:model-complexity}
For this section, let $\mathcal{M}$ be the set of all process models.
We define a model-complexity measure as a function $\mathcal{C}^M : \mathcal{M} \rightarrow \mathbb{R}_0^+$, assigning a non-negative, real-valued score to workflow nets.
For our analyses, we investigate the model complexity measures collected by Lieben~et~al~\cite{LieDJJ18} and redefined for workflow nets by Schalk~et~al.~\cite{SchBL24}.
To make this paper self-contained, we repeat their formal definitions here. 
Let $W = (P,T,F,\ell,A,p_i,p_o)$ be a workflow net.
\begin{itemize}
\item The \textbf{size}~\cite{Men08} $\size$ of the workflow net $W$ is the number of nodes in its graphical representation.
More precisely, the size of $W$ is its number of places plus its number of transitions, $\size(W) = |P| + |T|$.

\item The \textbf{connector mismatch}~\cite{Men08} $\mismatch$ aims to estimate the amount of \texttt{xor}-splits that were closed by \texttt{and}-joins and the amount of \texttt{and}-splits that were closed by \texttt{xor}-joins in $W$.
Such connector mismatches often occur in practice, but render a workflow net more complex.
To avoid checking all paths in a workflow nets to find these connector mismatches, we calculate the difference of arcs exiting \texttt{xor}-splits and of arcs entering \texttt{xor}-joins:
\[MM_{\texttt{xor}}^W := \left|{\sum}_{t \in \mathcal{S}_{\texttt{xor}}^W} |\post{t}| - {\sum}_{t \in \mathcal{J}_{\texttt{xor}}^W} |\pre{t}|\right|\]
Analogously, we calculate the difference of arcs exiting \texttt{and}-splits and of arcs entering \texttt{and}-joins, giving us: 
\[MM_{\text{and}}^W := \left|{\sum}_{t \in \mathcal{S}_{\text{and}}^W} |\post{t}| - {\sum}_{t \in \mathcal{J}_{\text{and}}^W} |\pre{t}|\right|\]
We combine these two sub-measures to the connector mismatch measure $\mismatch(M) = MM_{\texttt{xor}}^W + MM_{\texttt{and}}^W$.

\item The \textbf{connector heterogeneity}~\cite{Men08} $\connhet(W)$ of $W$ is the entropy of its connector types.
If the workflow net $W$ has only one type of connectors, the score of this measure is $0$.
On the other hand, if it contains every connector-type equally often, the score of this measure is $1$.
To achieve this, we define:
\[\connhet(W) = - \left(\frac{|\mathcal{C}_{\text{and}}^W|}{|\mathcal{C}^W|} \cdot \log_2\left(\frac{|\mathcal{C}_{\text{and}}^W|}{|\mathcal{C}^W|}\right) + \frac{|\mathcal{C}_{\text{xor}}^W|}{|\mathcal{C}^W|}\cdot \log_2\left(\frac{|\mathcal{C}_{\text{xor}}^W|}{|\mathcal{C}^W|}\right)\right)\]

\item The \textbf{cross-connectivity metric}~\cite{VanRMAC08} $\crossconn$ identifies how strong the connection between two nodes in $W$ is.
The idea is that two activities that always occur together in an execution sequence, are stronger connected than two activities that are independent of each other.
This means, alternative activities are loosely connected.
Accordingly, we define the weight of a transition in $W$ as:
\begin{equation*}
w_W(v) := \begin{cases}
\frac{1}{|\pre{v}| + |\post{v}|} & \text{ if } v \in \mathcal{C}_{\texttt{xor}}^W \\
1 & \text{ if } v \in \mathcal{C}_{\texttt{and}}^W \\
1 & \text{ otherwise.}
\end{cases}
\end{equation*}
Thus, places that have more than one outgoing or incoming arc get a weight less than $1$, while all other nodes in $W$ have weight $1$.
Weights are extended to the edges of the workflow net by defining $w_W((u,v)) = w_W(u) \cdot w_W(v)$ for any edge $(u,v) \in F$.
For a simple path $\rho = v_1, v_2, \dots, v_{k-1}, v_k$, we set its weight to $w_W(\rho) = w_W((v_1, v_2)) \cdot \ldots \cdot w_W((v_{k-1},v_k))$ and define the value of a connection as:
\begin{equation*}
V_W(v_i, v_j) := \max(\{w_W(\rho) \mid \rho \text{ is a simple path in } W \text{ from } v_i \text{ to } v_j\} \cup \{0\})
\end{equation*}
To calculate the score of the cross-connectivity metric, we take the average of all connection-values and subtract the result from $1$:
\[\crossconn(W) = 1 - \frac{\sum_{v_1, v_2 \in P \cup T} V_W(v_1, v_2)}{(|P| + |T|) \cdot (|P| + |T| - 1)}\]

\item The \textbf{token split}~\cite{Men08} $\tokensplit$ is the minimum amount of edges that need to be removed, such that the resulting net has no \texttt{and}-splits anymore.
In turn, $\tokensplit(W) = \sum_{t \in \mathcal{S}_{\texttt{and}}} (|\post{t}| - 1)$.

\item The \textbf{control flow complexity}~\cite{Car05} $\controlflow$ estimates the cognitive load of a person that tries to understand the workflow net.
The idea is that parallel splits add some complexity, but keep the amount of possible control flows unchanged.
Split-connectors that start exclusive choices, however, add $k$ possible control flows, where $k$ is the amount of edges leaving the connector node.
With this, $\controlflow(W) = |\mathcal{S}_{\texttt{and}}^W| + \sum_{p \in \mathcal{S}_{\texttt{xor}}^W} |\post{p}|$.

\item The \textbf{separability}~\cite{Men08} $\separability$ is the ratio of cut-vertices in the workflow net.
In graph-theory, a cut-vertex is a node whose removal results in an increase of the amount of connected components of the graph.
If the graph has many cut-vertices, there are fewer structures in the graph where all nodes are connected to each other.
Since the initial place $p_i$ and the output place $p_o$ can never be cut-vertices, we  calculate the ratio of cut-vertices by dividing by $|P| + |T| - 2$ and set $\separability(W) = 1 -\frac{|\{v \in P \cup T \mid v \text{ is a cut-vertex in } W\}|}{|P| + |T| - 2}$.

\item The \textbf{average connector degree}~\cite{Men08} $\avgconn$ is the average amount of incoming and outgoing arcs of connector nodes, $\avgconn(W) = \frac{1}{|\mathcal{C}^W|} \cdot \sum_{x \in \mathcal{C}^W} (|\pre{x}| + |\post{x}|)$.

\item The \textbf{maximum connector degree}~\cite{Men08} $\maxconn$ is the maximum amount of incoming and outgoing arcs of connector nodes, so we define this measure as $\maxconn(W) = \max\{(|\pre{x}| + |\post{x}|) \mid x \in \mathcal{C}^W\}$.

\item The \textbf{sequentiality}~\cite{Men08} $\sequentiality$ is the ratio of arcs between non-connector nodes, $\sequentiality(W) = 1 - \frac{1}{|F|} \cdot |\{(x,y) \in F \mid x,y \not\in \mathcal{C}^W\}|$.
The idea behind this measure is that sequences in a workflow net are easier to understand than parallelism or exclusive choices.

\item The \textbf{depth}~\cite{Men08} $\depth$ is the maximum nesting of connectors in the workflow net.
The depth can be calculated by taking the minimum of the in-depth and the out-depth.
Then, the in-depth of a node $v$ is the minimum amount of connectors encountered on a simple path from $p_i$ to $v$.
The out-depth of a node $v$ is tha minimum amount of connectors encountered on a simple path from $v$ to $p_o$.
More formally, let $\mathcal{S}^W := \mathcal{S}_{\texttt{and}}^W \cup \mathcal{S}_{\texttt{xor}}^W$ be the set of all split nodes in $W$ and $\mathcal{J}^W := \mathcal{J}_{\texttt{and}}^W \cup \mathcal{J}_{\texttt{xor}}^W$ the set of all join nodes in $W$.
For every simple path $\rho = (v_1, \dots, v_n)$ starting in $p_i$ and ending in $v$, we define:
\begin{align*}
\lambda_{W}(v_1) &= \lambda_{W}(p_i) := 0\\
\lambda_{p}(v_n) &:=
\begin{cases}
\lambda_{W}(v_{n-1}) + 1 &\text{if } v_{n-1} \in \mathcal{S}^W \land v_n \not \in \mathcal{J}^W\\ 
\lambda_{W}(v_{n-1}) &\text{if } v_{n-1} \in \mathcal{S}^W \land v_n \in \mathcal{J}^W\\ 
\lambda_{W}(v_{n-1}) &\text{if } v_{n-1} \not \in \mathcal{S}^W \land v_n \not \in \mathcal{J}^W\\ 
\lambda_{W}(v_{n-1}) - 1 &\text{if } v_{n-1} \not \in \mathcal{S}^W \land v_n \in \mathcal{J}^W\\ 
\end{cases}
\ \\
\lambda_{W}(v) &:= \max\left\{0,\max_{\rho \text{ a path from } p_i \text{ to } v} \lambda_{\rho}(v)\right\} \quad (\text{for any } v \neq p_i)
\end{align*}
We define the out-depth in the same way, but with the net $\overleftarrow{W}$, where all edge directions reversed and where $p_o$ takes the place of $p_i$.
With this, the depth of the workflow net $W$ is $\depth(W) = \max\{\min\{\lambda_{W}(v), \lambda_{\overleftarrow{W}}(v)\} \mid v \in P \cup T\}$.

\item The \textbf{diameter}~\cite{Men08} $\diameter$ is the length of the longest simple path in $W$. 
Thus, we define $\diameter(W) = \max\{|k| \mid v_1, \dots, v_k \text{ is a simple path from } p_i \text{ to } p_o\}$.

\item The \textbf{cyclicity}~\cite{Men08} $\cyclicity$ is the ratio of nodes in $W$ that lie on a cycle. 
Since the nodes $p_i$ and $p_o$ can never lie on a cycle by definition, we take this ratio by dividing by $|P| + |T| - 2$ and get the following formal definition for cyclicity: $\cyclicity(W) = \frac{1}{|P| + |T| - 2} \cdot |\{x \in P \cup T \mid x \text{ lies on a cycle in } W\}|$.

\item The \textbf{coefficient of network connectivity}~\cite{Men08} $\netconn$ relates the number of arcs to the number of nodes, i.e. $\netconn(W) = \frac{|F|}{|P| + |T|}$.

\item The \textbf{density}~\cite{Men08} $\density$ relates the number of arcs in $W$ to the total possible amount of arcs in $W$.
Since it is only possible to connect places to transitions and transitions to places, there are $2 \cdot |T| \cdot |P|$ possible arcs in a Petri net.
In a workflow net, however, the input place $p_i$ and the output place $p_o$ can only have at most $|T|$ incoming or outgoing edges each.
Thus, in total there can be $2 \cdot |T| \cdot (|P| - 1)$ edges in a workflow net.
With this, we define the density of a workflow net $W$ as $\density(W) = \frac{|F|}{2 \cdot |T| \cdot (|P| - 1)}$.

\item The \textbf{number of duplicate tasks}~\cite{LaWMHRA11} $\duplicate$ is the amount of repetitions in the transition labels.
There are two possible ways to define this measure: 
Either by counting all label repetitions, including duplicate $\tau$-labels, or by just counting label repetitions $\neq \tau$. 
The latter is useful in cases where silent $\tau$-transitions are only considered as routing mechanisms. 
In these cases, $\tau$-repetitions could be even beneficial for how easy $W$ is to understand.
Therefore, we define $\duplicate(W) = \sum_{a \in A} (\max\left(|\{t \in T \mid \ell(t) = a\}|, 1\right) - 1)$.

\item The \textbf{number of empty sequence flows}~\cite{GruL09} $\emptyseq$ is the number of places that have only \texttt{and}-splits in their preset and \texttt{and}-joins in their postset. 
Such places are often implicit and can be left out completely. 
Thus, we define this measure as $\emptyseq(W) = |\{p \in P \mid \pre{p} \subseteq \mathcal{S}_{\texttt{and}}^N \land \post{p} \subseteq \mathcal{J}_{\texttt{and}}^N\}|$.
\end{itemize}
The formal definitions of these complexity measures for workflow nets are reported in \cref{table:model-complexity-measures}.
\begin{table}[ht]
\caption{The complexity measures for workflow nets we investigate in this paper.}
\label{table:model-complexity-measures}
\centering
\renewcommand{\arraystretch}{1.6}
\begin{tabular}{clc} \toprule
\textbf{Measure} & \textbf{Definition} & \textbf{Reference} \\ \midrule
$\size(W)$ & $\pad |P| + |T| \pad$ & \cite[p.118]{Men08} \\
$\mismatch(W)$ & $\pad MM_{\texttt{xor}}^W + MM_{\texttt{and}}^W \pad$ & \cite[p.125]{Men08} \\
$\connhet(W)$ & $\pad - \left(\frac{|\mathcal{C}_{\text{and}}^W|}{|\mathcal{C}^W|} \cdot \log_2\left(\frac{|\mathcal{C}_{\text{and}}^W|}{|\mathcal{C}^W|}\right) + \frac{|\mathcal{C}_{\text{xor}}^W|}{|\mathcal{C}^W|}\cdot \log_2\left(\frac{|\mathcal{C}_{\text{xor}}^W|}{|\mathcal{C}^W|}\right)\right) \pad$ & \cite[p.126]{Men08} \\
$\crossconn(W)$ & $\pad 1 - \frac{\sum_{n_1, n_2 \in P \cup T} V_W(n_1, n_2)}{(|P| + |T|) \cdot (|P| + |T| - 1)} \pad$ & \cite{VanRMAC08} \\
$\tokensplit(W)$ & $\pad \sum_{t \in \mathcal{S}_{\texttt{and}}} (|\post{t}| - 1) \pad$ & \cite[p.128]{Men08} \\
$\controlflow(W)$ & $\pad |\mathcal{S}_{\text{and}}^W| + \sum_{p \in \mathcal{S}_{\text{xor}}^W} |\post{p}| \pad$ & \cite{Car05} \\
$\separability(W)$ & $\pad 1 - \frac{|\{v \in P \cup T \mid v \text{ is a cut-vertex in } W\}|}{|P| + |T| - 2} \pad$ & \cite[p.122]{Men08} \\
$\avgconn(W)$ & $\pad \frac{1}{|\mathcal{C}^W|} \cdot \sum_{x \in \mathcal{C}^W} (|\pre{x}| + |\post{x}|) \pad$ & \cite[p.120]{Men08} \\
$\maxconn(W)$ & $\pad \max\{(|\pre{x}| + |\post{x}|) \mid x \in \mathcal{C}^W\} \pad$ & \cite[p.121]{Men08} \\
$\sequentiality(W)$ & $\pad 1 - \frac{1}{|F|} \cdot |\{(x,y) \in F \mid x,y \not\in \mathcal{C}^W\}| \pad$ & \cite[p.123]{Men08} \\
$\depth(W)$ & $\pad \max\{\min\{\lambda_{W}(v), \lambda_{\overleftarrow{W}}(v)\} \mid v \in P \cup T\} \pad$ & \cite[p.124]{Men08} \\
$\diameter(W)$ & $\pad \max\{|k| \mid v_1, \dots, v_k \text{ is a simple path from } p_i \text{ to } p_o\} \pad$ & \cite[p.119]{Men08} \\
$\cyclicity(W)$ & $\pad \frac{1}{|P| + |T| - 2} \cdot |\{x \in P \cup T \mid x \text{ lies on a cycle in } W\}| \pad$ & \cite[p.127]{Men08} \\
$\netconn(W)$ & $\pad \frac{|F|}{|P| + |T|} \pad$ & \cite[p.120]{Men08} \\
$\density(W)$ & $\pad \frac{|F|}{2 \cdot |T| \cdot (|P| - 1)} \pad$ & \cite[p.120]{Men08} \\
$\duplicate(W)$ & $\pad \sum_{a \in A} (\max\left(|\{t \in T \mid \ell(t) = a\}|, 1\right) - 1) \pad$ & \cite{LaWMHRA11} \\
$\emptyseq(W)$ & $\pad |\{p \in P \mid \pre{p} \subseteq \mathcal{S}_{\texttt{and}}^N \land \post{p} \subseteq \mathcal{J}_{\texttt{and}}^N\}| \pad$ & \cite{GruL09} \\ \bottomrule
\end{tabular}
\end{table}
For later convenience, we define the set of all inspected model complexity measures of this paper as $\moc := \{\size,$ $\mismatch,$ $\connhet,$ $\crossconn,$ $\tokensplit,$ $\controlflow,$ $\separability,$ $\avgconn,$ $\maxconn,$ $\sequentiality,$ $\depth,$ $\diameter,$ $\cyclicity,$ $\netconn,$ $\density,$ $\duplicate,$ $\emptyseq\}$.
In the next subsection, we will present the complexity measures for event logs that we use for our analyses.

\subsection{Complexity of Event Logs}
\label{sec:log-complexity}
Let $\mathcal{L}$ be the set of all event logs.
Similar to model complexity measures, we define a log complexity measure as a function $\mathcal{C}^L : \mathcal{L} \rightarrow \mathbb{R}_0^+$.
Thus, a log complexity measure assings a non-negative, real-valued score to event logs.
In this paper, we investigate the log complexity measures collected by Augusto~et~al.~\cite{AugMVW22}.
In the following, let $L$ be an event log over a set of activities $A$.
\begin{itemize}
\item The \textbf{magnitude}~\cite{Guen09} $\magnitude$ is the total number of events in the event log.
In other words, the magnitude is the sum of trace-sizes in $L$, where duplicates are counted as well.
Thus, we set $\magnitude = \sum_{\sigma \in L} L(\sigma) \cdot |\sigma|$.
For the event log $L$ shown in \cref{fig:tracenet-example}, we have $\magnitude(L) = 3 \cdot 50 + 4 \cdot 30 + 4 \cdot 20 = 350$.

\item The \textbf{variety}~\cite{Guen09} $\variety$ is the number of distinct event names in an event log, so $\variety(L) = |\{a \in A \mid \exists \sigma \in L: \exists i \in \{1, \dots, |\sigma|\}: \sigma(i) = a\}|$.
For the event log $L$ shown in \cref{fig:tracenet-example}, we have $\variety(L) = |\{a,b,c,d\}| = 4$.

\item The \textbf{length}~\cite{Guen09} $\support$ is the number of traces in the event log, where duplicates are counted as well.
Thus, $\support(L) = \sum_{\sigma \in L} L(\sigma)$.
Note that the original paper~\cite{Guen09} and the paper by Augusto~et~al.~\cite{AugMVW22} call this measure the \textbf{support} of an event log.
To avoid confusion with the set of unique elements in a multiset, which we also call support, we renamed this measure to length.
For the event log $L$ shown in \cref{fig:tracenet-example}, we have $\support(L) = 50 + 30 + 20 = 100$.

\item The \textbf{minimum trace length}~\cite{AugMVW22} $\tlmin$ is the minimum length of a trace in the event log, $\tlmin(L) = \min\{|\sigma| \mid \sigma \in L\}$.
For the event log $L$ shown in \cref{fig:tracenet-example}, we have $\tlmin(L) = \min\{3,4,4\} = 3$.

\item The \textbf{average trace length}~\cite{Aal16} $\tlavg$ is the average length of the traces in the event log, $\tlavg(L) = \frac{\sum_{\sigma \in L} L(\sigma) \cdot |\sigma|}{\sum_{\sigma \in L} L(\sigma)}$.
For the event log $L$ shown in \cref{fig:tracenet-example}, we have $\tlavg(L) = \frac{50 \cdot 3 + 30 \cdot 4 + 20 \cdot 4}{50 + 30 + 20} = \frac{350}{100} = 3.5$.

\item The \textbf{maximum trace length}~\cite{AugMVW22} $\tlmax$ is the maximum length of a trace in the event logs, $\tlmax(L) = \max\{|\sigma| \mid \sigma \in L\}$.
For the event log $L$ shown in \cref{fig:tracenet-example}, we have $\tlmax(L) = \max\{3,4,4\} = 4$.

\item The \textbf{level of detail}~\cite{AugMVW22} $\levelofdetail$ is the amount of distinct simple paths in the DFG of $L$, so $\levelofdetail(L) = |\{p \mid p \text{ is a simple DFG-path from } \triangleright \text{ to } \square\}|$.
Note that Günther~\cite{Guen09} defines the level of detail as the average amount of distinct event names per trace.
We use the definition for the level of detail of Augusto~et~al.\cite{AugMVW22}, because their work is more recent and the work of Günther contains no complexity measures that counts the amount of distinct simple paths in the directly follows graph of $L$.
For the event log $L$ shown in \cref{fig:tracenet-example}, we get the directly follows graph $G$ shown in \cref{fig:dfg-example}. 
This directly follows graph contains $6$ distinct simple paths from $\triangleright$ to $\square$: $(\triangleright, a, c, \square)$, $(\triangleright, a, b, c, \square)$, $(\triangleright, a, b, d, \square)$, $(\triangleright, a, c, d, \square)$, $(\triangleright, a, b, c, d, \square)$, and $(\triangleright, a, c, b, d, \square)$.
Thus, $\levelofdetail(L) = 6$ for this example.

\item The \textbf{number of ties}~\cite{Aal16} $\numberofties$ is the amount of activity-pairs $(a,b)$, such that $a$ is followed by $b$ in some traces, but $b$ is never followed by $a$ in any trace.
With the notation of \cref{def:directly-follows-graph}, we define this complexity measures as $\numberofties(L) = |\{(a,b) \mid a >_L b \land b \not>_L a\}|$.
For the event log $L$ shown in \cref{fig:tracenet-example}, we have $\numberofties(L) = |\{(a,b), (a,c), (b,d), (c,d)\}| = 4$.

\item The \textbf{Lempel-Ziv complexity}~\cite{Pen03} $\lempelziv$ is based on the complexity measure $LZ$ for finite sequences, proposed by Lempel and Ziv~\cite{LemZ76}. 
This measure understands the event log as a single sequence by concatenating all traces and calculating the Lempel-Ziv complexity.
This is essentially the number of distinct prefixes found while scanning through the sequence from left to right.
With this, $\lempelziv(L) = LZ(\prod_{\sigma \in L} \sigma^{L(\sigma)})$.
For an example, consider the event log $L = [\langle a,b,c \rangle^2, \langle a,b,c,d \rangle, \langle a,c,b,d \rangle]$, where only the trace $\langle a,b,c \rangle$ occurs more than once in $L$.
We turn this event log into the finite sequence $abcabcabcdacbd$ and compute its Lempel-Ziv complexity.
We find the unique prefixes $a$, $b$, $c$, $d$, $ab$, $ac$, $bc$, $bd$, and $ca$, so $\lempelziv(L) = 9$.

\item The \textbf{number of distinct traces}~\cite{Aal16} $\numberuniquetraces$ is the amount of traces in the support of the event log, $\numberuniquetraces(L) = |supp(L)|$.
For the event log $L$ shown in \cref{fig:tracenet-example}, we have $\numberuniquetraces(L) = |\{\langle a,b,c \rangle, \langle a,b,c,d \rangle, \langle a,c,b,d \rangle\}| = 3$.

\item The \textbf{percentage of distinct traces}~\cite{AugMVW22} $\percentageuniquetraces$ is the amount of traces in the support of the event log, divided by the total amount of traces in the event log, duplicates included.
More formally, $\percentageuniquetraces(L) = \frac{|supp(L)|}{\sum_{\sigma \in L} L(\sigma)}$.
For the event log $L$ shown in \cref{fig:tracenet-example}, we have $\percentageuniquetraces(L) = \frac{3}{100} = 0.03$.

\item The \textbf{structure}~\cite{AugMVW22} $\structure$ is the average amount of distinct events per trace, $\structure(L) = \frac{\sum_{\sigma \in L} L(\sigma) \cdot |\{a \in A \mid \exists i \in \{1, \dots, |\sigma|\}: \sigma(i) = a\}|}{\sum_{\sigma \in L} L(\sigma)}$.
Note that Günther~\cite{Guen09} calls this measure level of detail instead of structure.
For Günther, the structure of an event log is the number of directly follows relations divided by the maximum number of possible directly follows relations.
Since we have a similar measure with $\numberofties$ and the work of Augusto~et~al. is more recent, we use their definition of the structure of an event log.
For the event log $L$ shown in \cref{fig:tracenet-example}, we have $\structure(L) = \frac{50 \cdot 3 + 30 \cdot 4 + 20 \cdot 4}{350} = 1$.

\item The \textbf{average affinity}~\cite{Guen09} $\affinity$ is the average amount of neighborhoods two traces of the event log have in common.
For a trace $\sigma \in L$, we define $F(\sigma) = \{(a,b) \mid \exists i \in \{1, \dots, |\sigma|-1\}: \sigma(i) = a \land \sigma(i+1) = b\}$ as the set of direct neighborhoods in $\sigma$.
Then, the affinity between two traces $\sigma_1, \sigma_2 \in L$ is defined as $A(\sigma_1, \sigma_2) = \frac{|F(\sigma_1) \cap F(\sigma_2)|}{|F(\sigma_1) \cup F(\sigma_2)|}$.
For the average affinity, we do not compare the affinity of a trace to itself, as this would yield $1$.
However, we do compare the affinity of a trace $\sigma$ with all other traces, even if they are copies of $\sigma$.
Thus, $\affinity(L) = \frac{\sum_{\sigma_1 \in L} \sum_{\sigma_2 \in (L - [\sigma])} A(\sigma_1, \sigma_2)}{\left(\sum_{\sigma \in L} L(\sigma)\right) \cdot \left(\left(\sum_{\sigma \in L} L(\sigma)\right) - 1\right)}$.
For the event log $L$ shown in \cref{fig:tracenet-example}, $A(\langle a,b,c \rangle, \langle a,b,c,d \rangle) = \frac{3}{4}$, $A(\langle a,b,c \rangle, \langle a,c,b,d \rangle) = \frac{0}{5}$, and $A(\langle a,b,c,d \rangle, \langle a,c,b,d \rangle) = \frac{0}{6}$.
Thus, for the average affinity score, we get $\affinity(L) = \frac{50 \cdot (49 \cdot 1 + 30 \cdot \frac{3}{4}) + 30 \cdot (50 \cdot \frac{3}{4} + 29 \cdot 1) + 20 \cdot (19 \cdot 1)}{100 \cdot 99} = \frac{5950}{9900} = 0.6\overline{01}$.

\item The \textbf{deviation from random}~\cite{Pen03} $\deviationfromrandom$ is an indicator for how far the event log deviates from a completely random log, where all possible neighborhoods occur equally often.
To define this measure, we start by defining the amount of total neighborhood-relations in $L$ as $n_{\rightarrow}(L) = \sum_{\sigma \in L} (|\sigma| - 1)$.
For activities $a_1, a_2$ and a trace $\sigma$, $n_{\rightarrow}^{(a_1, a_2)}(\sigma) = |\{i \mid \sigma(i) = a_1 \land \sigma(i+1) = a_2\}|$ denotes the number of times $a_1$ is directly followed by $a_2$ in $\sigma$.
This definition can be straightforwardly extended to the event log $L$ by setting
$n_{\rightarrow}^{(a_1, a_2)}(L) = \sum_{\sigma \in L} L(\sigma) \cdot n_{\rightarrow}^{(a_1, a_2)}(\sigma)$.
Now, the deviation from random of $L$ is $\deviationfromrandom(L) = 1 - \sqrt{\sum_{(a_1, a_2) \in A \times A} \left(\frac{n_{\rightarrow}^{(a_1, a_2)}(L) - \frac{n_{\rightarrow}(L)}{|A|^2}}{n_{\rightarrow}(L)}\right)^2}$.
For the event log $L$ shown in \cref{fig:tracenet-example}, we have $n_{\rightarrow}(L) = 250$, $n_{\rightarrow}^{(a,b)}(L) = n_{\rightarrow}^{(b,c)}(L) = 80$, $n_{\rightarrow}^{(a,c)}(L) = n_{\rightarrow}^{(b,d)}(L) = n_{\rightarrow}^{(c,b)}(L) = 20$, and $n_{\rightarrow}^{(c,d)}(L) = 30$. 
All other activity-pairs receive the value $0$.
In turn, we get the following complexity score for~$L$: $\deviationfromrandom(L) = 1 - \sqrt{2 \cdot \left(\frac{80 - \frac{250}{64}}{250}\right)^2 + 3 \cdot \left(\frac{20 - \frac{250}{64}}{250}\right)^2 + \left(\frac{30 - \frac{250}{64}}{250}\right)^2} \approx 0.5433$

\item The \textbf{average edit-distance}~\cite{Pen03} $\avgdist$ is the average amount of insert- and delete-operations needed to transform one trace into another.
More general, the edit distance $ED(v,w)$ between two words $v$ and $w$ is the amount of insert- and delete-operations needed to transform $v$ into $w$.
There are variants where a replace-operation is allowed as well.
Since every replace-operation can be simulated by a delete-operation, followed by an insert-operation, we do not consider this alternative and define the average edit distance of the event log $L$ as $\avgdist(L) = \frac{\sum_{\sigma_1 \in L} \sum_{\sigma_2 \in L - [\sigma_1]} ED(\sigma_1, \sigma_2)}{\left(\sum_{\sigma \in L} L(\sigma)\right) \cdot \left(\left(\sum_{\sigma \in L} L(\sigma)\right) - 1\right)}$.
For the event log $L$ shown in \cref{fig:tracenet-example}, we have $ED(\langle a,b,c \rangle, \langle a,b,c,d \rangle) = 1$, $ED(\langle a,b,c \rangle, \langle a,c,b,d \rangle) = 3$, and $ED(\langle a,b,c,d \rangle, \langle a,c,b,d \rangle) = 2$.
Thus, we get $\avgdist(L) = \frac{50 \cdot (30 \cdot 1 + 20 \cdot 3) + 30 \cdot (50 \cdot 1 + 20 \cdot 2) + 20 \cdot (50 \cdot 3 + 30 \cdot 2)}{100 \cdot 99} = \frac{11400}{9900} = 1.\overline{15}$.

\item The \textbf{variant-entropy}~\cite{AugMVW22} $\varentropy$ is based on the prefix automaton originally constructed for precision-estimation by Muñoz-Gama~et~al.~\cite{MunC10}. 
The prefix automaton ist a graph that contains all prefixes of traces in $L$.
Each node representing a prefix in the event log receives a weight corresponding to how often there is a trace with the prefix in the event log.
Two prefixes are connected by an edge with label $a$ in the automaton if adding $a$ to the end of the prefix in the source-node leads to the prefix in the target node.
\cref{fig:prefix-automaton-example} shows an example for the event log $L$ shown in \cref{fig:tracenet-example}.
\begin{figure}[ht]
	\centering
	\begin{tikzpicture}[node distance = 1.25cm,>=stealth',bend angle=0,auto]
		\fill[purple!20] (1,-1.5) rectangle +(4.75,1.75) node[right,purple,yshift=-0.25cm] {$P_1$};
		\fill[orange!20] (1.75,0.35) rectangle +(4,1) node[right,orange,yshift=-0.25cm] {$P_2$};
		\node[draw,rectangle,fill=white,rounded corners] (epsilon) {$\varepsilon$};
		\node[draw,rectangle,fill=white,rounded corners,right of=epsilon] (a) {$a$}
		edge [pre] node[above]{$a$} (epsilon);
		\node[draw,rectangle,fill=white,rounded corners,above right of=a] (ab) {$ab$}
		edge [pre] node[above left]{$b$} (a);
		\node[draw,rectangle,fill=white,rounded corners,below right of=a] (ac) {$ac$}
		edge [pre] node[below left]{$c$} (a);
		\node[draw,rectangle,fill=white,rounded corners,right of=ab,xshift=0.2cm] (abc) {$abc$}
		edge [pre] node[above]{$c$} (ab);
		\node[draw,rectangle,fill=white,rounded corners,right of=ac,xshift=0.2cm] (acb) {$acb$}
		edge [pre] node[below]{$b$} (ac);
		\node[draw,rectangle,fill=white,rounded corners,right of=abc,xshift=0.3cm] (abcd) {$abcd$}
		edge [pre] node[above]{$d$} (abc);
		\node[draw,rectangle,fill=white,rounded corners,right of=acb,xshift=0.3cm] (acbd) {$acbd$}
		edge [pre] node[below]{$d$} (acb);
		\node[cyan!50!blue,yshift=-0.4cm] at (epsilon) {\small $100$};
		\node[cyan!50!blue,yshift=-0.4cm] at (a) {\small $100$};
		\node[cyan!50!blue,yshift=-0.4cm] at (ab) {\small $80$};
		\node[cyan!50!blue,yshift=-0.4cm] at (ac) {\small $20$};
		\node[cyan!50!blue,yshift=-0.4cm] at (abc) {\small $80$};
		\node[cyan!50!blue,yshift=-0.4cm] at (acb) {\small $20$};
		\node[cyan!50!blue,yshift=-0.4cm] at (abcd) {\small $30$};
		\node[cyan!50!blue,yshift=-0.4cm] at (acbd) {\small $20$};
	\end{tikzpicture}
	\caption{The prefix automaton for the event log $L$ of \cref{fig:tracenet-example} with partitions $P_1, P_2$.}
	\label{fig:prefix-automaton-example}
\end{figure}
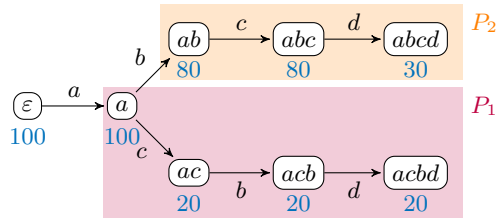
To calculate the variant entropy, we first take the set of nodes in the prefix automaton that are not labeled $\varepsilon$ and call it $S$. 
In the example above, we have $S = \{a, ab, ac, abc, acb, abcd, acbd\}$.
Then, for a partition $P_1, \dots, P_n$ of the graph defined by the extended prefix automaton, we calculate the variant entropy as $\varentropy(L) = |S| \cdot \ln(|S|) - \sum_{i = 1}^n (|P_i| \cdot \ln(|P_i|))$. 
In the example above, we would get $\varentropy(L) = 7 \cdot \ln(7) - 4 \cdot \ln(4) - 3 \cdot \ln(3) \approx 4.7804$.

\item The \textbf{normalized variant-entropy}~\cite{AugMVW22} $\normvarentropy$ follows the same ideas as its non-normalized counterpart, but makes sure that the returned scores lie between $0$ and $1$, so two entropy values are easier to compare to each other.
Formally, with the notions as defined for the variant entropy, we have $\normvarentropy(L) = \frac{|S| \cdot \ln(|S|) - \sum_{i = 1}^n (|P_i| \cdot \ln(|P_i|))}{|S| \cdot \ln(|S|)}$.
For the event log $L$ shown in \cref{fig:tracenet-example}, we would get $\normvarentropy(L) = \frac{7 \cdot \ln(7) - 4 \cdot \ln(4) - 3 \cdot \ln(3)}{7 \cdot \ln(7)} \approx 0.3509$.

\item The \textbf{sequence-entropy}~\cite{AugMVW22} $\seqentropy$ works similar as the variant entropy, but also uses the information about frequencies of traces in the event log $L$.
To do so, this measure assigns a weight $w(s)$ to each state $s$ in the prefix automaton, which corresponds to the amount of traces having the word represented by the state as a prefix.
In \cref{fig:prefix-automaton-example}, these weights are indicated as blue numbers below their states.
For the set of states $S$ in the prefix automaton that are not labeled $\varepsilon$, we set $W = \sum_{s \in S} w(s)$.
For a partition $P_i$ of the prefix automaton, we set $W_i = \sum_{p \in P_i} w(p)$.
Then, for $n$ partitions $P_1, \dots, P_n$ of the prefix automaton, the sequence entropy of the event log is $\seqentropy(L) = W \cdot \ln(W) - \sum_{i=1}^n W_i \cdot \ln(W_i)$.
For the event log $L$ of \cref{fig:tracenet-example}, we use the same prefix automaton as shown in \cref{fig:prefix-automaton-example} and get the complexity score $\seqentropy(L) = 350 \cdot \ln(350) - 160 \cdot \ln(160) - 190 \cdot \ln(190) \approx 241.3142$.

\item The \textbf{normalized sequence-entropy}~\cite{AugMVW22} $\normseqentropy$ is the normalized variant of the sequence-entropy, $\normseqentropy = \frac{W \cdot \ln(W) - \sum_{i=1}^n W_i \cdot \ln(W_i)}{W \cdot \ln(W)}$.
For the event log $L$ shown in \cref{fig:tracenet-example}, $\normseqentropy(L) = \frac{350 \cdot \ln(350) - 160 \cdot \ln(160) - 190 \cdot \ln(190)}{350 \cdot \ln(350)} \approx 0.1177$.
\end{itemize}
Like before, \cref{table:log-complexity-measures} reports the formal definitions of the log complexity measures we will analyze in this paper.
We define the set of all inspected log complexity measures as $\loc := \{\magnitude,$ $\variety,$ $\support,$ $\tlavg,$ $\tlmax,$ $\levelofdetail,$ $\numberofties,$ $\lempelziv,$ $\numberuniquetraces,$ $\percentageuniquetraces,$ $\structure,$ $\affinity,$ $\deviationfromrandom,$ $\avgdist,$ $\varentropy,$ $\normvarentropy,$ $\seqentropy,$ $\normseqentropy\}$.
Note that the log complexity measure $\tlmin$ is not part of this set. 
An explanation for this will follow in the next section.
We are now ready to dive into the analyses of the relationships between log- and model complexity.
While computing the log complexity scores, we use the \texttt{Python}-implementation of Vidgof~\cite{Vid24} to avoid calculation errors.
Since this implementation does not provide functions for calculating $\numberofties$, $\levelofdetail$, and $\avgdist$, as defined in~\cite{AugMVW22}, we added functions for these log complexity measures to the implementation.

\begin{table}[hp]
\caption{The complexity measures for event logs we investigate in this paper.}
\label{table:log-complexity-measures}
\centering
\renewcommand{\arraystretch}{1.6}
\begin{tabular}{clc} \toprule
\textbf{Measure} & \textbf{Definition} & \textbf{Reference} \\ \midrule
$\magnitude(L)$ & $\pad \sum_{\sigma \in L} L(\sigma) \cdot |\sigma| \pad$ & \cite[p.52]{Guen09}  \\
$\variety(L)$ & $\pad \left|\{a \in A \mid \exists \sigma \in L: \exists i \in \{1, \dots, |\sigma|\}: \sigma(i) = a\}\right| \pad$ & \cite[p.53]{Guen09} \\
$\support(L)$ & $\pad \sum_{\sigma \in L} L(\sigma) \pad$ & \cite[53]{Guen09} \\
$\tlmin(L)$ & $\pad \min\{|\sigma| \mid \sigma \in L\} \pad$ & \cite{AugMVW22} \\
$\tlavg(L)$ & $\pad \frac{\sum_{\sigma \in L} L(\sigma) \cdot |\sigma|}{\sum_{\sigma \in L} L(\sigma)} \pad$ & \cite[p.365]{Aal16} \\
$\tlmax(L)$ & $\pad \max\{|\sigma| \mid \sigma \in L\} \pad$ & \cite{AugMVW22} \\
$\levelofdetail(L)$ & $\pad |\{p \mid p \text{ is a simple DFG-path from } \triangleright \text{ to } \square\}| \pad$ & \cite{AugMVW22} \\
$\numberofties(L)$ & $\pad |\{(a,b) \mid a >_L b \land b \not>_L a\}| \pad$ & \cite[p.366]{Aal16} \\
$\lempelziv(L)$ & $\pad LZ(\prod_{\sigma \in L} \sigma^{L(\sigma)}) \pad$ & \cite{Pen03} \\
$\numberuniquetraces(L)$ & $\pad |supp(L)| \pad$ & \cite[p.366]{Aal16} \\
$\percentageuniquetraces(L)$ & $\pad \frac{|supp(L)|}{\sum_{\sigma \in L} L(\sigma)} \pad$ & \cite{AugMVW22} \\
$\structure(L)$ & $\pad \frac{\sum_{\sigma \in L} L(\sigma) \cdot |\{a \in A \mid \exists i \in \{1, \dots, |\sigma|\}: \sigma(i) = a\}|}{\sum_{\sigma \in L} L(\sigma)} \pad$ & \cite{AugMVW22} \\
$\affinity(L)$ & $\pad \frac{\sum_{\sigma_1 \in L} \sum_{\sigma_2 \in (L - [\sigma])} A(\sigma_1, \sigma_2)}{\left(\sum_{\sigma \in L} L(\sigma)\right) \cdot \left(\left(\sum_{\sigma \in L} L(\sigma)\right) - 1\right)} \pad$ & \cite[p.55]{Guen09} \\
$\deviationfromrandom(L)$ & $\pad 1 - \sqrt{\sum_{(a_1, a_2) \in A \times A} \left(\frac{n_{\rightarrow}^{(a_1, a_2)}(L) - \frac{n_{\rightarrow}(L)}{|A|^2}}{n_{\rightarrow}(L)}\right)^2} \pad$ & \cite{Pen03} \\
$\avgdist(L)$ & $\pad \frac{\sum_{\sigma_1 \in L} \sum_{\sigma_2 \in L - [\sigma_1]} ED(\sigma_1, \sigma_2)}{\left(\sum_{\sigma \in L} L(\sigma)\right) \cdot \left(\left(\sum_{\sigma \in L} L(\sigma)\right) - 1\right)} \pad$ & \cite{Pen03} \\
$\varentropy(L)$ & $\pad |S| \cdot \ln(|S|) - \sum_{i = 1}^n (|P_i| \cdot \ln(|P_i|)) \pad$ & \cite{AugMVW22} \\
$\normvarentropy(L)$ & $\pad \frac{|S| \cdot \ln(|S|) - \sum_{i = 1}^n (|P_i| \cdot \ln(|P_i|))}{|S| \cdot \ln(|S|)} \pad$ & \cite{AugMVW22} \\
$\seqentropy(L)$ & $\pad W \cdot \ln(W) - \sum_{i=1}^n W_i \cdot \ln(W_i) \pad$ & \cite{AugMVW22} \\
$\normseqentropy(L)$ & $\pad \frac{W \cdot \ln(W) - \sum_{i=1}^n W_i \cdot \ln(W_i)}{W \cdot \ln(W)} \pad$ & \cite{AugMVW22} \\ \bottomrule
\end{tabular}
\end{table}

\newpage
\section{Relationship of Log- and Model Complexity}
\label{sec:relationships}
As event logs grow over time, they typically become more complex, as they contain more behavior of the system.
Thus, we are interested in the question: 
For two event logs $L_1, L_2$ with $L_1 \sqsubset L_2$ and $\mathcal{C}^L(L_1) < \mathcal{C}^L(L_2)$, what can we say about the relation between $\mathcal{C}^M(M_1)$ and $\mathcal{C}^M(M_2)$, where $M_1$ is a model discovered for $L_1$ and $M_2$ is a model discovered for $L_2$?
A first intuition is that the model complexity should increase as well, i.e. $\mathcal{C}^M(M_1) < \mathcal{C}^M(M_2)$.
However, when the used discovery algorithm can filter out noise or infrequent behavior, this is not necessarily the case.
With noise-filtering, it is possible that we would like the model complexity to stay unchanged or even lower in certain cases.
We therefore need to be cautious which mining algorithms we investigate in our analyses.
In this paper, we solve this issue by understanding noise-filtering as a preprocessing step and expect the event logs to contain no noise at all.
Furthermore, we won't investigate the effects of changing the minimal trace length in the event log to model complexity, as $L_1 \sqsubset L_2$ directly implies $\tlmin(L_1) \geq \tlmin(L_2)$.

With these requirements, we would expect that $\mathcal{C}^L(L_1) < \mathcal{C}^L(L_2)$ implies $\mathcal{C}^M(M_1) < \mathcal{C}^M(M_2)$.
This section is therefore dedicated to find the relation $R \in \{\mless, \mleq, \meq, \mgeq, \mgreater, \norel\}$, such that $(\mathcal{C}^L, \mathcal{C}^M) \in R$, where
\begin{align*}
\mless = &\{(\mathcal{C}^L, \mathcal{C}^M) \mid \forall L_1, L_2: \mathcal{C}^L(L_1) < \mathcal{C}^L(L_2) \Rightarrow \mathcal{C}^M(M_1) < \mathcal{C}^M(M_2)\} \\
\mleq = &\{(\mathcal{C}^L, \mathcal{C}^M) \mid \forall L_1, L_2: \mathcal{C}^L(L_1) < \mathcal{C}^L(L_2) \Rightarrow \mathcal{C}^M(M_1) \leq \mathcal{C}^M(M_2)\} \setminus (\mless \cup \meq) \\
\meq = &\{(\mathcal{C}^L, \mathcal{C}^M) \mid \forall L_1, L_2: \mathcal{C}^L(L_1) < \mathcal{C}^L(L_2) \Rightarrow \mathcal{C}^M(M_1) = \mathcal{C}^M(M_2)\} \\
\mgeq = &\{(\mathcal{C}^L, \mathcal{C}^M) \mid \forall L_1, L_2: \mathcal{C}^L(L_1) < \mathcal{C}^L(L_2) \Rightarrow \mathcal{C}^M(M_1) \geq \mathcal{C}^M(M_2)\} \setminus (\mgreater \cup \meq) \\
\mgreater = &\{(\mathcal{C}^L, \mathcal{C}^M) \mid \forall L_1, L_2: \mathcal{C}^L(L_1) < \mathcal{C}^L(L_2) \Rightarrow \mathcal{C}^M(M_1) > \mathcal{C}^M(M_2)\} \\
\norel = &(\loc \times \moc) \setminus (\mless \cup \mleq \cup \meq \cup \mgeq \cup \mgreater) 
\end{align*}
In the remainder of this section, we will investigate which of these relations hold for five different discovery algorithms.
To do so, in each subsection, we first fix the investigated mining algorithm and find general properties for them.
For quick reference, we then report our findings in a table, before providing proofs for each entry in the table.
Note that, in the PDF-version of this paper, the entries in the tables can be clicked to show their respective proof.

\subsection{Flower Model}
\label{sec:flower}
As a first baseline mining algorithm, we investigate the algorithm that always returns the flower model for an input event log.
Thus, let $L$ be an event log over a set of activities $A = \{a_1, a_2, \dots, a_n\}$.
Then, the flower model is the net shown in \cref{fig:flower-model}, which allows for all behavior using only activities $a_1, a_2, \dots, a_n$.
\begin{figure}[ht]
	\centering
	\scalebox{\scalefactor}{
	\begin{tikzpicture}[node distance = 1.25cm,>=stealth',bend angle=0,auto]
		\node[place,tokens=1,label=left:$p_i$] (start) {};
		\node[transition,right of=start,label=below:$t_1$] (tau1) {$\tau$}
		edge [pre] (start);
		\node[place, above right of=tau1,label=below:$p$] (middle) {}
		edge [pre] (tau1);
		\node[transition,below right of=middle,label=below:$t_2$] (tau2) {$\tau$}
		edge [pre] (middle);
		\node[place,right of=tau2,label=right:$p_o$] (end) {}
		edge [pre] (tau2);
		\node[transition,left of=middle] (a1) {$a_1$}
		edge [pre,bend right=15] (middle)
		edge [post,bend left=15] (middle);
		\node[transition,above left of=middle] (a2) {$a_2$}
		edge [pre,bend right=15] (middle)
		edge [post,bend left=15] (middle);
		\node[transition,right of=middle] (an) {$a_n$}
		edge [pre,bend right=15] (middle)
		edge [post,bend left=15] (middle);
		\draw[dotted,line width=2pt,line cap=round,dash pattern=on 0pt off 2\pgflinewidth, bend left=10] ($(a2) + (1.05,0)$) to ($(an) + (-0.7,0.7)$);
	\end{tikzpicture}}
	\caption{The flower model for an event log $L$, using activities $A = \{a_1, a_2, \dots, a_n\}$.}
	\label{fig:flower-model}
\end{figure}
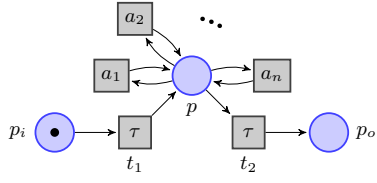
It is easy to see that the flower model is mostly affected by the amount of different activity names used in the underlying event log, $\variety$.
We find that all other log complexity measures are unaffected by the amount of different activity names in the event log.

\begin{lemma}
\label{lemma:not-variety-dependent}
Let $\mathcal{C}^L \in \loc \setminus \{\variety\}$ be a log complexity measure.
Then, there are event logs $L_1, L_2$ with $L_1 \sqsubset L_2$ and $\mathcal{C}^L(L_1) < \mathcal{C}^L(L_2)$, but $\variety(L_1) = \variety(L_2)$.
\end{lemma}
\begin{proof}
Consider the following event logs:
\begin{align*}
	L_1 &= [\langle a,b,c,d \rangle^{2}, \langle a,b,c,d,e \rangle^{2}, \langle d,e,a,b \rangle^{2}] \\
	L_2 &= L_1 + [\langle a,b,c,d,e \rangle^{2}, \langle d,e,a,b,c \rangle, \langle c,d,e,a,b \rangle, \langle e,c,d,a,b,c \rangle]
\end{align*}
These two event logs have the following log complexity scores:
\begin{center}
	\begin{tabular}{|c|c|c|c|c|c|c|c|c|c|c|}\hline
		 & $\magnitude$ & $\variety$ & $\support$ & $\tlavg$ & $\tlmax$ & $\levelofdetail$ & $\numberofties$ & $\lempelziv$ & $\numberuniquetraces$ & $\percentageuniquetraces$ \\ \hline
		$L_1$ & $\pad 26 \pad$ & $\pad 5 \pad$ & $\pad 6 \pad$ & $\pad 4.3333 \pad$ & $\pad 5 \pad$ & $\pad 6 \pad$ & $\pad 5 \pad$ & $\pad 13 \pad$ & $\pad 3 \pad$ & $\pad 0.5 \pad$ \\ \hline
		$L_2$ & $\pad 52 \pad$ & $\pad 5 \pad$ & $\pad 11 \pad$ & $\pad 4.7273 \pad$ & $\pad 6 \pad$ & $\pad 23 \pad$ & $\pad 7 \pad$ & $\pad 21 \pad$ & $\pad 6 \pad$ & $\pad 0.5455 \pad$ \\ \hline
	\end{tabular}
	
	\medskip
	
	\begin{tabular}{|c|c|c|c|c|c|c|c|c|} \hline
		 & $\structure$ & $\affinity$ & $\deviationfromrandom$ & $\avgdist$ & $\varentropy$ & $\normvarentropy$ & $\seqentropy$ & $\normseqentropy$ \\ \hline
		$L_1$ & $\pad 4.3333 \pad$ & $\pad 0.56 \pad$ & $\pad 0.5757 \pad$ & $\pad 2.6667 \pad$ & $\pad 6.1827 \pad$ & $\pad 0.3126 \pad$ & $\pad 16.0483 \pad$ & $\pad 0.1894 \pad$ \\ \hline
		$L_2$ & $\pad 4.6364 \pad$ & $\pad 0.5829 \pad$ & $\pad 0.6039 \pad$ & $\pad 2.9091 \pad$ & $\pad 29.0428 \pad$ & $\pad 0.4543 \pad$ & $\pad 60.0209 \pad$ & $\pad 0.2921 \pad$ \\ \hline
	\end{tabular}
\end{center}
Thus, all log complexity measures increased, except $\variety$, which is the same for $L_1$ and $L_2$.
Therefore, these event logs show the conjecture for every log complexity measure $\mathcal{C}^L \in \loc \setminus \{\variety\}$. \hfill$\square$
\end{proof}

For two event logs $L_1, L_2$ and their flower models $M_1, M_2$, we can conclude with \cref{lemma:not-variety-dependent} and the discussion above that $M_1$ and $M_2$ differ in their structure if and only if $\variety(L_1) \neq \variety(L_2)$.
Furthermore, we can see that an increase in variety means that the flower model receives a new transition, thus increasing most model complexity scores for the flower model.
If $L_1 \sqsubset L_2$, it is not possible that the model complexity scores of the flower model decrease.

\begin{lemma}
\label{lemma:flower-model-monotone-increasing}
Let $L_1, L_2$ be event logs with $L_1 \sqsubset L_2$. 
Let $M_1, M_2$ be the flower models for $L_1$ and $L_2$. 
Then, $\mathcal{C}^M(M_1) \leq \mathcal{C}^M(M_2)$ for any model complexity measure $\mathcal{C}^M \in \{\size,$ $\crossconn,$ $\controlflow,$ $\separability,$ $\avgconn,$ $\maxconn,$ $\sequentiality,$ $\cyclicity,$ $\netconn\}$.
\end{lemma}
\begin{proof}
Let $L_1, L_2$ be two event logs with $L_1 \sqsubset L_2$.
Then, $\variety(L_1) \leq \variety(L_2)$, since every trace in $L_1$ must be part of $L_2$, and thus, $L_1$ cannot contain any activity names that are not present in $L_2$.
With this observation, we prove $\mathcal{C}^M(M_1) \leq \mathcal{C}^M(M_2)$ for each of the model complexity measures separately.
\begin{itemize}
	\item \textbf{Size $\size$:}
	The flower model of an event log $L$ has exactly $3$ places and exactly $2 + \variety(L)$ transitions.
	Thus, we get
	\begin{align*}
	\size(M_1) = 5 + \variety(L_1) \overset{\variety(L_1) \leq \variety(L_2)}{\leq} 5 + \variety(L_2) = \size(M_2).
	\end{align*}
	
	\item \textbf{Cross Connectivity $\crossconn$:} 
	Let $M$ be the flower model for an event log $L$ and let $n := \variety(L)$.
	The only connector in the flower model is the place labeled $p$ in \cref{fig:flower-model}.
	Thus, to calculate the cross connectivity of the flower model, this place receives weight $\frac{1}{2n + 2}$, while all other nodes receive weight $1$.
	With this, we can calculate that 
	\[\crossconn(M) = \frac{4 n^4 + 44 n^3 + 143 n^2 + 164 n + 59}{4 (n + 1)^2 (n + 4) (n + 5)}\]
	which is monotonic increasing for increasing $n$, as
	\begin{align*}
	&\frac{\mathrm{d}}{\mathrm{d}n} \left(\frac{4 n^4 + 44 n^3 + 143 n^2 + 164 n + 59}{4 (n + 1)^2 (n + 4) (n + 5)}\right) \\
	= &\frac{389 + 729 n + 575 n^2 + 213 n^3 + 26 n^4}{4 (1 + n)^3 (4 + n)^2 (5 + n)^2} > 0.
	\end{align*}
	Thus, the cross connectivity of te flower model increases when the variety of the underlying event log does.
	Since $\variety(L_1) \leq \variety(L_2)$, we can therefore deduce that $\crossconn(M_1) \leq \crossconn(M_2)$.
	
	\item \textbf{Control Flow Complexity $\controlflow$:}
	The only connector in the flower model of an event log $L$ is the place labeled $p$ in \cref{fig:flower-model}.
	This place has $\variety(L) + 1$ outgoing edges, so we get
	\begin{align*}
	\controlflow(M_1) = \variety(L_1) + 1 \overset{\variety(L_1) \leq \variety(L_2)}{\leq} \variety(L_2) + 1 = \controlflow(M_2).
	\end{align*}
	
	\item \textbf{Separability $\separability$:}
	The flower model $M$ for an event log $L$ has exactly three cut-vertices, labeled $p$, $t_1$, and $t_2$ in \cref{fig:flower-model}.
	Since the flower model features $5 + \variety(L)$ nodes in total, we have
	\begin{align*}
	\separability(M_1) = \frac{\variety(L_1)}{\variety(L_1) + 3} \overset{\variety(L_1) \leq \variety(L_2)}{\leq} \frac{\variety(L_2)}{\variety(L_2) + 3} = \separability(M_2)
	\end{align*}
	since, in general, $\frac{x}{y} \leq \frac{x + a}{y + a}$ for any $x, y, a \in \mathbb{R}^+$ with $x \leq y$. 

	\item \textbf{Average Connector Degree $\avgconn$:}
	The only connector in the flower model $M$ for an event log $L$ is the place labeled $p$ in \cref{fig:flower-model}.
	This connector has $\variety(L) + 1$ incoming and $\variety(L) + 1$ outgoing edges, so
	\begin{align*}
	\avgconn(M_1) = 2 \variety(L_1) + 2 \overset{\variety(L_1) \leq \variety(L_2)}{\leq} 2 \variety(L_2) + 2 = \avgconn(M_2).
	\end{align*}
	
	\item \textbf{Maximum Connector Degree $\maxconn$:}
	The only connector in the flower model $M$ for an event log $L$ is the place labeled $p$ in \cref{fig:flower-model}.
	This connector has $\variety(L) + 1$ incoming and $\variety(L) + 1$ outgoing edges, so
	\begin{align*}
	\maxconn(M_1) = 2 \variety(L_1) + 2 \overset{\variety(L_1) \leq \variety(L_2)}{\leq} 2 \variety(L_2) + 2 = \maxconn(M_2).
	\end{align*}
	
	\item \textbf{Sequentiality $\sequentiality$:}
	There are exactly $2$ edges in the flower model between non-connector nodes: $(p_i, t_1)$ and $(t_2, p_o)$.
	In total, the flower model contains $2 \cdot \variety(L) + 4$ edges, so
	\begin{align*}
	\sequentiality(M_1) = \frac{2 \cdot \variety(L_1) + 2}{2 \cdot \variety(L_1) + 4} \overset{\variety(L_1) \leq \variety(L_2)}{\leq} \frac{2 \cdot \variety(L_2) + 2}{2 \cdot \variety(L_2) + 4} = \sequentiality(M_2).
	\end{align*}
	since, in general, $\frac{x}{y} \leq \frac{x + a}{y + a}$ for any $x, y, a \in \mathbb{R}^+$ with $x \leq y$.

	\item \textbf{Cyclicity $\cyclicity$:}
	In the flower model $M$ for an event log $L$, exactly $\variety(L) + 1$ nodes lie on a cycle.
	Since there are $5 + \variety(L)$ nodes in total, we have
	\begin{align*}
	\cyclicity(M_1) = \frac{\variety(L_1) + 1}{\variety(L_1) + 3} \overset{\variety(L_1) \leq \variety(L_2)}{\leq} \frac{\variety(L_2) + 1}{\variety(L_2) + 3} = \cyclicity(M_2)
	\end{align*}
	since, in general, $\frac{x}{y} \leq \frac{x + a}{y + a}$ for any $x, y, a \in \mathbb{R}^+$ with $x \leq y$. 
	
	\item \textbf{Coefficient of Network Connectivity $\netconn$:}
	The flower model for an event log $L$ has $2 \variety(L) + 4$ edges and $5 + \variety(L)$ nodes.
	Therefore
	\begin{align*}
	\netconn(M_1) = \frac{2 \variety(L_1) + 4}{\variety(L_1) + 5} \overset{\variety(L_1) \leq \variety(L_2)}{\leq} \frac{2 \variety(L_2) + 4}{\variety(L_2) + 5}
	\end{align*}
	since, in general, $\frac{x}{y} \leq \frac{x + 2a}{y + a}$ for any $x, y, a \in \mathbb{R}^+$ with $x \leq 2y$, which is true for $x = 2 \variety(L_1) + 4)$ and $y = 5 + \variety(L_1)$.
\end{itemize}
Thus, we showed that $\mathcal{C}^M(M_1) \leq \mathcal{C}^M(M_2)$ for any model complexity measure $\mathcal{C}^M \in \{\size,$ $\crossconn,$ $\controlflow,$ $\separability,$ $\avgconn,$ $\maxconn,$ $\sequentiality,$ $\cyclicity,$ $\netconn\}$. \hfill$\square$
\end{proof}

Next to these monotonic increasing model complexity measures, there are also measures that always return the same complexity score for a flower model.

\begin{lemma}
\label{lemma:flowermodel-constant-complexity}
Let $L_1, L_2$ be event logs and $M_1, M_2$ be the flower models for $L_1$ and $L_2$.
Then, we have $\mathcal{C}^M(M_1) = \mathcal{C}^M(M_2)$ for any model complexity measure $\mathcal{C}^M \in \{\mismatch, \connhet, \tokensplit, \depth, \diameter, \density, \duplicate, \emptyseq\}$.
\end{lemma}
\begin{proof}
Let $L_1, L_2, M_1, M_2$ and $\mathcal{C}^M$ be defined as stated by the lemma. 
We prove $\mathcal{C}^M(M_1) = \mathcal{C}^M(M_2)$ for each of the model complexity measures separately.
\begin{itemize}
	\item \textbf{Connector Mismatch $\mismatch$:}
	The flower model $M$ for an event log $L$ has no connector mismatches:
	The place labeled $p$ in \cref{fig:flower-model} is the only connector in the flower model, and has $\variety(L) + 1$ incoming and outgoing edges.
	Therefore, we have $\mismatch(M) = |(\variety(L) + 1) - (\variety(L) + 1)| = 0$ and thus $\mismatch(M_1) = 0 = \mismatch(M_2)$. 
	
	\item \textbf{Connector Heterogeneity $\connhet$:}
	The flower model $M$ for an event log $L$ has only one connector, which is the place labeled $p$ in \cref{fig:flower-model}. 
	Thus, every flower model has only one type of connector, leading to the complexity score $\connhet(M) = 1 \cdot \log_2(1) + 0 \cdot \log_2(0) = 0$.
	Therefore, $\connhet(M_1) = 0 = \connhet(M_2)$.
	
	\item \textbf{Token Split $\tokensplit$:}
	All transitions in the flower model $M$ for an event log have exactly one outgoing edge, so $\tokensplit(M) = 0$, and thus $\tokensplit(M_1) = 0 = \tokensplit(M_2)$.
	
	\item \textbf{Depth $\depth$:}
	In the flower model $M$ for an event log $L$, all nodes have depth $1$ since a path from $p_i$ or to $p_o$ must always contain the connector $p$, which is a split node and a join node.
	Thus, $\depth(M) = 1$ for any flower model $M$, so $\depth(M_1) = 1 = \depth(M_2)$.
	
	\item \textbf{Diameter $\diameter$:}
	The only simple path in the flower model $M$ for an event log is $(p_i, t_1, p, t_2, p_o)$.
	Therefore, the longest simple path in $M$ is always $\depth(M) = 5$.
	In turn, we have $\depth(M_1) = 5 = \depth(M_2)$.
	
	\item \textbf{Density $\density$:}
	The flower model $M$ for an event log $L$ always has exactly $2 \variety(L) + 4$ edges, $3$ places, and $\variety(L) + 2$ transitions.
	Thus, its density score is $\density(M) = \frac{2(\variety(L) + 2)}{2 \cdot (\variety(L) + 2) \cdot (3 - 1)} = \frac{1}{2}$, so $\density(M_1) = \frac{1}{2} = \density(M_2)$.
	
	\item \textbf{Number of Duplicate Tasks $\duplicate$:}
	The flower model $M$ for an event log contains one transition for each activity in the event log, as well as two $\tau$-transitions.
	Therefore, the only duplicate label in $M$ is the second $\tau$-label, leading to $\duplicate(M) = 1$.
	In turn, $\duplicate(M_1) = 1 = \duplicate(M_2)$.
	
	\item \textbf{Number of Empty Sequence Flows $\emptyseq$:}
	Since the flower model $M$ for an event log contains no parallel splits or joins, there cannot be any empty sequence flows in the flower model.
	Therefore, $\emptyseq(M) = 0$ for any flower model $M$, and thus $\emptyseq(M_1) = 0 = \emptyseq(M_2)$.
\end{itemize}
Thus, we showed that $\mathcal{C}^M(M_1) = \mathcal{C}^M(M_2)$ for any model complexity measure $\mathcal{C}^M \in \{\mismatch,$ $\connhet,$ $\tokensplit,$ $\depth,$ $\diameter,$ $\density,$ $\duplicate,$ $\emptyseq\}$. \hfill$\square$
\end{proof}

With these observations, we can now analyze the relations between log and model complexity for the flower model miner.
We start by showing the results in \cref{table:flowermodel-findings} and prove the relations shown in the table afterwards. 
\begin{table}[ht]
	\caption{The relations between the complexity scores of two flower-models $M_1$ and $M_2$ that were found for the event logs $L_1$ and $L_2$ respectively, where $L_1 \sqsubset L_2$ and the complexity of $L_1$ is lower than the complexity of $L_2$.}
	\label{table:flowermodel-findings}
	\resizebox{\textwidth}{!}{
	\begin{tabular}{|c|c|c|c|c|c|c|c|c|c|c|c|c|c|c|c|c|c|} \hline
		 & $\size$ & $\mismatch$ & $\connhet$ & $\crossconn$ & $\tokensplit$ & $\controlflow$ & $\separability$ & $\avgconn$ & $\maxconn$ & $\sequentiality$ & $\depth$ & $\diameter$ & $\cyclicity$ & $\netconn$ & $\density$ & $\duplicate$ & $\emptyseq$ \\ \hline
		$\magnitude$ & \hyperref[theo:flower-model-leq-entries]{$\mleq$} & \hyperref[theo:flower-model-equals-entries]{$\meq$} & \hyperref[theo:flower-model-equals-entries]{$\meq$} & \hyperref[theo:flower-model-leq-entries]{$\mleq$} & \hyperref[theo:flower-model-equals-entries]{$\meq$} & \hyperref[theo:flower-model-leq-entries]{$\mleq$} & \hyperref[theo:flower-model-leq-entries]{$\mleq$} & \hyperref[theo:flower-model-leq-entries]{$\mleq$} & \hyperref[theo:flower-model-leq-entries]{$\mleq$} & \hyperref[theo:flower-model-leq-entries]{$\mleq$} & \hyperref[theo:flower-model-equals-entries]{$\meq$} & \hyperref[theo:flower-model-equals-entries]{$\meq$} & \hyperref[theo:flower-model-leq-entries]{$\mleq$} & \hyperref[theo:flower-model-leq-entries]{$\mleq$} & \hyperref[theo:flower-model-equals-entries]{$\meq$} & \hyperref[theo:flower-model-equals-entries]{$\meq$} & \hyperref[theo:flower-model-equals-entries]{$\meq$} \\ \hline
		
		$\variety$ & \hyperref[theo:flower-model-less-entries]{$\mless$} & \hyperref[theo:flower-model-equals-entries]{$\meq$} & \hyperref[theo:flower-model-equals-entries]{$\meq$} & \hyperref[theo:flower-model-less-entries]{$\mless$} & \hyperref[theo:flower-model-equals-entries]{$\meq$} & \hyperref[theo:flower-model-less-entries]{$\mless$} & \hyperref[theo:flower-model-less-entries]{$\mless$} & \hyperref[theo:flower-model-less-entries]{$\mless$} & \hyperref[theo:flower-model-less-entries]{$\mless$} & \hyperref[theo:flower-model-less-entries]{$\mless$} & \hyperref[theo:flower-model-equals-entries]{$\meq$} & \hyperref[theo:flower-model-equals-entries]{$\meq$} & \hyperref[theo:flower-model-less-entries]{$\mless$} & \hyperref[theo:flower-model-less-entries]{$\mless$} & \hyperref[theo:flower-model-equals-entries]{$\meq$} & \hyperref[theo:flower-model-equals-entries]{$\meq$} & \hyperref[theo:flower-model-equals-entries]{$\meq$} \\ \hline
		
		$\support$ & \hyperref[theo:flower-model-leq-entries]{$\mleq$} & \hyperref[theo:flower-model-equals-entries]{$\meq$} & \hyperref[theo:flower-model-equals-entries]{$\meq$} & \hyperref[theo:flower-model-leq-entries]{$\mleq$} & \hyperref[theo:flower-model-equals-entries]{$\meq$} & \hyperref[theo:flower-model-leq-entries]{$\mleq$} & \hyperref[theo:flower-model-leq-entries]{$\mleq$} & \hyperref[theo:flower-model-leq-entries]{$\mleq$} & \hyperref[theo:flower-model-leq-entries]{$\mleq$} & \hyperref[theo:flower-model-leq-entries]{$\mleq$} & \hyperref[theo:flower-model-equals-entries]{$\meq$} & \hyperref[theo:flower-model-equals-entries]{$\meq$} & \hyperref[theo:flower-model-leq-entries]{$\mleq$} & \hyperref[theo:flower-model-leq-entries]{$\mleq$} & \hyperref[theo:flower-model-equals-entries]{$\meq$} & \hyperref[theo:flower-model-equals-entries]{$\meq$} & \hyperref[theo:flower-model-equals-entries]{$\meq$} \\ \hline
		
		$\tlavg$ & \hyperref[theo:flower-model-leq-entries]{$\mleq$} & \hyperref[theo:flower-model-equals-entries]{$\meq$} & \hyperref[theo:flower-model-equals-entries]{$\meq$} & \hyperref[theo:flower-model-leq-entries]{$\mleq$} & \hyperref[theo:flower-model-equals-entries]{$\meq$} & \hyperref[theo:flower-model-leq-entries]{$\mleq$} & \hyperref[theo:flower-model-leq-entries]{$\mleq$} & \hyperref[theo:flower-model-leq-entries]{$\mleq$} & \hyperref[theo:flower-model-leq-entries]{$\mleq$} & \hyperref[theo:flower-model-leq-entries]{$\mleq$} & \hyperref[theo:flower-model-equals-entries]{$\meq$} & \hyperref[theo:flower-model-equals-entries]{$\meq$} & \hyperref[theo:flower-model-leq-entries]{$\mleq$} & \hyperref[theo:flower-model-leq-entries]{$\mleq$} & \hyperref[theo:flower-model-equals-entries]{$\meq$} & \hyperref[theo:flower-model-equals-entries]{$\meq$} & \hyperref[theo:flower-model-equals-entries]{$\meq$} \\ \hline
		
		$\tlmax$ & \hyperref[theo:flower-model-leq-entries]{$\mleq$} & \hyperref[theo:flower-model-equals-entries]{$\meq$} & \hyperref[theo:flower-model-equals-entries]{$\meq$} & \hyperref[theo:flower-model-leq-entries]{$\mleq$} & \hyperref[theo:flower-model-equals-entries]{$\meq$} & \hyperref[theo:flower-model-leq-entries]{$\mleq$} & \hyperref[theo:flower-model-leq-entries]{$\mleq$} & \hyperref[theo:flower-model-leq-entries]{$\mleq$} & \hyperref[theo:flower-model-leq-entries]{$\mleq$} & \hyperref[theo:flower-model-leq-entries]{$\mleq$} & \hyperref[theo:flower-model-equals-entries]{$\meq$} & \hyperref[theo:flower-model-equals-entries]{$\meq$} & \hyperref[theo:flower-model-leq-entries]{$\mleq$} & \hyperref[theo:flower-model-leq-entries]{$\mleq$} & \hyperref[theo:flower-model-equals-entries]{$\meq$} & \hyperref[theo:flower-model-equals-entries]{$\meq$} & \hyperref[theo:flower-model-equals-entries]{$\meq$} \\ \hline
		
		$\levelofdetail$ & \hyperref[theo:flower-model-leq-entries]{$\mleq$} & \hyperref[theo:flower-model-equals-entries]{$\meq$} & \hyperref[theo:flower-model-equals-entries]{$\meq$} & \hyperref[theo:flower-model-leq-entries]{$\mleq$} & \hyperref[theo:flower-model-equals-entries]{$\meq$} & \hyperref[theo:flower-model-leq-entries]{$\mleq$} & \hyperref[theo:flower-model-leq-entries]{$\mleq$} & \hyperref[theo:flower-model-leq-entries]{$\mleq$} & \hyperref[theo:flower-model-leq-entries]{$\mleq$} & \hyperref[theo:flower-model-leq-entries]{$\mleq$} & \hyperref[theo:flower-model-equals-entries]{$\meq$} & \hyperref[theo:flower-model-equals-entries]{$\meq$} & \hyperref[theo:flower-model-leq-entries]{$\mleq$} & \hyperref[theo:flower-model-leq-entries]{$\mleq$} & \hyperref[theo:flower-model-equals-entries]{$\meq$} & \hyperref[theo:flower-model-equals-entries]{$\meq$} & \hyperref[theo:flower-model-equals-entries]{$\meq$} \\ \hline
		
		$\numberofties$ & \hyperref[theo:flower-model-leq-entries]{$\mleq$} & \hyperref[theo:flower-model-equals-entries]{$\meq$} & \hyperref[theo:flower-model-equals-entries]{$\meq$} & \hyperref[theo:flower-model-leq-entries]{$\mleq$} & \hyperref[theo:flower-model-equals-entries]{$\meq$} & \hyperref[theo:flower-model-leq-entries]{$\mleq$} & \hyperref[theo:flower-model-leq-entries]{$\mleq$} & \hyperref[theo:flower-model-leq-entries]{$\mleq$} & \hyperref[theo:flower-model-leq-entries]{$\mleq$} & \hyperref[theo:flower-model-leq-entries]{$\mleq$} & \hyperref[theo:flower-model-equals-entries]{$\meq$} & \hyperref[theo:flower-model-equals-entries]{$\meq$} & \hyperref[theo:flower-model-leq-entries]{$\mleq$} & \hyperref[theo:flower-model-leq-entries]{$\mleq$} & \hyperref[theo:flower-model-equals-entries]{$\meq$} & \hyperref[theo:flower-model-equals-entries]{$\meq$} & \hyperref[theo:flower-model-equals-entries]{$\meq$} \\ \hline
		
		$\lempelziv$ & \hyperref[theo:flower-model-leq-entries]{$\mleq$} & \hyperref[theo:flower-model-equals-entries]{$\meq$} & \hyperref[theo:flower-model-equals-entries]{$\meq$} & \hyperref[theo:flower-model-leq-entries]{$\mleq$} & \hyperref[theo:flower-model-equals-entries]{$\meq$} & \hyperref[theo:flower-model-leq-entries]{$\mleq$} & \hyperref[theo:flower-model-leq-entries]{$\mleq$} & \hyperref[theo:flower-model-leq-entries]{$\mleq$} & \hyperref[theo:flower-model-leq-entries]{$\mleq$} & \hyperref[theo:flower-model-leq-entries]{$\mleq$} & \hyperref[theo:flower-model-equals-entries]{$\meq$} & \hyperref[theo:flower-model-equals-entries]{$\meq$} & \hyperref[theo:flower-model-leq-entries]{$\mleq$} & \hyperref[theo:flower-model-leq-entries]{$\mleq$} & \hyperref[theo:flower-model-equals-entries]{$\meq$} & \hyperref[theo:flower-model-equals-entries]{$\meq$} & \hyperref[theo:flower-model-equals-entries]{$\meq$} \\ \hline
		
		$\numberuniquetraces$ & \hyperref[theo:flower-model-leq-entries]{$\mleq$} & \hyperref[theo:flower-model-equals-entries]{$\meq$} & \hyperref[theo:flower-model-equals-entries]{$\meq$} & \hyperref[theo:flower-model-leq-entries]{$\mleq$} & \hyperref[theo:flower-model-equals-entries]{$\meq$} & \hyperref[theo:flower-model-leq-entries]{$\mleq$} & \hyperref[theo:flower-model-leq-entries]{$\mleq$} & \hyperref[theo:flower-model-leq-entries]{$\mleq$} & \hyperref[theo:flower-model-leq-entries]{$\mleq$} & \hyperref[theo:flower-model-leq-entries]{$\mleq$} & \hyperref[theo:flower-model-equals-entries]{$\meq$} & \hyperref[theo:flower-model-equals-entries]{$\meq$} & \hyperref[theo:flower-model-leq-entries]{$\mleq$} & \hyperref[theo:flower-model-leq-entries]{$\mleq$} & \hyperref[theo:flower-model-equals-entries]{$\meq$} & \hyperref[theo:flower-model-equals-entries]{$\meq$} & \hyperref[theo:flower-model-equals-entries]{$\meq$} \\ \hline
		
		$\percentageuniquetraces$ & \hyperref[theo:flower-model-leq-entries]{$\mleq$} & \hyperref[theo:flower-model-equals-entries]{$\meq$} & \hyperref[theo:flower-model-equals-entries]{$\meq$} & \hyperref[theo:flower-model-leq-entries]{$\mleq$} & \hyperref[theo:flower-model-equals-entries]{$\meq$} & \hyperref[theo:flower-model-leq-entries]{$\mleq$} & \hyperref[theo:flower-model-leq-entries]{$\mleq$} & \hyperref[theo:flower-model-leq-entries]{$\mleq$} & \hyperref[theo:flower-model-leq-entries]{$\mleq$} & \hyperref[theo:flower-model-leq-entries]{$\mleq$} & \hyperref[theo:flower-model-equals-entries]{$\meq$} & \hyperref[theo:flower-model-equals-entries]{$\meq$} & \hyperref[theo:flower-model-leq-entries]{$\mleq$} & \hyperref[theo:flower-model-leq-entries]{$\mleq$} & \hyperref[theo:flower-model-equals-entries]{$\meq$} & \hyperref[theo:flower-model-equals-entries]{$\meq$} & \hyperref[theo:flower-model-equals-entries]{$\meq$} \\ \hline
		
		$\structure$ & \hyperref[theo:flower-model-leq-entries]{$\mleq$} & \hyperref[theo:flower-model-equals-entries]{$\meq$} & \hyperref[theo:flower-model-equals-entries]{$\meq$} & \hyperref[theo:flower-model-leq-entries]{$\mleq$} & \hyperref[theo:flower-model-equals-entries]{$\meq$} & \hyperref[theo:flower-model-leq-entries]{$\mleq$} & \hyperref[theo:flower-model-leq-entries]{$\mleq$} & \hyperref[theo:flower-model-leq-entries]{$\mleq$} & \hyperref[theo:flower-model-leq-entries]{$\mleq$} & \hyperref[theo:flower-model-leq-entries]{$\mleq$} & \hyperref[theo:flower-model-equals-entries]{$\meq$} & \hyperref[theo:flower-model-equals-entries]{$\meq$} & \hyperref[theo:flower-model-leq-entries]{$\mleq$} & \hyperref[theo:flower-model-leq-entries]{$\mleq$} & \hyperref[theo:flower-model-equals-entries]{$\meq$} & \hyperref[theo:flower-model-equals-entries]{$\meq$} & \hyperref[theo:flower-model-equals-entries]{$\meq$} \\ \hline
		
		$\affinity$ & \hyperref[theo:flower-model-leq-entries]{$\mleq$} & \hyperref[theo:flower-model-equals-entries]{$\meq$} & \hyperref[theo:flower-model-equals-entries]{$\meq$} & \hyperref[theo:flower-model-leq-entries]{$\mleq$} & \hyperref[theo:flower-model-equals-entries]{$\meq$} & \hyperref[theo:flower-model-leq-entries]{$\mleq$} & \hyperref[theo:flower-model-leq-entries]{$\mleq$} & \hyperref[theo:flower-model-leq-entries]{$\mleq$} & \hyperref[theo:flower-model-leq-entries]{$\mleq$} & \hyperref[theo:flower-model-leq-entries]{$\mleq$} & \hyperref[theo:flower-model-equals-entries]{$\meq$} & \hyperref[theo:flower-model-equals-entries]{$\meq$} & \hyperref[theo:flower-model-leq-entries]{$\mleq$} & \hyperref[theo:flower-model-leq-entries]{$\mleq$} & \hyperref[theo:flower-model-equals-entries]{$\meq$} & \hyperref[theo:flower-model-equals-entries]{$\meq$} & \hyperref[theo:flower-model-equals-entries]{$\meq$} \\ \hline
		
		$\deviationfromrandom$ & \hyperref[theo:flower-model-leq-entries]{$\mleq$} & \hyperref[theo:flower-model-equals-entries]{$\meq$} & \hyperref[theo:flower-model-equals-entries]{$\meq$} & \hyperref[theo:flower-model-leq-entries]{$\mleq$} & \hyperref[theo:flower-model-equals-entries]{$\meq$} & \hyperref[theo:flower-model-leq-entries]{$\mleq$} & \hyperref[theo:flower-model-leq-entries]{$\mleq$} & \hyperref[theo:flower-model-leq-entries]{$\mleq$} & \hyperref[theo:flower-model-leq-entries]{$\mleq$} & \hyperref[theo:flower-model-leq-entries]{$\mleq$} & \hyperref[theo:flower-model-equals-entries]{$\meq$} & \hyperref[theo:flower-model-equals-entries]{$\meq$} & \hyperref[theo:flower-model-leq-entries]{$\mleq$} & \hyperref[theo:flower-model-leq-entries]{$\mleq$} & \hyperref[theo:flower-model-equals-entries]{$\meq$} & \hyperref[theo:flower-model-equals-entries]{$\meq$} & \hyperref[theo:flower-model-equals-entries]{$\meq$} \\ \hline
		
		$\avgdist$ & \hyperref[theo:flower-model-leq-entries]{$\mleq$} & \hyperref[theo:flower-model-equals-entries]{$\meq$} & \hyperref[theo:flower-model-equals-entries]{$\meq$} & \hyperref[theo:flower-model-leq-entries]{$\mleq$} & \hyperref[theo:flower-model-equals-entries]{$\meq$} & \hyperref[theo:flower-model-leq-entries]{$\mleq$} & \hyperref[theo:flower-model-leq-entries]{$\mleq$} & \hyperref[theo:flower-model-leq-entries]{$\mleq$} & \hyperref[theo:flower-model-leq-entries]{$\mleq$} & \hyperref[theo:flower-model-leq-entries]{$\mleq$} & \hyperref[theo:flower-model-equals-entries]{$\meq$} & \hyperref[theo:flower-model-equals-entries]{$\meq$} & \hyperref[theo:flower-model-leq-entries]{$\mleq$} & \hyperref[theo:flower-model-leq-entries]{$\mleq$} & \hyperref[theo:flower-model-equals-entries]{$\meq$} & \hyperref[theo:flower-model-equals-entries]{$\meq$} & \hyperref[theo:flower-model-equals-entries]{$\meq$} \\ \hline
		
		$\varentropy$ & \hyperref[theo:flower-model-leq-entries]{$\mleq$} & \hyperref[theo:flower-model-equals-entries]{$\meq$} & \hyperref[theo:flower-model-equals-entries]{$\meq$} & \hyperref[theo:flower-model-leq-entries]{$\mleq$} & \hyperref[theo:flower-model-equals-entries]{$\meq$} & \hyperref[theo:flower-model-leq-entries]{$\mleq$} & \hyperref[theo:flower-model-leq-entries]{$\mleq$} & \hyperref[theo:flower-model-leq-entries]{$\mleq$} & \hyperref[theo:flower-model-leq-entries]{$\mleq$} & \hyperref[theo:flower-model-leq-entries]{$\mleq$} & \hyperref[theo:flower-model-equals-entries]{$\meq$} & \hyperref[theo:flower-model-equals-entries]{$\meq$} & \hyperref[theo:flower-model-leq-entries]{$\mleq$} & \hyperref[theo:flower-model-leq-entries]{$\mleq$} & \hyperref[theo:flower-model-equals-entries]{$\meq$} & \hyperref[theo:flower-model-equals-entries]{$\meq$} & \hyperref[theo:flower-model-equals-entries]{$\meq$} \\ \hline
	
		$\normvarentropy$ & \hyperref[theo:flower-model-leq-entries]{$\mleq$} & \hyperref[theo:flower-model-equals-entries]{$\meq$} & \hyperref[theo:flower-model-equals-entries]{$\meq$} & \hyperref[theo:flower-model-leq-entries]{$\mleq$} & \hyperref[theo:flower-model-equals-entries]{$\meq$} & \hyperref[theo:flower-model-leq-entries]{$\mleq$} & \hyperref[theo:flower-model-leq-entries]{$\mleq$} & \hyperref[theo:flower-model-leq-entries]{$\mleq$} & \hyperref[theo:flower-model-leq-entries]{$\mleq$} & \hyperref[theo:flower-model-leq-entries]{$\mleq$} & \hyperref[theo:flower-model-equals-entries]{$\meq$} & \hyperref[theo:flower-model-equals-entries]{$\meq$} & \hyperref[theo:flower-model-leq-entries]{$\mleq$} & \hyperref[theo:flower-model-leq-entries]{$\mleq$} & \hyperref[theo:flower-model-equals-entries]{$\meq$} & \hyperref[theo:flower-model-equals-entries]{$\meq$} & \hyperref[theo:flower-model-equals-entries]{$\meq$} \\ \hline
	
		$\seqentropy$ & \hyperref[theo:flower-model-leq-entries]{$\mleq$} & \hyperref[theo:flower-model-equals-entries]{$\meq$} & \hyperref[theo:flower-model-equals-entries]{$\meq$} & \hyperref[theo:flower-model-leq-entries]{$\mleq$} & \hyperref[theo:flower-model-equals-entries]{$\meq$} & \hyperref[theo:flower-model-leq-entries]{$\mleq$} & \hyperref[theo:flower-model-leq-entries]{$\mleq$} & \hyperref[theo:flower-model-leq-entries]{$\mleq$} & \hyperref[theo:flower-model-leq-entries]{$\mleq$} & \hyperref[theo:flower-model-leq-entries]{$\mleq$} & \hyperref[theo:flower-model-equals-entries]{$\meq$} & \hyperref[theo:flower-model-equals-entries]{$\meq$} & \hyperref[theo:flower-model-leq-entries]{$\mleq$} & \hyperref[theo:flower-model-leq-entries]{$\mleq$} & \hyperref[theo:flower-model-equals-entries]{$\meq$} & \hyperref[theo:flower-model-equals-entries]{$\meq$} & \hyperref[theo:flower-model-equals-entries]{$\meq$} \\ \hline
	
		$\normseqentropy$ & \hyperref[theo:flower-model-leq-entries]{$\mleq$} & \hyperref[theo:flower-model-equals-entries]{$\meq$} & \hyperref[theo:flower-model-equals-entries]{$\meq$} & \hyperref[theo:flower-model-leq-entries]{$\mleq$} & \hyperref[theo:flower-model-equals-entries]{$\meq$} & \hyperref[theo:flower-model-leq-entries]{$\mleq$} & \hyperref[theo:flower-model-leq-entries]{$\mleq$} & \hyperref[theo:flower-model-leq-entries]{$\mleq$} & \hyperref[theo:flower-model-leq-entries]{$\mleq$} & \hyperref[theo:flower-model-leq-entries]{$\mleq$} & \hyperref[theo:flower-model-equals-entries]{$\meq$} & \hyperref[theo:flower-model-equals-entries]{$\meq$} & \hyperref[theo:flower-model-leq-entries]{$\mleq$} & \hyperref[theo:flower-model-leq-entries]{$\mleq$} & \hyperref[theo:flower-model-equals-entries]{$\meq$} & \hyperref[theo:flower-model-equals-entries]{$\meq$} & \hyperref[theo:flower-model-equals-entries]{$\meq$} \\ \hline
	\end{tabular}
	}
\end{table}
For quick reference, the PDF-version of this paper allows to click on an entry to directly jump to its proof.

\begin{theorem}
\label{theo:flower-model-leq-entries}
Let $\mathcal{C}^L \in (\loc \setminus \{\variety\})$ be any log complexity measure and let $\mathcal{C}^M \in \{\size, \crossconn, \controlflow, \separability, \avgconn, \maxconn, \sequentiality, \cyclicity, \netconn\}$ be a model complexity measure.
Then, $(\mathcal{C}^L, \mathcal{C}^M) \in \mleq$.
\end{theorem}
\begin{proof}
By definition of $\mleq$, we need to show that for all logs $L_1 \sqsubset L_2$ and their flower models $M_1, M_2$, where $\mathcal{C}^L(L_1) < \mathcal{C}^L(L_2)$, we have $\mathcal{C}^M(M_1) < \mathcal{C}^M(M_2)$ or $\mathcal{C}^M(M_1) = \mathcal{C}^M(M_2)$, and that there are examples for both cases. 
By \cref{lemma:flower-model-monotone-increasing}, we already know that $\mathcal{C}^M(M_1) \leq \mathcal{C}^M(M_2)$ because $L_1 \sqsubset L_2$.
Furthermore, by \cref{lemma:not-variety-dependent}, we know that there are cases where $\mathcal{C}^L(L_1) < \mathcal{C}^L(L_2)$, but with $\variety(L_1) = \variety(L_2)$ and therefore $\mathcal{C}^M(M_1) = \mathcal{C}^M(M_2)$, since $M_1$ and $M_2$ are the same model.
To see that $\mathcal{C}^M(M_1) < \mathcal{C}^M(M_2)$ is also possible, consider the following event logs:
\begin{align*}
	L_1 &= [\langle a \rangle^{2}, \langle a,b,c,d \rangle^{3}] \\
	L_2 &= L_1 + [\langle e,a,b,c,d \rangle^{2}]
\end{align*}
These two event logs have the following log complexity scores:
\begin{center}
	\begin{tabular}{|c|c|c|c|c|c|c|c|c|c|c|}\hline
		 & $\magnitude$ & $\variety$ & $\support$ & $\tlavg$ & $\tlmax$ & $\levelofdetail$ & $\numberofties$ & $\lempelziv$ & $\numberuniquetraces$ & $\percentageuniquetraces$ \\ \hline
		$L_1$ & $\pad 14 \pad$ & $\pad 4 \pad$ & $\pad 5 \pad$ & $\pad 2.8 \pad$ & $\pad 4 \pad$ & $\pad 2 \pad$ & $\pad 3 \pad$ & $\pad 8 \pad$ & $\pad 2 \pad$ & $\pad 0.4 \pad$ \\ \hline
		$L_2$ & $\pad 24 \pad$ & $\pad 5 \pad$ & $\pad 7 \pad$ & $\pad 3.4286 \pad$ & $\pad 5 \pad$ & $\pad 4 \pad$ & $\pad 4 \pad$ & $\pad 11 \pad$ & $\pad 3 \pad$ & $\pad 0.4286 \pad$ \\ \hline
	\end{tabular}
	
	\medskip
	
	\begin{tabular}{|c|c|c|c|c|c|c|c|c|} \hline
		 & $\structure$ & $\affinity$ & $\deviationfromrandom$ & $\avgdist$ & $\varentropy$ & $\normvarentropy$ & $\seqentropy$ & $\normseqentropy$ \\ \hline
		$L_1$ & $\pad 2.8 \pad$ & $\pad 0.4 \pad$ & $\pad 0.4796 \pad$ & $\pad 1.8 \pad$ & $\pad 0 \pad$ & $\pad 0 \pad$ & $\pad 0 \pad$ & $\pad 0 \pad$ \\ \hline
		$L_2$ & $\pad 3.4286 \pad$ & $\pad 0.4524 \pad$ & $\pad 0.5169 \pad$ & $\pad 1.9048 \pad$ & $\pad 6.1827 \pad$ & $\pad 0.3126 \pad$ & $\pad 16.3006 \pad$ & $\pad 0.2137 \pad$ \\ \hline
	\end{tabular}
\end{center}
Thus, $\mathcal{C}^L(L_1) < \mathcal{C}^L(L_2)$ for any log complexity measure $\mathcal{C}^L \in (\log \setminus \{\variety\})$.
The flower models for $L_1$ and $L_2$ are shown in \cref{fig:flower-model-leq-entries}.
\begin{figure}[ht]
	\centering
	\begin{minipage}{0.45\textwidth}
	\centering
	\scalebox{\scalefactor}{
	\begin{tikzpicture}[node distance = 1.25cm,>=stealth',bend angle=0,auto]
		\node[place,tokens=1,label=below:$p_i$] (start) {};
		\node[transition,right of=start,label=below:$t_1$] (tau1) {$\tau$}
		edge [pre] (start);
		\node[place, above right of=tau1,label=below:$p$] (middle) {}
		edge [pre] (tau1);
		\node[transition,below right of=middle,label=below:$t_2$] (tau2) {$\tau$}
		edge [pre] (middle);
		\node[place,right of=tau2,label=below:$p_o$] (end) {}
		edge [pre] (tau2);
		\node[transition,left of=middle] (a) {$a$}
		edge [pre,bend right=15] (middle)
		edge [post,bend left=15] (middle);
		\node[transition,above left of=middle] (b) {$b$}
		edge [pre,bend right=15] (middle)
		edge [post,bend left=15] (middle);
		\node[transition,above right of=middle] (c) {$c$}
		edge [pre,bend right=15] (middle)
		edge [post,bend left=15] (middle);
		\node[transition,right of=middle] (d) {$d$}
		edge [pre,bend right=15] (middle)
		edge [post,bend left=15] (middle);
		\node[transition,opacity=0,above of=middle] (phantom) {};
	\end{tikzpicture}}
	\end{minipage}
	\begin{minipage}{0.45\textwidth}
	\centering
	\scalebox{\scalefactor}{
	\begin{tikzpicture}[node distance = 1.25cm,>=stealth',bend angle=0,auto]
		\node[place,tokens=1,label=below:$p_i$] (start) {};
		\node[transition,right of=start,label=below:$t_1$] (tau1) {$\tau$}
		edge [pre] (start);
		\node[place, above right of=tau1,label=below:$p$] (middle) {}
		edge [pre] (tau1);
		\node[transition,below right of=middle,label=below:$t_2$] (tau2) {$\tau$}
		edge [pre] (middle);
		\node[place,right of=tau2,label=below:$p_o$] (end) {}
		edge [pre] (tau2);
		\node[transition,left of=middle] (a) {$a$}
		edge [pre,bend right=15] (middle)
		edge [post,bend left=15] (middle);
		\node[transition,above left of=middle] (b) {$b$}
		edge [pre,bend right=15] (middle)
		edge [post,bend left=15] (middle);
		\node[transition,above of=middle] (c) {$c$}
		edge [pre,bend right=15] (middle)
		edge [post,bend left=15] (middle);
		\node[transition,above right of=middle] (d) {$d$}
		edge [pre,bend right=15] (middle)
		edge [post,bend left=15] (middle);
		\node[transition,right of=middle] (e) {$e$}
		edge [pre,bend right=15] (middle)
		edge [post,bend left=15] (middle);
	\end{tikzpicture}}
	\end{minipage}
	\caption{The flower models for the logs $L_1, L_2$ of \cref{theo:flower-model-leq-entries}.}
	\label{fig:flower-model-leq-entries}
\end{figure}
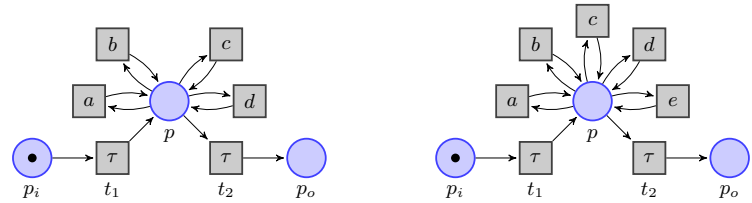
These models have the following model complexity scores:
\begin{center}
	\begin{tabular}{|c|c|c|c|c|c|c|c|c|c|}\hline
		 & $\size$ & $\crossconn$ & $\controlflow$ & $\separability$ & $\avgconn$ & $\maxconn$ & $\sequentiality$ & $\cyclicity$ & $\netconn$ \\ \hline
		$L_1$ & $\pad 9 \pad$ & $\pad 0.9504 \pad$ & $\pad 5 \pad$ & $\pad 0.5714 \pad$ & $\pad 10 \pad$ & $\pad 10 \pad$ & $\pad 0.8333 \pad$ & $\pad 0.7143 \pad$ & $\pad 1.3333 \pad$ \\ \hline
		$L_2$ & $\pad 10 \pad$ & $\pad 0.961 \pad$ & $\pad 6 \pad$ & $\pad 0.625 \pad$ & $\pad 12 \pad$ & $\pad 12 \pad$ & $\pad 0.8571 \pad$ & $\pad 0.75 \pad$ & $\pad 1.4 \pad$ \\ \hline
	\end{tabular}
\end{center}
Thus, $(\mathcal{C}^L, \mathcal{C}^M) \in \mleq$ for any log complexity measure $\mathcal{C}^L \in (\loc \setminus \{\variety\})$ and a measure $\mathcal{C}^M \in \{\size, \crossconn, \controlflow, \separability, \avgconn, \maxconn, \sequentiality, \cyclicity, \netconn\}$ for model complexity, as stated in the theorem. \hfill$\square$
\end{proof}

\begin{theorem}
\label{theo:flower-model-less-entries}
Let $\mathcal{C}^M \in \{\size, \crossconn, \controlflow, \separability, \avgconn, \maxconn, \sequentiality, \cyclicity, \netconn\}$ be a model complexity measure.
Then, $(\variety, \mathcal{C}^M) \in \mless$.
\end{theorem}
\begin{proof}
Let $L_1 \sqsubset L_2$ be event logs and $M_1, M_2$ be their flower models.
Then, $\variety(L_1) < \variety(L_2)$ implies that there is a new activity name in $L_2$ that is not present in $L_1$.
In turn, $M_2$ contains a transition that does not exist in $M_1$.
Since, by \cref{lemma:flower-model-monotone-increasing}, $\mathcal{C}^M(M_1) \leq \mathcal{C}^M(M_2)$, this means that $\mathcal{C}^M(M_1) < \mathcal{C}^M(M_2)$. \hfill$\square$
\end{proof}

\begin{theorem}
\label{theo:flower-model-equals-entries}
Let $\mathcal{C}^L \in \loc$ be any log complexity measure and let $\mathcal{C}^M$ be a model complexity measure with $\mathcal{C}^M \in \{\mismatch, \connhet, \tokensplit, \depth, \diameter, \density, \duplicate, \emptyseq\}$.
Then, $(\mathcal{C}^L, \mathcal{C}^M) \in \meq$.
\end{theorem}
\begin{proof}
By \cref{lemma:flowermodel-constant-complexity}, $\mathcal{C}^M(M_1) = \mathcal{C}^M(M_2)$ for any flower models $M_1, M_2$.
Therefore, the implication $\mathcal{C}^L(L_1) < \mathcal{C}^L(L_2) \Rightarrow \mathcal{C}^M(M_1) = \mathcal{C}^M(M_2)$ is true for all event logs $L_1, L_2$, where $M_1, M_2$ are the flower models for $L_1, L_2$. \hfill$\square$
\end{proof}

As \cref{table:flowermodel-findings} shows, the model complexity of the flower model is only dependent on the variety of the underlying event log.
In the remainder of this subsection, we will go even further and characterize the model complexity scores of the flower model by using the variety of the event log.
Note that some of the arguments we will show here already appeared in \cref{lemma:flower-model-monotone-increasing}.
In the following, let $L$ be an event log over a set of activities $A$ and $M$ be the flower model for $L$.
\begin{itemize}
	\item \textbf{Size $\size$:}
	The flower model has exactly $3$ places, labeled $p_i$, $p_o$ and $p$ in \cref{fig:flower-model}.
	Furthermore, it features two silent transitions, highlighted as $t_1$ and $t_2$ in the same figure.
	Every flower model has these $5$ nodes, independent of the event log.
	Apart from them, it contains a transition for each activity name in the event log, so $\size(M) = 5 + \variety(L)$.
	
	\item \textbf{Connector Mismatch $\mismatch$:}
	The flower model has exactly one connector, labeled $p$ in \cref{fig:flower-model}.
	This place has $|T| - 1$ incoming and $|T| - 1$ outgoing arcs, so $\mismatch(M) = \left| (|T| - 1) - (|T| - 1)\right| = 0$.
	
	\item \textbf{Connector Heterogeneity $\connhet$:}
	The only connector of the flower model is the place labeled $p$ in \cref{fig:flower-model}, which is an \texttt{xor}-connector.
	Since there are no other connectors in the flower model, there are also no \texttt{and}-connectors.
	Therefore, calculating the entropy of connector types in the flower model gives $\connhet(M) = -(1 \cdot \log_2(1) + 0 \cdot \log_2(0)) = 0$.
	
	\item \textbf{Cross Connectivity $\crossconn$:}
	Let $n := \variety(L)$ be the variety of the event log $L$.
	Then, the place $p$ receives weight $\frac{1}{2n+2}$, while all other nodes receive weight $1$.
	After calculating the weights of the edges and paths between all nodes, we receive the values shown in \cref{table:flowermodel-cross-conn}.
	\begin{table}[ht]
		\caption{The connection values of all node-pairs in a flower model.}
		\label{table:flowermodel-cross-conn}
		\centering
		\def\smallpad{\hspace*{1mm}}
		\def\bigpad{\hspace*{3mm}}
		\renewcommand{\arraystretch}{1.15}
		\begin{tabular}{|c|c|c|c|c|c|c|c|c|} \hline
		 & $\bigpad p_i \bigpad$ & $\bigpad t_1 \bigpad$ & $\bigpad p \bigpad$ & $a_1$ & $\smallpad\ldots\smallpad$ & $a_n$ & $t_2$ & $p_o$ \\ \hline
		$p_i$ & $0$ & $1$ & $\frac{1}{2n+2}$ & $\frac{1}{(2n+2)^2}$ & $\ldots$ & $\frac{1}{(2n+2)^2}$ & $\frac{1}{(2n+2)^2}$ & $\frac{1}{(2n+2)^2}$ \\ \hline
		$t_1$ & $0$ & $0$ & $\frac{1}{2n+2}$ & $\frac{1}{(2n+2)^2}$ & $\ldots$ & $\frac{1}{(2n+2)^2}$ & $\frac{1}{(2n+2)^2}$ & $\frac{1}{(2n+2)^2}$ \\ \hline
		$p$ & $0$ & $0$ & $\frac{1}{(2n+2)^2}$ & $\frac{1}{2n+2}$ & $\ldots$ & $\frac{1}{2n+2}$ & $\frac{1}{2n+2}$ & $\frac{1}{2n+2}$ \\ \hline
		$a_1$ & $0$ & $0$ & $\frac{1}{2n+2}$ & $\frac{1}{(2n+2)^2}$ & $\dots$ & $\frac{1}{(2n+2)^2}$ & $\frac{1}{(2n+2)^2}$ & $\frac{1}{(2n+2)^2}$ \\ \hline
		$\vdots$ & $\vdots$ & $\vdots$ & $\vdots$ & $\vdots$ & $\ddots$ & $\vdots$ & $\vdots$ & $\vdots$ \\ \hline
		$a_n$ & $0$ & $0$ & $\frac{1}{2n+2}$ & $\frac{1}{(2n+2)^2}$ & $\ldots$ & $\frac{1}{(2n+2)^2}$ & $\frac{1}{(2n+2)^2}$ & $\frac{1}{(2n+2)^2}$ \\ \hline
		$t_2$ & $0$ & $0$ & $0$ & $0$ & $\ldots$ & $0$ & $0$ & $1$ \\ \hline
		$p_o$ & $0$ & $0$ & $0$ & $0$ & $\ldots$ & $0$ & $0$ & $0$ \\ \hline
		\end{tabular}
	\end{table}
	This table contains the entry $1$ two times, the entry $\frac{1}{2n+2}$ exactly $2n+4$ times, and the entry $\frac{1}{(2n+2)^2}$ a total of $n^2 + 4n + 5$ times.
	Thus, for the cross connectivity score, we get $\crossconn(M) = 1 - \frac{2 + \frac{2n+4}{2n+2} + \frac{n^2 + 4n + 5}{4n^2 + 8n + 4}}{n^2 + 9n + 20} = \frac{4 n^4 + 44 n^3 + 143 n^2 + 164 n + 59}{4 (n + 1)^2 (n + 4) (n + 5)}$.
	
	\item \textbf{Token Split $\tokensplit$:}
	By construction, every node in the flower model has exactly one incoming and one outgoing edge.
	Thus, there are no transitions with more than one outgoing edge and we get $\tokensplit(M) = 0$.
	
	\item \textbf{Control Flow Complexity $\controlflow$:}
	The place labeled $p$ in \cref{fig:flower-model} is the only connector node of the flower model.
	This place is an \texttt{xor}-connector with $|T| - 1$ outgoing edges.
	Since $|T| = \variety(L) + 2$, we get the control flow complexity score $\controlflow(M) = \variety(L) + 1$.
	
	\item \textbf{Separability $\separability$:}
	The cut-vertices of the flower model are the nodes labeled $t_1$, $p$, and $t_2$ in \cref{fig:flower-model}.
	Thus, we have exactly $3$ cut-vertices in the flower model.
	Since there are $5 + \variety(L)$ nodes in total, we get the separability score $\separability(M) = 1 - \frac{3}{5 + \variety(L) - 2} = \frac{\variety(L)}{3 + \variety(L)}$.
	
	\item \textbf{Average Connector Degree $\avgconn$:}
	The place labeled $p$ in \cref{fig:flower-model} is the only connector of the flower model and has $|\pre{p}| + |\post{p}| = |T| - 1 + |T| - 1 = 2 |T| - 2$.
	Since $|T| = \variety(L) + 2$, we get $\avgconn(M) = 2 \variety(L) + 2$.
	
	\item \textbf{Maximum Connector Degree $\maxconn$:}
	The place labeled $p$ in \cref{fig:flower-model} is the only connector of the flower model and $|\pre{p}| + |\post{p}| = |T| - 1 + |T| - 1 = 2 |T| - 2$.
	Since $|T| = \variety(L) + 2$, we get $\maxconn(M) = 2 \variety(L) + 2$.
	
	\item \textbf{Sequentiality $\sequentiality$:}
	In the flower model, only the edges $(p_i, t_1$ and $(t_2, p_o)$ connect only non-connector nodes.
	In total, there are $2 |T| = 2 \variety(L) + 4$ edges, so we get $\sequentiality(M) = 1 - \frac{2}{2 \variety(L) + 4} = \frac{2 \variety(L) + 2}{2 \variety(L) + 4} = \frac{\variety(L) + 1}{\variety(L) + 2}$.
	
	\item \textbf{Depth $\depth$:}
	Let $A = \{a_1, \dots, a_n\}$ be the activity names that occur in $L$.
	Then, \cref{table:flowermodel-depth} shows the in- and out-depth of each node in the flower model.
	\begin{table}[ht]
	\caption{The in- and out-depths of all nodes in the flower model.}
	\label{table:flowermodel-depth}
	\centering
	\begin{tabular}{|c|c|c|} \hline
	\textbf{Nodes} & \textbf{In-Depth} & \textbf{Out-Depth} \\ \hline
	$p_i, t_1$ & $0$ & $1$ \\ \hline
	$p$ & $0$ & $0$ \\ \hline
	$a_1, \dots, a_n$ & $1$ & $1$ \\ \hline
	$p_o, t_2$ & $1$ & $0$ \\ \hline
	\end{tabular}
	\end{table}
	With this, we get $\depth(M) = 1$.
	
	\item \textbf{Diameter $\diameter$:}
	The longest simple path through the flower model is the path $(p_i, t_1, p, t_2, p_o)$, so $\diameter(M) = 5$.
	
	\item \textbf{Cyclicity $\cyclicity$:}
	With the labels shown in \cref{fig:flower-model}, only the nodes $a_1$, $\dots$, $a_n$, and $p$ lie on a cycle in the flower model.
	Since the model has $5 + \variety(L)$ nodes in total, we get $\cyclicity(M) = \frac{\variety(L) + 1}{5 + \variety(L) - 2} = \frac{\variety(L) + 1}{\variety(L) + 3}$.
	
	\item \textbf{Coefficient of Network Connectivity $\netconn$:}
	Since every transition in the flower model has exactly one incoming and one outgoing edge, it contains $2 |T|$ edges in total.
	With $|T| = \variety(L) + 2$ and the fact that the flower model has $5 + \variety(L)$ nodes in total, we get $\netconn(M) = \frac{2\variety(L) + 4}{\variety(L) + 5}$.
	
	\item \textbf{Density $\density$:}
	Since every transition in the flower model has exactly one incoming and one outgoing edge, it contains $2 |T|$ edges in total.
	With $|P| = 3$ and $|T| = \variety(L) + 2$, we therefore get $\density(M) = \frac{2 (\variety(L) + 2)}{2 \cdot (\variety(L) + 2) \cdot (3 - 1)} = \frac{1}{2}$.
	
	\item \textbf{Number of Duplicate Tasks $\duplicate$:}
	The only label repititions the flower model contains are the ones issued by the two silent transitions highlighted as $t_1$ and $t_2$ in \cref{fig:flower-model}.
	In turn, $\duplicate(M) = 1$. 
	
	\item \textbf{Number of Empty Sequence Flows $\emptyseq$:}
	Since the flower model does not contain any \texttt{and}-connectors, $\emptyseq(M) = 0$.
\end{itemize}
These findings conclude our analysis of the flower miner. 
\cref{table:flowermodel-model-complexity} summarizes these findings for quick reference.

\begin{table}[ht]
	\caption{The complexity scores of the flower model $M$ for an event log $L$ over $A$.}
	\label{table:flowermodel-model-complexity}
	\centering
	\def\pad{\hspace*{1.5mm}}
	\renewcommand{\arraystretch}{1.25}
	\begin{tabular}{|r|l|} \hline
	$\pad\size(M)\pad$ & $\pad 5 + \variety(L) \pad$ \\ \hline
	$\pad\mismatch(M)\pad$ & $\pad 0 \pad$ \\ \hline
	$\pad\connhet(M)\pad$ & $\pad 0 \pad$ \\ \hline
	$\pad\crossconn(M)\pad$ & $\pad \frac{4 \variety(L)^4 + 44 \variety(L)^3 + 143 \variety(L)^2 + 164 \variety(L) + 59}{4 (\variety(L) + 1)^2 (\variety(L) + 4) (\variety(L) + 5)} \pad$ \\ \hline
	$\pad\tokensplit(M)\pad$ & $\pad 0 \pad$ \\ \hline
	$\pad\controlflow(M)\pad$ & $\pad \variety(L) + 1 \pad$ \\ \hline
	$\pad\separability(M)\pad$ & $\pad \frac{\variety(L)}{3 + \variety(L)} \pad$ \\ \hline
	$\pad\avgconn(M)\pad$ & $\pad 2 \variety(L) + 2 \pad$ \\ \hline
	$\pad\maxconn(M)\pad$ & $\pad 2 \variety(L) + 2 \pad$ \\ \hline
	$\pad\sequentiality(M)\pad$ & $\pad \frac{\variety(L) + 1}{\variety(L) + 2} \pad$ \\ \hline
	$\pad\depth(M)\pad$ & $\pad 1 \pad$ \\ \hline
	$\pad\diameter(M)\pad$ & $\pad 5 \pad$ \\ \hline
	$\pad\cyclicity(M)\pad$ & $\pad \frac{\variety(L) + 1}{\variety(L) + 3} \pad$ \\ \hline
	$\pad\netconn(M)\pad$ & $\pad \frac{2\variety(L) + 4}{\variety(L) + 5} \pad$ \\ \hline
	$\pad\density(M)\pad$ & $\pad \frac{1}{2} \pad$ \\ \hline
	$\pad\duplicate(M)\pad$ & $\pad 1 \pad$ \\ \hline
	$\pad\emptyseq(M)\pad$ & $\pad 0 \pad$ \\ \hline
	\end{tabular}
\end{table}

\subsection{Trace Net}
\label{sec:trace}
For a second baseline mining algorithm, we investigate the trace-net miner. 
This miner takes an event log $L$ as input and outputs the trace net, where every trace of $L$ corresponds to a unique path from an initial place $p_i$ to a final place $p_o$.
\cref{fig:trace-net} shows the trace net for an event log $L$ with $supp(L) = \{\sigma_1, \sigma_2, \dots, \sigma_n\}$.
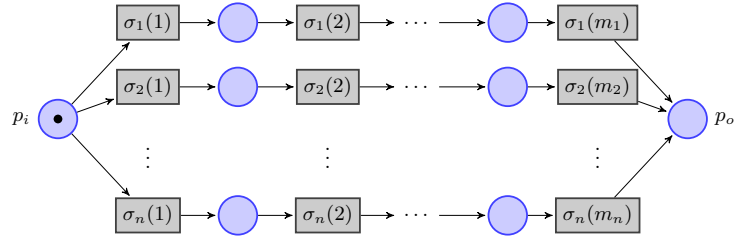
\begin{figure}[ht]
	\centering
	\scalebox{\scalefactor}{
	\begin{tikzpicture}[node distance = 1.4cm,>=stealth',bend angle=0,auto]
		\node[place,tokens=1,label=left:$p_i$] (start) {};
		\node[transition,right of=start,yshift=1.5cm] (sigma11) {$\sigma_1(1)$}
		edge [pre] (start);
		\node[transition,right of=start,yshift=0.5cm] (sigma21) {$\sigma_2(1)$}
		edge [pre] (start);
		\node[right of=start,yshift=-0.5cm] (dots) {$\vdots$};
		\node[transition,right of=start,yshift=-1.5cm] (sigman1) {$\sigma_n(1)$}
		edge [pre] (start);
		\node[place,right of=sigma11] (p11) {}
		edge [pre] (sigma11);
		\node[place,right of=sigma21] (p21) {}
		edge [pre] (sigma21);
		\node[right of=dots] (dots) {};
		\node[place,right of=sigman1] (pn1) {}
		edge [pre] (sigman1);
		\node[transition,right of=p11] (sigma12) {$\sigma_1(2)$}
		edge [pre] (p11);
		\node[transition,right of=p21] (sigma22) {$\sigma_2(2)$}
		edge [pre] (p21);
		\node[right of=dots] (dots) {$\vdots$};
		\node[transition,right of=pn1] (sigman2) {$\sigma_n(2)$}
		edge [pre] (pn1);
		\node[right of=sigma12] (dots1) {$\dots$}
		edge [pre] (sigma12);
		\node[right of=sigma22] (dots2) {$\dots$}
		edge [pre] (sigma22);
		\node[right of=dots] (dots) {};
		\node[right of=sigman2] (dotsn) {$\dots$}
		edge [pre] (sigman2);
		\node[place,right of=dots1] (p1m1) {}
		edge [pre] (dots1);
		\node[place,right of=dots2] (p2m2) {}
		edge [pre] (dots2);
		\node[right of=dots] (dots) {};
		\node[place,right of=dotsn] (pnmn) {}
		edge [pre] (dotsn);
		\node[transition,right of=p1m1] (sigma1m1) {$\sigma_1(m_1)$}
		edge [pre] (p1m1);
		\node[transition,right of=p2m2] (sigma2m2) {$\sigma_2(m_2)$}
		edge [pre] (p2m2);
		\node[right of=dots] (dots) {$\vdots$};
		\node[transition,right of=pnmn] (sigmanmn) {$\sigma_n(m_n)$}
		edge [pre] (pnmn);
		\node[place,right of=sigmanmn,yshift=1.5cm,label=right:$p_o$] (end) {}
		edge [pre] (sigma1m1)
		edge [pre] (sigma2m2)
		edge [pre] (sigmanmn);
	\end{tikzpicture}}
	\caption{The trace net for an event log $L$ with $supp(L) = \{\sigma_1, \sigma_2, \dots, \sigma_n\}$, where $|\sigma_i| =: m_i$ for all $i \in \{1, \dots, n\}$.}
	\label{fig:trace-net}
\end{figure}
In contrast to the flower model investigated in the previous subsection, the complexity of the trace net does not depend on the variety $\variety$ of the event log.
Instead, the amount of distinct traces in the event log, $\numberuniquetraces$, plays an important role in most model complexity scores for the trace net.
We will first observe that not all log complexity measures, an increase in log complexity means a change in the support of the event log.
Furthermore, we assume that there are no empty traces in the event log.

\begin{lemma}
\label{lemma:not-support-influencing}
Let $\mathcal{C}^L$ be a log complexity measure with $\mathcal{C}^L \in \{\magnitude, \support, \tlavg,$ $\lempelziv, \structure, \affinity, \deviationfromrandom, \avgdist,\seqentropy, \normseqentropy\}$.
Then, there are event logs $L_1, L_2$ with $L_1 \sqsubset L_2$ and with $\mathcal{C}^L(L_1) < \mathcal{C}^L(L_2)$, but $support(L_1) = support(L_2)$.
\end{lemma}
\begin{proof}
Consider the following event logs:
\begin{align*}
	L_1 &= [\langle a,b,c \rangle, \langle a,b,c,d \rangle^{2}, \langle a,b,c,d,e \rangle^{2}, \langle d,e,a,b \rangle] \\
	L_2 &= L_1 + [\langle a,b,c,d,e \rangle^{3}, \langle d,e,a,b \rangle^{3}]
\end{align*}
These two event logs have the following log complexity scores:
\begin{center}
	\begin{tabular}{|c|c|c|c|c|c|c|c|c|c|c|}\hline
		 & $\magnitude$ & $\variety$ & $\support$ & $\tlavg$ & $\tlmax$ & $\levelofdetail$ & $\numberofties$ & $\lempelziv$ & $\numberuniquetraces$ & $\percentageuniquetraces$ \\ \hline
		$L_1$ & $\pad 25 \pad$ & $\pad 5 \pad$ & $\pad 6 \pad$ & $\pad 4.1667 \pad$ & $\pad 5 \pad$ & $\pad 8 \pad$ & $\pad 5 \pad$ & $\pad 11 \pad$ & $\pad 4 \pad$ & $\pad 0.6667 \pad$ \\ \hline
		$L_2$ & $\pad 52 \pad$ & $\pad 5 \pad$ & $\pad 12 \pad$ & $\pad 4.3333 \pad$ & $\pad 5 \pad$ & $\pad 8 \pad$ & $\pad 5 \pad$ & $\pad 20 \pad$ & $\pad 4 \pad$ & $\pad 0.3333 \pad$ \\ \hline
	\end{tabular}
	
	\medskip
	
	\begin{tabular}{|c|c|c|c|c|c|c|c|c|} \hline
		 & $\structure$ & $\affinity$ & $\deviationfromrandom$ & $\avgdist$ & $\varentropy$ & $\normvarentropy$ & $\seqentropy$ & $\normseqentropy$ \\ \hline
		$L_1$ & $\pad 4.1667 \pad$ & $\pad 0.5856 \pad$ & $\pad 0.5517 \pad$ & $\pad 2.0667 \pad$ & $\pad 6.1827 \pad$ & $\pad 0.3126 \pad$ & $\pad 10.9917 \pad$ & $\pad 0.1366 \pad$ \\ \hline
		$L_2$ & $\pad 4.3333 \pad$ & $\pad 0.5899 \pad$ & $\pad 0.5743 \pad$ & $\pad 2.5152 \pad$ & $\pad 6.1827 \pad$ & $\pad 0.3126 \pad$ & $\pad 32.0966 \pad$ & $\pad 0.1562 \pad$ \\ \hline
	\end{tabular}
\end{center}
Thus, for all complexity measures $\mathcal{C}^L$ allowed by this theorem, we have that $\mathcal{C}^L(L_1) < \mathcal{C}^L(L_2)$.
Since, at the same time, $support(L_1) = support(L_2)$, these event logs prove the conjecture of this theorem. \hfill$\square$
\end{proof}

The fact that the complexity of the trace net is not dependent on the variety $\variety$ already shows that different mining algorithms require different log complexity measures to predict the complexity of their results.
For our analysis of the trace net, we first observe that some of its model complexity scores  must increase if more behavior is added to the underlying event log.
To avoid edge cases or trace nets where some complexity measures are undefined, we require for this entire subsection that $|supp(L)| > 1$ for any event log $L$. 
We allow this restriction, as event logs with just a single trace rarely occur in practice.

\begin{lemma}
\label{lemma:tracenet-monotone-increasing}
Let $L_1, L_2$ be event logs with $L_1 \sqsubset L_2$ and $|supp(L_1)| > 1$.
Let $M_1, M_2$ be the trace nets for $L_1$ and $L_2$. 
Then, $\mathcal{C}^M(M_1) \leq \mathcal{C}^M(M_2)$ for any model complexity measure $\mathcal{C}^M \in \{\size, \controlflow, \avgconn, \maxconn, \diameter, \duplicate\}$.
\end{lemma}
\begin{proof}
Let $L_1, L_2$ be two event logs with $L_1 \sqsubset L_2$. 
With this, we then know that $support(L_1) \subseteq support(L_2)$, since every unique trace in $L_1$ must also be present in $L_2$.
We abbreviate this observation by $(\star)$, and prove $\mathcal{C}^M(M_1) \leq \mathcal{C}^M(M_2)$ for each of the model complexity measures separately.
\begin{itemize}
	\item \textbf{Size $\size$:}
	The trace net contains the places $p_i$ and $p_o$, as well as a path of places and transitions for each trace of the event log.
	This means, in a trace net $M$ for an event log $L$, there are $\sum_{\sigma \in L} |\sigma|$ transitions and $2 + \sum_{\sigma \in L} (|\sigma| - 1)$ places.
	Thus, $\size(M) = 2 + \sum_{\sigma \in L} (2 |\sigma| - 1)$.
	Since $supp(L_1) \subseteq supp(L_2)$, this means:
	\begin{align*}
	\size(M_1) = 2 + \sum_{\sigma \in L_1} (2|\sigma| - 1) \overset{(\star)}{\leq} 2 + \sum_{\sigma \in L_2} (2|\sigma| - 1) = \size(M_2).
	\end{align*}

	\item \textbf{Control Flow Complexity $\controlflow$:}
	The only connector nodes in the trace net are $p_i$ and $p_o$.
	The node $p_i$ is a \texttt{xor}-split, while $p_o$ is a \texttt{xor}-join.
	In a trace net $M$ for an event log $L$, $p_i$ has $|supp(L)|$ outgoing edges, so we have $\controlflow(M) = |supp(L)|$, which means:
	\begin{align*}
	\controlflow(M_1) = |supp(L_1)| \overset{(\star)}{\leq} |supp(L_2)| = \controlflow(M_2).
	\end{align*}
	
	\item \textbf{Average Connector Degree $\avgconn$:}
	The only connector nodes in the trace net are $p_i$ and $p_o$.
	In a trace net $M$ for an event log $L$, $p_i$ and $p_o$ both have degree $|supp(L)|$, so $\avgconn(M) = \frac{1}{2} \cdot 2 \cdot |supp(L)| = |supp(L)|$, so we get:
	\begin{align*}
	\avgconn(M_1) = |supp(L_1)| \overset{(\star)}{\leq} |supp(L_2)| = \avgconn(M_2).
	\end{align*}
	
	\item \textbf{Maximum Connector Degree $\maxconn$:}
	The only connector nodes in the trace net are $p_i$ and $p_o$.
	In a trace net $M$ for an event log $L$, $p_i$ and $p_o$ both have degree $|supp(L)|$, so $\maxconn(M) = |supp(L)|$, leading to:
	\begin{align*}
	\maxconn(M_1) = |supp(L_1)| \overset{(\star)}{\leq} |supp(L_2)| = \maxconn(M_2).
	\end{align*}
	
	\item \textbf{Diameter $\diameter$:}
	In the trace net $M$ for an event log $L$, every trace $\sigma \in L$ creates a unique path $(p_i, \sigma(1), \dots, \sigma(|\sigma|), p_o)$ of length $2 \cdot |\sigma| + 1$.
	Thus, the length of the longest path in $M$ is $\diameter(M) = 2 \tlmax(L) + 1$.
	Since all traces in $L_1$ are also present in $L_2$, this means:
	\begin{align*}
	\diameter(M_1) = 2 \tlmax(L_1) + 1 \overset{(\star)}{\leq} 2 \tlmax(L_2) + 1 = \diameter(M_2).
	\end{align*}
	
	\item \textbf{Number of Duplicate Tasks $\duplicate$:}
	The number of duplicate tasks in the trace net $M$ for an event log $L$ is exactly the amount of activity name repetitions in the support of the event log $L$.
	Since $supp(L_1) \subseteq supp(L_2)$, this amount of repetitions can only be higher in $L_2$ than in $L_1$, so we get $\duplicate(M_1) \leq \duplicate(M_2)$.
\end{itemize}
Thus, we showed that $\mathcal{C}^M(M_1) \leq \mathcal{C}^M(M_2)$ for any model complexity measure $\mathcal{C}^M \in \{\size, \controlflow, \avgconn, \maxconn, \diameter, \duplicate\}$. \hfill$\square$
\end{proof}

Like for the flower model, there are some model complexity measures that always return the same value for a trace net. 
We will investigate these complexity measures in the next Lemma.

\begin{lemma}
\label{lemma:tracenet-constant-complexity}
Let $L_1, L_2$ be event logs and $M_1, M_2$ be the trace nets for $L_1$ and $L_2$.
Then, $\mathcal{C}^M(M_1) = \mathcal{C}^M(M_2)$, where $\mathcal{C}^M \in \{\mismatch, \connhet, \tokensplit, \separability, \depth, \cyclicity, \emptyseq\}$.
\end{lemma}
\begin{proof}
Let $L_1, L_2, M_1, M_2$ and $\mathcal{C}^M$ be defined as stated by the theorem.
We prove $\mathcal{C}^M(M_1) = \mathcal{C}^M(M_2)$ for each of the model complexity measures separately:
\begin{itemize}
	\item \textbf{Connector Mismatch $\mismatch$:}
	The trace net $M$ for an event log $L$ contains exactly two connectors: $p_i$ and $p_o$.
	$p_i$ has exactly $|supp(L)|$ outgoing edges, and $p_o$ has exactly $|supp(L)|$ incoming arcs, so its connector mismatch score is $\mismatch(M) = \left| |supp(L)| - |supp(L)| \right| = 0$.
	Therefore, we have that $\mismatch(M_1) = 0 = \mismatch(M_2)$.
	
	\item \textbf{Connector Heterogeneity $\connhet$:}
	The trace net $M$ for an event log $L$ has only the connectors $p_i$ and $p_o$. 
	Both of these connectors are \texttt{xor}-connectors, so $\connhet(M) = - (1 \cdot \log_2(1) + 0 \cdot \log_2(0)) = 0$.
	In turn, we know that $\connhet(M_1) = 0 = \connhet(M_2)$.
	
	\item \textbf{Token Split $\tokensplit$:}
	Every transition in the trace net $M$ for an event log $L$ has exactly one incoming and one outgoing edge.
	Therefore, there are no transitions in $M$ with more than one outgoing edge, leading to $\tokensplit(M) = 0$.
	Therefore, we get $\tokensplit(M_1) = 0 = \tokensplit(M_2)$.
	
	\item \textbf{Separability $\separability$:}
	Since we require $|supp(L_1)| > 1$, we know that $M_1$ does not contain any cut-vertices.
	$M_2$ also does not contain any cut-vertices, as $|supp(L_2)| \geq |supp(L_1) > 1$.
	Therefore, $\separability(M_1) = 1 = \separability(M_2)$.
	
	\item \textbf{Depth $\depth$:}
	In the trace net $M$ for an event log $L$, all nodes except $p_i$ and $p_o$ have in- and out-depth $1$, since $p_i$ and $p_o$ are connectors.
	$p_i$ and $p_o$ themselves, on the other hand, both have in- and out-depth $0$.
	Therefore, $\depth(M) = 1$ and, consequently, $\depth(M_1) = 1 = \depth(M_2)$.
	
	\item \textbf{Cyclicity $\cyclicity$:}
	The trace net $M$ for an event log $L$ does not contain any cycles, so $\cyclicity(M) = 0$.
	In turn, $\cyclicity(M_1) = 0 = \cyclicity(M_2)$.
	
	\item \textbf{Number of Empty Sequence Flows $\emptyseq$:}
	In the trace net $M$ for an event log $L$, every transition has exactly one incoming and one outgoing edge.
	Therefore, there are no \texttt{and}-connectors in $M$, which means $\emptyseq(M) = 0$.
	Consequently, $\emptyseq(M_1) = 0 = \emptyseq(M_2)$.
\end{itemize}
Thus, we showed that $\mathcal{C}^M(M_1) = \mathcal{C}^M(M_2)$ for any model complexity measure $\mathcal{C}^M \in \{\mismatch, \connhet, \tokensplit, \separability, \depth, \cyclicity, \emptyseq\}$. \hfill$\square$
\end{proof}

With these observations, we can now analyze the relations between log and model complexity for the trace net miner. 
We start by showing the results in \cref{table:tracenet-findings} and prove the relations shown in the table afterwards.
\begin{table}[ht]
	\caption{The relations between the complexity scores of two trace nets $M_1$ and $M_2$ that were found for the event logs $L_1$ and $L_2$ respectively, where $L_1 \sqsubset L_2$, $|supp(L_1)| > 1$, and the complexity of $L_1$ is lower than the complexity of $L_2$.}
	\label{table:tracenet-findings}
	\centering
	\resizebox{\textwidth}{!}{
	\begin{tabular}{|c|c|c|c|c|c|c|c|c|c|c|c|c|c|c|c|c|c|} \hline
		 & $\size$ & $\mismatch$ & $\connhet$ & $\crossconn$ & $\tokensplit$ & $\controlflow$ & $\separability$ & $\avgconn$ & $\maxconn$ & $\sequentiality$ & $\depth$ & $\diameter$ & $\cyclicity$ & $\netconn$ & $\density$ & $\duplicate$ & $\emptyseq$ \\ \hline
		$\magnitude$ & \hyperref[theo:tracenet-leq-entries]{$\mleq$} & \hyperref[theo:tracenet-equals-entries]{$\meq$} & \hyperref[theo:tracenet-equals-entries]{$\meq$} & \hyperref[theo:tracenet-crossconn-entries]{$\norel^*$} & \hyperref[theo:tracenet-equals-entries]{$\meq$} & \hyperref[theo:tracenet-leq-entries]{$\mleq$} & \hyperref[theo:tracenet-equals-entries]{$\meq$} & \hyperref[theo:tracenet-leq-entries]{$\mleq$} & \hyperref[theo:tracenet-leq-entries]{$\mleq$} & \hyperref[theo:tracenet-sequentiality-entries]{$\norel$} & \hyperref[theo:tracenet-equals-entries]{$\meq$} & \hyperref[theo:tracenet-diameter-leq-entries]{$\mleq$} & \hyperref[theo:tracenet-equals-entries]{$\meq$} & \hyperref[theo:tracenet-netconn-entries]{$\norel$} & \hyperref[theo:tracenet-density-geq-entries]{$\mgeq$} & \hyperref[theo:tracenet-duplicate-entries]{$\mleq$} & \hyperref[theo:tracenet-equals-entries]{$\meq$} \\ \hline
		
		$\variety$ & \hyperref[theo:tracnet-less-entries]{$\mless$} & \hyperref[theo:tracenet-equals-entries]{$\meq$} & \hyperref[theo:tracenet-equals-entries]{$\meq$} & \hyperref[theo:tracenet-crossconn-entries]{$\norel^*$} & \hyperref[theo:tracenet-equals-entries]{$\meq$} & \hyperref[theo:tracnet-less-entries]{$\mless$} & \hyperref[theo:tracenet-equals-entries]{$\meq$} & \hyperref[theo:tracnet-less-entries]{$\mless$} & \hyperref[theo:tracnet-less-entries]{$\mless$} & \hyperref[theo:tracenet-sequentiality-entries]{$\norel$} & \hyperref[theo:tracenet-equals-entries]{$\meq$} & \hyperref[theo:tracenet-diameter-leq-entries]{$\mleq$} & \hyperref[theo:tracenet-equals-entries]{$\meq$} & \hyperref[theo:tracenet-netconn-entries]{$\norel$} & \hyperref[theo:tracenet-density-geq-entries]{$\mgeq$} & \hyperref[theo:tracenet-duplicate-entries]{$\mleq$} & \hyperref[theo:tracenet-equals-entries]{$\meq$} \\ \hline
		
		$\support$ & \hyperref[theo:tracenet-leq-entries]{$\mleq$} & \hyperref[theo:tracenet-equals-entries]{$\meq$} & \hyperref[theo:tracenet-equals-entries]{$\meq$} & \hyperref[theo:tracenet-crossconn-entries]{$\norel^*$} & \hyperref[theo:tracenet-equals-entries]{$\meq$} & \hyperref[theo:tracenet-leq-entries]{$\mleq$} & \hyperref[theo:tracenet-equals-entries]{$\meq$} & \hyperref[theo:tracenet-leq-entries]{$\mleq$} & \hyperref[theo:tracenet-leq-entries]{$\mleq$} & \hyperref[theo:tracenet-sequentiality-entries]{$\norel$} & \hyperref[theo:tracenet-equals-entries]{$\meq$} & \hyperref[theo:tracenet-diameter-leq-entries]{$\mleq$} & \hyperref[theo:tracenet-equals-entries]{$\meq$} & \hyperref[theo:tracenet-netconn-entries]{$\norel$} & \hyperref[theo:tracenet-density-geq-entries]{$\mgeq$} & \hyperref[theo:tracenet-duplicate-entries]{$\mleq$} & \hyperref[theo:tracenet-equals-entries]{$\meq$} \\ \hline
		
		$\tlavg$ & \hyperref[theo:tracenet-leq-entries]{$\mleq$} & \hyperref[theo:tracenet-equals-entries]{$\meq$} & \hyperref[theo:tracenet-equals-entries]{$\meq$} & \hyperref[theo:tracenet-crossconn-entries]{$\norel^*$} & \hyperref[theo:tracenet-equals-entries]{$\meq$} & \hyperref[theo:tracenet-leq-entries]{$\mleq$} & \hyperref[theo:tracenet-equals-entries]{$\meq$} & \hyperref[theo:tracenet-leq-entries]{$\mleq$} & \hyperref[theo:tracenet-leq-entries]{$\mleq$} & \hyperref[theo:tracenet-sequentiality-entries]{$\norel$} & \hyperref[theo:tracenet-equals-entries]{$\meq$} & \hyperref[theo:tracenet-diameter-leq-entries]{$\mleq$} & \hyperref[theo:tracenet-equals-entries]{$\meq$} & \hyperref[theo:tracenet-netconn-entries]{$\norel$} & \hyperref[theo:tracenet-density-geq-entries]{$\mgeq$} & \hyperref[theo:tracenet-duplicate-entries]{$\mleq$} & \hyperref[theo:tracenet-equals-entries]{$\meq$} \\ \hline
		
		$\tlmax$ & \hyperref[theo:tracnet-less-entries]{$\mless$} & \hyperref[theo:tracenet-equals-entries]{$\meq$} & \hyperref[theo:tracenet-equals-entries]{$\meq$} & \hyperref[theo:tracenet-crossconn-entries]{$\norel^*$} & \hyperref[theo:tracenet-equals-entries]{$\meq$} & \hyperref[theo:tracnet-less-entries]{$\mless$} & \hyperref[theo:tracenet-equals-entries]{$\meq$} & \hyperref[theo:tracnet-less-entries]{$\mless$} & \hyperref[theo:tracnet-less-entries]{$\mless$} & \hyperref[theo:tracenet-sequentiality-entries]{$\norel$} & \hyperref[theo:tracenet-equals-entries]{$\meq$} & \hyperref[theo:tracenet-diameter-less-entries]{$\mless$} & \hyperref[theo:tracenet-equals-entries]{$\meq$} & \hyperref[theo:tracenet-netconn-entries]{$\norel$} & \hyperref[theo:tracenet-density-greater-entries]{$\mgreater$} & \hyperref[theo:tracenet-duplicate-entries]{$\mleq$} & \hyperref[theo:tracenet-equals-entries]{$\meq$} \\ \hline
		
		$\levelofdetail$ & \hyperref[theo:tracnet-less-entries]{$\mless$} & \hyperref[theo:tracenet-equals-entries]{$\meq$} & \hyperref[theo:tracenet-equals-entries]{$\meq$} & \hyperref[theo:tracenet-crossconn-entries]{$\norel^*$} & \hyperref[theo:tracenet-equals-entries]{$\meq$} & \hyperref[theo:tracnet-less-entries]{$\mless$} & \hyperref[theo:tracenet-equals-entries]{$\meq$} & \hyperref[theo:tracnet-less-entries]{$\mless$} & \hyperref[theo:tracnet-less-entries]{$\mless$} & \hyperref[theo:tracenet-sequentiality-entries]{$\norel$} & \hyperref[theo:tracenet-equals-entries]{$\meq$} & \hyperref[theo:tracenet-diameter-leq-entries]{$\mleq$} & \hyperref[theo:tracenet-equals-entries]{$\meq$} & \hyperref[theo:tracenet-netconn-entries]{$\norel$} & \hyperref[theo:tracenet-density-geq-entries]{$\mgeq$} & \hyperref[theo:tracenet-duplicate-entries]{$\mleq$} & \hyperref[theo:tracenet-equals-entries]{$\meq$} \\ \hline
		
		$\numberofties$ & \hyperref[theo:tracnet-less-entries]{$\mless$} & \hyperref[theo:tracenet-equals-entries]{$\meq$} & \hyperref[theo:tracenet-equals-entries]{$\meq$} & \hyperref[theo:tracenet-crossconn-entries]{$\norel^*$} & \hyperref[theo:tracenet-equals-entries]{$\meq$} & \hyperref[theo:tracnet-less-entries]{$\mless$} & \hyperref[theo:tracenet-equals-entries]{$\meq$} & \hyperref[theo:tracnet-less-entries]{$\mless$} & \hyperref[theo:tracnet-less-entries]{$\mless$} & \hyperref[theo:tracenet-sequentiality-entries]{$\norel$} & \hyperref[theo:tracenet-equals-entries]{$\meq$} & \hyperref[theo:tracenet-diameter-leq-entries]{$\mleq$} & \hyperref[theo:tracenet-equals-entries]{$\meq$} & \hyperref[theo:tracenet-netconn-entries]{$\norel$} & \hyperref[theo:tracenet-density-greater-entries]{$\mgreater$} & \hyperref[theo:tracenet-duplicate-entries]{$\mleq$} & \hyperref[theo:tracenet-equals-entries]{$\meq$} \\ \hline
		
		$\lempelziv$ & \hyperref[theo:tracenet-leq-entries]{$\mleq$} & \hyperref[theo:tracenet-equals-entries]{$\meq$} & \hyperref[theo:tracenet-equals-entries]{$\meq$} & \hyperref[theo:tracenet-crossconn-entries]{$\norel^*$} & \hyperref[theo:tracenet-equals-entries]{$\meq$} & \hyperref[theo:tracenet-leq-entries]{$\mleq$} & \hyperref[theo:tracenet-equals-entries]{$\meq$} & \hyperref[theo:tracenet-leq-entries]{$\mleq$} & \hyperref[theo:tracenet-leq-entries]{$\mleq$} & \hyperref[theo:tracenet-sequentiality-entries]{$\norel$} & \hyperref[theo:tracenet-equals-entries]{$\meq$} & \hyperref[theo:tracenet-diameter-leq-entries]{$\mleq$} & \hyperref[theo:tracenet-equals-entries]{$\meq$} & \hyperref[theo:tracenet-netconn-entries]{$\norel$} & \hyperref[theo:tracenet-density-geq-entries]{$\mgeq$} & \hyperref[theo:tracenet-duplicate-entries]{$\mleq$} & \hyperref[theo:tracenet-equals-entries]{$\meq$} \\ \hline
		
		$\numberuniquetraces$ & \hyperref[theo:tracnet-less-entries]{$\mless$} & \hyperref[theo:tracenet-equals-entries]{$\meq$} & \hyperref[theo:tracenet-equals-entries]{$\meq$} & \hyperref[theo:tracenet-crossconn-entries]{$\norel^*$} & \hyperref[theo:tracenet-equals-entries]{$\meq$} & \hyperref[theo:tracnet-less-entries]{$\mless$} & \hyperref[theo:tracenet-equals-entries]{$\meq$} & \hyperref[theo:tracnet-less-entries]{$\mless$} & \hyperref[theo:tracnet-less-entries]{$\mless$} & \hyperref[theo:tracenet-sequentiality-entries]{$\norel$} & \hyperref[theo:tracenet-equals-entries]{$\meq$} & \hyperref[theo:tracenet-diameter-leq-entries]{$\mleq$} & \hyperref[theo:tracenet-equals-entries]{$\meq$} & \hyperref[theo:tracenet-netconn-entries]{$\norel$} & \hyperref[theo:tracenet-density-geq-entries]{$\mgeq$} & \hyperref[theo:tracenet-duplicate-entries]{$\mleq$} & \hyperref[theo:tracenet-equals-entries]{$\meq$} \\ \hline
		
		$\percentageuniquetraces$ & \hyperref[theo:tracnet-less-entries]{$\mless$} & \hyperref[theo:tracenet-equals-entries]{$\meq$} & \hyperref[theo:tracenet-equals-entries]{$\meq$} & \hyperref[theo:tracenet-crossconn-entries]{$\norel^*$} & \hyperref[theo:tracenet-equals-entries]{$\meq$} & \hyperref[theo:tracnet-less-entries]{$\mless$} & \hyperref[theo:tracenet-equals-entries]{$\meq$} & \hyperref[theo:tracnet-less-entries]{$\mless$} & \hyperref[theo:tracnet-less-entries]{$\mless$} & \hyperref[theo:tracenet-sequentiality-entries]{$\norel$} & \hyperref[theo:tracenet-equals-entries]{$\meq$} & \hyperref[theo:tracenet-diameter-leq-entries]{$\mleq$} & \hyperref[theo:tracenet-equals-entries]{$\meq$} & \hyperref[theo:tracenet-netconn-entries]{$\norel$} & \hyperref[theo:tracenet-density-geq-entries]{$\mgeq$} & \hyperref[theo:tracenet-duplicate-entries]{$\mleq$} & \hyperref[theo:tracenet-equals-entries]{$\meq$} \\ \hline
		
		$\structure$ & \hyperref[theo:tracenet-leq-entries]{$\mleq$} & \hyperref[theo:tracenet-equals-entries]{$\meq$} & \hyperref[theo:tracenet-equals-entries]{$\meq$} & \hyperref[theo:tracenet-crossconn-entries]{$\norel^*$} & \hyperref[theo:tracenet-equals-entries]{$\meq$} & \hyperref[theo:tracenet-leq-entries]{$\mleq$} & \hyperref[theo:tracenet-equals-entries]{$\meq$} & \hyperref[theo:tracenet-leq-entries]{$\mleq$} & \hyperref[theo:tracenet-leq-entries]{$\mleq$} & \hyperref[theo:tracenet-sequentiality-entries]{$\norel$} & \hyperref[theo:tracenet-equals-entries]{$\meq$} & \hyperref[theo:tracenet-diameter-leq-entries]{$\mleq$} & \hyperref[theo:tracenet-equals-entries]{$\meq$} & \hyperref[theo:tracenet-netconn-entries]{$\norel$} & \hyperref[theo:tracenet-density-geq-entries]{$\mgeq$} & \hyperref[theo:tracenet-duplicate-entries]{$\mleq$} & \hyperref[theo:tracenet-equals-entries]{$\meq$} \\ \hline
		
		$\affinity$ & \hyperref[theo:tracenet-leq-entries]{$\mleq$} & \hyperref[theo:tracenet-equals-entries]{$\meq$} & \hyperref[theo:tracenet-equals-entries]{$\meq$} & \hyperref[theo:tracenet-crossconn-entries]{$\norel^*$} & \hyperref[theo:tracenet-equals-entries]{$\meq$} & \hyperref[theo:tracenet-leq-entries]{$\mleq$} & \hyperref[theo:tracenet-equals-entries]{$\meq$} & \hyperref[theo:tracenet-leq-entries]{$\mleq$} & \hyperref[theo:tracenet-leq-entries]{$\mleq$} & \hyperref[theo:tracenet-sequentiality-entries]{$\norel$} & \hyperref[theo:tracenet-equals-entries]{$\meq$} & \hyperref[theo:tracenet-diameter-leq-entries]{$\mleq$} & \hyperref[theo:tracenet-equals-entries]{$\meq$} & \hyperref[theo:tracenet-netconn-entries]{$\norel$} & \hyperref[theo:tracenet-density-geq-entries]{$\mgeq$} & \hyperref[theo:tracenet-duplicate-entries]{$\mleq$} & \hyperref[theo:tracenet-equals-entries]{$\meq$} \\ \hline
		
		$\deviationfromrandom$ & \hyperref[theo:tracenet-leq-entries]{$\mleq$} & \hyperref[theo:tracenet-equals-entries]{$\meq$} & \hyperref[theo:tracenet-equals-entries]{$\meq$} & \hyperref[theo:tracenet-crossconn-entries]{$\norel^*$} & \hyperref[theo:tracenet-equals-entries]{$\meq$} & \hyperref[theo:tracenet-leq-entries]{$\mleq$} & \hyperref[theo:tracenet-equals-entries]{$\meq$} & \hyperref[theo:tracenet-leq-entries]{$\mleq$} & \hyperref[theo:tracenet-leq-entries]{$\mleq$} & \hyperref[theo:tracenet-sequentiality-entries]{$\norel$} & \hyperref[theo:tracenet-equals-entries]{$\meq$} & \hyperref[theo:tracenet-diameter-leq-entries]{$\mleq$} & \hyperref[theo:tracenet-equals-entries]{$\meq$} & \hyperref[theo:tracenet-netconn-entries]{$\norel$} & \hyperref[theo:tracenet-density-geq-entries]{$\mgeq$} & \hyperref[theo:tracenet-duplicate-entries]{$\mleq$} & \hyperref[theo:tracenet-equals-entries]{$\meq$} \\ \hline
		
		$\avgdist$ & \hyperref[theo:tracenet-leq-entries]{$\mleq$} & \hyperref[theo:tracenet-equals-entries]{$\meq$} & \hyperref[theo:tracenet-equals-entries]{$\meq$} & \hyperref[theo:tracenet-crossconn-entries]{$\norel^*$} & \hyperref[theo:tracenet-equals-entries]{$\meq$} & \hyperref[theo:tracenet-leq-entries]{$\mleq$} & \hyperref[theo:tracenet-equals-entries]{$\meq$} & \hyperref[theo:tracenet-leq-entries]{$\mleq$} & \hyperref[theo:tracenet-leq-entries]{$\mleq$} & \hyperref[theo:tracenet-sequentiality-entries]{$\norel$} & \hyperref[theo:tracenet-equals-entries]{$\meq$} & \hyperref[theo:tracenet-diameter-leq-entries]{$\mleq$} & \hyperref[theo:tracenet-equals-entries]{$\meq$} & \hyperref[theo:tracenet-netconn-entries]{$\norel$} & \hyperref[theo:tracenet-density-geq-entries]{$\mgeq$} & \hyperref[theo:tracenet-duplicate-entries]{$\mleq$} & \hyperref[theo:tracenet-equals-entries]{$\meq$} \\ \hline
		
		$\varentropy$ & \hyperref[theo:tracnet-less-entries]{$\mless$} & \hyperref[theo:tracenet-equals-entries]{$\meq$} & \hyperref[theo:tracenet-equals-entries]{$\meq$} & \hyperref[theo:tracenet-crossconn-entries]{$\norel^*$} & \hyperref[theo:tracenet-equals-entries]{$\meq$} & \hyperref[theo:tracnet-less-entries]{$\mless$} & \hyperref[theo:tracenet-equals-entries]{$\meq$} & \hyperref[theo:tracnet-less-entries]{$\mless$} & \hyperref[theo:tracnet-less-entries]{$\mless$} & \hyperref[theo:tracenet-sequentiality-entries]{$\norel$} & \hyperref[theo:tracenet-equals-entries]{$\meq$} & \hyperref[theo:tracenet-diameter-leq-entries]{$\mleq$} & \hyperref[theo:tracenet-equals-entries]{$\meq$} & \hyperref[theo:tracenet-netconn-entries]{$\norel$} & \hyperref[theo:tracenet-density-geq-entries]{$\mgeq$} & \hyperref[theo:tracenet-duplicate-entries]{$\mleq$} & \hyperref[theo:tracenet-equals-entries]{$\meq$} \\ \hline
		
		$\normvarentropy$ & \hyperref[theo:tracnet-less-entries]{$\mless$} & \hyperref[theo:tracenet-equals-entries]{$\meq$} & \hyperref[theo:tracenet-equals-entries]{$\meq$} & \hyperref[theo:tracenet-crossconn-entries]{$\norel^*$} & \hyperref[theo:tracenet-equals-entries]{$\meq$} & \hyperref[theo:tracnet-less-entries]{$\mless$} & \hyperref[theo:tracenet-equals-entries]{$\meq$} & \hyperref[theo:tracnet-less-entries]{$\mless$} & \hyperref[theo:tracnet-less-entries]{$\mless$} & \hyperref[theo:tracenet-sequentiality-entries]{$\norel$} & \hyperref[theo:tracenet-equals-entries]{$\meq$} & \hyperref[theo:tracenet-diameter-leq-entries]{$\mleq$} & \hyperref[theo:tracenet-equals-entries]{$\meq$} & \hyperref[theo:tracenet-netconn-entries]{$\norel$} & \hyperref[theo:tracenet-density-geq-entries]{$\mgeq$} & \hyperref[theo:tracenet-duplicate-entries]{$\mleq$} & \hyperref[theo:tracenet-equals-entries]{$\meq$} \\ \hline
		
		$\seqentropy$ & \hyperref[theo:tracenet-leq-entries]{$\mleq$} & \hyperref[theo:tracenet-equals-entries]{$\meq$} & \hyperref[theo:tracenet-equals-entries]{$\meq$} & \hyperref[theo:tracenet-crossconn-entries]{$\norel^*$} & \hyperref[theo:tracenet-equals-entries]{$\meq$} & \hyperref[theo:tracenet-leq-entries]{$\mleq$} & \hyperref[theo:tracenet-equals-entries]{$\meq$} & \hyperref[theo:tracenet-leq-entries]{$\mleq$} & \hyperref[theo:tracenet-leq-entries]{$\mleq$} & \hyperref[theo:tracenet-sequentiality-entries]{$\norel$} & \hyperref[theo:tracenet-equals-entries]{$\meq$} & \hyperref[theo:tracenet-diameter-leq-entries]{$\mleq$} & \hyperref[theo:tracenet-equals-entries]{$\meq$} & \hyperref[theo:tracenet-netconn-entries]{$\norel$} & \hyperref[theo:tracenet-density-geq-entries]{$\mgeq$} & \hyperref[theo:tracenet-duplicate-entries]{$\mleq$} & \hyperref[theo:tracenet-equals-entries]{$\meq$} \\ \hline
		
		$\normseqentropy$ & \hyperref[theo:tracenet-leq-entries]{$\mleq$} & \hyperref[theo:tracenet-equals-entries]{$\meq$} & \hyperref[theo:tracenet-equals-entries]{$\meq$} & \hyperref[theo:tracenet-crossconn-entries]{$\norel^*$} & \hyperref[theo:tracenet-equals-entries]{$\meq$} & \hyperref[theo:tracenet-leq-entries]{$\mleq$} & \hyperref[theo:tracenet-equals-entries]{$\meq$} & \hyperref[theo:tracenet-leq-entries]{$\mleq$} & \hyperref[theo:tracenet-leq-entries]{$\mleq$} & \hyperref[theo:tracenet-sequentiality-entries]{$\norel$} & \hyperref[theo:tracenet-equals-entries]{$\meq$} & \hyperref[theo:tracenet-diameter-leq-entries]{$\mleq$} & \hyperref[theo:tracenet-equals-entries]{$\meq$} & \hyperref[theo:tracenet-netconn-entries]{$\norel$} & \hyperref[theo:tracenet-density-geq-entries]{$\mgeq$} & \hyperref[theo:tracenet-duplicate-entries]{$\mleq$} & \hyperref[theo:tracenet-equals-entries]{$\meq$} \\ \hline
	\end{tabular}}
	{\scriptsize ${}^*$ We did not find examples showing that $\mathcal{C}^L(L_1) < \mathcal{C}^L(L_2)$ and $\crossconn(M_1) = \crossconn(M_2)$ is possible.}
\end{table}
For quick navigation, the PDF-version of this paper enables its readers to click on the entries of the table to jump to the proof of the respective property.

\begin{theorem}
\label{theo:tracenet-leq-entries}
$(\mathcal{C}^L, \mathcal{C}^M) \in \mleq$ for any log cmplexity measure $\mathcal{C}^L \in \{\magnitude, \support,$ $\tlavg, \lempelziv, \structure, \affinity, \deviationfromrandom, \avgdist,\seqentropy, \normseqentropy\}$ and a model complexity measure $\mathcal{C}^M \in \{\size, \controlflow, \avgconn, \maxconn\}$.
\end{theorem}
\begin{proof}
Let $L_1, L_2,$ be event logs with $L_1 \sqsubset L_2$ and $|supp(L_1)| > 1$, and $M_1, M_2$ be the trace nets for $L_1$ and $L_2$.
By \cref{lemma:tracenet-monotone-increasing}, we know that $L_1 \sqsubset L_2$ and $|supp(L_1)| > 1$ implies $\mathcal{C}^M(M_1) \leq \mathcal{C}^M(M_2)$.
We now need to show that both $\mathcal{C}^M(M_1) < \mathcal{C}^M(M_2)$ and $\mathcal{C}^M(M_1) = \mathcal{C}^M(M_2)$ are possible.
For the former, take the following event logs:
\begin{align*}
	L_1 &= [\langle a,b,c \rangle, \langle a,b,c,d \rangle^{2}, \langle a,b,c,d,e \rangle^{2}, \langle d,e,a,b \rangle] \\
	L_2 &= L_1 + [\langle a,b,c,d,e,f \rangle, \langle a,a,b,c,d,e,f \rangle, \langle a,b,c,d,e,a,b \rangle]
\end{align*}
These two event logs have the following log complexity scores:
\begin{center}
	\begin{tabular}{|c|c|c|c|c|c|c|c|c|c|c|}\hline
		 & $\magnitude$ & $\variety$ & $\support$ & $\tlavg$ & $\tlmax$ & $\levelofdetail$ & $\numberofties$ & $\lempelziv$ & $\numberuniquetraces$ & $\percentageuniquetraces$ \\ \hline
		$L_1$ & $\pad 25 \pad$ & $\pad 5 \pad$ & $\pad 6 \pad$ & $\pad 4.1667 \pad$ & $\pad 5 \pad$ & $\pad 8 \pad$ & $\pad 5 \pad$ & $\pad 11 \pad$ & $\pad 4 \pad$ & $\pad 0.6667 \pad$ \\ \hline
		$L_2$ & $\pad 45 \pad$ & $\pad 6 \pad$ & $\pad 9 \pad$ & $\pad 5 \pad$ & $\pad 7 \pad$ & $\pad 10 \pad$ & $\pad 6 \pad$ & $\pad 18 \pad$ & $\pad 7 \pad$ & $\pad 0.7778 \pad$ \\ \hline
	\end{tabular}
	
	\medskip
	
	\begin{tabular}{|c|c|c|c|c|c|c|c|c|} \hline
		 & $\structure$ & $\affinity$ & $\deviationfromrandom$ & $\avgdist$ & $\varentropy$ & $\normvarentropy$ & $\seqentropy$ & $\normseqentropy$ \\ \hline
		$L_1$ & $\pad 4.1667 \pad$ & $\pad 0.5856 \pad$ & $\pad 0.5517 \pad$ & $\pad 2.0667 \pad$ & $\pad 6.1827 \pad$ & $\pad 0.3126 \pad$ & $\pad 10.9917 \pad$ & $\pad 0.1366 \pad$ \\ \hline
		$L_2$ & $\pad 4.6667 \pad$ & $\pad 0.5872 \pad$ & $\pad 0.5861 \pad$ & $\pad 2.5556 \pad$ & $\pad 23.5941 \pad$ & $\pad 0.4535 \pad$ & $\pad 38.233 \pad$ & $\pad 0.2232 \pad$ \\ \hline
	\end{tabular}
\end{center}
Thus, $\mathcal{C}^L(L_1) < \mathcal{C}^L(L_2)$ for any of the log complexity measures allowed by this theorem.
The trace nets for $L_1$ and $L_2$ are shown in \cref{fig:tracenet-leq-entries}.
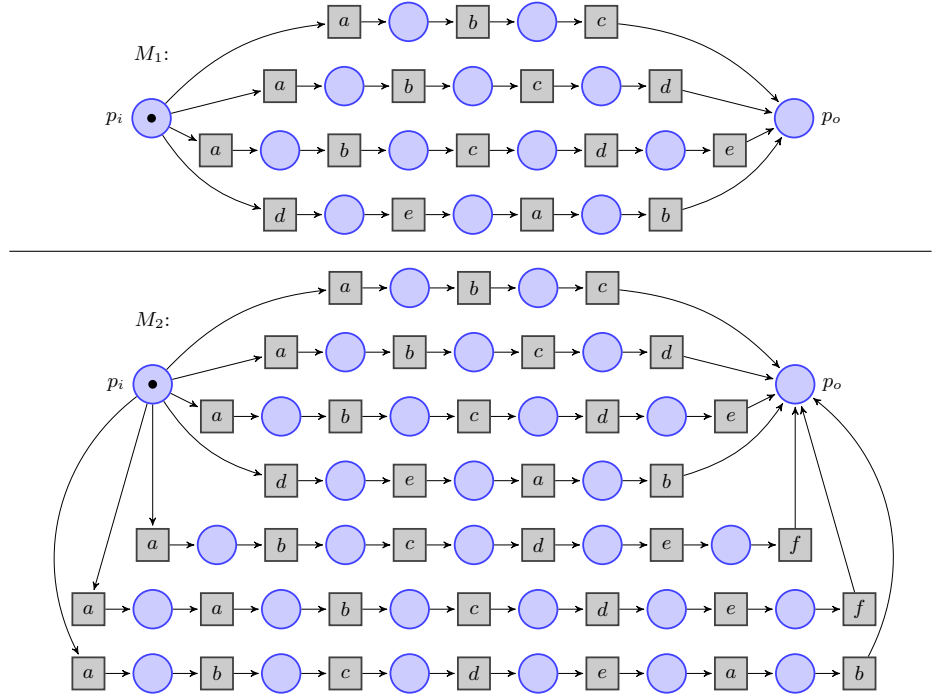
\begin{figure}[ht]
	\centering
	\scalebox{\scalefactor}{
	\begin{tikzpicture}[node distance = 1cm,>=stealth',bend angle=0,auto]
		\node[place,tokens=1,label=left:$p_i$] (start) {};
		\node[above of=start] {$M_1$:};
		\node[transition,right of=start,yshift=-0.5cm] (a) {$a$}
		edge [pre] (start);
		\node[place,right of=a] (p1) {}
		edge [pre] (a);
		\node[transition,right of=p1] (b) {$b$}
		edge [pre] (p1);
		\node[place,right of=b] (p2) {}
		edge [pre] (b);
		\node[transition,right of=p2] (c) {$c$}
		edge [pre] (p2);
		\node[place,right of=c] (p3) {}
		edge [pre] (c);
		\node[transition,right of=p3] (d) {$d$}
		edge [pre] (p3);
		\node[place,right of=d] (p4) {}
		edge [pre] (d);
		\node[transition,right of=p4] (e) {$e$}
		edge [pre] (p4);
		\node[place,right of=e,yshift=0.5cm,label=right:$p_o$] (end) {}
		edge [pre] (e);
		\node[transition,below of=p1] (d2) {$d$}
		edge [pre,bend left=20] (start);
		\node[place,right of=d2] (p12) {}
		edge [pre] (d2);
		\node[transition,right of=p12] (e2) {$e$}
		edge [pre] (p12);
		\node[place,right of=e2] (p22) {}
		edge [pre] (e2);
		\node[transition,right of=p22] (a2) {$a$}
		edge [pre] (p22);
		\node[place,right of=a2] (p32) {}
		edge [pre] (a2);
		\node[transition,right of=p32] (b2) {$b$}
		edge [pre] (p32)
		edge [post,bend right=20] (end);
		\node[transition,above of=p1] (a3) {$a$}
		edge [pre] (start);
		\node[place,right of=a3] (p13) {}
		edge [pre] (a3);
		\node[transition,right of=p13] (b3) {$b$}
		edge [pre] (p13);
		\node[place,right of=b3] (p23) {}
		edge [pre] (b3);
		\node[transition,right of=p23] (c3) {$c$}
		edge [pre] (p23);
		\node[place,right of=c3] (p33) {}
		edge [pre] (c3);
		\node[transition,right of=p33] (d3) {$d$}
		edge [pre] (p33)
		edge [post] (end);
		\node[transition,above of=p13] (a4) {$a$}
		edge [pre,bend right=20] (start);
		\node[place,right of=a4] (p14) {}
		edge [pre] (a4);
		\node[transition,right of=p14] (b4) {$b$}
		edge [pre] (p14);
		\node[place,right of=b4] (p24) {}
		edge [pre] (b4);
		\node[transition,right of=p24] (c4) {$c$}
		edge [pre] (p24)
		edge [post,bend left=20] (end);
	\end{tikzpicture}}
	
	\medskip
	\hrule
	\medskip
	
	\centering
	\scalebox{\scalefactor}{
	\begin{tikzpicture}[node distance = 1cm,>=stealth',bend angle=0,auto]
		\node[place,tokens=1,label=left:$p_i$] (start) {};
		\node[above of=start] {$M_2$:};
		\node[transition,right of=start,yshift=-0.5cm] (a) {$a$}
		edge [pre] (start);
		\node[place,right of=a] (p1) {}
		edge [pre] (a);
		\node[transition,right of=p1] (b) {$b$}
		edge [pre] (p1);
		\node[place,right of=b] (p2) {}
		edge [pre] (b);
		\node[transition,right of=p2] (c) {$c$}
		edge [pre] (p2);
		\node[place,right of=c] (p3) {}
		edge [pre] (c);
		\node[transition,right of=p3] (d) {$d$}
		edge [pre] (p3);
		\node[place,right of=d] (p4) {}
		edge [pre] (d);
		\node[transition,right of=p4] (e) {$e$}
		edge [pre] (p4);
		\node[place,right of=e,yshift=0.5cm,label=right:$p_o$] (end) {}
		edge [pre] (e);
		\node[transition,below of=p1] (d2) {$d$}
		edge [pre,bend left=20] (start);
		\node[place,right of=d2] (p12) {}
		edge [pre] (d2);
		\node[transition,right of=p12] (e2) {$e$}
		edge [pre] (p12);
		\node[place,right of=e2] (p22) {}
		edge [pre] (e2);
		\node[transition,right of=p22] (a2) {$a$}
		edge [pre] (p22);
		\node[place,right of=a2] (p32) {}
		edge [pre] (a2);
		\node[transition,right of=p32] (b2) {$b$}
		edge [pre] (p32)
		edge [post,bend right=20] (end);
		\node[transition,above of=p1] (a3) {$a$}
		edge [pre] (start);
		\node[place,right of=a3] (p13) {}
		edge [pre] (a3);
		\node[transition,right of=p13] (b3) {$b$}
		edge [pre] (p13);
		\node[place,right of=b3] (p23) {}
		edge [pre] (b3);
		\node[transition,right of=p23] (c3) {$c$}
		edge [pre] (p23);
		\node[place,right of=c3] (p33) {}
		edge [pre] (c3);
		\node[transition,right of=p33] (d3) {$d$}
		edge [pre] (p33)
		edge [post] (end);
		\node[transition,above of=p13] (a4) {$a$}
		edge [pre,bend right=20] (start);
		\node[place,right of=a4] (p14) {}
		edge [pre] (a4);
		\node[transition,right of=p14] (b4) {$b$}
		edge [pre] (p14);
		\node[place,right of=b4] (p24) {}
		edge [pre] (b4);
		\node[transition,right of=p24] (c4) {$c$}
		edge [pre] (p24)
		edge [post,bend left=20] (end);
		\node[transition,below of=d2] (b5) {$b$};
		\node[place,left of=b5] (p15) {}
		edge [post] (b5);
		\node[transition,left of=p15] (a5) {$a$}
		edge [post] (p15)
		edge [pre] (start);
		\node[place,right of=b5] (p25) {}
		edge [pre] (b5);
		\node[transition,right of=p25] (c5) {$c$}
		edge [pre] (p25);
		\node[place,right of=c5] (p35) {}
		edge [pre] (c5);
		\node[transition,right of=p35] (d5) {$d$}
		edge [pre] (p35);
		\node[place,right of=d5] (p45) {}
		edge [pre] (d5);
		\node[transition,right of=p45] (e5) {$e$}
		edge [pre] (p45);
		\node[place,right of=e5] (p55) {}
		edge [pre] (e5);
		\node[transition,right of=p55] (f5) {$f$}
		edge [pre] (p55)
		edge [post] (end);
		\node[place,below of=a5] (p16) {};
		\node[transition,left of=p16] (a6) {$a$}
		edge [post] (p16)
		edge [pre] (start);
		\node[transition,right of=p16] (aa6) {$a$}
		edge [pre] (p16);
		\node[place,right of=aa6] (p26) {}
		edge [pre] (aa6);
		\node[transition,right of=p26] (b6) {$b$}
		edge [pre] (p26);
		\node[place,right of=b6] (p36) {}
		edge [pre] (b6);
		\node[transition,right of=p36] (c6) {$c$}
		edge [pre] (p36);
		\node[place,right of=c6] (p46) {}
		edge [pre] (c6);
		\node[transition,right of=p46] (d6) {$d$}
		edge [pre] (p46);
		\node[place,right of=d6] (p56) {}
		edge [pre] (d6);
		\node[transition,right of=p56] (e6) {$e$}
		edge [pre] (p56);
		\node[place,right of=e6] (p66) {}
		edge [pre] (e6);
		\node[transition,right of=p66] (f6) {$f$}
		edge [pre] (p66)
		edge [post] (end);
		\node[transition,below of=a6] (a7) {$a$}
		edge [pre,bend left=40] (start);
		\node[place,right of=a7] (p17) {}
		edge [pre] (a7);
		\node[transition,right of=p17] (b7) {$b$}
		edge [pre] (p17);
		\node[place,right of=b7] (p27) {}
		edge [pre] (b7);
		\node[transition,right of=p27] (c7) {$c$}
		edge [pre] (p27);
		\node[place,right of=c7] (p37) {}
		edge [pre] (c7);
		\node[transition,right of=p37] (d7) {$d$}
		edge [pre] (p37);
		\node[place,right of=d7] (p47) {}
		edge [pre] (d7);
		\node[transition,right of=p47] (e7) {$e$}
		edge [pre] (p47);
		\node[place,right of=e7] (p57) {}
		edge [pre] (e7);
		\node[transition,right of=p57] (aa7) {$a$}
		edge [pre] (p57);
		\node[place,right of=aa7] (p67) {}
		edge [pre] (aa7);
		\node[transition,right of=p67] (bb7) {$b$}
		edge [pre] (p67)
		edge [post,bend right=40] (end);
	\end{tikzpicture}}
	\caption{The trace nets $M_1, M_2$ for the event logs $L_1, L_2$ of \cref{theo:tracenet-leq-entries}.}
	\label{fig:tracenet-leq-entries}
\end{figure}
These models have the following model complexity scores:
\begin{center}
	\begin{tabular}{|c|c|c|c|c|}\hline
		 & $\size$ & $\controlflow$ & $\avgconn$ & $\maxconn$ \\ \hline
		$L_1$ & $30$ & $4$ & $4$ & $4$ \\ \hline
		$L_2$ & $67$ & $7$ & $7$ & $7$ \\ \hline
	\end{tabular}
\end{center}
Thus, $\mathcal{C}^L(L_1) < \mathcal{C}^L(L_2)$ and $\mathcal{C}^M(M_1) < \mathcal{C}^M(M_2)$.
To see that $\mathcal{C}^L(L_1) < \mathcal{C}^L(L_2)$ and $\mathcal{C}^M(M_1) = \mathcal{C}^M(M_2)$ are also possible, consider the example used in the proof of \cref{lemma:not-support-influencing}.
Since both event logs have the same support, they have the same trace net, labeled $M_1$ in \cref{fig:tracenet-leq-entries}.
Thus, the model complexity of the trace nets stay the same, even though the log complexity score increased from the first to the second event log. \hfill$\square$
\end{proof}

\begin{theorem}
\label{theo:tracnet-less-entries}
Let $\mathcal{C}^M \in \{\size, \controlflow, \avgconn, \maxconn\}$ be a model complexity measure and let $\mathcal{C}^L \in \{\variety, \tlmax, \levelofdetail, \numberofties, \numberuniquetraces, \percentageuniquetraces, \varentropy, \normvarentropy\}$ and be a log complexity measure.
Then, $(\mathcal{C}^L, \mathcal{C}^M) \in \mless$.
\end{theorem}
\begin{proof}
Let $\mathcal{C}^L$ be a log complexity measure and $\mathcal{C}^M$ a model complexity measure allowed by this theorem.
Furthermore, let $L_1 \sqsubset L_2$ be event logs and $M_1, M_2$ their respective trace nets.
In this proof, we first show that $\mathcal{C}^L(L_1) < \mathcal{C}^L(L_2)$ implies $supp(L_1) \subsetneq supp(L_2)$ for all allowed log complexity measures.
\begin{itemize}
	\item \textbf{Variety $\variety$:}
	Suppose $\variety(L_1) < \variety(L_2)$.
	Since $L_1 \sqsubset L_2$, we know that $supp(L_1) \subseteq supp(L_2)$.
	What remains to be shown is $supp(L_1) \neq supp(L_2)$.
	By definition of $\variety$, and since $\variety(L_1) < \variety(L_2)$, there must be an activity name $a$ that occurs in $L_2$, but not in $L_1$.
	This is only possible if there is a trace $\sigma$, such that there is a $i \in \{1, \dots, |\sigma|\}$ with $\sigma(i) = a$, and such that $\sigma \in supp(L_2) \setminus supp(L_1)$.
	Thus, $supp(L_2) \setminus supp(L_1) \neq \emptyset$, and we get $supp(L_1) \neq supp(L_2)$.
	
	\item \textbf{Maximum Trace Length $\tlmax$:}
	Suppose $\tlmax(L_1) < \tlmax(L_2)$.
	Since $L_1 \sqsubset L_2$, we know that $supp(L_1) \subseteq supp(L_2)$.
	What remains to be shown is $supp(L_1) \neq supp(L_2)$.
	Since the length of the longest trace in $L_2$ is longer than the length of the longest trace in $L_1$, there must be a trace $\sigma \in supp(L_2) \setminus supp(L_1)$ with $|\sigma| > |\rho|$ for all $\rho \in L_1$.
	Thus, we know that $supp(L_2) \setminus supp(L_1) \neq \emptyset$, and therefore $supp(L_1) \neq supp(L_2)$.
	
	\item \textbf{Level of Detail $\levelofdetail$:}
	Suppose $\levelofdetail(L_1) < \levelofdetail(L_2)$.
	Since $L_1 \sqsubset L_2$, we know that $supp(L_1) \subseteq supp(L_2)$ is true.
	What remains to be shown is $supp(L_1) \neq supp(L_2)$.
	By definition of $\levelofdetail$, since $\levelofdetail(L_1) < \levelofdetail(L_2)$, the DFG of~$L_2$ contains at least one path that is not present in the DFG of~$L_1$.
	But this is only possible if there is at least one edge $(a,b)$ in the DFG of~$L_2$ that is not part of the DFG of~$L_1$.
	By construction of the directly follows graph, this means $a >_{L_2} b$, but $a \not>_{L_1} b$.
	Thus, a $\sigma \in supp(L_2) \setminus supp(L_1)$ must exist with $\sigma(i) = a$ and $\sigma(i+1) = b$ for some $i \in \{1, \dots |\sigma|-1\}$.
	Therefore, $supp(L_2) \setminus supp(L_1) \neq \emptyset$, and we get that $supp(L_1) \neq supp(L_2)$.
	
	\item \textbf{Number of Ties $\numberofties$:}
	Suppose that $\numberofties(L_1) < \numberofties(L_2)$.
	Since $L_1 \sqsubset L_2$, we know $supp(L_1) \subseteq supp(L_2)$.
	What remains to be shown is $supp(L_1) \neq supp(L_2)$.
	Since $\numberofties(L_1) < \numberofties(L_2)$, there are activity names $a,b$ with $a >_{L_2} b$ but $a \not>_{L_1} b$ or $b >_{L_1} a$.
	Since adding behavior to an event log cannot remove any direct neighborhoods of activities, we know that $a \not>_{L_1}$ is true.
	Then, there must be a trace $\sigma \in supp(L_2) \setminus supp(L_1)$ with $\sigma(i) = a$ and $\sigma(i+1) = b$ for some $i \in \{1, \dots |\sigma|-1\}$.
	Therefore, $supp(L_2) \setminus supp(L_1) \neq \emptyset$, and we get $supp(L_1) \neq supp(L_2)$.
	
	\item \textbf{Number of Distinct Traces $\numberuniquetraces$:}
	Suppose $\numberuniquetraces(L_1) < \numberuniquetraces(L_2)$.
	Since $L_1 \sqsubset L_2$, we know $supp(L_1) \subseteq supp(L_2)$.
	What remains to be shown is $supp(L_1) \neq supp(L_2)$.
	Since $\numberuniquetraces(L_1) < \numberuniquetraces(L_2)$, we know by definition that $|supp(L_1)| < |supp(L_2)|$.
	Thus, $supp(L_1) \neq supp(L_2)$ must be true.
	
	\item \textbf{Percentage of Distinct Traces $\percentageuniquetraces$:}
	Let $\percentageuniquetraces(L_1) < \percentageuniquetraces(L_2)$.
	Since $L_1 \sqsubset L_2$, we know $supp(L_1) \subseteq supp(L_2)$.
	What remains to be shown is $supp(L_1) \neq supp(L_2)$.
	Since $\percentageuniquetraces(L_1) < \percentageuniquetraces(L_2)$, we know by definition that $\frac{|supp(L_1)|}{\sum_{\sigma \in L_1} L_1(\sigma)} < \frac{|supp(L_2)|}{\sum_{\sigma \in L_2} L_2(\sigma)}$.
	But since $L_1 \sqsubset L_2$, we know that the inequality $\sum_{\sigma \in L_1} L_1(\sigma) < \sum_{\sigma \in L_2} L_2(\sigma)$ is true.
	Thus, the previous inequality can only be true if $|supp(L_1)| < |supp(L_2)|$, so $supp(L_1) \neq supp(L_2)$.
	
	\item \textbf{Variant Entropy $\varentropy$:}
	Suppose $\varentropy(L_1) < \varentropy(L_2)$.
	Since $L_1 \sqsubset L_2$, we know that $supp(L_1) \subseteq supp(L_2)$.
	What remains to be shown is that $supp(L_1) \neq supp(L_2)$.
	Since $\varentropy(L_1) < \varentropy(L_2)$, we know by definition that there must be a node in the prefix automaton of $L_2$ that is not present in the prefix automaton of $L_1$.
	In turn, a trace $\sigma \in supp(L_2) \setminus supp(L_1)$ must exist that deviates from all traces in $L_1$ after a (possibly empty) common prefix.
	Since $supp(L_2) \setminus supp(L_1) \neq \emptyset$, $supp(L_1) \neq supp(L_2)$.
	
	\item \textbf{Normalized Variant Entropy $\normvarentropy$:}
	Since $|S| \cdot \ln(|S|)$ can only increase for larger event logs, $\normvarentropy(L_1) < \normvarentropy(L_2)$ directly implies that $\varentropy(L_1) < \varentropy(L_2)$.
	But as we have already seen, the latter implies $supp(L_1) \neq supp(L_2)$.
\end{itemize}
Since the trace net $M$ for an event log $L$ includes a unique path for each trace in $supp(L)$, we can quickly verify that $\mathcal{C}^M(M_1) < \mathcal{C}^M(M_2)$ if $supp(L_1) \neq supp(L_2)$, where $\mathcal{C}^M \in \{\size, \controlflow, \avgconn, \maxconn\}$. \hfill$\square$
\end{proof}

\begin{theorem}
\label{theo:tracenet-equals-entries}
Let $\mathcal{C}^L \in \loc$ be any log complexity measure and $\mathcal{C}^M$ be a model complexity measure with $\mathcal{C}^M \in \{\mismatch, \connhet, \tokensplit, \separability, \depth, \cyclicity, \emptyseq\}$.
Then, we have $(\mathcal{C}^L, \mathcal{C}^M) \in \meq$.
\end{theorem}
\begin{proof}
By \cref{lemma:tracenet-constant-complexity}, $\mathcal{C}^M(M_1) = \mathcal{C}^M(M_2)$ for any trace nets $M_1, M_2$.
Therefore, the implication $\mathcal{C}^L(L_1) < \mathcal{C}^L(L_2) \Rightarrow \mathcal{C}^M(M_1) = \mathcal{C}^M(M_2)$ is true for all event logs $L_1, L_2$, where $M_1, M_2$ are the trace nets for $L_1, L_2$. \hfill$\square$
\end{proof}

\begin{theorem}
\label{theo:tracenet-crossconn-entries}
Let $\mathcal{C}^L \in \loc$ be a log complexity measure.
Then, $(\mathcal{C}^L, \crossconn) \in \norel$.
\end{theorem}
\begin{proof}
Consider the following event logs:
\begin{align*}
	L_1 &= [\langle a,b,c,d \rangle^2, \langle a,c,c,e \rangle^2, \langle a,a,a,a\rangle^2] \\
	L_2 &= L_1 + [\langle a,a,b,c,c,d,e,f \rangle] \\
	L_3 &= L_2 + [\langle g,a,a,b,c,c,d,e,f,a,a,b,c,c,d,e,f \rangle]
\end{align*}
Then, the trace nets $M_1, M_2, M_3$ for the event logs $L_1, L_2, L_3$ fulfill:
\begin{itemize}
	\item[•] $\crossconn(M_1) \approx 0.8476$,
	\item[•] $\crossconn(M_2) \approx 0.8677$,
	\item[•] $\crossconn(M_3) \approx 0.8544$,
\end{itemize}
and therefore,  $\crossconn(M_1) < \crossconn(M_2)$ and $\crossconn(M_2) > \crossconn(M_3)$.
But the following table shows $\mathcal{C}^L(L_1) < \mathcal{C}^L(L_2) < \mathcal{C}^L(L_3)$ for any $\mathcal{C}^L \in (\loc \setminus \{\normvarentropy\})$:
\begin{center}
	\begin{tabular}{|c|c|c|c|c|c|c|c|c|c|c|} \hline
		 & $\magnitude$ & $\variety$ & $\support$ & $\tlavg$ & $\tlmax$ & $\levelofdetail$ & $\numberofties$ & $\lempelziv$ & $\numberuniquetraces$ & $\percentageuniquetraces$ \\ \hline
		$L_1$ & $\pad 24 \pad$ & $\pad 5 \pad$ & $\pad 6 \pad$ & $\pad 4 \pad$ & $\pad 4 \pad$ & $\pad 5 \pad$ & $\pad 5 \pad$ & $\pad 12 \pad$ & $\pad 3 \pad$ & $\pad 0.5 \pad$ \\ \hline
		$L_2$ & $\pad 32 \pad$ & $\pad 6 \pad$ & $\pad 7 \pad$ & $\pad 4.5714 \pad$ & $\pad 8 \pad$ & $\pad 11 \pad$ & $\pad 7 \pad$ & $\pad 16 \pad$ & $\pad 4 \pad$ & $\pad 0.5714 \pad$ \\ \hline
		$L_3$ & $\pad 49 \pad$ & $\pad 7 \pad$ & $\pad 8 \pad$ & $\pad 6.125 \pad$ & $\pad 17 \pad$ & $\pad 22 \pad$ & $\pad 9 \pad$ & $\pad 23 \pad$ & $\pad 5 \pad$ & $\pad 0.625 \pad$ \\ \hline
	\end{tabular}
	
	\medskip
	
	\begin{tabular}{|c|c|c|c|c|c|c|c|c|} \hline
		 & $\structure$ & $\affinity$ & $\deviationfromrandom$ & $\avgdist$ & $\varentropy$ & $\normvarentropy$ & $\seqentropy$ & $\normseqentropy$ \\ \hline
		$L_1$ & $\pad 2.\overline{6} \pad$ & $\pad 0.2 \pad$ & $\pad 0.619 \pad$ & $\pad 4.2667 \pad$ & $\pad 10.889 \pad$ & $\pad 0.4729 \pad$ & $\pad 24.9533 \pad$ & $\pad 0.3272 \pad$ \\ \hline
		$L_2$ & $\pad 3.1429 \pad$ & $\pad 0.2079 \pad$ & $\pad 0.6475 \pad$ & $\pad 4.5714 \pad$ & $\pad 21.474 \pad$ & $\pad 0.4841 \pad$ & $\pad 42.4367 \pad$ & $\pad 0.3826 \pad$ \\ \hline
		$L_3$ & $\pad 3.625 \pad$ & $\pad 0.2219 \pad$ & $\pad 0.6776 \pad$ & $\pad 6.5357 \pad$ & $\pad 44.3327 \pad$ & $\pad 0.3842 \pad$ & $\pad 74.0677 \pad$ & $\pad 0.3884 \pad$ \\ \hline
	\end{tabular}
\end{center}
\noindent
For $\mathcal{C}^L = \normvarentropy$, we take the following event logs:
\begin{align*}
	L_1 &= [\langle a \rangle, \langle a,b,c \rangle] \\
	L_2 &= L_1 + [\langle a,b,c,d,e \rangle, \langle x,y,z \rangle] \\
	L_3 &= L_2 + [\langle f,g,h,i,j,k,l,m,n,o,p \rangle]
\end{align*}
Then, the trace nets $M_1, M_2, M_3$ for the event logs $L_1, L_2, L_3$ fulfill:
\begin{itemize}
	\item[•] $\crossconn(M_1) \approx 0.7098$,
	\item[•] $\crossconn(M_2) \approx 0.857$,
	\item[•] $\crossconn(M_3) \approx 0.8436$
\end{itemize}
and therefore,  $\crossconn(M_1) < \crossconn(M_2)$ and $\crossconn(M_2) > \crossconn(M_3)$, even though 
\begin{itemize}
	\item[•] $\normvarentropy(L_1) = 0$,
	\item[•] $\normvarentropy(L_2) \approx 0.3181$,
	\item[•] $\normvarentropy(L_3) \approx 0.3258$,
\end{itemize}
and therefore $\normvarentropy(L_1) < \normvarentropy(L_2) < \normvarentropy(L_3)$ is true. \hfill$\square$
\end{proof}

\begin{theorem}
\label{theo:tracenet-sequentiality-entries}
Let $\mathcal{C}^L \in \loc$ be a log complexity measure. 
Then, $(\mathcal{C}^L, \sequentiality) \in \norel$.
\end{theorem}
\begin{proof}
Consider the following event logs:
\begin{align*}
	L_1 &= [\langle a,b,c,d,e \rangle^{3}, \langle e,d,c,a,b \rangle^{3}] \\
	L_2 &= L_1 + [\langle a,f,e,d,c,b \rangle^{2}] \\
	L_3 &= L_2 + [\langle g,a,c,d,e,b,f \rangle^{2}, \langle a,b \rangle]
\end{align*}
Then, the trace nets $M_1, M_2, M_3$ for the event logs $L_1, L_2, L_3$ fulfill:
\begin{itemize}
	\item[•] $\sequentiality(M_1) = 0.2$,
	\item[•] $\sequentiality(M_2) \approx 0.1875$,
	\item[•] $\sequentiality(M_3) = 0.2$,
\end{itemize}
and so, $\sequentiality(M_1) > \sequentiality(M_2)$, $\sequentiality(M_2) < \sequentiality(M_3)$, and $\sequentiality(M_1) = \sequentiality(M_3)$.
But the next table shows $\mathcal{C}^L(L_1) < \mathcal{C}^L(L_2) < \mathcal{C}^L(L_3)$ for $\mathcal{C}^L \in (\loc \setminus \{\affinity\})$: 
\begin{center}
	\begin{tabular}{|c|c|c|c|c|c|c|c|c|c|c|} \hline
		 & $\magnitude$ & $\variety$ & $\support$ & $\tlavg$ & $\tlmax$ & $\levelofdetail$ & $\numberofties$ & $\lempelziv$ & $\numberuniquetraces$ & $\percentageuniquetraces$ \\ \hline
		$L_1$ & $\pad 30 \pad$ & $\pad 5 \pad$ & $\pad 6 \pad$ & $\pad 5 \pad$ & $\pad 5 \pad$ & $\pad 4 \pad$ & $\pad 3 \pad$ & $\pad 16 \pad$ & $\pad 2 \pad$ & $\pad 0.3333 \pad$ \\ \hline
		$L_2$ & $\pad 42 \pad$ & $\pad 6 \pad$ & $\pad 8 \pad$ & $\pad 5.25 \pad$ & $\pad 6 \pad$ & $\pad 7 \pad$ & $\pad 4 \pad$ & $\pad 21 \pad$ & $\pad 3 \pad$ & $\pad 0.375 \pad$ \\ \hline
		$L_3$ & $\pad 58 \pad$ & $\pad 7 \pad$ & $\pad 11 \pad$ & $\pad 5.2727 \pad$ & $\pad 7 \pad$ & $\pad 37 \pad$ & $\pad 6 \pad$ & $\pad 28 \pad$ & $\pad 5 \pad$ & $\pad 0.4545 \pad$ \\ \hline
	\end{tabular}
	
	\medskip
	
	\begin{tabular}{|c|c|c|c|c|c|c|c|c|} \hline
		 & $\structure$ & $\affinity$ & $\deviationfromrandom$ & $\avgdist$ & $\varentropy$ & $\normvarentropy$ & $\seqentropy$ & $\normseqentropy$ \\ \hline
		$L_1$ & $\pad 5 \pad$ & $\pad 0.4857 \pad$ & $\pad 0.659 \pad$ & $\pad 3.6 \pad$ & $\pad 6.9315 \pad$ & $\pad 0.301 \pad$ & $\pad 20.7944 \pad$ & $\pad 0.2038 \pad$ \\ \hline
		$L_2$ & $\pad 5.25 \pad$ & $\pad 0.3571 \pad$ & $\pad 0.7031 \pad$ & $\pad 4.0714 \pad$ & $\pad 16.4792 \pad$ & $\pad 0.4057 \pad$ & $\pad 45.1709 \pad$ & $\pad 0.2877 \pad$ \\ \hline
		$L_3$ & $\pad 5.2727 \pad$ & $\pad 0.2545 \pad$ & $\pad 0.7395 \pad$ & $\pad 4.5455 \pad$ & $\pad 30.24 \pad$ & $\pad 0.4447 \pad$ & $\pad 78.9679 \pad$ & $\pad 0.3353 \pad$ \\ \hline
	\end{tabular}
\end{center}
\noindent
For $\mathcal{C}^L = \affinity$, we take the following event logs:
\begin{align*}
	L_1 &= [\langle a,b,c,d,e \rangle^{3}, \langle e,d,c,a,b \rangle^{3}] \\
	L_2 &= L_1 + [\langle f,e,d,c,a,b \rangle^{2}] \\
	L_3 &= L_2 + [\langle g,f,e,d,c,a,b \rangle^{5}, \langle a,b \rangle]
\end{align*}
Then, the trace nets $M_1, M_2, M_3$ for the event logs $L_1, L_2, L_3$ fulfill:
\begin{itemize}
	\item[•] $\sequentiality(M_1) = 0.2$,
	\item[•] $\sequentiality(M_2) \approx 0.1875$,
	\item[•] $\sequentiality(M_3) = 0.2$,
\end{itemize}
and so, $\sequentiality(M_1) > \sequentiality(M_2)$,  $\sequentiality(M_2) < \sequentiality(M_3)$, and $\sequentiality(M_1) = \sequentiality(M_3)$, even though the affinity scores strictly increase:
\begin{itemize}
	\item[•] $\affinity(L_1) \approx 0.4857$,
	\item[•] $\affinity(L_2) \approx 0.4941$,
	\item[•] $\affinity(L_3) \approx 0.5117$,
\end{itemize}
and therefore $\affinity(L_1) < \affinity(L_2) < \affinity(L_3)$ is true. \hfill$\square$
\end{proof}

\begin{theorem}
\label{theo:tracenet-diameter-leq-entries}
Let $\mathcal{C}^L \in (\loc \setminus \{\tlmax\})$ be a log complexity measure.
Then, $(\mathcal{C}^L, \diameter) \in \mleq$.
\end{theorem}
\begin{proof}
Let $L_1, L_2$ be event logs with $L_1 \sqsubset L_2$ and $|supp(L_1)| > 1$, and $M_1, M_2$ be the trace nets for $L_1$ and $L_2$.
By \cref{lemma:tracenet-monotone-increasing}, we know that $L_1 \sqsubset L_2$ and $|supp(L_1)| > 1$ implies $\diameter(M_1) \leq \diameter(M_2)$.
We now show that both $\diameter(M_1) = \diameter(M_2)$ and $\diameter(M_1) < \diameter(M_2)$ are possible,
For the former, take the following event logs:
\begin{align*}
	L_1 &= [\langle a,b,c,d \rangle^{2}, \langle a,b,c,d,e \rangle^{2}, \langle d,e,a,b \rangle^{2}] \\
	L_2 &= L_1 + [\langle a,b,c,d,e \rangle^{2}, \langle d,e,a,b,c \rangle, \langle c,d,e,a,b \rangle, \langle f,c,d,a,b \rangle]
\end{align*}
These two event logs have the following log complexity scores:
\begin{center}
	\begin{tabular}{|c|c|c|c|c|c|c|c|c|c|c|c|c|} \hline
		 & $\magnitude$ & $\variety$ & $\support$ & $\tlavg$ & $\tlmax$ & $\levelofdetail$ & $\numberofties$ & $\lempelziv$ & $\numberuniquetraces$ & $\percentageuniquetraces$ \\ \hline
		$L_1$ & $\pad 26 \pad$ & $\pad 5 \pad$ & $\pad 6 \pad$ & $\pad 4.3333 \pad$ & $\pad 5 \pad$ & $\pad 6 \pad$ & $\pad 5 \pad$ & $\pad 13 \pad$ & $\pad 3 \pad$ & $\pad 0.5 \pad$ \\ \hline
		$L_2$ & $\pad 51 \pad$ & $\pad 6 \pad$ & $\pad 11 \pad$ & $\pad 4.6364 \pad$ & $\pad 5 \pad$ & $\pad 20 \pad$ & $\pad 7 \pad$ & $\pad 21 \pad$ & $\pad 6 \pad$ & $\pad 0.5455 \pad$ \\ \hline
	\end{tabular}
		
	\medskip
		
	\begin{tabular}{|c|c|c|c|c|c|c|c|c|} \hline
		 & $\structure$ & $\affinity$ & $\deviationfromrandom$ & $\avgdist$ & $\varentropy$ & $\normvarentropy$ & $\seqentropy$ & $\normseqentropy$ \\ \hline
		$L_1$ & $\pad 4.3333 \pad$ & $\pad 0.56 \pad$ & $\pad 0.5757 \pad$ & $\pad 2.6667 \pad$ & $\pad 6.1827 \pad$ & $\pad 0.3126 \pad$ & $\pad 16.0483 \pad$ & $\pad 0.1894 \pad$ \\ \hline
		$L_2$ & $\pad 4.6364 \pad$ & $\pad 0.5626 \pad$ & $\pad 0.5880 \pad$ & $\pad 2.9818 \pad$ & $\pad 27.7259 \pad$ & $\pad 0.4628 \pad$ & $\pad 57.7827 \pad$ & $\pad 0.2882 \pad$ \\ \hline
	\end{tabular}
\end{center}
Thus, $\mathcal{C}^L(L_1) < \mathcal{C}^L(L_2)$ for any of the log complexity measures allowed by this theorem.
But the trace nets $M_1, M_2$ for the event logs $L_1, L_2$ fulfill the property $\diameter(L_1) = 11 = \diameter(L_2)$.

To see that the diameter can also increase, take the following event logs:
\begin{align*}
	L_1 &= [\langle a,b,c,d \rangle^{2}, \langle a,b,c,d,e \rangle^{2}, \langle d,e,a,b \rangle^{2}] \\
	L_2 &= L_1 + [\langle a,b,c,d,e \rangle^{2}, \langle d,e,a,b,c \rangle, \langle c,d,e,a,b \rangle, \langle f,c,d,a,b,c \rangle]
\end{align*}
Note that $L_1$ did not change in contrast to the previous log with the same name, while in $L_2$, the trace $\langle f,c,d,a,b \rangle$ became $\langle f,c,d,a,b,c \rangle$.
These two event logs have the following log complexity scores:
\begin{center}
	\begin{tabular}{|c|c|c|c|c|c|c|c|c|c|c|c|c|} \hline
		 & $\magnitude$ & $\variety$ & $\support$ & $\tlavg$ & $\tlmax$ & $\levelofdetail$ & $\numberofties$ & $\lempelziv$ & $\numberuniquetraces$ & $\percentageuniquetraces$ \\ \hline
		$L_1$ & $\pad 26 \pad$ & $\pad 5 \pad$ & $\pad 6 \pad$ & $\pad 4.3333 \pad$ & $\pad 5 \pad$ & $\pad 6 \pad$ & $\pad 5 \pad$ & $\pad 13 \pad$ & $\pad 3 \pad$ & $\pad 0.5 \pad$ \\ \hline
		$L_2$ & $\pad 52 \pad$ & $\pad 6 \pad$ & $\pad 11 \pad$ & $\pad 4.7273 \pad$ & $\pad 6 \pad$ & $\pad 20 \pad$ & $\pad 7 \pad$ & $\pad 21 \pad$ & $\pad 6 \pad$ & $\pad 0.5455 \pad$ \\ \hline
	\end{tabular}
		
	\medskip
		
	\begin{tabular}{|c|c|c|c|c|c|c|c|c|} \hline
		 & $\structure$ & $\affinity$ & $\deviationfromrandom$ & $\avgdist$ & $\varentropy$ & $\normvarentropy$ & $\seqentropy$ & $\normseqentropy$ \\ \hline
		$L_1$ & $\pad 4.3333 \pad$ & $\pad 0.56 \pad$ & $\pad 0.5757 \pad$ & $\pad 2.6667 \pad$ & $\pad 6.1827 \pad$ & $\pad 0.3126 \pad$ & $\pad 16.0483 \pad$ & $\pad 0.1894 \pad$ \\ \hline
		$L_2$ & $\pad 4.6364 \pad$ & $\pad 0.5829 \pad$ & $\pad 0.5887 \pad$ & $\pad 2.9091 \pad$ & $\pad 29.0428 \pad$ & $\pad 0.4543 \pad$ & $\pad 60.0209 \pad$ & $\pad 0.2921 \pad$ \\ \hline
	\end{tabular}
\end{center}
Thus, $\mathcal{C}^L(L_1) < \mathcal{C}^L(L_2)$ for any of the log complexity measures allowed by this theorem.
But the trace nets $M_1, M_2$ for the event logs $L_1, L_2$ fulfill the property $\diameter(M_1) = 11 < 13 = \diameter(M_2)$.\hfill$\square$
\end{proof}

\begin{theorem}
\label{theo:tracenet-diameter-less-entries}
$(\tlmax, \density) \in \mless$.
\end{theorem}
\begin{proof}
Let $L_1, L_2$ be event logs with $L_1 \sqsubset L_2$.
Further, let $M_1, M_2$ be the trace nets for $L_1, L_2$.
Suppose $\tlmax(L_1) < \tlmax(L_2)$.
Since the trace net contains a unique path from the start node to the end node for each trace, and no other paths from the start to the end node exist, all lengths of paths are dependent on the lengths of the traces they enable.
Because we know that $\tlmax(L_1) < \tlmax(L_1)$, there is a trace $\sigma \in L_2$ with $|\sigma| > |\rho|$ for all $\rho \in L_1$.
Thus, the length of the path for $\sigma$ in $M_2$ is longer than any path in $M_1$, which means $\diameter(M_1) < \diameter(M_2)$. \hfill$\square$
\end{proof}

\begin{theorem}
\label{theo:tracenet-netconn-entries}
$(\mathcal{C}^L, \netconn) \in \norel$ for any log complexity measure $\mathcal{C}^L \in \loc$.
\end{theorem}
\begin{proof}
Consider the following event logs:
\begin{align*}
	L_1 &= [\langle a,b,c,d \rangle^{2}, \langle a,c,c,e \rangle^{2}, \langle a,a,a,a \rangle^{2}] \\
	L_2 &= L_1 + [\langle a,a,b,c,c,d,e,f \rangle] \\
	L_3 &= L_2 + [\langle g,a,a,b,c,c,d,e,f,a,a,b,c,c,d,f \rangle]
\end{align*}
Then, the trace nets $M_1, M_2, M_3$ for the event logs $L_1, L_2, L_3$ fulfill:
\begin{itemize}
	\item[•] $\netconn(M_1) \approx 1.0435$,
	\item[•] $\netconn(M_2) \approx 1.0526$,
	\item[•] $\netconn(M_3) \approx 1.0435$,
\end{itemize}
so we can see that $\netconn(M_1) < \netconn(M_2)$, $\netconn(M_2) > \netconn(M_3)$, and $\netconn(M_1) = \netconn(M_3)$.
But the next table showsn $\mathcal{C}^L(L_1) < \mathcal{C}^L(L_2) < \mathcal{C}^L(L_3)$ for any $\mathcal{C}^L \in (\loc \setminus \{\normvarentropy\})$:
\begin{center}
	\def\pad{\hspace*{1.5mm}}
	\begin{tabular}{|c|c|c|c|c|c|c|c|c|c|c|} \hline
		 & $\magnitude$ & $\variety$ & $\support$ & $\tlavg$ & $\tlmax$ & $\levelofdetail$ & $\numberofties$ & $\lempelziv$ & $\numberuniquetraces$ & $\percentageuniquetraces$ \\ \hline
		$L_1$ & $\pad 24 \pad$ & $\pad 5 \pad$ & $\pad 6 \pad$ & $\pad 4 \pad$ & $\pad 4 \pad$ & $\pad 5 \pad$ & $\pad 5 \pad$ & $\pad 12 \pad$ & $\pad 3 \pad$ & $\pad 0.5 \pad$ \\ \hline
		$L_2$ & $\pad 32 \pad$ & $\pad 6 \pad$ & $\pad 7 \pad$ & $\pad 4.5714 \pad$ & $\pad 8 \pad$ & $\pad 11 \pad$ & $\pad 7 \pad$ & $\pad 16 \pad$ & $\pad 4 \pad$ & $\pad 0.5714 \pad$ \\ \hline
		$L_3$ & $\pad 48 \pad$ & $\pad 7 \pad$ & $\pad 8 \pad$ & $\pad 6 \pad$ & $\pad 16 \pad$ & $\pad 26 \pad$ & $\pad 10 \pad$ & $\pad 23 \pad$ & $\pad 5 \pad$ & $\pad 0.625 \pad$ \\ \hline
	\end{tabular}
	
	\medskip
	
	\begin{tabular}{|c|c|c|c|c|c|c|c|c|} \hline
		 & $\structure$ & $\affinity$ & $\deviationfromrandom$ & $\avgdist$ & $\varentropy$ & $\normvarentropy$ & $\seqentropy$ & $\normseqentropy$ \\ \hline
		$L_1$ & $\pad 2.6667 \pad$ & $\pad 0.2 \pad$ & $\pad 0.619 \pad$ & $\pad 4.2667 \pad$ & $\pad 10.889 \pad$ & $\pad 0.4729 \pad$ & $\pad 24.9533 \pad$ & $\pad 0.3272 \pad$ \\ \hline
		$L_2$ & $\pad 3.1429 \pad$ & $\pad 0.2079 \pad$ & $\pad 0.6475 \pad$ & $\pad 4.5714 \pad$ & $\pad 21.474 \pad$ & $\pad 0.4841 \pad$ & $\pad 42.4367 \pad$ & $\pad 0.3826 \pad$ \\ \hline
		$L_3$ & $\pad 3.625 \pad$ & $\pad 0.2154 \pad$ & $\pad 0.6766 \pad$ & $\pad 6.2857 \pad$ & $\pad 43.6547 \pad$ & $\pad 0.3936 \pad$ & $\pad 72.9894 \pad$ & $\pad 0.3928 \pad$ \\ \hline
	\end{tabular}
\end{center}
For $\mathcal{C}^L = \normvarentropy$, we take the following event logs:
\begin{align*}
	L_1 &= [\langle a,b \rangle, \langle a,b,c,d \rangle, \langle a,b,c,e \rangle] \\
	L_2 &= L_1 + [\langle s,t,u,v,w,x,y,z \rangle] \\
	L_3 &= L_2 + [\langle b,c,d,e,f,g,h,i,j,k,l,m \rangle]
\end{align*}
Then, the trace nets $M_1, M_2, M_3$ for the event logs $L_1, L_2, L_3$ fulfill:
\begin{itemize}
	\item[•] $\netconn(M_1) \approx 1.0526$,
	\item[•] $\netconn(M_2) \approx 1.0588$,
	\item[•] $\netconn(M_3) \approx 1.0526$,
\end{itemize}
and therefore, we have that $\netconn(M_1) < \netconn(M_2)$, $\netconn(M_2) > \netconn(M_3)$, and $\netconn(M_1) = \netconn(M_3)$, even though
\begin{itemize}
	\item[•] $\normvarentropy(L_1) \approx 0.3109$,
	\item[•] $\normvarentropy(L_2) \approx 0.3348$,
	\item[•] $\normvarentropy(L_3) \approx 0.3538$,
\end{itemize}
and therefore $\normvarentropy(L_1) < \normvarentropy(L_2) < \normvarentropy(L_3)$ is true. \hfill$\square$
\end{proof}

\begin{theorem}
\label{theo:tracenet-density-geq-entries}
Let $\mathcal{C}^L \in (\loc \setminus \{\tlmax, \numberofties\})$ be a log complexity measure. 
Then, $(\mathcal{C}^L, \density) \in \mgeq$.
\end{theorem}
\begin{proof}
Let $L$ be an event log and $M$ be its trace net.
Since every transition in $M$ has exactly one incoming and one outgoing edge by definition, we have $\density(M) = \frac{2|T|}{2|T|(|P|-1)} = \frac{1}{|P| - 1}$.
Because $M$ contains $2 + \sum_{\sigma \in L} (|\sigma| - 1)$ places, we get, for two trace nets $M_1, M_2$ of event logs $L_1, L_2$ with $L_1 \sqsubset L_2$:
\begin{align*}
\density(M_1) = \frac{1}{1 + \sum_{\sigma \in L_1} (|\sigma| - 1)} \overset{L_1 \sqsubset L_2}{\geq} \frac{1}{1 + \sum_{\sigma \in L_2} (|\sigma| - 1)} = \density(L_2).
\end{align*}
What remains to be shown is that both $\density(M_1) > \density(M_2)$ and and also $\density(M_1) = \density(M_2)$ are possible.
For the former, take the following logs:
\begin{align*}
	L_1 &= [\langle a,b,c,d \rangle^{2}, \langle a,b,c,d,e \rangle^{2}, \langle d,e,a,b \rangle^{2}] \\
	L_2 &= L_1 + [\langle a,b,c,d,e \rangle^{2}, \langle d,e,a,b,c \rangle, \langle c,d,e,a,b \rangle, \langle f,c,d,a,b,c \rangle]
\end{align*}
These two event logs have the following log complexity scores:
\begin{center}
	\begin{tabular}{|c|c|c|c|c|c|c|c|c|c|c|c|c|} \hline
		 & $\magnitude$ & $\variety$ & $\support$ & $\tlavg$ & $\tlmax$ & $\levelofdetail$ & $\numberofties$ & $\lempelziv$ & $\numberuniquetraces$ & $\percentageuniquetraces$ \\ \hline
		$L_1$ & $\pad 26 \pad$ & $\pad 5 \pad$ & $\pad 6 \pad$ & $\pad 4.3333 \pad$ & $\pad 5 \pad$ & $\pad 6 \pad$ & $\pad 5 \pad$ & $\pad 13 \pad$ & $\pad 3 \pad$ & $\pad 0.5 \pad$ \\ \hline
		$L_2$ & $\pad 52 \pad$ & $\pad 6 \pad$ & $\pad 11 \pad$ & $\pad 4.7273 \pad$ & $\pad 6 \pad$ & $\pad 20 \pad$ & $\pad 7 \pad$ & $\pad 21 \pad$ & $\pad 6 \pad$ & $\pad 0.5455 \pad$ \\ \hline
	\end{tabular}
		
	\medskip
		
	\begin{tabular}{|c|c|c|c|c|c|c|c|c|} \hline
		 & $\structure$ & $\affinity$ & $\deviationfromrandom$ & $\avgdist$ & $\varentropy$ & $\normvarentropy$ & $\seqentropy$ & $\normseqentropy$ \\ \hline
		$L_1$ & $\pad 4.3333 \pad$ & $\pad 0.56 \pad$ & $\pad 0.5757 \pad$ & $\pad 2.6667 \pad$ & $\pad 6.1827 \pad$ & $\pad 0.3126 \pad$ & $\pad 16.0483 \pad$ & $\pad 0.1894 \pad$ \\ \hline
		$L_2$ & $\pad 4.6364 \pad$ & $\pad 0.5829 \pad$ & $\pad 0.5887 \pad$ & $\pad 2.9091 \pad$ & $\pad 29.0428 \pad$ & $\pad 0.4543 \pad$ & $\pad 60.0209 \pad$ & $\pad 0.2921 \pad$ \\ \hline
	\end{tabular}
\end{center}
Thus, $\mathcal{C}^L(L_1) < \mathcal{C}^L(L_2)$ for any of the log complexity measures allowed by this theorem.
But the trace nets $M_1, M_2$ for the event logs $L_1, L_2$ fulfill the property $\density(M_1) \approx 0.0909 > 0.0417 \approx \density(M_2)$.

To see that $\density(M_1) = \density(M_2)$ is possible, take the following event logs:
\begin{align*}
	L_1 &= [\langle a,e \rangle^{4}, \langle a,b,c,d,e \rangle] \\
	L_2 &= L_1 + [\langle a,b,c,d,e \rangle, \langle f \rangle]
\end{align*}
These two event logs have the following log complexity scores:
\begin{center}
	\begin{tabular}{|c|c|c|c|c|c|c|c|c|c|c|c|c|} \hline
		 & $\magnitude$ & $\variety$ & $\support$ & $\tlavg$ & $\tlmax$ & $\levelofdetail$ & $\numberofties$ & $\lempelziv$ & $\numberuniquetraces$ & $\percentageuniquetraces$ \\ \hline
		$L_1$ & $\pad 13 \pad$ & $\pad 5 \pad$ & $\pad 5 \pad$ & $\pad 2.6 \pad$ & $\pad 5 \pad$ & $\pad 2 \pad$ & $\pad 5 \pad$ & $\pad 8 \pad$ & $\pad 2 \pad$ & $\pad 0.4 \pad$ \\ \hline
		$L_2$ & $\pad 19 \pad$ & $\pad 6 \pad$ & $\pad 7 \pad$ & $\pad 2.7143 \pad$ & $\pad 5 \pad$ & $\pad 3 \pad$ & $\pad 5 \pad$ & $\pad 11 \pad$ & $\pad 3 \pad$ & $\pad 0.4286 \pad$ \\ \hline
	\end{tabular}
		
	\medskip
		
	\begin{tabular}{|c|c|c|c|c|c|c|c|c|} \hline
		 & $\structure$ & $\affinity$ & $\deviationfromrandom$ & $\avgdist$ & $\varentropy$ & $\normvarentropy$ & $\seqentropy$ & $\normseqentropy$ \\ \hline
		$L_1$ & $\pad 2.6 \pad$ & $\pad 0.6 \pad$ & $\pad 0.478 \pad$ & $\pad 1.2 \pad$ & $\pad 3.8191 \pad$ & $\pad 0.3552 \pad$ & $\pad 8.0241 \pad$ & $\pad 0.2406 \pad$ \\ \hline
		$L_2$ & $\pad 2.7143 \pad$ & $\pad 0.3333 \pad$ & $\pad 0.559 \pad$ & $\pad 2.2857 \pad$ & $\pad 6.6899 \pad$ & $\pad 0.4911 \pad$ & $\pad 16.283 \pad$ & $\pad 0.2911 \pad$ \\ \hline
	\end{tabular}
\end{center}
Thus, $\mathcal{C}^L(L_1) < \mathcal{C}^L(L_2)$ for any allowed log complexity measure except affinity $\affinity$.
But the trace nets $M_1, M_2$ for the event logs $L_1, L_2$ fulfill the property $\density(M_1) = 0.1\overline{6} = \density(M_2)$.
For $\mathcal{C}^L = \affinity$, we take the following logs:
\begin{align*}
	L_1 &= [\langle a,e \rangle, \langle a,b,c,d,e \rangle] \\
	L_2 &= L_1 + [\langle a,b,c,d,e \rangle, \langle f \rangle]
\end{align*}
Compared to the previous event logs, only the frequencies of traces changed, so for the trace nets $M_1, M_2$ for $L_1, L_2$ we still have $\density(M_1) = 0.1\overline{6} = \density(M_2)$.
But now, $\affinity(L_1) = 0 < 0.1667 \approx \affinity(L_2)$, showing that equal density is also possible when affinity increases. \hfill$\square$
\end{proof}

\begin{theorem}
\label{theo:tracenet-density-greater-entries}
Let $\mathcal{C}^L \in \{\tlmax, \numberofties\}$ be a log complexity measure.
Then, $(\mathcal{C}^L, \density) \in \mgreater$.
\end{theorem}
\begin{proof}
As argued in the proof of \cref{theo:tracenet-density-geq-entries}, the density of a trace net $M$ for an event log $L$ is $\density(M) = \frac{1}{1 + \sum_{\sigma \in L} (|\sigma| - 1)}$.
Thus, the density of $M$ lowers if $1 + \sum_{\sigma \in L} (|\sigma| - 1)$ increases.
\begin{itemize}
	\item \textbf{Maximum Trace Length $\tlmax$:}
	Let $L_1, L_2$ be two event logs with $L_1 \sqsubset L_2$ and $\tlmax(L_1) < \tlmax(L_2)$.
	Then, there must be a trace $\sigma \in supp(L_2) \setminus supp(L_1)$ with $|\sigma| > 2$, since all traces in $L_1$ have length at least $1$.
	But then, $|\sigma| - 1 \geq 1$ and therefore
	\begin{align*}
	1 + \sum_{\sigma \in L_1} (|\sigma| - 1) < 1 + \sum_{\sigma \in L_2} (|\sigma| - 1).
	\end{align*}
	
	\item \textbf{Number of Ties $\numberofties$:}
	Let $L_1, L_2$ be two event logs with $L_1 \sqsubset L_2$ and $\numberofties(L_1) < \numberofties(L_2)$.
	Then, there must be two activity names $a, b$ with $a >_{L_2} b$ and $b \not>_{L_2} a$, but with $a \not>_{L_1} b$ or $b >_{L_1} a$.
	Since $L_1 \sqsubset L_2$ and $b \not>_{L_2} a$, out of the latter two, only $a \not>_{L_1} b$ can be true.
	In turn, there must be a trace $\sigma \in supp(L_2) \setminus supp(L_1)$ with $\sigma(i) = a$ and $\sigma(i+1) = b$ for some $i \in \{1, \dots, |\sigma| - 1\}$.
	But then, $|\sigma| \geq 2$ and therefore $|\sigma| - 1 \geq 1$, so
	\begin{align*}
	1 + \sum_{\sigma \in L_1} (|\sigma| - 1) < 1 + \sum_{\sigma \in L_2} (|\sigma| - 1).
	\end{align*}
\end{itemize}
Thus, for any event logs $L_1, L_2$ and their trace nets $M_1, M_2$, we have shown that
$\tlmax(L_1) < \tlmax(L_2) \Rightarrow \density(M_1) > \density(M_2)$ and, similarly, $\numberofties(L_1) < \numberofties(L_2) \Rightarrow \density(M_1) > \density(M_2)$. \hfill$\square$
\end{proof}

\begin{theorem}
\label{theo:tracenet-duplicate-entries}
$(\mathcal{C}^L, \duplicate) \in \mleq$ for any log complexity measure $\mathcal{C}^L \in \loc$
\end{theorem}
\begin{proof}
By \cref{lemma:tracenet-monotone-increasing}, we know that $\duplicate(M_1) < \duplicate(M_2)$ for trace nets $M_1, M_2$ of event logs $L_1, L_2$ with $L_1 \sqsubset L_2$.
What remains to be shown is that both $\duplicate(M_1) = \duplicate(M_2)$ and $\duplicate(M_1) < \duplicate(M_2)$ are possible.
For the former, take the following event logs:
\begin{align*}
	L_1 &= [\langle a \rangle, \langle a,b,c,d \rangle] \\
	L_2 &= L_1 + [\langle u,v,w,x,y,z \rangle^{2}]
\end{align*}
These two event logs have the following log complexity scores:
\begin{center}
	\begin{tabular}{|c|c|c|c|c|c|c|c|c|c|c|c|c|} \hline
		 & $\magnitude$ & $\variety$ & $\support$ & $\tlavg$ & $\tlmax$ & $\levelofdetail$ & $\numberofties$ & $\lempelziv$ & $\numberuniquetraces$ & $\percentageuniquetraces$ \\ \hline
		$L_1$ & $\pad 5 \pad$ & $\pad 4 \pad$ & $\pad 2 \pad$ & $\pad 2.5 \pad$ & $\pad 4 \pad$ & $\pad 2 \pad$ & $\pad 3 \pad$ & $\pad 4 \pad$ & $\pad 2 \pad$ & $\pad 1 \pad$ \\ \hline
		$L_2$ & $\pad 17 \pad$ & $\pad 10 \pad$ & $\pad 4 \pad$ & $\pad 4.25 \pad$ & $\pad 6 \pad$ & $\pad 3 \pad$ & $\pad 8 \pad$ & $\pad 13 \pad$ & $\pad 3 \pad$ & $\pad 0.75 \pad$ \\ \hline
	\end{tabular}
		
	\medskip
		
	\begin{tabular}{|c|c|c|c|c|c|c|c|c|} \hline
		 & $\structure$ & $\affinity$ & $\deviationfromrandom$ & $\avgdist$ & $\varentropy$ & $\normvarentropy$ & $\seqentropy$ & $\normseqentropy$ \\ \hline
		$L_1$ & $\pad 2.5 \pad$ & $\pad 0 \pad$ & $\pad 0.4796 \pad$ & $\pad 3 \pad$ & $\pad 0 \pad$ & $\pad 0 \pad$ & $\pad 0 \pad$ & $\pad 0 \pad$ \\ \hline
		$L_2$ & $\pad 4.25 \pad$ & $\pad 0.1667 \pad$ & $\pad 0.6449 \pad$ & $\pad 6.1667 \pad$ & $\pad 6.7301 \pad$ & $\pad 0.2923 \pad$ & $\pad 10.2986 \pad$ & $\pad 0.2138 \pad$ \\ \hline
	\end{tabular}
\end{center}
Thus, $\mathcal{C}^L(L_1) < \mathcal{C}^L(L_2)$ for any log complexity measure except the percentage of distinct traces $\percentageuniquetraces$.
However, the number of duplicate tasks in the trace nets $M_1, M_2$ for the logs $L_1, L_2$ are the same: $\duplicate(M_1) = 1 = \duplicate(M_2)$.
To see that there is also such an example for the percentage of distinct traces $\percentageuniquetraces$, we take the event logs above and change their frequencies:
\begin{align*}
	L_1 &= [\langle a \rangle^{4}, \langle a,b,c,d \rangle] \\
	L_2 &= L_1 + [\langle u,v,w,x,y,z \rangle^{2}]
\end{align*}
Then, $\percentageuniquetraces(L_1) = 0.4 < 0.4286 \approx \percentageuniquetraces(L_2)$, but $\duplicate(M_1) = 1 = \duplicate(M_2)$.

To see that $\duplicate$ can also increase, take the following event logs:
\begin{align*}
	L_1 &= [\langle a \rangle^{2}, \langle a,b,c,d \rangle^{3}] \\
	L_2 &= L_1 + [\langle e,a,b,c,d \rangle^{2}]
\end{align*}
These two event logs have the following log complexity scores:
\begin{center}
	\begin{tabular}{|c|c|c|c|c|c|c|c|c|c|c|c|c|} \hline
		 & $\magnitude$ & $\variety$ & $\support$ & $\tlavg$ & $\tlmax$ & $\levelofdetail$ & $\numberofties$ & $\lempelziv$ & $\numberuniquetraces$ & $\percentageuniquetraces$ \\ \hline
		$L_1$ & $\pad 14 \pad$ & $\pad 4 \pad$ & $\pad 5 \pad$ & $\pad 2.8 \pad$ & $\pad 4 \pad$ & $\pad 2 \pad$ & $\pad 3 \pad$ & $\pad 8 \pad$ & $\pad 2 \pad$ & $\pad 0.4 \pad$ \\ \hline
		$L_2$ & $\pad 24 \pad$ & $\pad 5 \pad$ & $\pad 7 \pad$ & $\pad 3.4286 \pad$ & $\pad 5 \pad$ & $\pad 4 \pad$ & $\pad 4 \pad$ & $\pad 11 \pad$ & $\pad 3 \pad$ & $\pad 0.4286 \pad$ \\ \hline
	\end{tabular}
		
	\medskip
		
	\begin{tabular}{|c|c|c|c|c|c|c|c|c|} \hline
		 & $\structure$ & $\affinity$ & $\deviationfromrandom$ & $\avgdist$ & $\varentropy$ & $\normvarentropy$ & $\seqentropy$ & $\normseqentropy$ \\ \hline
		$L_1$ & $\pad 2.8 \pad$ & $\pad 0.4 \pad$ & $\pad 0.4796 \pad$ & $\pad 1.8 \pad$ & $\pad 0 \pad$ & $\pad 0 \pad$ & $\pad 0 \pad$ & $\pad 0 \pad$ \\ \hline
		$L_2$ & $\pad 3.4286 \pad$ & $\pad 0.4524 \pad$ & $\pad 0.5169 \pad$ & $\pad 1.9048 \pad$ & $\pad 6.1827 \pad$ & $\pad 0.3126 \pad$ & $\pad 16.3006 \pad$ & $\pad 0.2137 \pad$ \\ \hline
	\end{tabular}
\end{center}
Thus, $\mathcal{C}^L(L_1) < \mathcal{C}^L(L_2)$ for every log complexity measure $\mathcal{C}^L \in \loc$. 
But the trace nets $M_1, M_2$ for the logs $L_1, L_2$ fulfill $\duplicate(M_1) = 1 < 5 = \duplicate(M_2)$. \hfill$\square$
\end{proof}

For the remainder of this subsection, we will analyze the model complexity of the trace net in more depth and characterize the model complexity scores of the trace net by using log complexity measures.
Since many complexity scores for the trace net are dependent of the amount of places in the trace net, we define 
\[\mathcal{N}(L) := \sum_{\sigma \in L} (|\sigma| - 1)\]
for an event log $L$ as the total number of neighborhoods in distinct traces of $L$.
Since two transitions in the trace net are connected via a place, the total amount of places in the trace net is $2 + \mathcal{N}(L)$.
We need to increase $\mathcal{N}(L)$ by two for the initial place $p_i$ and the final place $p_o$.
With this notion, we can now analyze the model complexity scores of the trace net $M$ for an event log $L$ over a set of activities $A$.

\begin{itemize}
	\item \textbf{Size $\size$:}
	As argued before, the trace net contains $2 + \mathcal{N}(L)$ places.
	Furthermore, it contains $\sum_{\sigma \in L} |\sigma|$ transitions.
	Thus, we have:
	\begin{align*}
	\size(M) &= 2 + \sum_{\sigma \in L} (2|\sigma| + 1) = 2 + \left(\sum_{\sigma \in L} 2(|\sigma| - 1)\right) + |supp(L)| \\
	&= 2 + 2\mathcal{N}(L) + \numberuniquetraces(L)
	\end{align*}
	
	\item \textbf{Connector Mismatch $\mismatch$:}
	If $|supp(L)| = 1$, there are no connectors in $M$ and $\mismatch(M) = 0$.
	Otherwise, the only connectors in $M$ are $p_i$ and $p_o$.
	The place $p_i$ has $|supp(L)|$ outgoing, while the place $p_o$ has $|supp(L)|$ incoming edges.
	Thus, we get $\mismatch(M) = \left| |supp(L)| - |supp(L)| \right| = 0$.
	
	\item \textbf{Connector Heterogeneity $\connhet$:}
	If $|supp(L)| = 1$, there are no connectors in $M$ and $\connhet$ is undefined.
	Otherwise, the only connectors in $M$ are $p_i$ and $p_o$.
	Both of these connectors are \texttt{xor}-connectors, so we get the connector heterogeneity score $\connhet(M) = - \left(1 \cdot \log_2(1) + 0 \cdot \log_2(0)\right) = 0$.
	
	\item \textbf{Cross Connectivity $\crossconn$:}
	For readability, let $n := |supp(L)|$. 
	There are only two nodes in $M$ that have a weight $\neq 1$: $p_i$ and $p_o$.
	This results in only the edges leaving $p_i$ and the edges entering $p_o$ having weight $\frac{1}{n}$, while all other edges in $M$ have weight $1$.
	Thus, or the connection values in the trace net $M$, we get:
	\begin{itemize}
		\item $V(p_i, x) = \frac{1}{n}$ for all nodes $x$ of $M$ with $x \neq p_i$ and $x \neq p_o$,
		\item $V(p_i, p_o) = \frac{1}{n^2}$,
		\item Let $\sigma \in L$ be a trace and $i \in \{1, \dots, |\sigma|\}$. 
		Then, the transition for $\sigma(i)$ has $|\sigma| - i$ times the value $1$ with succeeding transitions, $|\sigma| - 1$ times value $1$ with succeeding places except $p_o$, and value $\frac{1}{n}$ with the place $p_o$.
		\item Let $\sigma \in L$ be a trace and $i \in \{1, \dots, |\sigma|\}$. 
		Then, the place $p_{\sigma(i)}$ in the postset of the transition for $\sigma(i)$ has $|\sigma| - i$ times the value $1$ with succeeding transitions, $|\sigma| - i - 1$ times the value $1$ with succeeding places except $p_o$, and value $\frac{1}{n}$ with the place $p_o$. 
	\end{itemize}
	Since all other connections have value $0$, we get, for the sum of these values:
	\begin{align*}
	&\frac{1}{n^2} + \sum_{\sigma \in L} \left(\sum_{i = 1}^{|\sigma|} \left(2 (|\sigma| - i) + \frac{2}{n}\right) + \sum_{j = 1}^{|\sigma|-1} \left(2\left(|\sigma| - j\right) - 1 + \frac{2}{n}\right)\right) \\
	=\: &\frac{1}{n^2} + \sum_{\sigma \in L} \frac{(2|\sigma| - 1) \cdot (n(|\sigma| - 1) + 2)}{n} \\
	=\: &\frac{1}{n} \cdot \left(\frac{1}{n} + \sum_{\sigma \in L} (2|\sigma| - 1) \cdot (n(|\sigma| - 1) + 2)\right) \\
	=\: &\frac{1}{n} \cdot \left(\frac{1}{n} + n \cdot \sum_{\sigma \in L} (2|\sigma| - 1) \cdot \left(|\sigma| - 1 + \frac{2}{n}\right)\right)
	\end{align*}
	In turn, the cross connectivity of the trace net is:
	\[\crossconn(M) = 1 - \frac{\frac{1}{n^2} + \sum_{\sigma \in L} (2|\sigma| - 1) \cdot \left(|\sigma| - 1 + \frac{2}{n}\right)}{\left(2 + 2\mathcal{N}(L) + \numberuniquetraces(L)\right) \cdot \left(1 + 2\mathcal{N}(L) + \numberuniquetraces(L))\right)}\]
	
	\item \textbf{Token Split $\tokensplit$:}
	Since every transition in $M$ has exactly one incoming and one outgoing edge, it contains no \texttt{and}-splits.
	Therefore, $\tokensplit(M) = 0$.
	
	\item \textbf{Control Flow Complexity $\controlflow$:}
	If $|supp(L)| = 1$, the trace net $M$ does not contain any connectors, and thus $\controlflow(M) = 0$.
	Otherwise, the only connector nodes in $M$ are $p_i$ and $p_o$.
	$p_i$ is an \texttt{xor}-split and $p_o$ an \texttt{xor}-join, so $\controlflow(M) = |\post{p_i}| = |supp(L)| = \numberuniquetraces(L)$.
	
	\item \textbf{Separability $\separability$:}
	If $|supp(L)| = 1$, every node in $M$ except $p_i$ and $p_o$ is a cut-vertex, so $\separability(M) = 0$.
	Otherwise, $M$ does not contain any cut-vertices, so $\separability(M) = 1$.
	
	\item \textbf{Average Connector Degree $\avgconn$:}
	If $|supp(L)| = 1$, the trace net $M$ contains no connectors and thus, the average connector degree is undefined.
	Otherwise, only the places $p_i$ and $p_o$ are connectors in $M$.
	Both places have degree $|supp(L)|$, so the average connector degree of the trace net is $\avgconn(M) = \frac{|supp(L)| + |supp(L)|}{2} = |supp(L)| = \numberuniquetraces(L)$.
	
	\item \textbf{Maximum Connector Degree $\maxconn$:}
	If $|supp(L)| = 1$, the trace net $M$ contains no connectors and thus, the maximum connector degree is undefined.
	Otherwise, only the places $p_i$ and $p_o$ are connectors in $M$.
	Both places have degree $|supp(L)|$, so the maximum connector degree of the trace net is $\maxconn(M) = |supp(L)| = \numberuniquetraces(L)$.
	
	\item \textbf{Sequentiality $\sequentiality$:}
	If $|supp(L)| = 1$, the trace net $M$ contains no connectors and thus, every edge in $M$ connects two non-connector nodes, leading to $\sequentiality(M) = 0$.
	Otherwise, only the edges leaving $p_i$ or entering $p_o$ have a connector node at their head or tail.
	Since $M$ contains $2|T|$ edges in total, we get $\sequentiality(M) = 1 - \frac{2|T| - 2|supp(L)|}{2|T|} = \frac{2|supp(L)|}{2|T|} = \frac{\numberuniquetraces(L)}{\mathcal{N}(L) + \numberuniquetraces(L)}$.
	
	\item \textbf{Depth $\depth$:}
	If $|supp(L)| = 1$, there trace net $M$ does not contain any connectors, and thus, the in- and out-depth of every node in $M$ is $0$, leading to $\depth(M) = 0$.
	Otherwise, every node except $p_i$ and $p_o$ have in- and out-depth $1$, while $p_i$ and $p_o$ have in- and out-depth $0$.
	Thus, $\depth(M) = \max\{0,1\} = 1$.
	
	\item \textbf{Diameter $\diameter$:}
	The diameter of the trace net $M$ is dependent on the length of the longest trace in $L$.
	Let $\sigma \in L$ be a trace with maximum length in $L$.
	Then, one of the paths through the trace net $M$ with maximal length is the path $(p_i, \sigma(1), p_{\sigma(1)}, \sigma(2), \dots, p_{|\sigma|-1}, \sigma(|\sigma|), p_o)$, where $p_{\sigma(i)}$ is the place in the postset of the transition for $\sigma(i)$ for any $i \in \{1, \dots, |\sigma\}$. 
	The length of this path is $\diameter(M) = 1 + 2 \max\{|\sigma| \mid \sigma \in L\} = 1 + 2 \cdot \tlmax(L)$.
	
	\item \textbf{Cyclicity $\cyclicity$:}
	The trace net $M$ does not introduce any cycles, so it has no nodes that lie on such cycles.
	Therefore, $\cyclicity(M) = 0$.
	
	\item \textbf{Coefficent of Network Connectivity $\netconn$:}
	Since by construction, every transition in $M$ has exactly one incoming and one outgoing edge, $M$ contains $2|T|$ edges in total.
	Thus, $\netconn(M) = \frac{2|T|}{|P| + |T|} = \frac{2(\mathcal{N}(L) + \numberuniquetraces(L))}{2 + 2 \mathcal{N}(L) + \numberuniquetraces(L)}$.
	
	\item \textbf{Density $\density$:}
	Since by construction, every transition in $M$ has exactly one incomin and one outgoing edge, $M$ contains $2|T|$ edges in total.
	Thus, $\density(M) = \frac{2|T|}{2|T|(|P| - 1)} = \frac{1}{|P| - 1} = \frac{1}{1 + \mathcal{N}(L)}$.
	
	\item \textbf{Number of Duplicate Tasks $\duplicate$:}
	The number of duplicate tasks in the trace net $M$ is exactly the amount of event name repetitions in the support of the event log $L$.
	Thus, $\duplicate(M) = \sum_{a \in A} (|\{(i,j) \mid \sigma_i \in L: \sigma_i(j) = a\}| - 1)$.
	
	\item \textbf{Number of Empty Sequence Flows $\emptyseq$:}
	Since the trace net $M$ does not contain any \texttt{and}-connectors, the number of empty sequence flows in $M$ is $\emptyseq(M) = 0$ by definition.
\end{itemize}
These findings conclude our analysis of the trace net miner.
\cref{table:tracenet-model-complexity} summarizes these findings for quick reference.

\begin{table}[p]
	\caption{The complexity scores of the trace net $M$ for an event log $L$ over $A$.}
	\label{table:tracenet-model-complexity}
	\centering
	\def\pad{\hspace*{1.5mm}}
	\renewcommand{\arraystretch}{1.25}
	\begin{tabular}{|r|l|} \hline
	$\pad\size(M)\pad$ & $\pad 2 + 2\mathcal{N}(L) + \numberuniquetraces(L) \pad$ \\ \hline
	$\pad\mismatch(M)\pad$ & $\pad 0 \pad$ \\ \hline
	$\pad\connhet(M)\pad$ & $\pad 0 \pad$ \\ \hline
	$\pad\crossconn(M)\pad$ & $\pad 1 - \frac{\frac{1}{\numberuniquetraces(L)^2} + \sum_{\sigma \in L} (2|\sigma| - 1) \cdot \left(|\sigma| - 1 + \frac{2}{\numberuniquetraces(L)}\right)}{\left(2 + 2\mathcal{N}(L) + \numberuniquetraces(L)\right) \cdot \left(1 + 2\mathcal{N}(L) + \numberuniquetraces(L))\right)} \pad$ \\ \hline
	$\pad\tokensplit(M)\pad$ & $\pad 0 \pad$ \\ \hline
	$\pad\controlflow(M)\pad$ & $\pad \begin{cases} 0 & \text{if } \numberuniquetraces(L) = 1 \\ \numberuniquetraces(L) & \text{otherwise} \end{cases} \pad$ \\ \hline
	$\pad\separability(M)\pad$ & $\pad \begin{cases} 0 & \text{if } \numberuniquetraces(L) = 1 \\ 1 & \text{otherwise} \end{cases} \pad$ \\ \hline
	$\pad\avgconn(M)\pad$ & $\pad \numberuniquetraces(L) \pad$ \\ \hline
	$\pad\maxconn(M)\pad$ & $\pad \numberuniquetraces(L) \pad$ \\ \hline
	$\pad\sequentiality(M)\pad$ & $\pad \begin{cases} 0 & \text{if } \numberuniquetraces(L) = 1 \\ \frac{\numberuniquetraces(L)}{\mathcal{N}(L) + \numberuniquetraces(L)} & \text{otherwise} \end{cases} \pad$ \\ \hline
	$\pad\depth(M)\pad$ & $\pad \begin{cases} 0 & \text{if } \numberuniquetraces(L) = 1 \\ 1 & \text{otherwise} \end{cases} \pad$ \\ \hline
	$\pad\diameter(M)\pad$ & $\pad 1 + 2 \cdot \tlmax(L) \pad$ \\ \hline
	$\pad\cyclicity(M)\pad$ & $\pad 0 \pad$ \\ \hline
	$\pad\netconn(M)\pad$ & $\pad \frac{2(\mathcal{N}(L) + \numberuniquetraces(L))}{2 + 2 \mathcal{N}(L) + \numberuniquetraces(L)} \pad$ \\ \hline
	$\pad\density(M)\pad$ & $\pad \frac{1}{1 + \mathcal{N}(L)} \pad$ \\ \hline
	$\pad\duplicate(M)\pad$ & $\pad \sum_{a \in A} (|\{(i,j) \mid \sigma_i \in L: \sigma_i(j) = a\}| - 1) \pad$ \\ \hline
	$\pad\emptyseq(M)\pad$ & $\pad 0 \pad$ \\ \hline
	\end{tabular}
\end{table}

\newpage
\subsection{Alpha Miner}
\label{sec:alpha}
The alpha miner~\cite{AalWM04} is one of the first algorithms introduced for process discovery.
It calculates a Petri net for an event log by first constructing the causal footprint of the log and then analyzing which activities should directly follow each other.
As a first example, take the following event log:
\[L_1 = [\langle a,b,c,d,e \rangle, \langle a,b,d,c,e \rangle, \langle a,u,v,x,y,z \rangle]\]
The set of activities occuring in $L_1$ is $A_{L_1} = \{a,b,c,d,e,u,v,x,y,z\}$.
For each of these activities, we create a row and a cell in a matrix we call the causal footprint, and use it as a table to show the relation between two activities.
\begin{center}
	\begin{tabular}{|c|c|c|c|c|c|c|c|c|c|c|} \hline
	 & $a$ & $b$ & $c$ & $d$ & $e$ & $u$ & $v$ & $x$ & $y$ & $z$ \\ \hline
	$a$ & $\#$ & $\rightarrow$ & $\#$ & $\#$ & $\#$ & $\rightarrow$ & $\#$ & $\#$ & $\#$ & $\#$ \\ \hline
	$b$ & $\leftarrow$ & $\#$ & $\rightarrow$ & $\rightarrow$ & $\#$ & $\#$ & $\#$ & $\#$ & $\#$ & $\#$ \\ \hline
	$c$ & $\#$ & $\leftarrow$ & $\#$ & $||$ & $\rightarrow$ & $\#$ & $\#$ & $\#$ & $\#$ & $\#$ \\ \hline
	$d$ & $\#$ & $\leftarrow$ & $||$ & $\#$ & $\rightarrow$ & $\#$ & $\#$ & $\#$ & $\#$ & $\#$ \\ \hline
	$e$ & $\#$ & $\#$ & $\leftarrow$ & $\leftarrow$ & $\#$ & $\#$ & $\#$ & $\#$ & $\#$ & $\#$ \\ \hline
	$u$ & $\leftarrow$ & $\#$ & $\#$ & $\#$ & $\#$ & $\#$ & $\rightarrow$ & $\#$ & $\#$ & $\#$ \\ \hline
	$v$ & $\#$ & $\#$ & $\#$ & $\#$ & $\#$ & $\leftarrow$ & $\#$ & $\rightarrow$ & $\#$ & $\#$ \\ \hline
	$x$ & $\#$ & $\#$ & $\#$ & $\#$ & $\#$ & $\#$ & $\leftarrow$ & $\#$ & $\rightarrow$ & $\#$ \\ \hline
	$y$ & $\#$ & $\#$ & $\#$ & $\#$ & $\#$ & $\#$ & $\#$ & $\leftarrow$ & $\#$ & $\rightarrow$ \\ \hline
	$z$ & $\#$ & $\#$ & $\#$ & $\#$ & $\#$ & $\#$ & $\#$ & $\#$ & $\leftarrow$ & $\#$ \\ \hline
	\end{tabular}
\end{center}
The general idea is to create transitions $a$ for every $a \in A_{L_1}$ and connect two transitions $a, b$ via a place if $a \rightarrow b$.
To do this, first define
\begin{align*}
X_{L} = \{&(B, C) \mid B \subseteq A_{L} \land B \neq \emptyset \land C \subseteq A_{L} \land C \neq \emptyset\, \land \\
&\forall b \in B, c \in C: b \rightarrow c \land \forall b_1, b_2 \in B: b_1 \# b_2 \land \forall c_1, c_2 \in C: c_1 \# c_2\}
\end{align*}
for any event log $L$ over a set of activities $A_L$. 
Intuitively, $X_{L}$ contains all pairs of activity-name-sets where all activities of the first set are in directly follows relation ($\rightarrow$) to all activities of the second set.
In order to model concurrency correctly, all elements of one set must be incomparable ($\#$) to each other.
In the example above, $(\{b\}, \{c\}), (\{b\}, \{d\}) \in X_{L_1}$, but $(\{b\}, \{c,d\}) \not\in X_{L_1}$, because $c$ and $d$ are parallel to each other, and therefore do not fulfill $c \# d$.
Using this set to define the places of the output net would result in many implicit places, so the alpha miner instead uses the most expressive tuples of $X_L$:
\begin{align*}
Y_L = \{(B,C) \in X_L \mid \forall (B',C') \in X_L: (B \subseteq B' \land C \subseteq C') \Rightarrow (B,C) = (B',C')\}
\end{align*}
Thus, we only keep tuples that are maximal in the sense that the sets of no other tuple contain the sets of the maximal tuple.
In the example above, this means that, even though $(\{a\}, \{b\}), (\{a\}, \{u\}) \in X_{L_1}$, both of these tuples are not included in $Y_{L_1}$, because $(\{a\}, \{b,u\}) \in  X_{L_1}$ and we have that $\{a\} \subseteq \{a\}$ and $\{b\},\{u\} \subseteq \{b,u\}$.
As mentioned earlier, each of the tuples in $Y_L$ correspond to a place in the resulting Petri net.
On top of that, we have two special places $p_i$ and $p_o$, where $p_i$ is the initially marked input place and $p_o$ is the place that defines the final marking of the net.
The alpha miner creates edges from $p_i$ to all transitions whose label occurs first in any trace, i.e., $\post{p_i} = A_I = \{a \mid \exists \sigma \in L: \sigma(1) = a\}$.
Furthermore, it creates edges to $p_o$ from all transitions whose label occurs last in any trace, i.e., $\pre{p_o} = A_O = \{a \mid \exists \sigma \in L: \sigma(|\sigma|) = a\}$.
\cref{fig:alpha-example-L1} shows the Petri net found by the alpha miner for the input event log $L_1$.

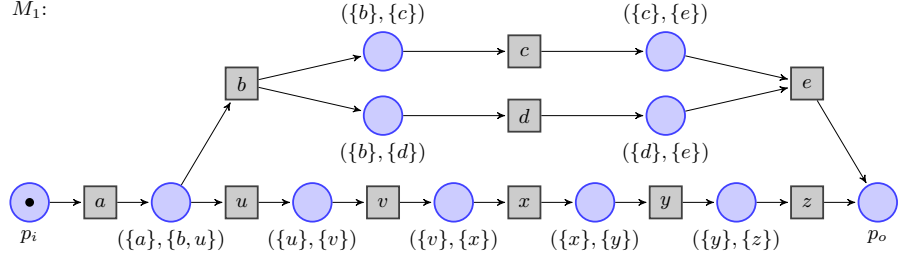
\begin{figure}
	\centering
	\scalebox{\scalefactor}{
	\begin{tikzpicture}[node distance = 1.1cm,>=stealth',bend angle=0,auto]
		\node[place,tokens=1,label=below:$p_i$] (start) {};
		\node[yshift=3cm] at (start) {$M_1$:};
		\node[transition,right of=start] (a) {$a$}
		edge [pre] (start);
		\node[place,right of=a,label=below:{$(\{a\}, \{b,u\})$}] (p1) {}
		edge [pre] (a);
		\node[transition,right of=p1] (u) {$u$}
		edge [pre] (p1);
		\node[place,right of=u,label=below:{$(\{u\}, \{v\})$}] (p2) {}
		edge [pre] (u);
		\node[transition,right of=p2] (v) {$v$}
		edge [pre] (p2);
		\node[place,right of=v,label=below:{$(\{v\}, \{x\})$}] (p3) {}
		edge [pre] (v);
		\node[transition,right of=p3] (x) {$x$}
		edge [pre] (p3);
		\node[place,right of=x,label=below:{$(\{x\}, \{y\})$}] (p4) {}
		edge [pre] (x);
		\node[transition,right of=p4] (y) {$y$}
		edge [pre] (p4);
		\node[place,right of=y,label=below:{$(\{y\}, \{z\})$}] (p5) {}
		edge [pre] (y);
		\node[transition,right of=p5] (z) {$z$}
		edge [pre] (p5);
		\node[place,right of=z,label=below:$p_o$] (end) {}
		edge [pre] (z);
		\node[transition,above of=u,yshift=0.75cm] (b) {$b$}
		edge [pre] (p1);
		\node[place,above of=v,yshift=1.25cm,label=above:{$(\{b\}, \{c\})$}] (p6) {}
		edge [pre] (b);
		\node[place,above of=v,yshift=0.25cm,label=below:{$(\{b\}, \{d\})$}] (p7) {}
		edge [pre] (b);
		\node[transition,above of=x,yshift=1.25cm] (c) {$c$}
		edge [pre] (p6);
		\node[transition,above of=x,yshift=0.25cm] (d) {$d$}
		edge [pre] (p7);
		\node[place,above of=y,yshift=1.25cm,label=above:{$(\{c\}, \{e\})$}] (p8) {}
		edge [pre] (c);
		\node[place,above of=y,yshift=0.25cm,label=below:{$(\{d\}, \{e\})$}] (p9) {}
		edge [pre] (d);
		\node[transition,above of=z,yshift=0.75cm] (e) {$e$}
		edge [pre] (p8)
		edge [pre] (p9)
		edge [post] (end);
	\end{tikzpicture}
	}
	\caption{The output of the alpha algorithm for the input event log $L_1$.}
	\label{fig:alpha-example-L1}
\end{figure}

The result of the alpha algorithm is not always sound, which we will use to our advantage during the analyses of this section.
For example, take the following event log $L_2$, which is a proper superset of the event log $L_1$:
\[L_2 = [\langle a,b,c,d,e \rangle, \langle a,b,d,c,e \rangle, \langle a,u,v,x,y,z \rangle, \langle a,b,c,d,e,f,g,h \rangle]\]
The only change to the event log $L_1$ is that the trace $\langle a,b,c,d,e \rangle$ can be extended by the events $f$, $g$, and $h$ in that order.
The result of the alpha miner for this event log is shown in \cref{fig:alpha-example-L2}.
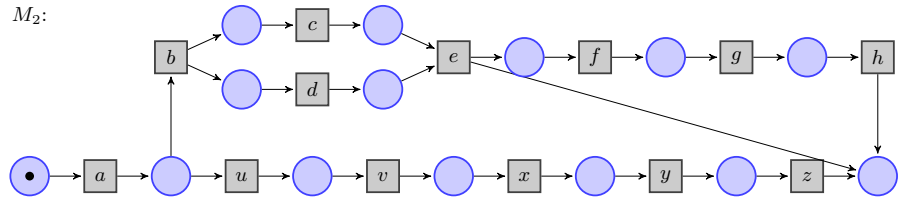
\begin{figure}
	\centering
	\scalebox{\scalefactor}{
	\begin{tikzpicture}[node distance = 1.1cm,>=stealth',bend angle=0,auto]
		\node[place,tokens=1] (start) {};
		\node[yshift=2.5cm] at (start) {$M_2$:};
		\node[transition,right of=start] (a) {$a$}
		edge [pre] (start);
		\node[place,right of=a] (p1) {}
		edge [pre] (a);
		\node[transition,right of=p1] (u) {$u$}
		edge [pre] (p1);
		\node[place,right of=u] (p2) {}
		edge [pre] (u);
		\node[transition,right of=p2] (v) {$v$}
		edge [pre] (p2);
		\node[place,right of=v] (p3) {}
		edge [pre] (v);
		\node[transition,right of=p3] (x) {$x$}
		edge [pre] (p3);
		\node[place,right of=x] (p4) {}
		edge [pre] (x);
		\node[transition,right of=p4] (y) {$y$}
		edge [pre] (p4);
		\node[place,right of=y] (p5) {}
		edge [pre] (y);
		\node[transition,right of=p5] (z) {$z$}
		edge [pre] (p5);
		\node[place,right of=z] (end) {}
		edge [pre] (z);
		\node[transition,above of=p1,yshift=0.75cm] (b) {$b$}
		edge [pre] (p1);
		\node[place,above of=u,yshift=1.25cm] (p6) {}
		edge [pre] (b);
		\node[place,above of=u,yshift=0.25cm] (p7) {}
		edge [pre] (b);
		\node[transition,above of=p2,yshift=1.25cm] (c) {$c$}
		edge [pre] (p6);
		\node[transition,above of=p2,yshift=0.25cm] (d) {$d$}
		edge [pre] (p7);
		\node[place,above of=v,yshift=1.25cm] (p8) {}
		edge [pre] (c);
		\node[place,above of=v,yshift=0.25cm] (p9) {}
		edge [pre] (d);
		\node[transition,above of=p3,yshift=0.75cm] (e) {$e$}
		edge [pre] (p8)
		edge [pre] (p9)
		edge [post] (end);
		\node[place,right of=e] (p10) {}
		edge [pre] (e);
		\node[transition,right of=p10] (f) {$f$}
		edge [pre] (p10);
		\node[place,right of=f] (p11) {}
		edge [pre] (f);
		\node[transition,right of=p11] (g) {$g$}
		edge [pre] (p11);
		\node[place,right of=g] (p12) {}
		edge [pre] (g);
		\node[transition,right of=p12] (h) {$h$}
		edge [pre] (p12)
		edge [post] (end);
	\end{tikzpicture}
	}
	\caption{The output of the alpha algorithm for the input event log $L_2$.}
	\label{fig:alpha-example-L2}
\end{figure}
In this Petri net, the final place $p_o$ can contain $2$ tokens at once, when the transitions $a,b,c,d,e,f,g,h$ fire in that sequence.
We can use this property to increase the token split or connector mismatch score without changing much of the behavior.
Sometimes, the output of the alpha miner is not a workflow net, as it can contain isolated nodes.
For example:
\[L_3 = L_2 + [\langle d \rangle, \langle g \rangle, \langle c,e \rangle, \langle a,b,c,d,e \rangle^2, \langle b,b,c,d,d,e,f,f,g,g,h,h \rangle]\]
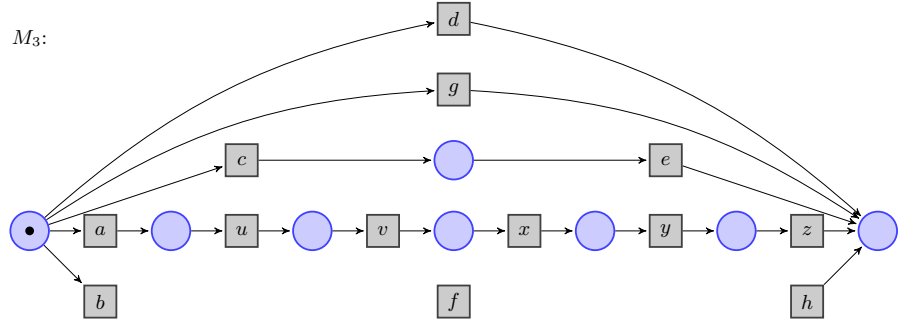
\begin{figure}
	\centering
	\scalebox{\scalefactor}{
	\begin{tikzpicture}[node distance = 1.1cm,>=stealth',bend angle=0,auto]
		\node[place,tokens=1] (start) {};
		\node[yshift=3cm] at (start) {$M_3$:};
		\node[transition,right of=start] (a) {$a$}
		edge [pre] (start);
		\node[place,right of=a] (p1) {}
		edge [pre] (a);
		\node[transition,right of=p1] (u) {$u$}
		edge [pre] (p1);
		\node[place,right of=u] (p2) {}
		edge [pre] (u);
		\node[transition,right of=p2] (v) {$v$}
		edge [pre] (p2);
		\node[place,right of=v] (p3) {}
		edge [pre] (v);
		\node[transition,right of=p3] (x) {$x$}
		edge [pre] (p3);
		\node[place,right of=x] (p4) {}
		edge [pre] (x);
		\node[transition,right of=p4] (y) {$y$}
		edge [pre] (p4);
		\node[place,right of=y] (p5) {}
		edge [pre] (y);
		\node[transition,right of=p5] (z) {$z$}
		edge [pre] (p5);
		\node[place,right of=z] (end) {}
		edge [pre] (z);
		\node[transition,below of=a] (b) {$b$}
		edge [pre] (start);
		\node[transition,below of=z] (h) {$h$}
		edge [post] (end);
		\node[transition,below of=p3] (f) {$f$};
		\node[transition,above of=u] (c) {$c$}
		edge [pre] (start);
		\node[place, above of=p3] (p) {}
		edge [pre] (c);
		\node[transition,above of=y] (e) {$e$}
		edge [pre] (p)
		edge [post] (end);
		\node[transition,above of=p] (g) {$g$}
		edge [pre,bend right=15] (start)
		edge [post,bend left=15] (end);
		\node[transition,above of=g] (d) {$d$}
		edge [pre,bend right=15] (start)
		edge [post,bend left=15] (end);
	\end{tikzpicture}
	}
	\caption{The output of the alpha algorithm for the input event log $L_3$.}
	\label{fig:alpha-example-L3}
\end{figure}
\cref{fig:alpha-example-L3} shows the result of the alpha miner for the event log $L_3$, containing isolated nodes.
Due to this behavior of the alpha algorithm, we can find counter-examples showing that none of the log complexity measures can predict the model complexity of the alpha algorithm's output.
The only exception is the number of duplicate tasks $\duplicate$, which is always $0$, since we create exactly one transition for each distinct activity in the event log.
\cref{table:alphaminer-findings} shows our findings.

\medskip

\begin{table}[ht]
	\caption{The relations between the complexity scores of two nets $M_1$ and $M_2$ found by the alpha miner for the event logs $L_1$ and $L_2$ as input respectively, where $L_1 \sqsubset L_2$ and the complexity of $L_1$ is lower than the complexity of $L_2$.}
	\label{table:alphaminer-findings}
	\centering
	\resizebox{\textwidth}{!}{
	\begin{tabular}{|c|c|c|c|c|c|c|c|c|c|c|c|c|c|c|c|c|c|} \hline
		 & $\size$ & $\mismatch$ & $\connhet$ & $\crossconn$ & $\tokensplit$ & $\controlflow$ & $\separability$ & $\avgconn$ & $\maxconn$ & $\sequentiality$ & $\depth$ & $\diameter$ & $\cyclicity$ & $\netconn$ & $\density$ & $\duplicate$ & $\emptyseq$ \\ \hline
		$\magnitude$ & \hyperref[theo:alphaminer-size-ts-cfc-acd-mcd-cyc-empty-entries]{$\norel$} & \hyperref[theo:alphaminer-mm-ch-depth-entries]{$\norel$} & \hyperref[theo:alphaminer-mm-ch-depth-entries]{$\norel$} & \hyperref[theo:alpha-cc-seq-entries]{$\norel^*$} & \hyperref[theo:alphaminer-size-ts-cfc-acd-mcd-cyc-empty-entries]{$\norel$} & \hyperref[theo:alphaminer-size-ts-cfc-acd-mcd-cyc-empty-entries]{$\norel$} & \hyperref[theo:alphaminer-sep-entries]{$\norel$} & \hyperref[theo:alphaminer-size-ts-cfc-acd-mcd-cyc-empty-entries]{$\norel$} & \hyperref[theo:alphaminer-size-ts-cfc-acd-mcd-cyc-empty-entries]{$\norel$} & \hyperref[theo:alpha-cc-seq-entries]{$\norel$} & \hyperref[theo:alphaminer-mm-ch-depth-entries]{$\norel$} & \hyperref[theo:alphaminer-diam-entries]{$\norel$} & \hyperref[theo:alphaminer-size-ts-cfc-acd-mcd-cyc-empty-entries]{$\norel$} & \hyperref[theo:alphaminer-cnc-entries]{$\norel$} & \hyperref[theo:alphaminer-dens-entries]{$\norel$} & \hyperref[theo:alphaminer-duplicate-entries]{$\meq$} & \hyperref[theo:alphaminer-size-ts-cfc-acd-mcd-cyc-empty-entries]{$\norel$} \\ \hline
		
		$\variety$ & \hyperref[theo:alphaminer-size-ts-cfc-acd-mcd-cyc-empty-entries]{$\norel$} & \hyperref[theo:alphaminer-mm-ch-depth-entries]{$\norel$} & \hyperref[theo:alphaminer-mm-ch-depth-entries]{$\norel$} & \hyperref[theo:alpha-cc-seq-entries]{$\norel^*$} & \hyperref[theo:alphaminer-size-ts-cfc-acd-mcd-cyc-empty-entries]{$\norel$} & \hyperref[theo:alphaminer-size-ts-cfc-acd-mcd-cyc-empty-entries]{$\norel$} & \hyperref[theo:alphaminer-sep-entries]{$\norel$} & \hyperref[theo:alphaminer-size-ts-cfc-acd-mcd-cyc-empty-entries]{$\norel$} & \hyperref[theo:alphaminer-size-ts-cfc-acd-mcd-cyc-empty-entries]{$\norel$} & \hyperref[theo:alpha-cc-seq-entries]{$\norel$} & \hyperref[theo:alphaminer-mm-ch-depth-entries]{$\norel$} & \hyperref[theo:alphaminer-diam-entries]{$\norel$} & \hyperref[theo:alphaminer-size-ts-cfc-acd-mcd-cyc-empty-entries]{$\norel$} & \hyperref[theo:alphaminer-cnc-entries]{$\norel$} & \hyperref[theo:alphaminer-dens-entries]{$\norel$} & \hyperref[theo:alphaminer-duplicate-entries]{$\meq$} & \hyperref[theo:alphaminer-size-ts-cfc-acd-mcd-cyc-empty-entries]{$\norel$} \\ \hline
		
		$\support$ & \hyperref[theo:alphaminer-size-ts-cfc-acd-mcd-cyc-empty-entries]{$\norel$} & \hyperref[theo:alphaminer-mm-ch-depth-entries]{$\norel$} & \hyperref[theo:alphaminer-mm-ch-depth-entries]{$\norel$} & \hyperref[theo:alpha-cc-seq-entries]{$\norel^*$} & \hyperref[theo:alphaminer-size-ts-cfc-acd-mcd-cyc-empty-entries]{$\norel$} & \hyperref[theo:alphaminer-size-ts-cfc-acd-mcd-cyc-empty-entries]{$\norel$} & \hyperref[theo:alphaminer-sep-entries]{$\norel$} & \hyperref[theo:alphaminer-size-ts-cfc-acd-mcd-cyc-empty-entries]{$\norel$} & \hyperref[theo:alphaminer-size-ts-cfc-acd-mcd-cyc-empty-entries]{$\norel$} & \hyperref[theo:alpha-cc-seq-entries]{$\norel$} & \hyperref[theo:alphaminer-mm-ch-depth-entries]{$\norel$} & \hyperref[theo:alphaminer-diam-entries]{$\norel$} & \hyperref[theo:alphaminer-size-ts-cfc-acd-mcd-cyc-empty-entries]{$\norel$} & \hyperref[theo:alphaminer-cnc-entries]{$\norel$} & \hyperref[theo:alphaminer-dens-entries]{$\norel$} & \hyperref[theo:alphaminer-duplicate-entries]{$\meq$} & \hyperref[theo:alphaminer-size-ts-cfc-acd-mcd-cyc-empty-entries]{$\norel$} \\ \hline
		
		$\tlavg$ & \hyperref[theo:alphaminer-size-ts-cfc-acd-mcd-cyc-empty-entries]{$\norel$} & \hyperref[theo:alphaminer-mm-ch-depth-entries]{$\norel$} & \hyperref[theo:alphaminer-mm-ch-depth-entries]{$\norel$} & \hyperref[theo:alpha-cc-seq-entries]{$\norel^*$} & \hyperref[theo:alphaminer-size-ts-cfc-acd-mcd-cyc-empty-entries]{$\norel$} & \hyperref[theo:alphaminer-size-ts-cfc-acd-mcd-cyc-empty-entries]{$\norel$} & \hyperref[theo:alphaminer-sep-entries]{$\norel$} & \hyperref[theo:alphaminer-size-ts-cfc-acd-mcd-cyc-empty-entries]{$\norel$} & \hyperref[theo:alphaminer-size-ts-cfc-acd-mcd-cyc-empty-entries]{$\norel$} & \hyperref[theo:alpha-cc-seq-entries]{$\norel$} & \hyperref[theo:alphaminer-mm-ch-depth-entries]{$\norel$} & \hyperref[theo:alphaminer-diam-entries]{$\norel$} & \hyperref[theo:alphaminer-size-ts-cfc-acd-mcd-cyc-empty-entries]{$\norel$} & \hyperref[theo:alphaminer-cnc-entries]{$\norel$} & \hyperref[theo:alphaminer-dens-entries]{$\norel$} & \hyperref[theo:alphaminer-duplicate-entries]{$\meq$} & \hyperref[theo:alphaminer-size-ts-cfc-acd-mcd-cyc-empty-entries]{$\norel$} \\ \hline
		
		$\tlmax$ & \hyperref[theo:alphaminer-size-ts-cfc-acd-mcd-cyc-empty-entries]{$\norel$} & \hyperref[theo:alphaminer-mm-ch-depth-entries]{$\norel$} & \hyperref[theo:alphaminer-mm-ch-depth-entries]{$\norel$} & \hyperref[theo:alpha-cc-seq-entries]{$\norel^*$} & \hyperref[theo:alphaminer-size-ts-cfc-acd-mcd-cyc-empty-entries]{$\norel$} & \hyperref[theo:alphaminer-size-ts-cfc-acd-mcd-cyc-empty-entries]{$\norel$} & \hyperref[theo:alphaminer-sep-entries]{$\norel$} & \hyperref[theo:alphaminer-size-ts-cfc-acd-mcd-cyc-empty-entries]{$\norel$} & \hyperref[theo:alphaminer-size-ts-cfc-acd-mcd-cyc-empty-entries]{$\norel$} & \hyperref[theo:alpha-cc-seq-entries]{$\norel$} & \hyperref[theo:alphaminer-mm-ch-depth-entries]{$\norel$} & \hyperref[theo:alphaminer-diam-entries]{$\norel$} & \hyperref[theo:alphaminer-size-ts-cfc-acd-mcd-cyc-empty-entries]{$\norel$} & \hyperref[theo:alphaminer-cnc-entries]{$\norel$} & \hyperref[theo:alphaminer-dens-entries]{$\norel$} & \hyperref[theo:alphaminer-duplicate-entries]{$\meq$} & \hyperref[theo:alphaminer-size-ts-cfc-acd-mcd-cyc-empty-entries]{$\norel$} \\ \hline
		
		$\levelofdetail$ & \hyperref[theo:alphaminer-size-ts-cfc-acd-mcd-cyc-empty-entries]{$\norel$} & \hyperref[theo:alphaminer-mm-ch-depth-entries]{$\norel$} & \hyperref[theo:alphaminer-mm-ch-depth-entries]{$\norel$} & \hyperref[theo:alpha-cc-seq-entries]{$\norel^*$} & \hyperref[theo:alphaminer-size-ts-cfc-acd-mcd-cyc-empty-entries]{$\norel$} & \hyperref[theo:alphaminer-size-ts-cfc-acd-mcd-cyc-empty-entries]{$\norel$} & \hyperref[theo:alphaminer-sep-entries]{$\norel$} & \hyperref[theo:alphaminer-size-ts-cfc-acd-mcd-cyc-empty-entries]{$\norel$} & \hyperref[theo:alphaminer-size-ts-cfc-acd-mcd-cyc-empty-entries]{$\norel$} & \hyperref[theo:alpha-cc-seq-entries]{$\norel$} & \hyperref[theo:alphaminer-mm-ch-depth-entries]{$\norel$} & \hyperref[theo:alphaminer-diam-entries]{$\norel$} & \hyperref[theo:alphaminer-size-ts-cfc-acd-mcd-cyc-empty-entries]{$\norel$} & \hyperref[theo:alphaminer-cnc-entries]{$\norel$} & \hyperref[theo:alphaminer-dens-entries]{$\norel$} & \hyperref[theo:alphaminer-duplicate-entries]{$\meq$} & \hyperref[theo:alphaminer-size-ts-cfc-acd-mcd-cyc-empty-entries]{$\norel$} \\ \hline
		
		$\numberofties$ & \hyperref[theo:alphaminer-size-ts-cfc-acd-mcd-cyc-empty-entries]{$\norel$} & \hyperref[theo:alphaminer-mm-ch-depth-entries]{$\norel$} & \hyperref[theo:alphaminer-mm-ch-depth-entries]{$\norel$} & \hyperref[theo:alpha-cc-seq-entries]{$\norel^*$} & \hyperref[theo:alphaminer-size-ts-cfc-acd-mcd-cyc-empty-entries]{$\norel$} & \hyperref[theo:alphaminer-size-ts-cfc-acd-mcd-cyc-empty-entries]{$\norel$} & \hyperref[theo:alphaminer-sep-entries]{$\norel$} & \hyperref[theo:alphaminer-size-ts-cfc-acd-mcd-cyc-empty-entries]{$\norel$} & \hyperref[theo:alphaminer-size-ts-cfc-acd-mcd-cyc-empty-entries]{$\norel$} & \hyperref[theo:alpha-cc-seq-entries]{$\norel$} & \hyperref[theo:alphaminer-mm-ch-depth-entries]{$\norel$} & \hyperref[theo:alphaminer-diam-entries]{$\norel$} & \hyperref[theo:alphaminer-size-ts-cfc-acd-mcd-cyc-empty-entries]{$\norel$} & \hyperref[theo:alphaminer-cnc-entries]{$\norel$} & \hyperref[theo:alphaminer-dens-entries]{$\norel$} & \hyperref[theo:alphaminer-duplicate-entries]{$\meq$} & \hyperref[theo:alphaminer-size-ts-cfc-acd-mcd-cyc-empty-entries]{$\norel$} \\ \hline
		
		$\lempelziv$ & \hyperref[theo:alphaminer-size-ts-cfc-acd-mcd-cyc-empty-entries]{$\norel$} & \hyperref[theo:alphaminer-mm-ch-depth-entries]{$\norel$} & \hyperref[theo:alphaminer-mm-ch-depth-entries]{$\norel$} & \hyperref[theo:alpha-cc-seq-entries]{$\norel^*$} & \hyperref[theo:alphaminer-size-ts-cfc-acd-mcd-cyc-empty-entries]{$\norel$} & \hyperref[theo:alphaminer-size-ts-cfc-acd-mcd-cyc-empty-entries]{$\norel$} & \hyperref[theo:alphaminer-sep-entries]{$\norel$} & \hyperref[theo:alphaminer-size-ts-cfc-acd-mcd-cyc-empty-entries]{$\norel$} & \hyperref[theo:alphaminer-size-ts-cfc-acd-mcd-cyc-empty-entries]{$\norel$} & \hyperref[theo:alpha-cc-seq-entries]{$\norel$} & \hyperref[theo:alphaminer-mm-ch-depth-entries]{$\norel$} & \hyperref[theo:alphaminer-diam-entries]{$\norel$} & \hyperref[theo:alphaminer-size-ts-cfc-acd-mcd-cyc-empty-entries]{$\norel$} & \hyperref[theo:alphaminer-cnc-entries]{$\norel$} & \hyperref[theo:alphaminer-dens-entries]{$\norel$} & \hyperref[theo:alphaminer-duplicate-entries]{$\meq$} & \hyperref[theo:alphaminer-size-ts-cfc-acd-mcd-cyc-empty-entries]{$\norel$} \\ \hline
		
		$\numberuniquetraces$ & \hyperref[theo:alphaminer-size-ts-cfc-acd-mcd-cyc-empty-entries]{$\norel$} & \hyperref[theo:alphaminer-mm-ch-depth-entries]{$\norel$} & \hyperref[theo:alphaminer-mm-ch-depth-entries]{$\norel$} & \hyperref[theo:alpha-cc-seq-entries]{$\norel^*$} & \hyperref[theo:alphaminer-size-ts-cfc-acd-mcd-cyc-empty-entries]{$\norel$} & \hyperref[theo:alphaminer-size-ts-cfc-acd-mcd-cyc-empty-entries]{$\norel$} & \hyperref[theo:alphaminer-sep-entries]{$\norel$} & \hyperref[theo:alphaminer-size-ts-cfc-acd-mcd-cyc-empty-entries]{$\norel$} & \hyperref[theo:alphaminer-size-ts-cfc-acd-mcd-cyc-empty-entries]{$\norel$} & \hyperref[theo:alpha-cc-seq-entries]{$\norel$} & \hyperref[theo:alphaminer-mm-ch-depth-entries]{$\norel$} & \hyperref[theo:alphaminer-diam-entries]{$\norel$} & \hyperref[theo:alphaminer-size-ts-cfc-acd-mcd-cyc-empty-entries]{$\norel$} & \hyperref[theo:alphaminer-cnc-entries]{$\norel$} & \hyperref[theo:alphaminer-dens-entries]{$\norel$} & \hyperref[theo:alphaminer-duplicate-entries]{$\meq$} & \hyperref[theo:alphaminer-size-ts-cfc-acd-mcd-cyc-empty-entries]{$\norel$} \\ \hline
		
		$\percentageuniquetraces$ & \hyperref[theo:alphaminer-size-ts-cfc-acd-mcd-cyc-empty-entries]{$\norel$} & \hyperref[theo:alphaminer-mm-ch-depth-entries]{$\norel$} & \hyperref[theo:alphaminer-mm-ch-depth-entries]{$\norel$} & \hyperref[theo:alpha-cc-seq-entries]{$\norel^*$} & \hyperref[theo:alphaminer-size-ts-cfc-acd-mcd-cyc-empty-entries]{$\norel$} & \hyperref[theo:alphaminer-size-ts-cfc-acd-mcd-cyc-empty-entries]{$\norel$} & \hyperref[theo:alphaminer-sep-entries]{$\norel$} & \hyperref[theo:alphaminer-size-ts-cfc-acd-mcd-cyc-empty-entries]{$\norel$} & \hyperref[theo:alphaminer-size-ts-cfc-acd-mcd-cyc-empty-entries]{$\norel$} & \hyperref[theo:alpha-cc-seq-entries]{$\norel$} & \hyperref[theo:alphaminer-mm-ch-depth-entries]{$\norel$} & \hyperref[theo:alphaminer-diam-entries]{$\norel$} & \hyperref[theo:alphaminer-size-ts-cfc-acd-mcd-cyc-empty-entries]{$\norel$} & \hyperref[theo:alphaminer-cnc-entries]{$\norel$} & \hyperref[theo:alphaminer-dens-entries]{$\norel$} & \hyperref[theo:alphaminer-duplicate-entries]{$\meq$} & \hyperref[theo:alphaminer-size-ts-cfc-acd-mcd-cyc-empty-entries]{$\norel$} \\ \hline
		
		$\structure$ & \hyperref[theo:alphaminer-size-ts-cfc-acd-mcd-cyc-empty-entries]{$\norel$} & \hyperref[theo:alphaminer-mm-ch-depth-entries]{$\norel$} & \hyperref[theo:alphaminer-mm-ch-depth-entries]{$\norel$} & \hyperref[theo:alpha-cc-seq-entries]{$\norel^*$} & \hyperref[theo:alphaminer-size-ts-cfc-acd-mcd-cyc-empty-entries]{$\norel$} & \hyperref[theo:alphaminer-size-ts-cfc-acd-mcd-cyc-empty-entries]{$\norel$} & \hyperref[theo:alphaminer-sep-entries]{$\norel$} & \hyperref[theo:alphaminer-size-ts-cfc-acd-mcd-cyc-empty-entries]{$\norel$} & \hyperref[theo:alphaminer-size-ts-cfc-acd-mcd-cyc-empty-entries]{$\norel$} & \hyperref[theo:alpha-cc-seq-entries]{$\norel$} & \hyperref[theo:alphaminer-mm-ch-depth-entries]{$\norel$} & \hyperref[theo:alphaminer-diam-entries]{$\norel$} & \hyperref[theo:alphaminer-size-ts-cfc-acd-mcd-cyc-empty-entries]{$\norel$} & \hyperref[theo:alphaminer-cnc-entries]{$\norel$} & \hyperref[theo:alphaminer-dens-entries]{$\norel$} & \hyperref[theo:alphaminer-duplicate-entries]{$\meq$} & \hyperref[theo:alphaminer-size-ts-cfc-acd-mcd-cyc-empty-entries]{$\norel$} \\ \hline
		
		$\affinity$ & \hyperref[theo:alphaminer-size-ts-cfc-acd-mcd-cyc-empty-entries]{$\norel$} & \hyperref[theo:alphaminer-mm-ch-depth-entries]{$\norel$} & \hyperref[theo:alphaminer-mm-ch-depth-entries]{$\norel$} & \hyperref[theo:alpha-cc-seq-entries]{$\norel^*$} & \hyperref[theo:alphaminer-size-ts-cfc-acd-mcd-cyc-empty-entries]{$\norel$} & \hyperref[theo:alphaminer-size-ts-cfc-acd-mcd-cyc-empty-entries]{$\norel$} & \hyperref[theo:alphaminer-sep-entries]{$\norel$} & \hyperref[theo:alphaminer-size-ts-cfc-acd-mcd-cyc-empty-entries]{$\norel$} & \hyperref[theo:alphaminer-size-ts-cfc-acd-mcd-cyc-empty-entries]{$\norel$} & \hyperref[theo:alpha-cc-seq-entries]{$\norel$} & \hyperref[theo:alphaminer-mm-ch-depth-entries]{$\norel$} & \hyperref[theo:alphaminer-diam-entries]{$\norel$} & \hyperref[theo:alphaminer-size-ts-cfc-acd-mcd-cyc-empty-entries]{$\norel$} & \hyperref[theo:alphaminer-cnc-entries]{$\norel$} & \hyperref[theo:alphaminer-dens-entries]{$\norel$} & \hyperref[theo:alphaminer-duplicate-entries]{$\meq$} & \hyperref[theo:alphaminer-size-ts-cfc-acd-mcd-cyc-empty-entries]{$\norel$} \\ \hline
		
		$\deviationfromrandom$ & \hyperref[theo:alphaminer-size-ts-cfc-acd-mcd-cyc-empty-entries]{$\norel$} & \hyperref[theo:alphaminer-mm-ch-depth-entries]{$\norel$} & \hyperref[theo:alphaminer-mm-ch-depth-entries]{$\norel$} & \hyperref[theo:alpha-cc-seq-entries]{$\norel^*$} & \hyperref[theo:alphaminer-size-ts-cfc-acd-mcd-cyc-empty-entries]{$\norel$} & \hyperref[theo:alphaminer-size-ts-cfc-acd-mcd-cyc-empty-entries]{$\norel$} & \hyperref[theo:alphaminer-sep-entries]{$\norel$} & \hyperref[theo:alphaminer-size-ts-cfc-acd-mcd-cyc-empty-entries]{$\norel$} & \hyperref[theo:alphaminer-size-ts-cfc-acd-mcd-cyc-empty-entries]{$\norel$} & \hyperref[theo:alpha-cc-seq-entries]{$\norel$} & \hyperref[theo:alphaminer-mm-ch-depth-entries]{$\norel$} & \hyperref[theo:alphaminer-diam-entries]{$\norel$} & \hyperref[theo:alphaminer-size-ts-cfc-acd-mcd-cyc-empty-entries]{$\norel$} & \hyperref[theo:alphaminer-cnc-entries]{$\norel$} & \hyperref[theo:alphaminer-dens-entries]{$\norel$} & \hyperref[theo:alphaminer-duplicate-entries]{$\meq$} & \hyperref[theo:alphaminer-size-ts-cfc-acd-mcd-cyc-empty-entries]{$\norel$} \\ \hline
		
		$\avgdist$ & \hyperref[theo:alphaminer-size-ts-cfc-acd-mcd-cyc-empty-entries]{$\norel$} & \hyperref[theo:alphaminer-mm-ch-depth-entries]{$\norel$} & \hyperref[theo:alphaminer-mm-ch-depth-entries]{$\norel$} & \hyperref[theo:alpha-cc-seq-entries]{$\norel^*$} & \hyperref[theo:alphaminer-size-ts-cfc-acd-mcd-cyc-empty-entries]{$\norel$} & \hyperref[theo:alphaminer-size-ts-cfc-acd-mcd-cyc-empty-entries]{$\norel$} & \hyperref[theo:alphaminer-sep-entries]{$\norel$} & \hyperref[theo:alphaminer-size-ts-cfc-acd-mcd-cyc-empty-entries]{$\norel$} & \hyperref[theo:alphaminer-size-ts-cfc-acd-mcd-cyc-empty-entries]{$\norel$} & \hyperref[theo:alpha-cc-seq-entries]{$\norel$} & \hyperref[theo:alphaminer-mm-ch-depth-entries]{$\norel$} & \hyperref[theo:alphaminer-diam-entries]{$\norel$} & \hyperref[theo:alphaminer-size-ts-cfc-acd-mcd-cyc-empty-entries]{$\norel$} & \hyperref[theo:alphaminer-cnc-entries]{$\norel$} & \hyperref[theo:alphaminer-dens-entries]{$\norel$} & \hyperref[theo:alphaminer-duplicate-entries]{$\meq$} & \hyperref[theo:alphaminer-size-ts-cfc-acd-mcd-cyc-empty-entries]{$\norel$} \\ \hline
		
		$\varentropy$ & \hyperref[theo:alphaminer-size-ts-cfc-acd-mcd-cyc-empty-entries]{$\norel$} & \hyperref[theo:alphaminer-mm-ch-depth-entries]{$\norel$} & \hyperref[theo:alphaminer-mm-ch-depth-entries]{$\norel$} & \hyperref[theo:alpha-cc-seq-entries]{$\norel^*$} & \hyperref[theo:alphaminer-size-ts-cfc-acd-mcd-cyc-empty-entries]{$\norel$} & \hyperref[theo:alphaminer-size-ts-cfc-acd-mcd-cyc-empty-entries]{$\norel$} & \hyperref[theo:alphaminer-sep-entries]{$\norel$} & \hyperref[theo:alphaminer-size-ts-cfc-acd-mcd-cyc-empty-entries]{$\norel$} & \hyperref[theo:alphaminer-size-ts-cfc-acd-mcd-cyc-empty-entries]{$\norel$} & \hyperref[theo:alpha-cc-seq-entries]{$\norel$} & \hyperref[theo:alphaminer-mm-ch-depth-entries]{$\norel$} & \hyperref[theo:alphaminer-diam-entries]{$\norel$} & \hyperref[theo:alphaminer-size-ts-cfc-acd-mcd-cyc-empty-entries]{$\norel$} & \hyperref[theo:alphaminer-cnc-entries]{$\norel$} & \hyperref[theo:alphaminer-dens-entries]{$\norel$} & \hyperref[theo:alphaminer-duplicate-entries]{$\meq$} & \hyperref[theo:alphaminer-size-ts-cfc-acd-mcd-cyc-empty-entries]{$\norel$} \\ \hline
		
		$\normvarentropy$ & \hyperref[theo:alphaminer-size-ts-cfc-acd-mcd-cyc-empty-entries]{$\norel$} & \hyperref[theo:alphaminer-mm-ch-depth-entries]{$\norel$} & \hyperref[theo:alphaminer-mm-ch-depth-entries]{$\norel$} & \hyperref[theo:alpha-cc-seq-entries]{$\norel^*$} & \hyperref[theo:alphaminer-size-ts-cfc-acd-mcd-cyc-empty-entries]{$\norel$} & \hyperref[theo:alphaminer-size-ts-cfc-acd-mcd-cyc-empty-entries]{$\norel$} & \hyperref[theo:alphaminer-sep-entries]{$\norel$} & \hyperref[theo:alphaminer-size-ts-cfc-acd-mcd-cyc-empty-entries]{$\norel$} & \hyperref[theo:alphaminer-size-ts-cfc-acd-mcd-cyc-empty-entries]{$\norel$} & \hyperref[theo:alpha-cc-seq-entries]{$\norel$} & \hyperref[theo:alphaminer-mm-ch-depth-entries]{$\norel$} & \hyperref[theo:alphaminer-diam-entries]{$\norel$} & \hyperref[theo:alphaminer-size-ts-cfc-acd-mcd-cyc-empty-entries]{$\norel$} & \hyperref[theo:alphaminer-cnc-entries]{$\norel$} & \hyperref[theo:alphaminer-dens-entries]{$\norel$} & \hyperref[theo:alphaminer-duplicate-entries]{$\meq$} & \hyperref[theo:alphaminer-size-ts-cfc-acd-mcd-cyc-empty-entries]{$\norel$} \\ \hline
		
		$\seqentropy$ & \hyperref[theo:alphaminer-size-ts-cfc-acd-mcd-cyc-empty-entries]{$\norel$} & \hyperref[theo:alphaminer-mm-ch-depth-entries]{$\norel$} & \hyperref[theo:alphaminer-mm-ch-depth-entries]{$\norel$} & \hyperref[theo:alpha-cc-seq-entries]{$\norel^*$} & \hyperref[theo:alphaminer-size-ts-cfc-acd-mcd-cyc-empty-entries]{$\norel$} & \hyperref[theo:alphaminer-size-ts-cfc-acd-mcd-cyc-empty-entries]{$\norel$} & \hyperref[theo:alphaminer-sep-entries]{$\norel$} & \hyperref[theo:alphaminer-size-ts-cfc-acd-mcd-cyc-empty-entries]{$\norel$} & \hyperref[theo:alphaminer-size-ts-cfc-acd-mcd-cyc-empty-entries]{$\norel$} & \hyperref[theo:alpha-cc-seq-entries]{$\norel$} & \hyperref[theo:alphaminer-mm-ch-depth-entries]{$\norel$} & \hyperref[theo:alphaminer-diam-entries]{$\norel$} & \hyperref[theo:alphaminer-size-ts-cfc-acd-mcd-cyc-empty-entries]{$\norel$} & \hyperref[theo:alphaminer-cnc-entries]{$\norel$} & \hyperref[theo:alphaminer-dens-entries]{$\norel$} & \hyperref[theo:alphaminer-duplicate-entries]{$\meq$} & \hyperref[theo:alphaminer-size-ts-cfc-acd-mcd-cyc-empty-entries]{$\norel$} \\ \hline
		
		$\normseqentropy$ & \hyperref[theo:alphaminer-size-ts-cfc-acd-mcd-cyc-empty-entries]{$\norel$} & \hyperref[theo:alphaminer-mm-ch-depth-entries]{$\norel$} & \hyperref[theo:alphaminer-mm-ch-depth-entries]{$\norel$} & \hyperref[theo:alpha-cc-seq-entries]{$\norel^*$} & \hyperref[theo:alphaminer-size-ts-cfc-acd-mcd-cyc-empty-entries]{$\norel$} & \hyperref[theo:alphaminer-size-ts-cfc-acd-mcd-cyc-empty-entries]{$\norel$} & \hyperref[theo:alphaminer-sep-entries]{$\norel$} & \hyperref[theo:alphaminer-size-ts-cfc-acd-mcd-cyc-empty-entries]{$\norel$} & \hyperref[theo:alphaminer-size-ts-cfc-acd-mcd-cyc-empty-entries]{$\norel$} & \hyperref[theo:alpha-cc-seq-entries]{$\norel$} & \hyperref[theo:alphaminer-mm-ch-depth-entries]{$\norel$} & \hyperref[theo:alphaminer-diam-entries]{$\norel$} & \hyperref[theo:alphaminer-size-ts-cfc-acd-mcd-cyc-empty-entries]{$\norel$} & \hyperref[theo:alphaminer-cnc-entries]{$\norel$} & \hyperref[theo:alphaminer-dens-entries]{$\norel$} & \hyperref[theo:alphaminer-duplicate-entries]{$\meq$} & \hyperref[theo:alphaminer-size-ts-cfc-acd-mcd-cyc-empty-entries]{$\norel$} \\ \hline
	\end{tabular}
	}
	{\scriptsize ${}^*$We did not find examples showing that $\mathcal{C}^L(L_1) < \mathcal{C}^L(L_2)$ and $\crossconn(M_1) = \crossconn(M_2)$ is possible.}
\end{table}

\begin{theorem}
\label{theo:alphaminer-size-ts-cfc-acd-mcd-cyc-empty-entries}
Let $\mathcal{C}^L \in \loc$ be any log complexity measure and let $\mathcal{C}^M$ be a model complexity measure with $\mathcal{C}^M \in \{\size, \tokensplit, \controlflow, \avgconn, \maxconn, \cyclicity, \emptyseq\}$.
Then, $(\mathcal{C}^L, \mathcal{C}^M) \in \norel$.
\end{theorem}
\begin{proof}
Consider the following event logs:
\begin{align*}
	L_1 &= [\langle a,b,c,d,e \rangle^{3}, \langle e \rangle^{2}] \\
	L_2 &= L_1 + [\langle a,b,c,d,b,c,d,e,f \rangle^{2}] \\
	L_3 &= L_2 + [\langle a,b,c,d,b,c,d,b,c,d,e \rangle^{2}, \langle a,b,c,d,b,c,d,b,c,d,b,c,d,d,e \rangle, \\
	&\phantom{= L_2 + [} \hspace*{1mm} \langle a,a,b,b,c,c,d,d,e,e,f,f,g,g,h,h,i,i \rangle]
\end{align*}
\cref{fig:alphaminer-counterexample-size} shows the models $M_1, M_2, M_3$ found by the alpha miner for $L_1, L_2, L_3$.
\begin{figure}[ht]
	\centering
	\scalebox{\scalefactor}{
	\begin{tikzpicture}[node distance = 1.1cm,>=stealth',bend angle=0,auto]
		\node[place,tokens=1] (start) {};
		\node[yshift=1cm] at (start) {$M_1$:};
		\node[transition,right of=start] (a) {$a$}
		edge [pre] (start);
		\node[place,right of=a] (p1) {}
		edge [pre] (a);
		\node[transition,right of=p1] (b) {$b$}
		edge [pre] (p1);
		\node[place,right of=b] (p2) {}
		edge [pre] (b);
		\node[transition,right of=p2] (c) {$c$}
		edge [pre] (p2);
		\node[place,right of=c] (p3) {}
		edge [pre] (c);
		\node[transition,right of=p3] (d) {$d$}
		edge [pre] (p3);
		\node[place,right of=d] (p4) {}
		edge [pre] (d);
		\node[transition,right of=p4] (e) {$e$}
		edge [pre] (p4)
		edge [pre,bend right=18] (start);
		\node[place,right of=e] (end) {}
		edge [pre] (e);
	\end{tikzpicture}}
	
	\medskip
	\hrule
	\medskip
	
	\scalebox{\scalefactor}{
	\begin{tikzpicture}[node distance = 1.1cm,>=stealth',bend angle=0,auto]
		\node[place,tokens=1] (start) {};
		\node[yshift=1cm] at (start) {$M_2$:};
		\node[transition,right of=start] (a) {$a$}
		edge [pre] (start);
		\node[place,right of=a] (p1) {}
		edge [pre] (a);
		\node[transition,right of=p1] (b) {$b$}
		edge [pre] (p1);
		\node[place,right of=b] (p2) {}
		edge [pre] (b);
		\node[transition,right of=p2] (c) {$c$}
		edge [pre] (p2);
		\node[place,right of=c] (p3) {}
		edge [pre] (c);
		\node[transition,right of=p3] (d) {$d$}
		edge [pre] (p3)
		edge [post,bend right=20] (p1);
		\node[place,right of=d] (p4) {}
		edge [pre] (d)
		edge [post,bend left=20] (b);
		\node[transition,right of=p4] (e) {$e$}
		edge [pre] (p4)
		edge [pre,bend right=20] (start);
		\node[place,right of=e] (p5) {}
		edge [pre] (e);
		\node[transition,right of=p5] (f) {$f$}
		edge [pre] (p5);
		\node[place,right of=f] (end) {}
		edge [pre] (f)
		edge [pre,bend left=20] (e);
	\end{tikzpicture}}
	
	\medskip
	\hrule
	\medskip
	
	\scalebox{\scalefactor}{
	\begin{tikzpicture}[node distance = 1.1cm,>=stealth',bend angle=0,auto]
		\node[place,tokens=1] (start) {};
		\node[yshift=1cm] at (start) {$M_3$:};
		\node[transition,right of=start,yshift=1cm] (a) {$a$}
		edge [pre] (start);
		\node[transition,right of=start] (e) {$e$}
		edge [pre] (start);
		\node[place,right of=e,yshift=-1cm] (end) {}
		edge [pre] (e);
		\node[transition,left of=end] (f) {$f$}
		edge [post] (end);
		\node[transition,left of=end,yshift=-1cm] (i) {$i$}
		edge [post] (end);
		\node[transition,right of=end,yshift=0.5cm] (b) {$b$};
		\node[transition,right of=b] (c) {$c$};
		\node[transition,right of=c] (d) {$d$};
		\node[transition,right of=d] (g) {$g$};
		\node[transition,right of=g] (h) {$h$};
	\end{tikzpicture}}
	\caption{The results of the alpha algorithm for the input logs $L_1, L_2, L_3$ from the example in \cref{theo:alphaminer-size-ts-cfc-acd-mcd-cyc-empty-entries}. $M_1$ is the model mined from the log $L_1$, $M_2$ the model mined from the log $L_2$, and $M_3$ the model mined from the log $L_3$.}
	\label{fig:alphaminer-counterexample-size}
\end{figure}
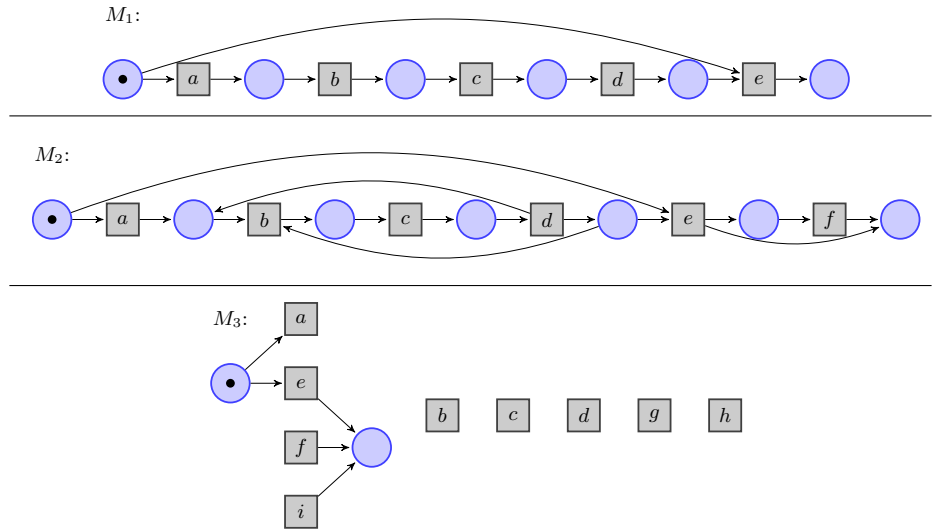
These models have the following complexity scores:
\begin{center}
	\begin{tabular}{|c|c|c|c|c|c|c|c|} \hline
	& $\size$ & $\tokensplit$ & $\controlflow$ & $\avgconn$ & $\maxconn$ & $\cyclicity$ & $\emptyseq$ \\ \hline
	$M_1$ & $\pad 11 \pad$ & $\pad 0 \pad$ & $\pad 2 \pad$ & $\pad 2.5 \pad$ & $\pad 3 \pad$ & $\pad 0 \pad$ & $\pad 0 \pad$ \\ \hline
	$M_2$ & $\pad 13 \pad$ & $\pad 2 \pad$ & $\pad 6 \pad$ & $\pad 2.8571 \pad$ & $\pad 4 \pad$ & $\pad 0.6364 \pad$ & $\pad 1 \pad$ \\ \hline
	$M_3$ & $\pad 11 \pad$ & $\pad 0 \pad$ & $\pad 2 \pad$ & $\pad 2.5 \pad$ & $\pad 3 \pad$ & $\pad 0 \pad$ & $\pad 0 \pad$ \\ \hline
	\end{tabular}
\end{center}
Thus, $\mathcal{C}^M(M_1) < \mathcal{C}^M(M_2)$, $\mathcal{C}^M(M_2) > \mathcal{C}^M(M_3)$, and $\mathcal{C}^M(M_1) = \mathcal{C}^M(M_3)$ for all $\mathcal{C}^M \in \{\size, \tokensplit, \controlflow, \avgconn, \maxconn, \cyclicity, \emptyseq\}$.
But the event logs $L_1, L_2, L_3$ have the following log complexity scores:
\begin{center}
	\def\pad{\hspace*{1.5mm}}
	\begin{tabular}{|c|c|c|c|c|c|c|c|c|c|c|c|c|} \hline
		 & $\magnitude$ & $\variety$ & $\support$ & $\tlavg$ & $\tlmax$ & $\levelofdetail$ & $\numberofties$ & $\lempelziv$ & $\numberuniquetraces$ & $\percentageuniquetraces$ \\ \hline
		$L_1$ & $\pad 17 \pad$ & $\pad 5 \pad$ & $\pad 5 \pad$ & $\pad 3.4 \pad$ & $\pad 5 \pad$ & $\pad 2 \pad$ & $\pad 4 \pad$ & $\pad 11 \pad$ & $\pad 2 \pad$ & $\pad 0.4 \pad$ \\ \hline
		$L_2$ & $\pad 35 \pad$ & $\pad 6 \pad$ & $\pad 7 \pad$ & $\pad 5 \pad$ & $\pad 9 \pad$ & $\pad 4 \pad$ & $\pad 6 \pad$ & $\pad 18 \pad$ & $\pad 3 \pad$ & $\pad 0.4286 \pad$ \\ \hline
		$L_3$ & $\pad 90 \pad$ & $\pad 9 \pad$ & $\pad 11 \pad$ & $\pad 8.1818 \pad$ & $\pad 18 \pad$ & $\pad 6 \pad$ & $\pad 9 \pad$ & $\pad 37 \pad$ & $\pad 6 \pad$ & $\pad 0.5455 \pad$ \\ \hline
		\end{tabular}
		
		\medskip
		
		\begin{tabular}{|c|c|c|c|c|c|c|c|c|} \hline
		 & $\structure$ & $\affinity$ & $\deviationfromrandom$ & $\avgdist$ & $\varentropy$ & $\normvarentropy$ & $\seqentropy$ & $\normseqentropy$ \\ \hline
		$L_1$ & $\pad 3.4 \pad$ & $\pad 0.4 \pad$ & $\pad 0.5417 \pad$ & $\pad 2.4 \pad$ & $\pad 2.7034 \pad$ & $\pad 0.2515 \pad$ & $\pad 6.1576 \pad$ & $\pad 0.1278 \pad$ \\ \hline
		$L_2$ & $\pad 4.1429 \pad$ & $\pad 0.4286 \pad$ & $\pad 0.5862 \pad$ & $\pad 3.8095 \pad$ & $\pad 10.2825 \pad$ & $\pad 0.3898\pad $ & $\pad 27.9087 \pad$ & $\pad 0.2243 \pad$ \\ \hline
		$L_3$ & $\pad 4.8182 \pad$ & $\pad 0.4584 \pad$ & $\pad 0.6336 \pad$ & $\pad 7.2727 \pad$ & $\pad 55.7526 \pad$ & $\pad 0.4173 \pad$ & $\pad 136.0569 \pad$ & $\pad 0.3360 \pad$ \\ \hline
		\end{tabular}
\end{center}
Since $\mathcal{C}^L(L_1) < \mathcal{C}^L(L_2) < \mathcal{C}^L(L_3)$ for any log coplexity measure $\mathcal{C}^L \in \loc$, we have thus shown that $(\mathcal{C}^L, \mathcal{C}^M) \in \norel$ for any model complexity measure $\mathcal{C}^M \in \{\size, \tokensplit, \controlflow, \avgconn, \maxconn, \cyclicity, \emptyseq\}$. \hfill$\square$
\end{proof}

\begin{theorem}
\label{theo:alphaminer-mm-ch-depth-entries}
Let $\mathcal{C}^L \in \loc$ be any log complexity measure and let $\mathcal{C}^M$ be a model complexity measure with $\mathcal{C}^M \in \{\mismatch, \connhet, \depth\}$.
Then, $(\mathcal{C}^L, \mathcal{C}^M) \in \norel$.
\end{theorem}
\begin{proof}
Consider the following event logs:
\begin{align*}
	L_1 &= [\langle a,b,c,d \rangle^{3}, \langle e \rangle^{2}] \\
	L_2 &= L_1 + [\langle a,b,c,d \rangle^{3}, \langle a,c,b,d \rangle, \langle a,b,c,b,c,d \rangle, \langle b,c,b,c,b,c,d \rangle, \\
	&\phantom{= L_1 + [}\hspace*{1mm} \langle a,b,c,f,e,f,e \rangle] \\
	L_3 &= L_2 + [\langle a,b,c,b,c,b,c,b,c,d \rangle^{3}, \langle a,b,c,b,c,b,c,b,c,b,c,d \rangle, \\
	&\phantom{= L_2 + [}\hspace*{1mm} \langle a,a,b,b,c,c,d,d \rangle, \langle e,e,f,f,g,g \rangle]
\end{align*}
\cref{fig:alphaminer-counterexample-mismatch} shows the models $M_1, M_2, M_3$ found by the alpha miner for $L_1, L_2, L_3$.
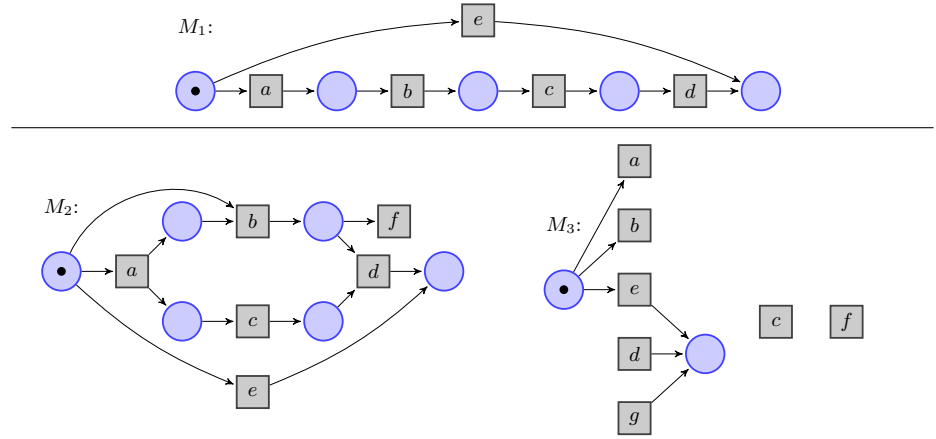
\begin{figure}[ht]
	\centering
	\scalebox{\scalefactor}{
	\begin{tikzpicture}[node distance = 1.1cm,>=stealth',bend angle=0,auto]
		\node[place,tokens=1] (start) {};
		\node[yshift=1cm] at (start) {$M_1$:};
		\node[transition,right of=start] (a) {$a$}
		edge [pre] (start);
		\node[place,right of=a] (p1) {}
		edge [pre] (a);
		\node[transition,right of=p1] (b) {$b$}
		edge [pre] (p1);
		\node[place,right of=b] (p2) {}
		edge [pre] (b);
		\node[transition,right of=p2] (c) {$c$}
		edge [pre] (p2);
		\node[place,right of=c] (p3) {}
		edge [pre] (c);
		\node[transition,right of=p3] (d) {$d$}
		edge [pre] (p3);
		\node[place,right of=d] (end) {}
		edge [pre] (d);
		\node[transition,above of=p2] (e) {$e$}
		edge [pre,bend right=10] (start)
		edge [post,bend left=10] (end);
	\end{tikzpicture}}
	
	\medskip
	\hrule
	\medskip
	
	\begin{minipage}{0.48\textwidth}
	\centering
	\scalebox{\scalefactor}{
	\begin{tikzpicture}[node distance = 1.1cm,>=stealth',bend angle=0,auto]
		\node[place,tokens=1] (start) {};
		\node[yshift=1cm] at (start) {$M_2$:};
		\node[transition,right of=start] (a) {$a$}
		edge [pre] (start);
		\node[place,above right of=a] (p1) {}
		edge [pre] (a);
		\node[place,below right of=a] (p2) {}
		edge [pre] (a);
		\node[transition,right of=p1] (b) {$b$}
		edge [pre] (p1)
		edge [pre,bend right=50] (start);
		\node[transition,right of=p2] (c) {$c$}
		edge [pre] (p2);
		\node[place,right of=b] (p3) {}
		edge [pre] (b);
		\node[place,right of=c] (p4) {}
		edge [pre] (c);
		\node[transition,below right of=p3] (d) {$d$}
		edge [pre] (p3)
		edge [pre] (p4);
		\node[place,right of=d] (end) {}
		edge [pre] (d);
		\node[transition,right of=p3] (f) {$f$}
		edge [pre] (p3);
		\node[transition,below of=c] (e) {$e$}
		edge [pre,bend left=10] (start)
		edge [post,bend right=10] (end);
	\end{tikzpicture}}
	\end{minipage}
	\begin{minipage}{0.48\textwidth}
	\centering
	\scalebox{\scalefactor}{
	\begin{tikzpicture}[node distance = 1.1cm,>=stealth',bend angle=0,auto]
		\node[place,tokens=1] (start) {};
		\node[yshift=1cm] at (start) {$M_3$:};
		\node[transition,right of=start,yshift=2cm] (a) {$a$}
		edge [pre] (start);
		\node[transition,right of=start,yshift=1cm] (b) {$b$}
		edge [pre] (start);
		\node[transition,right of=start] (e) {$e$}
		edge [pre] (start);
		\node[place,right of=e,yshift=-1cm] (end) {}
		edge [pre] (e);
		\node[transition,left of=end] (d) {$d$}
		edge [post] (end);
		\node[transition,left of=end,yshift=-1cm] (g) {$g$}
		edge [post] (end);
		\node[transition,right of=end,yshift=0.5cm] (c) {$c$};
		\node[transition,right of=c] (f) {$f$};
	\end{tikzpicture}}
	\end{minipage}
	\caption{The results of the alpha algorithm for the input logs $L_1, L_2, L_3$ from the example in \cref{theo:alphaminer-mm-ch-depth-entries}. $M_1$ is the model mined from the log $L_1$, $M_2$ the model mined from the log $L_2$, and $M_3$ the model mined from the log $L_3$.}
	\label{fig:alphaminer-counterexample-mismatch}
\end{figure}
These models have the following complexity scores:
\begin{center}
	\begin{tabular}{|c|c|c|c|} \hline
	& $\mismatch$ & $\connhet$ & $\depth$ \\ \hline
	$M_1$ & $0$ & $0$ & $1$ \\ \hline
	$M_2$ & $5$ & $1$ & $2$ \\ \hline
	$M_3$ & $0$ & $0$ & $1$ \\ \hline
	\end{tabular}
\end{center}
Thus, $\mathcal{C}^M(M_1) < \mathcal{C}^M(M_2)$, $\mathcal{C}^M(M_2) > \mathcal{C}^M(M_3)$, and $\mathcal{C}^M(M_1) = \mathcal{C}^M(M_3)$ for all $\mathcal{C}^M \in \{\mismatch, \connhet, \depth\}$.
But the event logs $L_1, L_2, L_3$ have the following log complexity scores:
\begin{center}
	\def\pad{\hspace*{1.5mm}}
	\begin{tabular}{|c|c|c|c|c|c|c|c|c|c|c|}\hline
		 & $\magnitude$ & $\variety$ & $\support$ & $\tlavg$ & $\tlmax$ & $\levelofdetail$ & $\numberofties$ & $\lempelziv$ & $\numberuniquetraces$ & $\percentageuniquetraces$ \\ \hline
		$L_1$ & $\pad 14 \pad$ & $\pad 5 \pad$ & $\pad 5 \pad$ & $\pad 2.8 \pad$ & $\pad 4 \pad$ & $\pad 2 \pad$ & $\pad 3 \pad$ & $\pad 9 \pad$ & $\pad 2 \pad$ & $\pad 0.4 \pad$ \\ \hline
		$L_2$ & $\pad 62 \pad$ & $\pad 6 \pad$ & $\pad 14 \pad$ & $\pad 4.4286 \pad$ & $\pad 7 \pad$ & $\pad 10 \pad$ & $\pad 5 \pad$ & $\pad 25 \pad$ & $\pad 6 \pad$ & $\pad 0.4286 \pad$ \\ \hline
		$L_3$ & $\pad 118 \pad$ & $\pad 7 \pad$ & $\pad 20 \pad$ & $\pad 5.9 \pad$ & $\pad 12 \pad$ & $\pad 14 \pad$ & $\pad 6 \pad$ & $\pad 43 \pad$ & $\pad 10 \pad$ & $\pad 0.5 \pad$ \\ \hline
	\end{tabular}
	
	\medskip
	
	\begin{tabular}{|c|c|c|c|c|c|c|c|c|} \hline
		 & $\structure$ & $\affinity$ & $\deviationfromrandom$ & $\avgdist$ & $\varentropy$ & $\normvarentropy$ & $\seqentropy$ & $\normseqentropy$ \\ \hline
		$L_1$ & $\pad 2.8 \pad$ & $\pad 0.4 \pad$ & $\pad 0.4584 \pad$ & $\pad 3 \pad$ & $\pad 2.502 \pad$ & $\pad 0.3109 \pad$ & $\pad 5.7416 \pad$ & $\pad 0.1554 \pad$ \\ \hline
		$L_2$ & $\pad 3.5714 \pad$ & $\pad 0.4555 \pad$ & $\pad 0.565 \pad$ & $\pad 3.3626 \pad$ & $\pad 36.6995 \pad$ & $\pad 0.5397 \pad$ & $\pad 78.6547 \pad$ & $\pad 0.3074 \pad$ \\ \hline
		$L_3$ & $\pad 3.65 \pad$ & $\pad 0.4632 \pad$ & $\pad 0.5683 \pad$ & $\pad 5.3579 \pad$ & $\pad 89.9638 \pad$ & $\pad 0.5731 \pad$ & $\pad 207.215 \pad$ & $\pad 0.3681 \pad$ \\ \hline
	\end{tabular}
\end{center}
Since $\mathcal{C}^L(L_1) < \mathcal{C}^L(L_2) < \mathcal{C}^L(L_3)$ for any log complexity measure $\mathcal{C}^L \in \loc$, we have thus shown that $(\mathcal{C}^L, \mathcal{C}^M) \in \norel$ for any model complexity measure $\mathcal{C}^M \in \{\mismatch, \connhet, \depth\}$. \hfill$\square$
\end{proof}

\begin{theorem}
\label{theo:alpha-cc-seq-entries}
Let $\mathcal{C}^L \in \loc$ be any log complexity measure and let $\mathcal{C}^M$ be a model complexity measure with $\mathcal{C}^M \in \{\crossconn, \sequentiality\}$.
Then, $(\mathcal{C}^L, \mathcal{C}^M) \in \norel$.
\end{theorem}
\begin{proof}
Consider the following event logs:
\begin{align*}
	L_1 &= [\langle a,b,d \rangle^{2}, \langle a,c,d \rangle^{2}, \langle e \rangle] \\
	L_2 &= L_1 + [\langle a,b,d,e \rangle, \langle a,c,d,e \rangle, \langle a,b,c,d \rangle, \langle a,b,c,b,d,e,f \rangle, \\
	&\phantom{= L_1 + [}\hspace*{1mm} \langle a,b,c,b,c,b,d,e,f \rangle] \\
	L_3 &= L_2 + [\langle a,c,b,d \rangle, \langle a,c,b,c,b,d,e \rangle, \langle a,b,c,b,c,b,c,d \rangle, \langle a,b,c,b,c,b,c,b,c,d \rangle, \\
	&\phantom{= L_2 + [}\hspace*{1mm} \langle a,a,b,b,c,c,d,d,e,e,f,f,g,g \rangle]
\end{align*}
\cref{fig:alphaminer-counterexample-crossconn} shows the models $M_1, M_2, M_3$ found by the alpha miner for $L_1, L_2, L_3$.
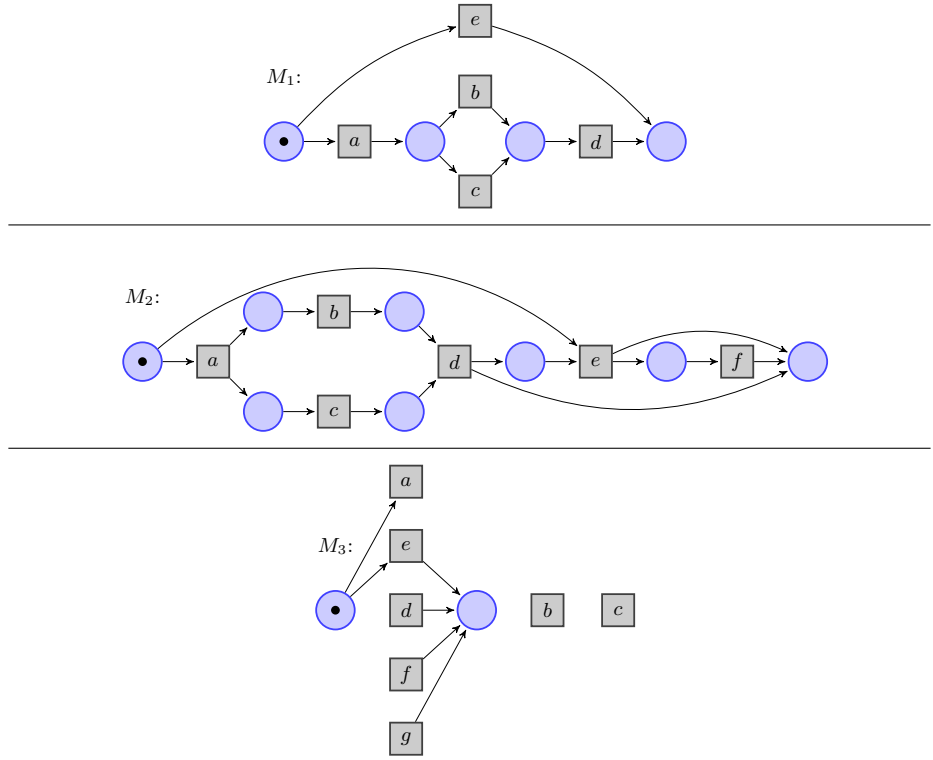
\begin{figure}[ht]
	\centering
	\scalebox{\scalefactor}{
	\begin{tikzpicture}[node distance = 1.1cm,>=stealth',bend angle=0,auto]
		\node[place,tokens=1] (start) {};
		\node[yshift=1cm] at (start) {$M_1$:};
		\node[transition,right of=start] (a) {$a$}
		edge [pre] (start);
		\node[place,right of=a] (p1) {}
		edge [pre] (a);
		\node[transition,above right of=p1] (b) {$b$}
		edge [pre] (p1);
		\node[transition,below right of=p1] (c) {$c$}
		edge [pre] (p1);
		\node[place,below right of=b] (p2) {}
		edge [pre] (b)
		edge [pre] (c);
		\node[transition,right of=p2] (d) {$d$}
		edge [pre] (p2);
		\node[place,right of=d] (end) {}
		edge [pre] (d);
		\node[transition,above of=b] (e) {$e$}
		edge [pre,bend right=15] (start)
		edge [post,bend left=15] (end);
	\end{tikzpicture}}
	
	\medskip
	\hrule
	\medskip
	
	\scalebox{\scalefactor}{
	\begin{tikzpicture}[node distance = 1.1cm,>=stealth',bend angle=0,auto]
		\node[place,tokens=1] (start) {};
		\node[yshift=1cm] at (start) {$M_2$:};
		\node[transition,right of=start] (a) {$a$}
		edge [pre] (start);
		\node[place, above right of=a] (p1) {}
		edge [pre] (a);
		\node[place,below right of=a] (p2) {}
		edge [pre] (a);
		\node[transition,right of=p1] (b) {$b$}
		edge [pre] (p1);
		\node[transition,right of=p2] (c) {$c$}
		edge [pre] (p2);
		\node[place,right of=b] (p3) {}
		edge [pre] (b);
		\node[place,right of=c] (p4) {}
		edge [pre] (c);
		\node[transition,below right of=p3] (d) {$d$}
		edge [pre] (p3)
		edge [pre] (p4);
		\node[place,right of=d] (p5) {}
		edge [pre] (d);
		\node[transition,right of=p5] (e) {$e$}
		edge [pre,bend right=40] (start)
		edge [pre] (p5);
		\node[place,right of=e] (p6) {}
		edge [pre] (e);
		\node[transition,right of=p6] (f) {$f$}
		edge [pre] (p6);
		\node[place,right of=f] (end) {}
		edge [pre] (f)
		edge [pre,bend right=25] (e)
		edge [pre, bend left=25] (d);
	\end{tikzpicture}}
	
	\medskip
	\hrule
	\medskip
	
	\scalebox{\scalefactor}{
	\begin{tikzpicture}[node distance = 1.1cm,>=stealth',bend angle=0,auto]
		\node[place,tokens=1] (start) {};
		\node[yshift=1cm] at (start) {$M_3$:};
		\node[transition,right of=start,yshift=2cm] (a) {$a$}
		edge [pre] (start);
		\node[transition,right of=start,yshift=1cm] (e) {$e$}
		edge [pre] (start);
		\node[transition,right of=start] (d) {$d$};
		\node[place,right of=d] (end) {}
		edge [pre] (d)
		edge [pre] (e);
		\node[transition,left of=end,yshift=-1cm] (f) {$f$}
		edge [post] (end);
		\node[transition,left of=end,yshift=-2cm] (g) {$g$}
		edge [post] (end);
		\node[transition,right of=end] (b) {$b$};
		\node[transition,right of=b] (c) {$c$};
	\end{tikzpicture}}
	\caption{The results of the alpha algorithm for the input logs $L_1, L_2, L_3$ from the example in \cref{theo:alpha-cc-seq-entries}. $M_1$ is the model mined from the log $L_1$, $M_2$ the model mined from the log $L_2$, and $M_3$ the model mined from the log $L_3$.}
	\label{fig:alphaminer-counterexample-crossconn}
\end{figure}
These models have the following complexity scores:
\begin{center}
	\begin{tabular}{|c|c|c|} \hline
	& $\crossconn$ & $\sequentiality$ \\ \hline
	$M_1$ & $0.9237$ & $1$ \\ \hline
	$M_2$ & $0.631$ & $0.7059$ \\ \hline
	$M_3$ & $0.9705$ & $1$ \\ \hline
	\end{tabular}
\end{center}
Thus, we have $\crossconn(M_1) > \crossconn(M_2)$ and $\crossconn(M_2) < \crossconn(M_3)$, as well as $\sequentiality(M_1) > \sequentiality(M_2)$, $\sequentiality(M_2) < \sequentiality(M_3)$, and $\sequentiality(M_1) = \sequentiality(M_3)$.
But the event logs $L_1, L_2, L_3$ have the following log complexity scores:
\begin{center}
	\def\pad{\hspace*{1.5mm}}
	\begin{tabular}{|c|c|c|c|c|c|c|c|c|c|c|}\hline
		 & $\magnitude$ & $\variety$ & $\support$ & $\tlavg$ & $\tlmax$ & $\levelofdetail$ & $\numberofties$ & $\lempelziv$ & $\numberuniquetraces$ & $\percentageuniquetraces$ \\ \hline
		$L_1$ & $\pad 13 \pad$ & $\pad 5 \pad$ & $\pad 5 \pad$ & $\pad 2.6 \pad$ & $\pad 3 \pad$ & $\pad 3 \pad$ & $\pad 4 \pad$ & $\pad 8 \pad$ & $\pad 3 \pad$ & $\pad 0.6 \pad$ \\ \hline
		$L_2$ & $\pad 41 \pad$ & $\pad 6 \pad$ & $\pad 10 \pad$ & $\pad 4.1 \pad$ & $\pad 9 \pad$ & $\pad 14 \pad$ & $\pad 6 \pad$ & $\pad 18 \pad$ & $\pad 8 \pad$ & $\pad 0.8 \pad$ \\ \hline
		$L_3$ & $\pad 84 \pad$ & $\pad 7 \pad$ & $\pad 15 \pad$ & $\pad 5.6 \pad$ & $\pad 14 \pad$ & $\pad 19 \pad$ & $\pad 7 \pad$ & $\pad 35 \pad$ & $\pad 13 \pad$ & $\pad 0.8667 \pad$ \\ \hline
	\end{tabular}
	
	\medskip
	
	\begin{tabular}{|c|c|c|c|c|c|c|c|c|} \hline
		 & $\structure$ & $\affinity$ & $\deviationfromrandom$ & $\avgdist$ & $\varentropy$ & $\normvarentropy$ & $\seqentropy$ & $\normseqentropy$ \\ \hline
		$L_1$ & $\pad 2.6 \pad$ & $\pad 0.2 \pad$ & $\pad 0.5417 \pad$ & $\pad 2.4 \pad$ & $\pad 6.0684 \pad$ & $\pad 0.5645 \pad$ & $\pad 11.1636 \pad$ & $\pad 0.3348 \pad$ \\ \hline
		$L_2$ & $\pad 3.7 \pad$ & $\pad 0.2316 \pad$ & $\pad 0.6705 \pad$ & $\pad 3.1333 \pad$ & $\pad 32.1247 \pad$ & $\pad 0.5742 \pad$ & $\pad 61.0512 \pad$ & $\pad 0.401 \pad$ \\ \hline
		$L_3$ & $\pad 4.0667 \pad$ & $\pad 0.237 \pad$ & $\pad 0.6926 \pad$ & $\pad 4.7429 \pad$ & $\pad 92.954 \pad$ & $\pad 0.5747 \pad$ & $\pad 174.779 \pad$ & $\pad 0.4696 \pad$ \\ \hline
	\end{tabular}
\end{center}
Since $\mathcal{C}^L(L_1) < \mathcal{C}^L(L_2) < \mathcal{C}^L(L_3)$ for any log complexity measure $\mathcal{C}^L \in \loc$, we have thus shown that $(\mathcal{C}^L, \mathcal{C}^M) \in \norel$ for any model complexity measure $\mathcal{C}^M \in \{\crossconn, \sequentiality\}$. \hfill$\square$
\end{proof}

\begin{theorem}
\label{theo:alphaminer-sep-entries}
$(\mathcal{C}^L, \separability) \in \norel$ for any log complexity measure $\mathcal{C}^L \in \loc$.
\end{theorem}
\begin{proof}
Consider the following event logs:
\begin{align*}
	L_1 &= [\langle a,b,c,d,e \rangle^{3}, \langle e,d,c,a,b \rangle^{3}] \\
	L_2 &= L_1 + [\langle a,f,e,d,c,b \rangle^{2}] \\
	L_3 &= L_2 + [\langle g,b,c,d,e,f,c \rangle^{2}]
\end{align*}
\cref{fig:alphaminer-counterexample-separability} shows the models $M_1, M_2, M_3$ found by the alpha miner for $L_1, L_2, L_3$.
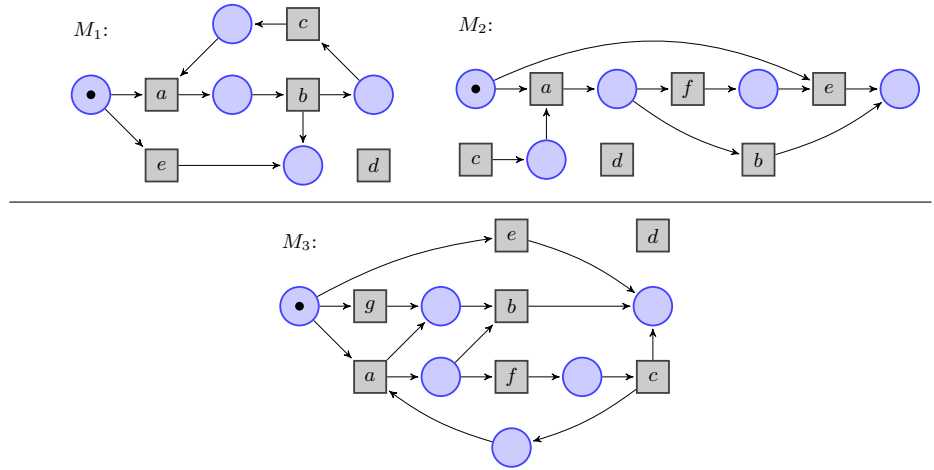
\begin{figure}[ht]
	\centering
	\begin{minipage}{0.45\textwidth}
	\centering
	\scalebox{\scalefactor}{
	\begin{tikzpicture}[node distance = 1.1cm,>=stealth',bend angle=0,auto]
		\node[place,tokens=1] (start) {};
		\node[yshift=1cm] at (start) {$M_1$:};
		\node[transition,right of=start] (a) {$a$}
		edge [pre] (start);
		\node[place,right of=a] (p1) {}
		edge [pre] (a);
		\node[transition,right of=p1] (b) {$b$}
		edge [pre] (p1);
		\node[place,right of=b] (p2) {}
		edge [pre] (b);
		\node[transition,above of=b] (c) {$c$}
		edge [pre] (p2);
		\node[place,left of=c] (p3) {}
		edge [pre] (c)
		edge [post] (a);
		\node[transition,below of=a] (e) {$e$}
		edge [pre] (start);
		\node[place,below of=b] (end) {}
		edge [pre] (b)
		edge [pre] (e);
		\node[transition,right of=end] (d) {$d$};
	\end{tikzpicture}}
	\end{minipage}
	\begin{minipage}{0.52\textwidth}
	\centering
	\scalebox{\scalefactor}{
	\begin{tikzpicture}[node distance = 1.1cm,>=stealth',bend angle=0,auto]
		\node[place,tokens=1] (start) {};
		\node[yshift=1cm] at (start) {$M_2$:};
		\node[transition,right of=start] (a) {$a$}
		edge [pre] (start);
		\node[place,right of=a] (p1) {}
		edge [pre] (a);
		\node[transition,right of=p1] (f) {$f$}
		edge [pre] (p1);
		\node[place,right of=f] (p2) {}
		edge [pre] (f);
		\node[transition,right of=p2] (e) {$e$}
		edge [pre] (p2)
		edge [pre,bend right=25] (start);
		\node[place,right of=e] (end) {}
		edge [pre] (e);
		\node[transition,below of=start] (c) {$c$};
		\node[place,right of=c] (p3) {}
		edge [pre] (c)
		edge [post] (a);
		\node[transition,right of=p3] (d) {$d$};
		\node[transition,below of=p2] (b) {$b$}
		edge [pre,bend left=10] (p1)
		edge [post,bend right=10] (end);
	\end{tikzpicture}}
	\end{minipage}
	
	\medskip
	\hrule
	\medskip
	
	\scalebox{\scalefactor}{
	\begin{tikzpicture}[node distance = 1.1cm,>=stealth',bend angle=0,auto]
		\node[place,tokens=1] (start) {};
		\node[yshift=1cm] at (start) {$M_3$:};
		\node[transition,right of=start] (g) {$g$}
		edge [pre] (start);
		\node[transition,below of=g] (a) {$a$}
		edge [pre] (start);
		\node[place,right of=g] (p1) {}
		edge [pre] (a)
		edge [pre] (g);
		\node[place, right of=a] (p2) {}
		edge [pre] (a);
		\node[transition,right of=p1] (b) {$b$}
		edge [pre] (p1)
		edge [pre] (p2);
		\node[transition,right of=p2] (f) {$f$}
		edge [pre] (p2);
		\node[place,right of=f] (p3) {}
		edge [pre] (f);
		\node[transition,right of=p3] (c) {$c$}
		edge [pre] (p3);
		\node[place,above of=c] (end) {}
		edge [pre] (b)
		edge [pre] (c);
		\node[place, below of=f] (p4) {}
		edge [pre,bend right=10] (c)
		edge [post,bend left=10] (a);
		\node[transition,above of=b] (e) {$e$}
		edge [pre,bend right=10] (start)
		edge [post,bend left=10] (end);
		\node[transition,above of=end] (d) {$d$};
	\end{tikzpicture}}
	\caption{The results of the alpha algorithm for the input logs $L_1, L_2, L_3$ from the example in \cref{theo:alphaminer-sep-entries}. $M_1$ is the model mined from the log $L_1$, $M_2$ the model mined from the log $L_2$, and $M_3$ the model mined from the log $L_3$.}
	\label{fig:alphaminer-counterexample-separability}
\end{figure}
These models have the following separability scores:
\begin{itemize}
	\item[•] $\separability(M_1) = 1$,
	\item[•] $\separability(M_2) \approx 0.7778$,
	\item[•] $\separability(M_3) = 1$,
\end{itemize}
so $\separability(M_1) > \separability(M_2)$, $\separability(M_2) < \separability(M_3)$, and $\separability(M_1) = \separability(M_3)$.
But the event logs $L_1, L_2, L_3$ have the following log complexity scores:
\begin{center}
	\def\pad{\hspace*{1.5mm}}
	\begin{tabular}{|c|c|c|c|c|c|c|c|c|c|c|}\hline
		 & $\magnitude$ & $\variety$ & $\support$ & $\tlavg$ & $\tlmax$ & $\levelofdetail$ & $\numberofties$ & $\lempelziv$ & $\numberuniquetraces$ & $\percentageuniquetraces$ \\ \hline
		$L_1$ & $\pad 30 \pad$ & $\pad 5 \pad$ & $\pad 6 \pad$ & $\pad 5 \pad$ & $\pad 5 \pad$ & $\pad 4 \pad$ & $\pad 3 \pad$ & $\pad 16 \pad$ & $\pad 2 \pad$ & $\pad 0.3333 \pad$ \\ \hline
		$L_2$ & $\pad 42 \pad$ & $\pad 6 \pad$ & $\pad 8 \pad$ & $\pad 5.25 \pad$ & $\pad 6 \pad$ & $\pad 7 \pad$ & $\pad 4 \pad$ & $\pad 21 \pad$ & $\pad 3 \pad$ & $\pad 0.375 \pad$ \\ \hline
		$L_3$ & $\pad 56 \pad$ & $\pad 7 \pad$ & $\pad 10 \pad$ & $\pad 5.6 \pad$ & $\pad 7 \pad$ & $\pad 20 \pad$ & $\pad 5 \pad$ & $\pad 27 \pad$ & $\pad 4 \pad$ & $\pad 0.4 \pad$ \\ \hline
	\end{tabular}
	
	\medskip
	
	\begin{tabular}{|c|c|c|c|c|c|c|c|c|} \hline
		 & $\structure$ & $\affinity$ & $\deviationfromrandom$ & $\avgdist$ & $\varentropy$ & $\normvarentropy$ & $\seqentropy$ & $\normseqentropy$ \\ \hline
		$L_1$ & $\pad 5 \pad$ & $\pad 0.4857 \pad$ & $\pad 0.659 \pad$ & $\pad 3.6 \pad$ & $\pad 6.9315 \pad$ & $\pad 0.301 \pad$ & $\pad 20.7944 \pad$ & $\pad 0.2038 \pad$ \\ \hline
		$L_2$ & $\pad 5.25 \pad$ & $\pad 0.3571 \pad$ & $\pad 0.7031 \pad$ & $\pad 4.0714 \pad$ & $\pad 16.4792 \pad$ & $\pad 0.4057 \pad$ & $\pad 45.1709 \pad$ & $\pad 0.2877 \pad$ \\ \hline
		$L_3$ & $\pad 5.4 \pad$ & $\pad 0.3016 \pad$ & $\pad 0.733 \pad$ & $\pad 4.9333 \pad$ & $\pad 30.24 \pad$ & $\pad 0.4447 \pad$ & $\pad 76.6617 \pad$ & $\pad 0.3401 \pad$ \\ \hline
	\end{tabular}
\end{center}
Thus, $\mathcal{C}^L(L_1) < \mathcal{C}^L(L_2) < \mathcal{C}^L(L_3)$ for all $\mathcal{C}^L \in (\loc \setminus \{\affinity\})$.
For affinity $\affinity$, we can change the frequencies of the traces and get the event logs:
\begin{align*}
	L_1 &= [\langle a,b,c,d,e \rangle, \langle e,d,c,a,b \rangle] \\
	L_2 &= L_1 + [\langle a,f,e,d,c,b \rangle^{2}] \\
	L_3 &= L_2 + [\langle g,b,c,d,e,f,c \rangle^{4}]
\end{align*}
For these event logs, we have:
\begin{itemize}
	\item[•] $\affinity(L_1) \approx 0.1429$,
	\item[•] $\affinity(L_2) \approx 0.2857$, and
	\item[•] $\affinity(L_3) \approx 0.3367$,
\end{itemize}
but the same outputs $M_1, M_2, M_3$ of the alpha miner as with the previous event logs.
Therefore, the separability scores from above are also valid for these logs.
Thus, we showed that $(\mathcal{C}^L, \separability) \in \norel$ for all $\mathcal{C}^L \in \loc$. \hfill$\square$
\end{proof}

\begin{theorem}
\label{theo:alphaminer-diam-entries}
$(\mathcal{C}^L, \diameter) \in \norel$ for any log complexity measure $\mathcal{C}^L \in \loc$.
\end{theorem}
\begin{proof}
Consider the following event logs:
\begin{align*}
	L_1 &= [\langle a,b,c,d \rangle^{3}, \langle e \rangle^{2}] \\
	L_2 &= L_1 + [\langle a,b,c,d \rangle^{3}, \langle a,b,c,b,c,d \rangle^{3}, \langle a,c,b,d \rangle, \langle b,c,b,c,b,c,d \rangle, \\
	&\phantom{= L_1 + [}\hspace*{1mm} \langle a,b,c,f,e,f,e \rangle] \\
	L_3 &= L_2 + [\langle a,b,c,b,c,b,c,b,c,d \rangle^{3}, \langle a,b,c,b,c,b,c,b,c,b,c,d \rangle, \\
	&\phantom{= L_2 + [}\hspace*{1mm} \langle a,a,b,b,c,c,d,d \rangle, \langle e,e,f,f,g,g \rangle, \langle h,i,j,k \rangle]
\end{align*}
\cref{fig:alphaminer-counterexample-diameter} shows the models $M_1, M_2, M_3$ found by the alpha miner for $L_1, L_2, L_3$.
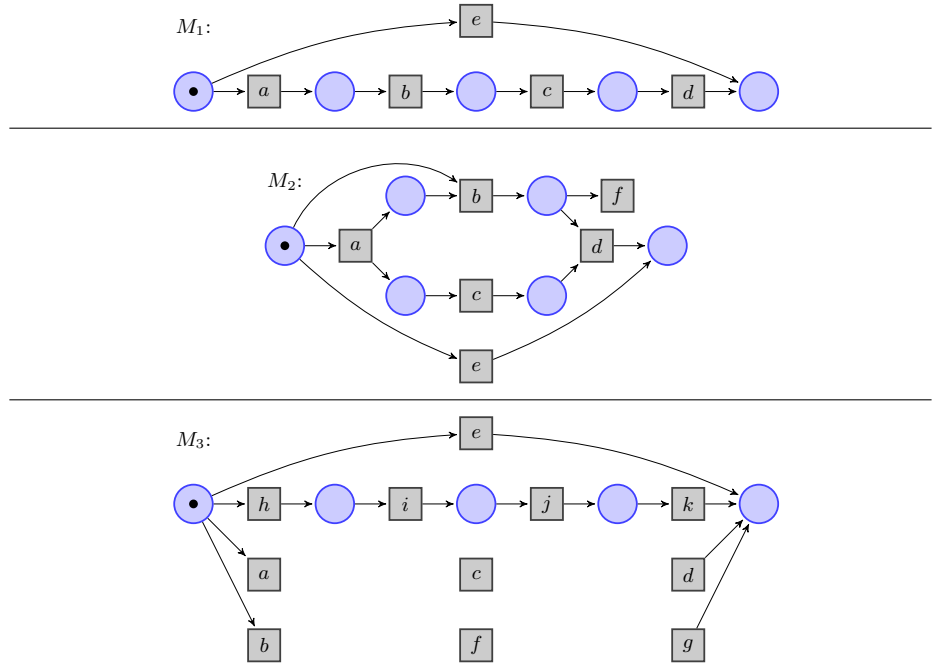
\begin{figure}[ht]
	\centering
	\scalebox{\scalefactor}{
	\begin{tikzpicture}[node distance = 1.1cm,>=stealth',bend angle=0,auto]
		\node[place,tokens=1] (start) {};
		\node[yshift=1cm] at (start) {$M_1$:};
		\node[transition,right of=start] (a) {$a$}
		edge [pre] (start);
		\node[place,right of=a] (p1) {}
		edge [pre] (a);
		\node[transition,right of=p1] (b) {$b$}
		edge [pre] (p1);
		\node[place,right of=b] (p2) {}
		edge [pre] (b);
		\node[transition,right of=p2] (c) {$c$}
		edge [pre] (p2);
		\node[place,right of=c] (p3) {}
		edge [pre] (c);
		\node[transition,right of=p3] (d) {$d$}
		edge [pre] (p3);
		\node[place,right of=d] (end) {}
		edge [pre] (d);
		\node[transition,above of=p2] (e) {$e$}
		edge [pre,bend right=10] (start)
		edge [post,bend left=10] (end);
	\end{tikzpicture}}
	
	\medskip
	\hrule
	\medskip
	
	\scalebox{\scalefactor}{
	\begin{tikzpicture}[node distance = 1.1cm,>=stealth',bend angle=0,auto]
		\node[place,tokens=1] (start) {};
		\node[yshift=1cm] at (start) {$M_2$:};
		\node[transition,right of=start] (a) {$a$}
		edge [pre] (start);
		\node[place,above right of=a] (p1) {}
		edge [pre] (a);
		\node[place,below right of=a] (p2) {}
		edge [pre] (a);
		\node[transition,right of=p1] (b) {$b$}
		edge [pre] (p1)
		edge [pre,bend right=50] (start);
		\node[transition,right of=p2] (c) {$c$}
		edge [pre] (p2);
		\node[place,right of=b] (p3) {}
		edge [pre] (b);
		\node[place,right of=c] (p4) {}
		edge [pre] (c);
		\node[transition,below right of=p3] (d) {$d$}
		edge [pre] (p3)
		edge [pre] (p4);
		\node[place,right of=d] (end) {}
		edge [pre] (d);
		\node[transition,right of=p3] (f) {$f$}
		edge [pre] (p3);
		\node[transition,below of=c] (e) {$e$}
		edge [pre,bend left=10] (start)
		edge [post,bend right=10] (end);
	\end{tikzpicture}}
	
	\medskip
	\hrule
	\medskip
	
	\scalebox{\scalefactor}{
	\begin{tikzpicture}[node distance = 1.1cm,>=stealth',bend angle=0,auto]
		\node[place,tokens=1] (start) {};
		\node[yshift=1cm] at (start) {$M_3$:};
		\node[transition,right of=start] (h) {$h$}
		edge [pre] (start);
		\node[place,right of=h] (p1) {}
		edge [pre] (h);
		\node[transition,right of=p1] (i) {$i$}
		edge [pre] (p1);
		\node[place,right of=i] (p2) {}
		edge [pre] (i);
		\node[transition,right of=p2] (j) {$j$}
		edge [pre] (p2);
		\node[place,right of=j] (p3) {}
		edge [pre] (j);
		\node[transition,right of=p3] (k) {$k$}
		edge [pre] (p3);
		\node[place,right of=k] (end) {}
		edge [pre] (k);
		\node[transition,above of=p2] (e) {$e$}
		edge [pre,bend right=10] (start)
		edge [post,bend left=10] (end);
		\node[transition,below of=h] (a) {$a$}
		edge [pre] (start);
		\node[transition,below of=a] (b) {$b$}
		edge [pre] (start);
		\node[transition,below of=p2] (c) {$c$};
		\node[transition,below of=c] (f) {$f$};
		\node[transition,below of=k] (d) {$d$}
		edge [post] (end);
		\node[transition,below of=d] (g) {$g$}
		edge [post] (end);
	\end{tikzpicture}}
	\caption{The results of the alpha algorithm for the input logs $L_1, L_2, L_3$ from the example in \cref{theo:alphaminer-diam-entries}. $M_1$ is the model mined from the log $L_1$, $M_2$ the model mined from the log $L_2$, and $M_3$ the model mined from the log $L_3$.}
	\label{fig:alphaminer-counterexample-diameter}
\end{figure}
These models have the following diameter scores:
\begin{itemize}
	\item[•] $\diameter(M_1) = 9$,
	\item[•] $\diameter(M_2) = 7$,
	\item[•] $\diameter(M_3) = 9$,
\end{itemize}
so these models fulfill $\diameter(M_1) > \diameter(M_2)$, $\diameter(M_2) < \diameter(M_3)$, and $\diameter(M_1) = \diameter(M_3)$.
But the event logs $L_1, L_2, L_3$ have the following complexity scores:
\begin{center}
	\def\pad{\hspace*{1.5mm}}
	\begin{tabular}{|c|c|c|c|c|c|c|c|c|c|c|}\hline
		 & $\magnitude$ & $\variety$ & $\support$ & $\tlavg$ & $\tlmax$ & $\levelofdetail$ & $\numberofties$ & $\lempelziv$ & $\numberuniquetraces$ & $\percentageuniquetraces$ \\ \hline
		$L_1$ & $\pad 14 \pad$ & $\pad 5 \pad$ & $\pad 5 \pad$ & $\pad 2.8 \pad$ & $\pad 4 \pad$ & $\pad 2 \pad$ & $\pad 3 \pad$ & $\pad 9 \pad$ & $\pad 2 \pad$ & $\pad 0.4 \pad$ \\ \hline
		$L_2$ & $\pad 62 \pad$ & $\pad 6 \pad$ & $\pad 14 \pad$ & $\pad 4.4286 \pad$ & $\pad 7 \pad$ & $\pad 10 \pad$ & $\pad 5 \pad$ & $\pad 25 \pad$ & $\pad 6 \pad$ & $\pad 0.4286 \pad$ \\ \hline
		$L_3$ & $\pad 122 \pad$ & $\pad 11 \pad$ & $\pad 21 \pad$ & $\pad 5.8095 \pad$ & $\pad 12 \pad$ & $\pad 15 \pad$ & $\pad 9 \pad$ & $\pad 47 \pad$ & $\pad 11 \pad$ & $\pad 0.5238 \pad$ \\ \hline
	\end{tabular}
	
	\medskip
	
	\begin{tabular}{|c|c|c|c|c|c|c|c|c|} \hline
		 & $\structure$ & $\affinity$ & $\deviationfromrandom$ & $\avgdist$ & $\varentropy$ & $\normvarentropy$ & $\seqentropy$ & $\normseqentropy$ \\ \hline
		$L_1$ & $\pad 2.8 \pad$ & $\pad 0.4 \pad$ & $\pad 0.4584 \pad$ & $\pad 3 \pad$ & $\pad 2.502 \pad$ & $\pad 0.3109 \pad$ & $\pad 5.7416 \pad$ & $\pad 0.1554 \pad$ \\ \hline
		$L_2$ & $\pad 3.5714 \pad$ & $\pad 0.4555 \pad$ & $\pad 0.565 \pad$ & $\pad 3.3626 \pad$ & $\pad 36.6995 \pad$ & $\pad 0.5397 \pad$ & $\pad 78.6547 \pad$ & $\pad 0.3074 \pad$ \\ \hline
		$L_3$ & $\pad 3.6667 \pad$ & $\pad 0.4191 \pad$ & $\pad 0.5679 \pad$ & $\pad 5.7905 \pad$ & $\pad 103.554 \pad$ & $\pad 0.588 \pad$ & $\pad 224.82 \pad$ & $\pad 0.3836 \pad$ \\ \hline
	\end{tabular}
\end{center}
Thus, $\mathcal{C}^L(L_1) < \mathcal{C}^L(L_2) < \mathcal{C}^L(L_3)$ for all $\mathcal{C}^L \in (\loc \setminus \{\affinity\}$.
For affinity $\affinity$, we can use the event logs $L_1, L_2, L_3$ from the introductory example of this subsection, whose models $M_1, M_2, M_3$ found by the alpha algorithm are shown in \cref{fig:alpha-example-L1}, \cref{fig:alpha-example-L2}, and \cref{fig:alpha-example-L3}.
For these event logs, we have that $\affinity(L_1) = 0.0476 < \affinity(L_2) = 0.1357 < \affinity(L_3) = 0.1498$, but $\diameter(M_1) = 13 < 15 = \diameter(M_2)$, $\diameter(M_2) = 15 > 13 = \diameter(M_3)$, and $\diameter(M_1) = 13 = \diameter(M_3)$.
Thus, we showed $(\mathcal{C}^L, \diameter) \in \norel$ for all log complexity measures $\mathcal{C}^L \in \loc$. \hfill$\square$
\end{proof}

\begin{theorem}
\label{theo:alphaminer-cnc-entries}
$(\mathcal{C}^L, \netconn) \in \norel$ for any event log complexity measure $\mathcal{C}^L \in \loc$.
\end{theorem}
\begin{proof}
Consider the following event logs:
\begin{align*}
	L_1 &= [\langle a,c,d \rangle^{4}, \langle b \rangle] \\
	L_2 &= L_1 + [\langle a,c,d,e \rangle, \langle b,c,d,e \rangle, \langle b,c,e,d \rangle] \\
	L_3 &= L_2 + [\langle a,c,d,e\rangle, \langle a,c,e,d \rangle, \langle a,b,c,d \rangle, \langle a,a,b,b,c,c,d,d,e,e,f,f \rangle, \langle c \rangle]
\end{align*}
\cref{fig:alphaminer-counterexample-netconn} shows the models $M_1, M_2, M_3$ found by the alpha miner for $L_1, L_2, L_3$.
\begin{figure}[ht]
	\centering
	\begin{minipage}{0.6\textwidth}
	\centering
	\scalebox{\scalefactor}{
	\begin{tikzpicture}[node distance = 1.1cm,>=stealth',bend angle=0,auto]
		\node[place,tokens=1] (start) {};
		\node[yshift=1cm] at (start) {$M_1$:};
		\node[transition,right of=start] (a) {$a$}
		edge [pre] (start);
		\node[place,right of=a] (p1) {}
		edge [pre] (a);
		\node[transition,right of=p1] (c) {$c$}
		edge [pre] (p1);
		\node[place,right of=c] (p2) {}
		edge [pre] (c);
		\node[transition,right of=p2] (d) {$d$}
		edge [pre] (p2);
		\node[place,right of=d] (end) {}
		edge [pre] (d);
		\node[transition,above of=c] (b) {$b$}
		edge [pre,bend right=10] (start)
		edge [post,bend left=10] (end);
	\end{tikzpicture}}
	
	\medskip
	\hrule
	\medskip
	
	\scalebox{\scalefactor}{
	\begin{tikzpicture}[node distance = 1.1cm,>=stealth',bend angle=0,auto]
		\node[place,tokens=1] (start) {};
		\node[yshift=1cm] at (start) {$M_2$:};
		\node[transition,right of=start] (a) {$a$}
		edge [pre] (start);
		\node[place,right of=a] (p1) {}
		edge [pre] (a);
		\node[transition,right of=p1] (c) {$c$}
		edge [pre] (p1);
		\node[place,right of=c] (p2) {}
		edge [pre] (c);
		\node[place,below of=p2] (p3) {}
		edge [pre] (c);
		\node[transition,right of=p2] (e) {$e$}
		edge [pre] (p2);
		\node[transition,right of=p3] (d) {$d$}
		edge [pre] (p3);
		\node[place,right of=e] (end) {}
		edge [pre] (d)
		edge [pre] (e);
		\node[transition,above of=c] (b) {$b$}
		edge [pre,bend right=10] (start)
		edge [post,bend left=10] (end)
		edge [post] (p1);
	\end{tikzpicture}}
	\end{minipage}
	\begin{minipage}{0.3\textwidth}
	\centering
	\scalebox{\scalefactor}{
	\begin{tikzpicture}[node distance = 1.1cm,>=stealth',bend angle=0,auto]
		\node[place,tokens=1] (start) {};
		\node[yshift=1cm] at (start) {$M_3$:};
		\node[transition,right of=start,yshift=2.5cm] (a) {$a$}
		edge [pre] (start);
		\node[transition,right of=start,yshift=1.5cm] (b) {$b$}
		edge [pre] (start);
		\node[transition,right of=start,yshift=0.5cm] (c) {$c$}
		edge [pre] (start);
		\node[transition,right of=start,yshift=-0.5cm] (d) {$d$};
		\node[transition,right of=start,yshift=-1.5cm] (e) {$e$};
		\node[transition,right of=start,yshift=-2.5cm] (f) {$f$};
		\node[place,right of=c,yshift=-0.5cm] (end) {}
		edge [pre] (b)
		edge [pre] (c)
		edge [pre] (d)
		edge [pre] (e)
		edge [pre] (f);
	\end{tikzpicture}}
	\end{minipage}
	\caption{The results of the alpha algorithm for the input logs $L_1, L_2, L_3$ from the example in \cref{theo:alphaminer-cnc-entries}. $M_1$ is the model mined from the log $L_1$, $M_2$ the model mined from the log $L_2$, and $M_3$ the model mined from the log $L_3$.}
	\label{fig:alphaminer-counterexample-netconn}
\end{figure}
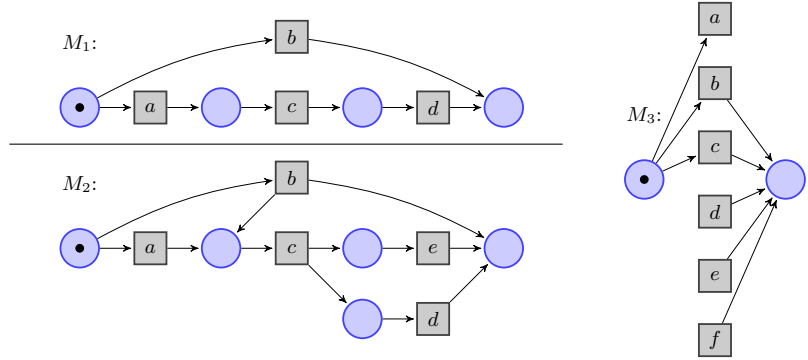
These models have the following complexity scores:
\begin{itemize}
	\item[•] $\netconn(M_1) = 1$,
	\item[•] $\netconn(M_2) = 1.2$,
	\item[•] $\netconn(M_3) = 1$,
\end{itemize}
so with this, we have $\netconn(M_1) > \netconn(M_2)$, $\netconn(M_2) < \netconn(M_3)$, and $\netconn(M_1) = \netconn(M_3)$.
But the event logs $L_1, L_2, L_3$ have the following complexity scores:
\begin{center}
	\def\pad{\hspace*{1.5mm}}
	\begin{tabular}{|c|c|c|c|c|c|c|c|c|c|c|}\hline
		 & $\magnitude$ & $\variety$ & $\support$ & $\tlavg$ & $\tlmax$ & $\levelofdetail$ & $\numberofties$ & $\lempelziv$ & $\numberuniquetraces$ & $\percentageuniquetraces$ \\ \hline
		$L_1$ & $\pad 13 \pad$ & $\pad 4 \pad$ & $\pad 5 \pad$ & $\pad 2.6 \pad$ & $\pad 3 \pad$ & $\pad 2 \pad$ & $\pad 2 \pad$ & $\pad 8 \pad$ & $\pad 2 \pad$ & $\pad 0.4 \pad$ \\ \hline
		$L_2$ & $\pad 25 \pad$ & $\pad 5 \pad$ & $\pad 8 \pad$ & $\pad 3.125 \pad$ & $\pad 4 \pad$ & $\pad 9 \pad$ & $\pad 4 \pad$ & $\pad 12 \pad$ & $\pad 5 \pad$ & $\pad 0.625 \pad$ \\ \hline
		$L_3$ & $\pad 50 \pad$ & $\pad 6 \pad$ & $\pad 13 \pad$ & $\pad 3.8462 \pad$ & $\pad 12 \pad$ & $\pad 30 \pad$ & $\pad 6 \pad$ & $\pad 23 \pad$ & $\pad 9 \pad$ & $\pad 0.6923 \pad$ \\ \hline
	\end{tabular}
	
	\medskip
	
	\begin{tabular}{|c|c|c|c|c|c|c|c|c|} \hline
		 & $\structure$ & $\affinity$ & $\deviationfromrandom$ & $\avgdist$ & $\varentropy$ & $\normvarentropy$ & $\seqentropy$ & $\normseqentropy$ \\ \hline
		$L_1$ & $\pad 2.6 \pad$ & $\pad 0.6 \pad$ & $\pad 0.3386 \pad$ & $\pad 1.6 \pad$ & $\pad 2.2493 \pad$ & $\pad 0.4056 \pad$ & $\pad 3.5255 \pad$ & $\pad 0.1057 \pad$ \\ \hline
		$L_2$ & $\pad 3.125 \pad$ & $\pad 0.3702 \pad$ & $\pad 0.5465 \pad$ & $\pad 2.25 \pad$ & $\pad 10.5492 \pad$ & $\pad 0.4581 \pad$ & $\pad 21.1028 \pad$ & $\pad 0.2622 \pad$ \\ \hline
		$L_3$ & $\pad 3.3846 \pad$ & $\pad 0.2541 \pad$ & $\pad 0.6768 \pad$ & $\pad 3.2308 \pad$ & $\pad 45.452 \pad$ & $\pad 0.5108 \pad$ & $\pad 73.2612 \pad$ & $\pad 0.3745 \pad$ \\ \hline
	\end{tabular}
\end{center}
Thus, $\mathcal{C}^L(L_1) < \mathcal{C}^L(L_2) < \mathcal{C}^L(L_3)$ for all $\mathcal{C}^L \in (\loc \setminus \{\affinity\}$.
For affinity $\affinity$, we can use the event logs $L_1, L_2, L_3$ from the introductory example of this subsection, whose models $M_1, M_2, M_3$ found by the alpha algorithm are shown in \cref{fig:alpha-example-L1}, \cref{fig:alpha-example-L2}, and \cref{fig:alpha-example-L3}.
For these event logs, we have that $\affinity(L_1) = 0.0476 < \affinity(L_2) = 0.1357 < \affinity(L_3) = 0.1498$, but at the same time we have $\netconn(M_1) \approx 1.0476 < 1.0741 \approx \netconn(M_2)$, $\netconn(M_2) \approx 1.0741 > 1.0476 \approx \netconn(M_3)$, and, furthermore, the property $\netconn(M_1) \approx 1.0476 \approx \netconn(M_3)$.
Thus, we showed $(\mathcal{C}^L, \netconn) \in \norel$ for all log complexity measures $\mathcal{C}^L \in \loc$. \hfill$\square$
\end{proof}

\begin{theorem}
\label{theo:alphaminer-dens-entries}
$(\mathcal{C}^L, \density) \in \norel$ for any log complexity measure $\mathcal{C}^L \in \loc$.
\end{theorem}
\begin{proof}
Consider the following event logs:
\begin{align*}
	L_1 &= [\langle a,b,c,d \rangle^{3}, \langle e \rangle^{2}] \\
	L_2 &= L_1 + [\langle a,b,c,d \rangle^{3}, \langle a,c,b,d \rangle, \langle a,b,c,b,c,d \rangle^{3}, \langle b,c,b,c,b,c,d \rangle, \\
	&\phantom{= L_1 + [}\hspace*{1mm} \langle a,b,c,f,e,f,e \rangle] \\
	L_3 &= L_2 + [\langle a,b,c,d \rangle^{3}, \langle a,b,c,b,c,b,c,b,c,d \rangle^{3}, \langle a,b,c,b,c,b,c,b,c,b,c,d \rangle, \\
	&\phantom{= L_2 + [}\hspace*{1mm} \langle a,a,b,b,c,c,d,d \rangle, \langle e,e,f,f,g,g,e,e \rangle, \langle a,a,h,h,i,i,j,j,e,e \rangle]
\end{align*}
\cref{fig:alphaminer-counterexample-density} shows the models $M_1, M_2, M_3$ found by the alpha miner for $L_1, L_2, L_3$.
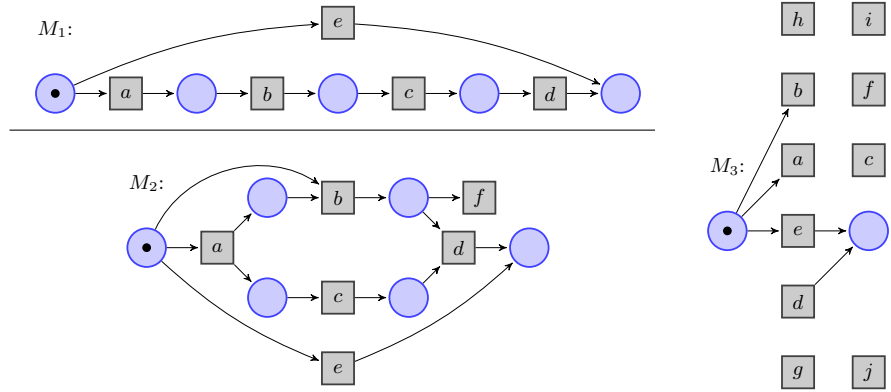
\begin{figure}[ht]
	\centering
	\begin{minipage}{0.7\textwidth}
	\centering
	\scalebox{\scalefactor}{
	\begin{tikzpicture}[node distance = 1.1cm,>=stealth',bend angle=0,auto]
		\node[place,tokens=1] (start) {};
		\node[yshift=1cm] at (start) {$M_1$:};
		\node[transition,right of=start] (a) {$a$}
		edge [pre] (start);
		\node[place,right of=a] (p1) {}
		edge [pre] (a);
		\node[transition,right of=p1] (b) {$b$}
		edge [pre] (p1);
		\node[place,right of=b] (p2) {}
		edge [pre] (b);
		\node[transition,right of=p2] (c) {$c$}
		edge [pre] (p2);
		\node[place,right of=c] (p3) {}
		edge [pre] (c);
		\node[transition,right of=p3] (d) {$d$}
		edge [pre] (p3);
		\node[place,right of=d] (end) {}
		edge [pre] (d);
		\node[transition,above of=p2] (e) {$e$}
		edge [pre,bend right=10] (start)
		edge [post,bend left=10] (end);
	\end{tikzpicture}}
	
	\medskip
	\hrule
	\medskip
	
	\scalebox{\scalefactor}{
	\begin{tikzpicture}[node distance = 1.1cm,>=stealth',bend angle=0,auto]
		\node[place,tokens=1] (start) {};
		\node[yshift=1cm] at (start) {$M_2$:};
		\node[transition,right of=start] (a) {$a$}
		edge [pre] (start);
		\node[place,above right of=a] (p1) {}
		edge [pre] (a);
		\node[place,below right of=a] (p2) {}
		edge [pre] (a);
		\node[transition,right of=p1] (b) {$b$}
		edge [pre] (p1)
		edge [pre,bend right=50] (start);
		\node[transition,right of=p2] (c) {$c$}
		edge [pre] (p2);
		\node[place,right of=b] (p3) {}
		edge [pre] (b);
		\node[place,right of=c] (p4) {}
		edge [pre] (c);
		\node[transition,below right of=p3] (d) {$d$}
		edge [pre] (p3)
		edge [pre] (p4);
		\node[place,right of=d] (end) {}
		edge [pre] (d);
		\node[transition,right of=p3] (f) {$f$}
		edge [pre] (p3);
		\node[transition,below of=c] (e) {$e$}
		edge [pre,bend left=10] (start)
		edge [post,bend right=10] (end);
	\end{tikzpicture}}
	\end{minipage}
	\begin{minipage}{0.28\textwidth}
	\centering
	\scalebox{\scalefactor}{
	\begin{tikzpicture}[node distance = 1.1cm,>=stealth',bend angle=0,auto]
		\node[place,tokens=1] (start) {};
		\node[yshift=1cm] at (start) {$M_3$:};
		\node[transition,right of=start] (e) {$e$}
		edge [pre] (start);
		\node[place,right of=e] (end) {}
		edge [pre] (e);
		\node[transition,above of=e] (a) {$a$}
		edge [pre] (start);
		\node[transition,above of=a] (b) {$b$}
		edge [pre] (start);
		\node[transition,above of=b] (h) {$h$};
		\node[transition,below of=e] (d) {$d$}
		edge [post] (end);
		\node[transition,below of=d] (g) {$g$};
		\node[transition,right of=g] (j) {$j$};
		\node[transition,right of=a] (c) {$c$};
		\node[transition,right of=b] (f) {$f$};
		\node[transition,right of=h] (i) {$i$};
	\end{tikzpicture}}
	\end{minipage}
	\caption{The results of the alpha algorithm for the input logs $L_1, L_2, L_3$ from the example in \cref{theo:alphaminer-dens-entries}. $M_1$ is the model mined from the log $L_1$, $M_2$ the model mined from the log $L_2$, and $M_3$ the model mined from the log $L_3$.}
	\label{fig:alphaminer-counterexample-density}
\end{figure}

\newpage
\noindent
These models have the following density scores:
\begin{itemize}
	\item[•] $\density(M_1) = 0.25$,
	\item[•] $\density(M_2) \approx 0.2333$,
	\item[•] $\density(M_3) = 0.25$,
\end{itemize}
so these models fulfill $\density(M_1) > \density(M_2)$, $\density(M_2) < \density(M_3)$, and $\density(M_1) = \density(M_3)$.
But the event logs $L_1, L_2, L_3$ have the following log complexity scores:
\begin{center}
	\def\pad{\hspace*{1.5mm}}
	\begin{tabular}{|c|c|c|c|c|c|c|c|c|c|c|}\hline
		 & $\magnitude$ & $\variety$ & $\support$ & $\tlavg$ & $\tlmax$ & $\levelofdetail$ & $\numberofties$ & $\lempelziv$ & $\numberuniquetraces$ & $\percentageuniquetraces$ \\ \hline
		$L_1$ & $\pad 14 \pad$ & $\pad 5 \pad$ & $\pad 5 \pad$ & $\pad 2.8 \pad$ & $\pad 4 \pad$ & $\pad 2 \pad$ & $\pad 3 \pad$ & $\pad 9 \pad$ & $\pad 2 \pad$ & $\pad 0.4 \pad$ \\ \hline
		$L_2$ & $\pad 62 \pad$ & $\pad 6 \pad$ & $\pad 14 \pad$ & $\pad 4.4286 \pad$ & $\pad 7 \pad$ & $\pad 10 \pad$ & $\pad 5 \pad$ & $\pad 25 \pad$ & $\pad 6 \pad$ & $\pad 0.4286 \pad$ \\ \hline
		$L_3$ & $\pad 142 \pad$ & $\pad 10 \pad$ & $\pad 24 \pad$ & $\pad 5.9167 \pad$ & $\pad 12 \pad$ & $\pad 14 \pad$ & $\pad 11 \pad$ & $\pad 51 \pad$ & $\pad 11 \pad$ & $\pad 0.4583 \pad$ \\ \hline
	\end{tabular}
	
	\medskip
	
	\begin{tabular}{|c|c|c|c|c|c|c|c|c|} \hline
		 & $\structure$ & $\affinity$ & $\deviationfromrandom$ & $\avgdist$ & $\varentropy$ & $\normvarentropy$ & $\seqentropy$ & $\normseqentropy$ \\ \hline
		$L_1$ & $\pad 2.8 \pad$ & $\pad 0.4 \pad$ & $\pad 0.4584 \pad$ & $\pad 3 \pad$ & $\pad 2.502 \pad$ & $\pad 0.3109 \pad$ & $\pad 5.7416 \pad$ & $\pad 0.1554 \pad$ \\ \hline
		$L_2$ & $\pad 3.5714 \pad$ & $\pad 0.4555 \pad$ & $\pad 0.565 \pad$ & $\pad 3.3626 \pad$ & $\pad 36.6995 \pad$ & $\pad 0.5397 \pad$ & $\pad 78.6547 \pad$ & $\pad 0.3074 \pad$ \\ \hline
		$L_3$ & $\pad 3.75 \pad$ & $\pad 0.4662 \pad$ & $\pad 0.5956 \pad$ & $\pad 5.7029 \pad$ & $\pad 115.926 \pad$ & $\pad 0.5642 \pad$ & $\pad 256.546 \pad$ & $\pad 0.3646 \pad$ \\ \hline
	\end{tabular}
\end{center}
Thus, $\mathcal{C}^L(L_1) < \mathcal{C}^L(L_2) < \mathcal{C}^L(L_3)$ for all $\mathcal{C}^L \in \loc$, so we have just shown that $(\mathcal{C}^L, \density) \in \norel$. \hfill$\square$
\end{proof}

\begin{theorem}
\label{theo:alphaminer-duplicate-entries}
$(\mathcal{C}^L, \duplicate) \in \meq$ for any log complexity measure $\mathcal{C}^L \in \loc$.
\end{theorem}
\begin{proof}
The alpha miner constructs exactly one transition for each activity name in the event log.
Since no other transitions are constructed by the algorithm, a model $M$ found by the alpha algorithm always has $\duplicate(M) = 0$. \hfill$\square$
\end{proof}

Except for $\duplicate$, none of the complexity scores of models found by the alpha miner can be described with current log complexity measures.
This is because the structure of these models highly depend on the set $Y_L$, which is not covered by current log complexity measures.
In fact, a model $M$ constructed by the alpha miner for an event log $L$ has exactly $2 + |Y_L|$ places and $\variety(L)$ transitions.
The edges present in $M$ are also encoded in $Y_L$, as every element $(A, B) \in Y_L$ issues $|A| + |B|$ edges being constructed in $M$.
Furthermore, $M$ contains $|A_I|$ edges starting from $p_i$ and $|A_O|$ edges ending in $p_o$.
Regarding the connectors in the model $M$ found by the alpha algorithm, we have:
\begin{align*}
	S_{\texttt{xor}}^M &= \{(B,C) \in Y_L \mid 1 < |C|\} \cup \{(\emptyset, A_I) \mid 1 < |A_I|\} \\
	J_{\texttt{xor}}^M &= \{(B,C) \in Y_L \mid 1 < |B|\} \cup \{(A_O, \emptyset) \mid 1 < |A_O|\} \\
	S_{\texttt{and}}^M &= \{b \in A \mid 1 < |\{(B, C) \in Y_L \mid b \in B\}|\} \\
	J_{\texttt{and}}^M &= \{c \in A \mid 1 < |\{(B, C) \in Y_L \mid c \in C\}|\}
\end{align*}
We will now describe the model complexity scores of a model $M$ found by the alpha algorithm for an event log $L$ over $A$.

\begin{itemize}
	\item \textbf{Size $\size$:}
	As argued before, $M$ contains $2 + |Y_L|$ places and $\variety(L)$ transitions.
	Thus, $\size(M) = 2 + |Y_L| + \variety(L)$.
	
	\item \textbf{Mismatch $\mismatch$:}
	With the notions above, the amount of mismatches between \texttt{xor}-connectors is  $MM_{\texttt{xor}} = \left| \sum_{(B,C) \in S_{\texttt{xor}}} |C| - |B| \right|$, while the amount of mismatches between \texttt{and}-connectors is 
	\[MM_{\texttt{and}} = \left|\sum_{a \in A} |\{(B,C) \in Y_L \mid a \in B\}| - |\{(B,C) \in Y_L \mid a \in C\}| \right|\]
	With these notions, $\mismatch(M) = MM_{\texttt{xor}} + MM_{\texttt{and}}$.
	
	\item \textbf{Connector Heterogeneity $\connhet$:}
	For the connector heterogeneity score, we take $r_{\texttt{xor}}^M = \frac{|S_{\texttt{xor}}^M \cup J_{\texttt{xor}}^M|}{|S_{\texttt{xor}}^M \cup J_{\texttt{xor}}^M \cup S_{\texttt{and}}^M \cup J_{\texttt{and}}^M|}$ and $r_{\texttt{and}}^M = \frac{|S_{\texttt{and}}^M \cup J_{\texttt{and}}^M|}{|S_{\texttt{xor}}^M \cup J_{\texttt{xor}}^M \cup S_{\texttt{and}}^M \cup J_{\texttt{and}}^M|}$ to calculate the connector heterogeneity $\connhet(M) = -(r_{\texttt{xor}}^M \cdot \log_2(r_{\texttt{xor}}^M) + r_{\texttt{and}}^M \cdot \log_2(r_{\texttt{and}}^M))$.
	
	\item \textbf{Cross Connectivity $\crossconn$:}
	The cross connectivity metric depends not only on properties of single nodes, but instead on all paths through the net. 
	While it would be possible to describe the scores of this measure with just $Y_L$ and $\variety(L)$, we doubt that such a description would yield any value due to its complexity, and therefore skip this metric.
	
	\item \textbf{Token Split $\tokensplit$:}
	With the notions above, we can describe the score of the token split measure as $\tokensplit(M) = \sum_{a \in S_{\texttt{and}}^M} \left(|\{(B,C) \in Y_L \mid a \in B\}| - 1\right)$.
	
	\item \textbf{Control Flow Complexity $\controlflow$:}
	With the notions above, we describe $M$'s control flow complexity score by $\controlflow(M) = |S_{\texttt{and}}^M| + \sum_{(B,C) \in S_{\texttt{xor}}^M} |B|$.
	
	\item \textbf{Separability $\separability$:}
	Like the cross connectivity metric, separability depends on the structure of the whole result, rather than properties of single nodes.
	A description for this measure would be highly complex and therefore of little value, so we skip this measure.
	
	\item \textbf{Average Connector Degree $\avgconn$:}
	With the previous notions, we define $C_{\texttt{xor}}^M = S_{\texttt{xor}}^M \cup J_{\texttt{xor}}^M$ as the set of all \texttt{xor}-connectors, and $C_{\texttt{and}}^M = S_{\texttt{and}}^M \cup J_{\texttt{and}}^M$ as the set of all \texttt{and}-connectors.
	The degree of an \texttt{xor}-connector $(B,C)$ in $M$ is $\text{deg}((B,C)) = |B| + |C|$, while the degree of an \texttt{and}-connector $a$ in $M$ is $\text{deg}(a) = |\{(B,C) \in Y_L \mid a \in B\}| + |\{(B,C) \in Y_L \mid a \in C\}|$.
	With this, we can describe the average connector degree of $M$ as: 
	\[\avgconn(M) = \frac{\sum_{(B,C) \in C_{\texttt{xor}}^M} \text{deg}((B,C)) + \sum_{a \in C_{\texttt{and}}^M} \text{deg}(a)}{|C_{\texttt{xor}}^M| + |C_{\texttt{and}}^M|}.\]
	
	\item \textbf{Maximum Connector Degree $\maxconn$:}
	With the same definitions for $C_{\texttt{xor}}^M$, $C_{\texttt{and}}^M$, $deg((B,C))$ for some $(B,C) \in C_{\texttt{xor}}^M$, and $deg(a)$ for some $a \in C_{\texttt{and}}^M$, we can describe the maximum connector degree as
	\[\maxconn(M) = \max(\{\text{deg}((B,C)) \mid (B,C) \in C_{\texttt{xor}}^M\} \cup \{\text{deg}(a) \mid a \in C_{\texttt{and}}^M\}).\]
	
	\item \textbf{Sequentiality $\sequentiality$}:
	With $C_{\texttt{xor}}^M = S_{\texttt{xor}}^M \cup J_{\texttt{xor}}^M$ and $C_{\texttt{and}}^M = S_{\texttt{and}}^M \cup J_{\texttt{and}}^M$, we can describe the sequentiality score of the alpha miner result $M$ as
	\[\sequentiality(M) = \sum_{(B,C) \in (Y_L \setminus C_{\texttt{xor}}^M)} |\{b \in B \mid b \not\in C_{\texttt{and}}^M\}| + |\{c \in C \mid c \not\in C_{\texttt{and}}^M\}|.\]
	
	\item \textbf{Depth $\depth$:}
	Since the depth of a node is dependent on the paths through $M$, we cannot describe the depth of $M$ in simple terms. 
	Therefore, we will skip this measure.
	
	\item \textbf{Diameter $\diameter$:}
	The diameter of the net is dependent on all paths through $M$, so we cannot describe it for $M$ in simple terms.
	Therefore, we will skip this measure.
	
	\item \textbf{Cyclicity $\cyclicity$:}
	Which nodes lie on cycles depends on the cyclic paths in the net $M$. 
	We cannot describe this notion in simple terms, so we will skip this measure.
	
	\item \textbf{Coefficient of Network Connectivity $\netconn$:}
	By the previous discussions, we know that $M$ contains $2 + |Y_L| + \variety(L)$ nodes and $|A_I| + |A_O| + \sum_{(B,C) \in Y_L} |B| + |C|$ edges, so its coefficient of network connectivity is 
	\[\netconn(M) = \frac{|A_I| + |A_O| + \sum_{(B,C) \in Y_L} |B| + |C|}{2 + |Y_L| + \variety(L)}.\] 
	
	\item \textbf{Density $\density$:}
	By the previous discussions, we know that $M$ contains exactly $2 + |Y_L|$ places, $\variety(L)$ transitions, and $|A_I| + |A_O| + \sum_{(B,C) \in Y_L} |B| + |C|$ edges. 
	Thus, its density is 
	\[\density(M) = \frac{|A_I| + |A_O| + \sum_{(B,C) \in Y_L} |B| + |C|}{2 \cdot \variety(L) \cdot (1 + |Y_L|)}.\]
	
	\item \textbf{Number of Duplicate Tasks $\duplicate$:}
	The alpha miner constructs exactly one transition for every activity name in the event log $L$, and no transitions beyond that.
	Therefore, in every model found by the alpha algorithm, each transition label occurs exactly once, giving us $\duplicate(M) = 0$.
	
	\item \textbf{Number of Empty Sequence Flows $\emptyseq$:}
	The number of empty sequence flows can be described as $\emptyseq(M) = |\{(B,C) \in Y_L \mid B \subseteq S_{\texttt{and}}^M \land C \subseteq J_{\texttt{and}}^M\}|$.
\end{itemize}
These findings conclude our analysis of the alpha miner. 
\cref{table:alphaminer-model-complexity} summarizes these findings for quick reference.

\begin{table}[htp]
	\caption{The complexity scores of the alpha-model $M$ for an event log $L$ over $A$.}
	\label{table:alphaminer-model-complexity}
	\centering
	\def\pad{\hspace*{1.5mm}}
	\renewcommand{\arraystretch}{2}
	\begin{tabular}{|r|l|} \hline
	$\pad\size(M)\pad$ & $\pad 2 + |Y_L| + \variety(L) \pad$ \\ \hline
	$\pad\mismatch(M)\pad$ & $\pad MM_{\texttt{xor}} + MM_{\texttt{and}} \pad$ \\ \hline
	$\pad\connhet(M)\pad$ & $\pad -\left(r_{\texttt{xor}}^M \cdot \log_2(r_{\texttt{xor}}^M) + r_{\texttt{and}}^M \cdot \log_2(r_{\texttt{and}}^M)\right) \pad$ \\ \hline
	$\pad\tokensplit(M)\pad$ & $\pad \sum_{a \in S_{\texttt{and}}^M} \left(|\{(B,C) \in Y_L \mid a \in B\}| - 1\right) \pad$ \\ \hline
	$\pad\controlflow(M)\pad$ & $\pad |S_{\texttt{and}}^M| + \sum_{(B,C) \in S_{\texttt{xor}}^M} |B| \pad$ \\ \hline
	$\pad\avgconn(M)\pad$ & $\pad \frac{\sum_{(B,C) \in C_{\texttt{xor}}^M} \text{deg}((B,C)) + \sum_{a \in C_{\texttt{and}}^M} \text{deg}(a)}{|C_{\texttt{xor}}^M| + |C_{\texttt{and}}^M|} \pad$ \\ \hline
	$\pad\maxconn(M)\pad$ & $\pad \max(\{\text{deg}((B,C)) \mid (B,C) \in C_{\texttt{xor}}^M\} \cup \{\text{deg}(a) \mid a \in C_{\texttt{and}}^M\}) \pad$ \\ \hline
	$\pad\sequentiality(M)\pad$ & $\pad \sum_{(B,C) \in (Y_L \setminus C_{\texttt{xor}}^M)} |\{b \in B \mid b \not\in C_{\texttt{and}}^M\}| + |\{c \in C \mid c \not\in C_{\texttt{and}}^M\}| \pad$ \\ \hline
	$\pad\netconn(M)\pad$ & $\pad \frac{|A_I| + |A_O| + \sum_{(B,C) \in Y_L} |B| + |C|}{2 + |Y_L| + \variety(L)} \pad$ \\ \hline
	$\pad\density(M)\pad$ & $\pad \frac{|A_I| + |A_O| + \sum_{(B,C) \in Y_L} |B| + |C|}{2 \cdot \variety(L) \cdot (1 + |Y_L|)} \pad$ \\ \hline
	$\pad\duplicate(M)\pad$ & $\pad 0 \pad$ \\ \hline
	$\pad\emptyseq(M)\pad$ & $\pad |\{(B,C) \in Y_L \mid B \subseteq S_{\texttt{and}}^M \land C \subseteq J_{\texttt{and}}^M\}| \pad$ \\ \hline
	\end{tabular}
\end{table}

\newpage
\subsection{Directly Follows Graph}
\label{sec:dfg}
Often, organizations prefer the directly follows graph over Petri nets to model the behavior of their systems.
This is due to the semantics of the directly follows graph (DFG) being easy to understand and requiring no further training for process analysts that have to work with the model.
The graph contains one node for every activity name in the event log, alongside with a special start node $\triangleright$ and a special end node $\square$.
Two activity nodes $a, b$ are connected by an edge $(a,b)$, if there is a trace $\sigma$ in the event log, where for some $i \in \{1, \dots, |\sigma| - 1\}$, $\sigma(i) = a$ and $\sigma(i+1) = b$.
In other words, an edge $(a,b)$ in the directly follows graph signals that $a$ can be directly followed by $b$ in the event log.
Similarly, an edge $(\triangleright, a)$, for an activity name $a$, signals that there is a trace in the event log that starts with $a$.
An edge $(a, \square)$, on the other hand, signals that there is a trace in the event log that ends with $a$.
In this subsection, we will assume that $|supp(L)| \geq 1$ for all event logs $L$ whose directly follows graph we compute, to avoid graphs consisting of just two nodes without any edges.

Directly follows graphs are not as expressive as Petri nets.
By design, they can model exclusive choices, but not concurrency.
Because this modelling language is frequently used in practice, we extend our analyses to it.
To start, we first need to translate the model complexity measures to DFGs.
Let $G = (V,E)$ be the directly follows graph for an event log $L$ over a set of activities $A$.
For a node $v \in V$, let $\text{indeg}(v) = |\{w \mid (v,w) \in E\}|$ and $\text{outdeg}(v) = |\{u \mid (u,v) \in E\}$, as well as $\text{deg}(v) = \text{indeg}(v) + \text{outdeg}(v)$.
For simplicity, we define the node sets
\begin{align*}
	S_{\texttt{xor}}^G &= \{v \in V \mid \text{outdeg}(v) > 1\} \\
	J_{\texttt{xor}}^G &= \{v \in V \mid \text{indeg}(v) > 1\}
\end{align*}
as the set of \texttt{xor}-splits and \texttt{xor}-joins, as well as $C_{\texttt{xor}}^G = S_{\texttt{xor}}^G \cup J_{\texttt{xor}}^G$ as the set of all connector nodes in the DFG $G$.
\begin{itemize}
	\item \textbf{Size $\size$:}
	Similarly to a Petri net, we define the size of the DFG as the amount of its nodes, i.e., $\size(G) = |V|$.
	
	\item \textbf{Connector Mismatch $\mismatch$:}
	Since $G$ does not contain any \texttt{and}-connectors, the amount of total connector mismatches is the amount of mismatches between \texttt{xor}-connectors.
	Thus, we define the connector mismatch of the directly follows graph $G$ as $\mismatch(G) = \left| \sum_{v \in S_{\texttt{xor}}^G} \text{outdeg}(v) - \sum_{v \in J_{\texttt{xor}}^G} \text{indeg}(v) \right|$.
	
	\item \textbf{Connector Heterogeneity $\connhet$:}
	Since $G$ only contains \texttt{xor}-connectors, it does not make sense to analyze the entropy of connector types in $G$.
	We will therefore omit this complexity measure for our analyses of the directly follows graph.
	
	\item \textbf{Cross Connectivity $\crossconn$:}
	Since the cross connectivity metric is independent of the modelling language, and thus works for any graph, we its the definition in \cref{sec:model-complexity} for the DFG $G$.
	
	\item \textbf{Token Split $\tokensplit$:}
	Since $G$ does not contain any \texttt{and}-connectors, asking for the amount of edges introducing concurrency does not make sense for the DFG.
	We will thus omit this complexity measure in our analyses of the directly follows graph.
	
	\item \textbf{Control Flow Complexity $\controlflow$:}
	By ignoring the part of control flow complexity that evaluates the cognitive load needed for parallel splits, we get $\controlflow(G) = \sum_{v \in S_{\texttt{xor}}^G} \text{outdeg}(v)$.
	
	\item \textbf{Separability $\separability$:}
	The separability measure is independent of the modeling type, as cut-vertices can occur in every graph.
	Thus, we use the definition of separability in \cref{sec:model-complexity} for the DFG $G$.
	
	\item \textbf{Average Connector Degree $\avgconn$:}
	With our definition of connectors $C_{\texttt{xor}}^G$ in the DFG $G$, we get $\avgconn(G) = \frac{\sum_{v \in C_{\texttt{xor}}^M} \text{deg}(v)}{|C_{\texttt{xor}}^G|}$.
	
	\item \textbf{Maximum Connector Degree $\maxconn$:}
	With our definition of connectors $C_{\texttt{xor}}^G$ in the DFG $G$, we get $\maxconn(G) = \max\{\text{deg}(v) \mid v \in C_{\texttt{xor}}^M\}$.
	
	\item \textbf{Sequentiality $\sequentiality$:}
	With the definition of $C_{\texttt{xor}}^G$, we can define the sequentiality of a DFG $G$ as $\sequentiality(G) = 1 - \frac{1}{|E|} \cdot |\{(u,v) \in E \mid u,v \not\in C_{\texttt{xor}}^G\}|$.
	
	\item \textbf{Depth $\depth$:}
	We reuse the definition of depth shown in \cref{sec:model-complexity} by setting $\mathcal{S}^G = S_{\texttt{xor}}^G$ and $\mathcal{J}^G = J_{\texttt{xor}}^G$.
	
	\item \textbf{Diameter $\diameter$:}
	Since the length of the longest path through the net is independent of the modelling language, we can reuse the definition for $\depth$ from \cref{sec:model-complexity}.
	
	\item \textbf{Cyclicity $\cyclicity$:}
	The notion of cycles is independent of the modelling language and can be used on any graph.
	Since the special nodes $\triangleright$ and $\square$ of $G$ can never lie on a cycle, we reuse the definition from \cref{sec:model-complexity} and define $\cyclicity(G) = \frac{|\{v \in V \mid v \text{ lies on a cycle in } G\}|}{|V| - 2}$.
	
	\item \textbf{Coefficient of Network Connectivity $\netconn$:}
	Similar to the complexity measure for Petri nets, we define $\netconn(G) = \frac{|E|}{|V|}$.
	
	\item \textbf{Density $\density$:}
	In contrast to Petri nets, the DFG can contain edges between all nodes, with two exceptions:
	The start node $\triangleright$ can have only outgoing edges, so $(a, \triangleright) \not\in E$ for all $a \in A \cup \{\square\}$.
	The end node $\square$ can have only incoming edges, so $(\square, a) \not\in E$ for all $a \in A \cup \{\triangleright\}$.
	Thus, we define $\density(G) = \frac{|E|}{|V| \cdot (|V| - 1)}$.
	
	\item \textbf{Number of Duplicate Tasks $\duplicate$:}
	By construction, $G$ cannot contain duplicate labels in nodes, as $V = A \cup \{\triangleright, \square\}$.
	Therefore, it makes no sense to ask for the number of duplicate tasks in the DFG, and we omit this complexity measure for our analyses of the DFG.
	
	\item \textbf{Number of Empty Sequence Flows $\emptyseq$:}
	Since the directly follows graph does not contain any \texttt{and}-connectors, it makes no sense to ask for the number of empty sequence flows. 
	Thus, we will omit this complexity measure for our analyses of the DFG.
\end{itemize}
With these complexity measures for the directly follows graph, we can start our analyses by first observing that the increase of some log complexity scores has no effect on the directly follows graph.

\begin{lemma}
\label{lemma:dfg-unchanging}
Let $\mathcal{C}^L \in (\loc \setminus \{\variety, \levelofdetail, \numberofties\})$.
Then, there are logs $L_1, L_2$ with $L_1 \sqsubset L_2$ and $\mathcal{C}^L(L_1) < \mathcal{C}^L(L_2)$ such that the DFG for $L_1$ is the same as the one for $L_2$.
\end{lemma}
\begin{proof}
Consider the following event logs:
\begin{align*}
	L_1 &= [\langle a,b,c,c \rangle^{2}, \langle c,c,d,e \rangle] \\
	L_2 &= L_1 + [\langle a,b,c,d,e \rangle]
\end{align*}
These event logs have the following log complexity scores:
\begin{center}
	\begin{tabular}{|c|c|c|c|c|c|c|c|c|c|c|c|c|} \hline
		 & $\magnitude$ & $\variety$ & $\support$ & $\tlavg$ & $\tlmax$ & $\levelofdetail$ & $\numberofties$ & $\lempelziv$ & $\numberuniquetraces$ & $\percentageuniquetraces$ \\ \hline
		$L_1$ & $\pad 12 \pad$ & $\pad 5 \pad$ & $\pad 3 \pad$ & $\pad 4 \pad$ & $\pad 4 \pad$ & $\pad 4 \pad$ & $\pad 4 \pad$ & $\pad 8 \pad$ & $\pad 2 \pad$ & $\pad 0.6667 \pad$ \\ \hline
		$L_2$ & $\pad 17 \pad$ & $\pad 5 \pad$ & $\pad 4 \pad$ & $\pad 4.25 \pad$ & $\pad 5 \pad$ & $\pad 4 \pad$ & $\pad 4 \pad$ & $\pad 10 \pad$ & $\pad 3 \pad$ & $\pad 0.75 \pad$ \\ \hline
	\end{tabular}
		
	\medskip
		
	\begin{tabular}{|c|c|c|c|c|c|c|c|c|} \hline
		 & $\structure$ & $\affinity$ & $\deviationfromrandom$ & $\avgdist$ & $\varentropy$ & $\normvarentropy$ & $\seqentropy$ & $\normseqentropy$ \\ \hline
		$L_1$ & $\pad 3 \pad$ & $\pad 0.4667 \pad$ & $\pad 0.5589 \pad$ & $\pad 2.6667 \pad$ & $\pad 5.5452 \pad$ & $\pad 0.3333 \pad$ & $\pad 7.6382 \pad$ & $\pad 0.2562 \pad$ \\ \hline
		$L_2$ & $\pad 3.5 \pad$ & $\pad 0.4333 \pad$ & $\pad 0.5912 \pad$ & $\pad 2.8333 \pad$ & $\pad 10.5492 \pad$ & $\pad 0.4581 \pad$ & $\pad 14.8563 \pad$ & $\pad 0.3084 \pad$ \\ \hline
	\end{tabular}
\end{center}
Thus, all log complexity scores except $\affinity$ increase.
But the directly follows graphs for $L_1$ and $L_2$ are the same, shown in \cref{fig:directly-follows-graph}.
\begin{figure}[h]
	\centering
	\begin{tikzpicture}[node distance = 1.1cm,>=stealth',bend angle=0,auto]
		\node[transition] (start) {$\triangleright$};
		\node[above of=start,yshift=-0.5cm] {$G$:};
		\node[transition,right of=start] (a) {$a$}
		edge [pre] (start);
		\node[transition,right of=a] (b) {$b$}
		edge [pre] (a);
		\node[transition,right of=b] (c) {$c$}
		edge [pre,bend left=30] (start)
		edge [pre] (b)
		edge [pre,loop,out=60,in=120,looseness=6] (c);
		\node[transition,right of=c] (d) {$d$}
		edge [pre] (c);
		\node[transition,right of=d] (e) {$e$}
		edge [pre] (d);
		\node[transition,right of=e] (end) {$\square$}
		edge [pre,bend right=30] (c)
		edge [pre] (e);
	\end{tikzpicture}
	\caption{The directly follows graph of event logs $L_1$ and $L_2$ in \cref{lemma:dfg-unchanging}.}
	\label{fig:directly-follows-graph}
\end{figure}
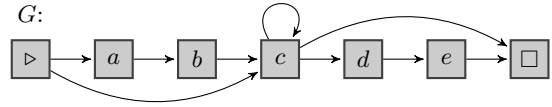
For $\affinity$, take the following event logs:
\begin{align*}
	L_1 &= [\langle a,b,c,c \rangle, \langle c,c,d,e \rangle] \\
	L_2 &= L_1 + [\langle a,b,c,d,e \rangle]
\end{align*}
Then, $\affinity(G_1) = 0.2 < 0.\overline{3} = \affinity(G_2)$, but the directly follows graphs for $L_1$ and $L_2$ are the same, shown in \cref{fig:directly-follows-graph}.
\end{proof}

Next, we find that some complexity measures are monotone increasing when behavior is added to the underlying event log.
To make sure that all complexity scores are well-defined, we require $|supp(L)| > 1$ for all of our investigated event logs $L$, as logs containing only one trace seldom occur in practice and are thus not as interesting to investigate.

\begin{lemma}
\label{lemma:dfg-monotone-increasing}
Let $L_1, L_2$ be event logs with $L_1 \sqsubset L_2$ and $|supp(L_1)| > 1$.
Let $G_1, G_2$ be the directly follows graphs for $L_1$ and $L_2$.
Then, $\mathcal{C}^M(G_1) \leq \mathcal{C}^M(G_2)$ for any model complexity measure $\mathcal{C}^M \in \{\size, \controlflow, \maxconn, \diameter\}$.
\end{lemma}
\begin{proof}
We prove the conjecture for each of the complexity measures separately.
\begin{itemize}
	\item \textbf{Size $\size$:}
	Since $L_1 \sqsubset L_2$, every activity name in $L_1$ is also present in $L_2$.
	Therefore, $G_2$ must contain all nodes from $G_1$ and we thus get that $\size(G_1) \leq \size(G_2)$.
	
	\item \textbf{Control Flow Complexity $\controlflow$:}
	Since $L_1 \sqsubset L_2$, every direct neighborhood in $L_1$ also occurs in $L_2$.
	This means, if two activities $a, b$ can occur directly after one another in $L_1$, this is also true for $L_2$, since the respective trace is contained in both event logs.
	Thus, if $G_1 = (V_1, E_1)$ and $G_2 = (V_2, E_2)$, we know that $V_1 \subseteq V_2$ and $E_1 \subseteq E_2$.
	Therefore, for all nodes $v_1 \in V_1$ and $v_2 \in V_2$, we have $\text{outdeg}(v_1) \leq \text{outdeg}(v_2)$.
	Therefore, every node classified as an \texttt{xor}-split in $G_1$ must also be classified as such in $G_2$.
	This and the fact that these nodes have the same out-degree in $G_1$ and $G_2$ leads to $\controlflow(G_1) \leq \controlflow(G_2)$.
	
	\item \textbf{Maximum Connector Degree $\maxconn$:}
	Since $L_1 \sqsubset L_2$, every direct neighborhood in $L_1$ also occurs in $L_2$.
	This means, if two activities $a, b$ can occur directly after one another in $L_1$, this is also true for $L_2$, since the respective trace is contained in both event logs.
	Thus, if $G_1 = (V_1, E_1)$ and $G_2 = (V_2, E_2)$, we know that $V_1 \subseteq V_2$ and $E_1 \subseteq E_2$.
	Therefore, for all nodes $v_1 \in V_1$ and $v_2 \in V_2$, we have $\text{deg}(v_1) \leq \text{deg}(v_2)$.
	Since all nodes classified as an \texttt{xor}-split in $G_1$ must also be classified as such in $G_2$, we get $\maxconn(G_1) \leq \maxconn(G_2)$.
	
	\item \textbf{Diameter $\diameter$:}
	Since $L_1 \sqsubset L_2$, every direct neighborhood in $L_1$ also occurs in $L_2$.
	This means, if two activities $a, b$ can occur directly after one another in $L_1$, this is also true for $L_2$, since the respective trace is contained in both event logs.
	Thus, if $G_1 = (V_1, E_1)$ and $G_2 = (V_2, E_2)$, we know that $V_1 \subseteq V_2$ and $E_1 \subseteq E_2$.
	In turn, every path in $G_1$ is also a path in $G_2$, so the length of the longest path in $G_2$ is at least as long as the length of the longest path in $G_1$, i.e., $\diameter(G_1) \leq \diameter(G_2)$.
\end{itemize}
Thus, we showed that $\mathcal{C}^M(G_1) \leq \mathcal{C}^M(G_2)$ for any model complexity measure $\mathcal{C}^M \in \{\size, \controlflow, \maxconn, \diameter\}$. \hfill$\square$
\end{proof}

In the directly follows graph, none of the investigated complexity measures always return the same value.
Thus, we can now analyze the relations between log and model complexity for the directly follows graph.
We start by showing the results in \cref{table:dfg-findings} and prove the relations in the table afterwards.
For quick navigation, the PDF-version of this paper enables its readers to click on the entries of the table to jump to the proof of the respective property.

\begin{table}[ht]
	\caption{The relations between the complexity scores of two directly follows graphs $G_1$ and $G_2$ for the event logs $L_1$ and $L_2$, where $L_1 \sqsubset L_2$, $|supp(L_1)| > 1$, and the complexity of $L_1$ is lower than the complexity of $L_2$.}
	\label{table:dfg-findings}
	\centering
	\scalebox{\scalefactor}{
	\begin{tabular}{|c|c|c|c|c|c|c|c|c|c|c|c|c|c|} \hline
		 & $\size$ & $\mismatch$ & $\crossconn$ & $\controlflow$ & $\separability$ & $\avgconn$ & $\maxconn$ & $\sequentiality$ & $\depth$ & $\diameter$ & $\cyclicity$ & $\netconn$ & $\density$ \\ \hline
		$\magnitude$ & \hyperref[theo:dfg-size-leq-entries]{$\mleq$} & \hyperref[theo:dfg-mismatch-entries]{$\norel$} & \hyperref[theo:dfg-crossconn-entries]{$\norel^*$} & \hyperref[theo:dfg-cfc-leq-entries]{$\mleq$} & \hyperref[theo:dfg-sep-entries]{$\norel$} & \hyperref[theo:dfg-acd-entries]{$\norel$} & \hyperref[theo:dfg-mcd-diam-entries]{$\mleq$} & \hyperref[theo:dfg-seq-entries]{$\norel$} & \hyperref[theo:dfg-depth-entries]{$\norel$} & \hyperref[theo:dfg-mcd-diam-entries]{$\mleq$} & \hyperref[theo:dfg-cyc-entries]{$\norel$} & \hyperref[theo:dfg-cnc-entries]{$\norel$} & \hyperref[theo:dfg-dens-entries]{$\norel$} \\ \hline
		
		$\variety$ & \hyperref[theo:size-less-entry]{$\mless$} & \hyperref[theo:dfg-mismatch-entries]{$\norel$} & \hyperref[theo:dfg-crossconn-entries]{$\norel^*$} & \hyperref[theo:dfg-cfc-less-entries]{$\mless$} & \hyperref[theo:dfg-sep-entries]{$\norel$} & \hyperref[theo:dfg-acd-entries]{$\norel$} & \hyperref[theo:dfg-mcd-diam-entries]{$\mleq$} & \hyperref[theo:dfg-seq-entries]{$\norel$} & \hyperref[theo:dfg-depth-entries]{$\norel$} & \hyperref[theo:dfg-mcd-diam-entries]{$\mleq$} & \hyperref[theo:dfg-cyc-entries]{$\norel$} & \hyperref[theo:dfg-cnc-entries]{$\norel$} & \hyperref[theo:dfg-dens-entries]{$\norel$} \\ \hline
		
		$\support$ & \hyperref[theo:dfg-size-leq-entries]{$\mleq$} & \hyperref[theo:dfg-mismatch-entries]{$\norel$} & \hyperref[theo:dfg-crossconn-entries]{$\norel^*$} & \hyperref[theo:dfg-cfc-leq-entries]{$\mleq$} & \hyperref[theo:dfg-sep-entries]{$\norel$} & \hyperref[theo:dfg-acd-entries]{$\norel$} & \hyperref[theo:dfg-mcd-diam-entries]{$\mleq$} & \hyperref[theo:dfg-seq-entries]{$\norel$} & \hyperref[theo:dfg-depth-entries]{$\norel$} & \hyperref[theo:dfg-mcd-diam-entries]{$\mleq$} & \hyperref[theo:dfg-cyc-entries]{$\norel$} & \hyperref[theo:dfg-cnc-entries]{$\norel$} & \hyperref[theo:dfg-dens-entries]{$\norel$} \\ \hline
		
		$\tlavg$ & \hyperref[theo:dfg-size-leq-entries]{$\mleq$} & \hyperref[theo:dfg-mismatch-entries]{$\norel$} & \hyperref[theo:dfg-crossconn-entries]{$\norel^*$} & \hyperref[theo:dfg-cfc-leq-entries]{$\mleq$} & \hyperref[theo:dfg-sep-entries]{$\norel$} & \hyperref[theo:dfg-acd-entries]{$\norel$} & \hyperref[theo:dfg-mcd-diam-entries]{$\mleq$} & \hyperref[theo:dfg-seq-entries]{$\norel$} & \hyperref[theo:dfg-depth-entries]{$\norel$} & \hyperref[theo:dfg-mcd-diam-entries]{$\mleq$} & \hyperref[theo:dfg-cyc-entries]{$\norel$} & \hyperref[theo:dfg-cnc-entries]{$\norel$} & \hyperref[theo:dfg-dens-entries]{$\norel$} \\ \hline
		
		$\tlmax$ & \hyperref[theo:dfg-size-leq-entries]{$\mleq$} & \hyperref[theo:dfg-mismatch-entries]{$\norel$} & \hyperref[theo:dfg-crossconn-entries]{$\norel^*$} & \hyperref[theo:dfg-cfc-leq-entries]{$\mleq$} & \hyperref[theo:dfg-sep-entries]{$\norel$} & \hyperref[theo:dfg-acd-entries]{$\norel$} & \hyperref[theo:dfg-mcd-diam-entries]{$\mleq$} & \hyperref[theo:dfg-seq-entries]{$\norel$} & \hyperref[theo:dfg-depth-entries]{$\norel$} & \hyperref[theo:dfg-mcd-diam-entries]{$\mleq$} & \hyperref[theo:dfg-cyc-entries]{$\norel$} & \hyperref[theo:dfg-cnc-entries]{$\norel$} & \hyperref[theo:dfg-dens-entries]{$\norel$} \\ \hline
		
		$\levelofdetail$ & \hyperref[theo:dfg-size-leq-entries]{$\mleq$} & \hyperref[theo:dfg-mismatch-entries]{$\norel$} & \hyperref[theo:dfg-crossconn-entries]{$\norel^*$} & \hyperref[theo:dfg-cfc-less-entries]{$\mless$} & \hyperref[theo:dfg-sep-entries]{$\norel$} & \hyperref[theo:dfg-acd-entries]{$\norel$} & \hyperref[theo:dfg-mcd-diam-entries]{$\mleq$} & \hyperref[theo:dfg-seq-entries]{$\norel$} & \hyperref[theo:dfg-depth-entries]{$\norel$} & \hyperref[theo:dfg-mcd-diam-entries]{$\mleq$} & \hyperref[theo:dfg-cyc-entries]{$\norel$} & \hyperref[theo:dfg-cnc-entries]{$\norel$} & \hyperref[theo:dfg-dens-entries]{$\norel$} \\ \hline
		
		$\numberofties$ & \hyperref[theo:dfg-size-leq-entries]{$\mleq$} & \hyperref[theo:dfg-mismatch-entries]{$\norel$} & \hyperref[theo:dfg-crossconn-entries]{$\norel^*$} & \hyperref[theo:dfg-cfc-less-entries]{$\mless$} & \hyperref[theo:dfg-sep-entries]{$\norel$} & \hyperref[theo:dfg-acd-entries]{$\norel$} & \hyperref[theo:dfg-mcd-diam-entries]{$\mleq$} & \hyperref[theo:dfg-seq-entries]{$\norel$} & \hyperref[theo:dfg-depth-entries]{$\norel$} & \hyperref[theo:dfg-mcd-diam-entries]{$\mleq$} & \hyperref[theo:dfg-cyc-entries]{$\norel$} & \hyperref[theo:dfg-cnc-entries]{$\norel$} & \hyperref[theo:dfg-dens-entries]{$\norel$} \\ \hline
		
		$\lempelziv$ & \hyperref[theo:dfg-size-leq-entries]{$\mleq$} & \hyperref[theo:dfg-mismatch-entries]{$\norel$} & \hyperref[theo:dfg-crossconn-entries]{$\norel^*$} & \hyperref[theo:dfg-cfc-leq-entries]{$\mleq$} & \hyperref[theo:dfg-sep-entries]{$\norel$} & \hyperref[theo:dfg-acd-entries]{$\norel$} & \hyperref[theo:dfg-mcd-diam-entries]{$\mleq$} & \hyperref[theo:dfg-seq-entries]{$\norel$} & \hyperref[theo:dfg-depth-entries]{$\norel$} & \hyperref[theo:dfg-mcd-diam-entries]{$\mleq$} & \hyperref[theo:dfg-cyc-entries]{$\norel$} & \hyperref[theo:dfg-cnc-entries]{$\norel$} & \hyperref[theo:dfg-dens-entries]{$\norel$} \\ \hline
		
		$\numberuniquetraces$ & \hyperref[theo:dfg-size-leq-entries]{$\mleq$} & \hyperref[theo:dfg-mismatch-entries]{$\norel$} & \hyperref[theo:dfg-crossconn-entries]{$\norel^*$} & \hyperref[theo:dfg-cfc-leq-entries]{$\mleq$} & \hyperref[theo:dfg-sep-entries]{$\norel$} & \hyperref[theo:dfg-acd-entries]{$\norel$} & \hyperref[theo:dfg-mcd-diam-entries]{$\mleq$} & \hyperref[theo:dfg-seq-entries]{$\norel$} & \hyperref[theo:dfg-depth-entries]{$\norel$} & \hyperref[theo:dfg-mcd-diam-entries]{$\mleq$} & \hyperref[theo:dfg-cyc-entries]{$\norel$} & \hyperref[theo:dfg-cnc-entries]{$\norel$} & \hyperref[theo:dfg-dens-entries]{$\norel$} \\ \hline
		
		$\percentageuniquetraces$ & \hyperref[theo:dfg-size-leq-entries]{$\mleq$} & \hyperref[theo:dfg-mismatch-entries]{$\norel$} & \hyperref[theo:dfg-crossconn-entries]{$\norel^*$} & \hyperref[theo:dfg-cfc-leq-entries]{$\mleq$} & \hyperref[theo:dfg-sep-entries]{$\norel$} & \hyperref[theo:dfg-acd-entries]{$\norel$} & \hyperref[theo:dfg-mcd-diam-entries]{$\mleq$} & \hyperref[theo:dfg-seq-entries]{$\norel$} & \hyperref[theo:dfg-depth-entries]{$\norel$} & \hyperref[theo:dfg-mcd-diam-entries]{$\mleq$} & \hyperref[theo:dfg-cyc-entries]{$\norel$} & \hyperref[theo:dfg-cnc-entries]{$\norel$} & \hyperref[theo:dfg-dens-entries]{$\norel$} \\ \hline
		
		$\structure$ & \hyperref[theo:dfg-size-leq-entries]{$\mleq$} & \hyperref[theo:dfg-mismatch-entries]{$\norel$} & \hyperref[theo:dfg-crossconn-entries]{$\norel^*$} & \hyperref[theo:dfg-cfc-leq-entries]{$\mleq$} & \hyperref[theo:dfg-sep-entries]{$\norel$} & \hyperref[theo:dfg-acd-entries]{$\norel$} & \hyperref[theo:dfg-mcd-diam-entries]{$\mleq$} & \hyperref[theo:dfg-seq-entries]{$\norel$} & \hyperref[theo:dfg-depth-entries]{$\norel$} & \hyperref[theo:dfg-mcd-diam-entries]{$\mleq$} & \hyperref[theo:dfg-cyc-entries]{$\norel$} & \hyperref[theo:dfg-cnc-entries]{$\norel$} & \hyperref[theo:dfg-dens-entries]{$\norel$} \\ \hline
		
		$\affinity$ & \hyperref[theo:dfg-size-leq-entries]{$\mleq$} & \hyperref[theo:dfg-mismatch-entries]{$\norel$} & \hyperref[theo:dfg-crossconn-entries]{$\norel^*$} & \hyperref[theo:dfg-cfc-leq-entries]{$\mleq$} & \hyperref[theo:dfg-sep-entries]{$\norel$} & \hyperref[theo:dfg-acd-entries]{$\norel$} & \hyperref[theo:dfg-mcd-diam-entries]{$\mleq$} & \hyperref[theo:dfg-seq-entries]{$\norel$} & \hyperref[theo:dfg-depth-entries]{$\norel$} & \hyperref[theo:dfg-mcd-diam-entries]{$\mleq$} & \hyperref[theo:dfg-cyc-entries]{$\norel$} & \hyperref[theo:dfg-cnc-entries]{$\norel$} & \hyperref[theo:dfg-dens-entries]{$\norel$} \\ \hline
		
		$\deviationfromrandom$ & \hyperref[theo:dfg-size-leq-entries]{$\mleq$} & \hyperref[theo:dfg-mismatch-entries]{$\norel$} & \hyperref[theo:dfg-crossconn-entries]{$\norel^*$} & \hyperref[theo:dfg-cfc-leq-entries]{$\mleq$} & \hyperref[theo:dfg-sep-entries]{$\norel$} & \hyperref[theo:dfg-acd-entries]{$\norel$} & \hyperref[theo:dfg-mcd-diam-entries]{$\mleq$} & \hyperref[theo:dfg-seq-entries]{$\norel$} & \hyperref[theo:dfg-depth-entries]{$\norel$} & \hyperref[theo:dfg-mcd-diam-entries]{$\mleq$} & \hyperref[theo:dfg-cyc-entries]{$\norel$} & \hyperref[theo:dfg-cnc-entries]{$\norel$} & \hyperref[theo:dfg-dens-entries]{$\norel$} \\ \hline
		
		$\avgdist$ & \hyperref[theo:dfg-size-leq-entries]{$\mleq$} & \hyperref[theo:dfg-mismatch-entries]{$\norel$} & \hyperref[theo:dfg-crossconn-entries]{$\norel^*$} & \hyperref[theo:dfg-cfc-leq-entries]{$\mleq$} & \hyperref[theo:dfg-sep-entries]{$\norel$} & \hyperref[theo:dfg-acd-entries]{$\norel$} & \hyperref[theo:dfg-mcd-diam-entries]{$\mleq$} & \hyperref[theo:dfg-seq-entries]{$\norel$} & \hyperref[theo:dfg-depth-entries]{$\norel$} & \hyperref[theo:dfg-mcd-diam-entries]{$\mleq$} & \hyperref[theo:dfg-cyc-entries]{$\norel$} & \hyperref[theo:dfg-cnc-entries]{$\norel$} & \hyperref[theo:dfg-dens-entries]{$\norel$} \\ \hline
		
		$\varentropy$ & \hyperref[theo:dfg-size-leq-entries]{$\mleq$} & \hyperref[theo:dfg-mismatch-entries]{$\norel$} & \hyperref[theo:dfg-crossconn-entries]{$\norel^*$} & \hyperref[theo:dfg-cfc-leq-entries]{$\mleq$} & \hyperref[theo:dfg-sep-entries]{$\norel$} & \hyperref[theo:dfg-acd-entries]{$\norel$} & \hyperref[theo:dfg-mcd-diam-entries]{$\mleq$} & \hyperref[theo:dfg-seq-entries]{$\norel$} & \hyperref[theo:dfg-depth-entries]{$\norel$} & \hyperref[theo:dfg-mcd-diam-entries]{$\mleq$} & \hyperref[theo:dfg-cyc-entries]{$\norel$} & \hyperref[theo:dfg-cnc-entries]{$\norel$} & \hyperref[theo:dfg-dens-entries]{$\norel$} \\ \hline
		
		$\normvarentropy$ & \hyperref[theo:dfg-size-leq-entries]{$\mleq$} & \hyperref[theo:dfg-mismatch-entries]{$\norel$} & \hyperref[theo:dfg-crossconn-entries]{$\norel^*$} & \hyperref[theo:dfg-cfc-leq-entries]{$\mleq$} & \hyperref[theo:dfg-sep-entries]{$\norel$} & \hyperref[theo:dfg-acd-entries]{$\norel$} & \hyperref[theo:dfg-mcd-diam-entries]{$\mleq$} & \hyperref[theo:dfg-seq-entries]{$\norel$} & \hyperref[theo:dfg-depth-entries]{$\norel$} & \hyperref[theo:dfg-mcd-diam-entries]{$\mleq$} & \hyperref[theo:dfg-cyc-entries]{$\norel$} & \hyperref[theo:dfg-cnc-entries]{$\norel$} & \hyperref[theo:dfg-dens-entries]{$\norel$} \\ \hline
		
		$\seqentropy$ & \hyperref[theo:dfg-size-leq-entries]{$\mleq$} & \hyperref[theo:dfg-mismatch-entries]{$\norel$} & \hyperref[theo:dfg-crossconn-entries]{$\norel^*$} & \hyperref[theo:dfg-cfc-leq-entries]{$\mleq$} & \hyperref[theo:dfg-sep-entries]{$\norel$} & \hyperref[theo:dfg-acd-entries]{$\norel$} & \hyperref[theo:dfg-mcd-diam-entries]{$\mleq$} & \hyperref[theo:dfg-seq-entries]{$\norel$} & \hyperref[theo:dfg-depth-entries]{$\norel$} & \hyperref[theo:dfg-mcd-diam-entries]{$\mleq$} & \hyperref[theo:dfg-cyc-entries]{$\norel$} & \hyperref[theo:dfg-cnc-entries]{$\norel$} & \hyperref[theo:dfg-dens-entries]{$\norel$} \\ \hline
		
		$\normseqentropy$ & \hyperref[theo:dfg-size-leq-entries]{$\mleq$} & \hyperref[theo:dfg-mismatch-entries]{$\norel$} & \hyperref[theo:dfg-crossconn-entries]{$\norel^*$} & \hyperref[theo:dfg-cfc-leq-entries]{$\mleq$} & \hyperref[theo:dfg-sep-entries]{$\norel$} & \hyperref[theo:dfg-acd-entries]{$\norel$} & \hyperref[theo:dfg-mcd-diam-entries]{$\mleq$} & \hyperref[theo:dfg-seq-entries]{$\norel$} & \hyperref[theo:dfg-depth-entries]{$\norel$} & \hyperref[theo:dfg-mcd-diam-entries]{$\mleq$} & \hyperref[theo:dfg-cyc-entries]{$\norel$} & \hyperref[theo:dfg-cnc-entries]{$\norel$} & \hyperref[theo:dfg-dens-entries]{$\norel$} \\ \hline
	\end{tabular}}
	
	{\scriptsize ${}^*$We did not find examples showing that $\mathcal{C}^L(L_1) < \mathcal{C}^L(L_2)$ and $\crossconn(M_1) = \crossconn(M_2)$ is possible.}
\end{table}

\newpage
\begin{theorem}
\label{theo:dfg-size-leq-entries}
Let $\mathcal{C}^L \in (\loc \setminus \{\variety\})$ be an event log complexity measure.
Then, $(\mathcal{C}^L, \size) \in \mleq$.
\end{theorem}
\begin{proof}
Let $L_1, L_2$ be event logs with $L_1 \sqsubset L_2$ and $|supp(L_1)| > 1$, and $G_1, G_2$ be their directly follows graphs.
By \cref{lemma:dfg-monotone-increasing}, we know that $\size(G_1) \leq \size(G_2)$.
What remains to be shown is that with the property $\mathcal{C}^L(L_1) < \mathcal{C}^L(L_2)$, both $\size(G_1) = \size(G_2)$ and $\size(G_1) < \size(G_2)$ are possible.
For the former, take the following event logs:
\begin{align*}
	L_1 &= [\langle a,b,c,d \rangle^{2}, \langle a,b,c,d,e \rangle^{2}, \langle d,e,a,b \rangle^{2}] \\
	L_2 &= L_1 + [\langle a,b,c,d,e \rangle^{2}, \langle d,e,a,b,c \rangle, \langle c,d,e,a,b \rangle, \langle e,c,d,a,b,c \rangle]
\end{align*}
These two event logs have the following log complexity scores:
\begin{center}
	\begin{tabular}{|c|c|c|c|c|c|c|c|c|c|c|c|c|} \hline
		 & $\magnitude$ & $\variety$ & $\support$ & $\tlavg$ & $\tlmax$ & $\levelofdetail$ & $\numberofties$ & $\lempelziv$ & $\numberuniquetraces$ & $\percentageuniquetraces$ \\ \hline
		$L_1$ & $\pad 26 \pad$ & $\pad 5 \pad$ & $\pad 6 \pad$ & $\pad 4.3333 \pad$ & $\pad 5 \pad$ & $\pad 6 \pad$ & $\pad 5 \pad$ & $\pad 13 \pad$ & $\pad 3 \pad$ & $\pad 0.5 \pad$ \\ \hline
		$L_2$ & $\pad 52 \pad$ & $\pad 5 \pad$ & $\pad 11 \pad$ & $\pad 4.7273 \pad$ & $\pad 6 \pad$ & $\pad 23 \pad$ & $\pad 7 \pad$ & $\pad 21 \pad$ & $\pad 6 \pad$ & $\pad 0.5455 \pad$ \\ \hline
	\end{tabular}
		
	\medskip
		
	\begin{tabular}{|c|c|c|c|c|c|c|c|c|} \hline
		 & $\structure$ & $\affinity$ & $\deviationfromrandom$ & $\avgdist$ & $\varentropy$ & $\normvarentropy$ & $\seqentropy$ & $\normseqentropy$ \\ \hline
		$L_1$ & $\pad 4.3333 \pad$ & $\pad 0.56 \pad$ & $\pad 0.5757 \pad$ & $\pad 2.6667 \pad$ & $\pad 6.1827 \pad$ & $\pad 0.3126 \pad$ & $\pad 16.0483 \pad$ & $\pad 0.1894 \pad$ \\ \hline
		$L_2$ & $\pad 4.6364 \pad$ & $\pad 0.5829 \pad$ & $\pad 0.6039 \pad$ & $\pad 2.9091 \pad$ & $\pad 29.0428 \pad$ & $\pad 0.4543 \pad$ & $\pad 60.0209 \pad$ & $\pad 0.2921 \pad$ \\ \hline
	\end{tabular}
\end{center}
Thus, $\mathcal{C}^L(L_1) < \mathcal{C}^L(L_2)$ for any $\mathcal{C}^L \in (\loc \setminus \{\variety\})$.
Ignoring the node labeled $f$ and its adjacent edges, \cref{fig:dfg-size} shows the directly follows graphs $G_1$ and $G_2$ for $L_1$ and $L_2$.
\begin{figure}[ht]
	\centering
	\begin{tikzpicture}[node distance = 1.1cm,>=stealth',bend angle=0,auto]
		\node[transition] (start) {$\triangleright$};
		\node[above of=start,yshift=-0.5cm] {$G_1$:};
		\node[transition,right of=start] (a) {$a$}
		edge [pre] (start);
		\node[transition,right of=a] (b) {$b$}
		edge [pre] (a);
		\node[transition,right of=b] (c) {$c$}
		edge [pre] (b);
		\node[transition,right of=c] (d) {$d$}
		edge [pre,bend right=40] (start)
		edge [pre] (c);
		\node[transition,right of=d] (e) {$e$}
		edge [post,bend left=30] (a)
		edge [pre] (d);
		\node[transition,right of=e] (end) {$\square$}
		edge [pre,bend right=40] (b)
		edge [pre,bend right=30] (d)
		edge [pre] (e);
	\end{tikzpicture}
	
	\medskip
	\hrule
	\medskip
	
	\begin{tikzpicture}[node distance = 1.1cm,>=stealth',bend angle=0,auto]
		\node[transition] (start) {$\triangleright$};
		\node[above of=start,yshift=-0.5cm] {$G_2$:};
		\node[transition,right of=start] (a) {$a$}
		edge [pre] (start);
		\node[transition,right of=a] (b) {$b$}
		edge [pre] (a);
		\node[transition,right of=b] (c) {$c$}
		edge [pre,bend left=30] (start)
		edge [pre] (b);
		\node[transition,right of=c] (d) {$d$}
		edge [pre,bend right=40] (start)
		edge [post,bend right=30] (a)
		edge [pre] (c);
		\node[transition,right of=d] (e) {$e$}
		edge [pre,bend left=40] (start)
		edge [post,bend left=30] (a)
		edge [post,bend right=40] (c)
		edge [pre] (d);
		\node[transition,right of=e] (end) {$\square$}
		edge [pre,bend right=40] (b)
		edge [pre,bend left=40] (c)
		edge [pre,bend right=30] (d)
		edge [pre] (e);
		\node[transition,below of=c,yshift=-0.75cm,opacity=0.5] (f) {$f$}
		edge [pre,opacity=0.5] (c)
		edge [post,bend right=30,opacity=0.5] (end);
	\end{tikzpicture}
	\caption{The directly follows graphs for the logs $L_1, L_2$ from the example in \cref{theo:dfg-size-leq-entries}. $G_1$ is the DFG for $L_1$, and $G_2$ the one for $L_2$.}
	\label{fig:dfg-size}
\end{figure}
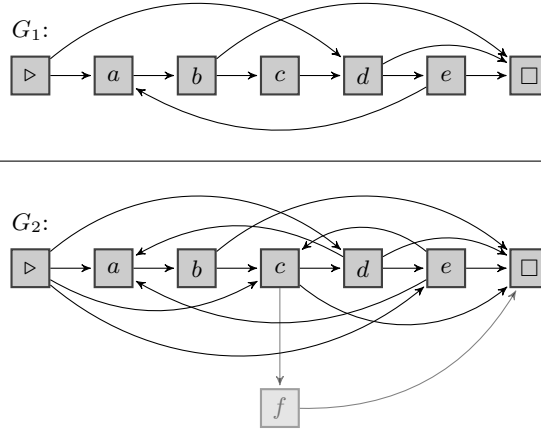
$G_1$ and $G_2$ fulfill $\size(G_1) = 7 = \size(G_2)$, so $\mathcal{C}^L(L_1) < \mathcal{C}^L(L_2)$ and $\size(G_1) = \size(G_2)$ are possible.
To see that $\mathcal{C}^L(L_1) < \mathcal{C}^L(L_2)$ and at the same time $\size(G_1) < \size(G_2)$ is also possible, consider the following logs:
\begin{align*}
	L_1 &= [\langle a,b,c,d \rangle^{2}, \langle a,b,c,d,e \rangle^{2}, \langle d,e,a,b \rangle^{2}] \\
	L_2 &= L_1 + [\langle a,b,c,d,e \rangle^{2}, \langle d,e,a,b,c \rangle, \langle c,d,e,a,b \rangle, \langle e,c,d,a,b,c,f \rangle]
\end{align*}
These two event logs have the following log complexity scores:
\begin{center}
	\begin{tabular}{|c|c|c|c|c|c|c|c|c|c|c|c|c|} \hline
		 & $\magnitude$ & $\variety$ & $\support$ & $\tlavg$ & $\tlmax$ & $\levelofdetail$ & $\numberofties$ & $\lempelziv$ & $\numberuniquetraces$ & $\percentageuniquetraces$ \\ \hline
		$L_1$ & $\pad 26 \pad$ & $\pad 5 \pad$ & $\pad 6 \pad$ & $\pad 4.3333 \pad$ & $\pad 5 \pad$ & $\pad 6 \pad$ & $\pad 5 \pad$ & $\pad 13 \pad$ & $\pad 3 \pad$ & $\pad 0.5 \pad$ \\ \hline
		$L_2$ & $\pad 53 \pad$ & $\pad 6 \pad$ & $\pad 11 \pad$ & $\pad 4.8182 \pad$ & $\pad 7 \pad$ & $\pad 30 \pad$ & $\pad 8 \pad$ & $\pad 22 \pad$ & $\pad 6 \pad$ & $\pad 0.5455 \pad$ \\ \hline
	\end{tabular}
		
	\medskip
		
	\begin{tabular}{|c|c|c|c|c|c|c|c|c|} \hline
		 & $\structure$ & $\affinity$ & $\deviationfromrandom$ & $\avgdist$ & $\varentropy$ & $\normvarentropy$ & $\seqentropy$ & $\normseqentropy$ \\ \hline
		$L_1$ & $\pad 4.3333 \pad$ & $\pad 0.56 \pad$ & $\pad 0.5757 \pad$ & $\pad 2.6667 \pad$ & $\pad 6.1827 \pad$ & $\pad 0.3126 \pad$ & $\pad 16.0483 \pad$ & $\pad 0.1894 \pad$ \\ \hline
		$L_2$ & $\pad 4.7273 \pad$ & $\pad 0.5721 \pad$ & $\pad 0.5995 \pad$ & $\pad 3.0909 \pad$ & $\pad 30.24 \pad$ & $\pad 0.4447 \pad$ & $\pad 62.1108 \pad$ & $\pad 0.2952 \pad$ \\ \hline
	\end{tabular}
\end{center}
Thus, $\mathcal{C}^L(L_1) < \mathcal{C}^L(L_2)$ for any $\mathcal{C}^L \in (\loc \setminus \{\variety\})$. 
\cref{fig:dfg-size} shows the directly follows graphs $G_1$ and $G_2$ for $L_1$ and $L_2$.
As can easily be seen, these models fulfill $\size(G_1) = 7 < 8 = \size(G_2)$, so $\mathcal{C}^L(L_1) < \mathcal{C}^L(L_2)$ and $\size(G_1) < \size(G_2)$ are also possible. \hfill$\square$
\end{proof}

\begin{theorem}
\label{theo:size-less-entry}
$(\variety, \size) \in \mless$.
\end{theorem}
\begin{proof}
Let $L$ be an event log and $G$ its directly follows graph.
For each activity name occurring in $L$, there is exactly one node in $G$.
Beside these nodes for activity names, there are only the nodes $\triangleright$ and $\square$ in the directly follows graph $G$.
Thus, $\size(G) = \variety(L) + 2$, so for two event logs $L_1, L_2$ with $L_1 \sqsubset L_2$, and their respective directly follows graphs, $G_1$ and $G_2$, we get that the property $\size(L_1) = \variety(L_1) + 2 < \variety(L_2) + 2 = \size(L_2)$ always holds. \hfill$\square$
\end{proof}

\newpage
\begin{theorem}
\label{theo:dfg-mismatch-entries}
$(\mathcal{C}^L, \mismatch) \in \norel$ for any log complexity measure $\mathcal{C}^L \in \loc$.
\end{theorem}
\begin{proof}
Consider the following event logs:
\begin{align*}
	L_1 &= [\langle a,b,d \rangle^{2}, \langle a,c,d \rangle^{2}, \langle e \rangle] \\
	L_2 &= L_1 + [\langle a,b,d,e \rangle, \langle a,c,d,e \rangle, \langle a,b,c,d \rangle, \langle a,b,c,b,d,e,f \rangle, \\
	&\phantom{= L_1 + [}\hspace*{1mm}\langle a,b,c,b,c,b,d,e,f \rangle] \\
	L_3 &= L_2 + [\langle a,c,b,d \rangle, \langle a,c,b,c,b,d,e \rangle, \langle a,b,c,b,c,b,c,d \rangle, \langle a,b,c,b,c,b,c,b,c,d \rangle, \\
	&\phantom{= L_2 + [}\hspace*{1mm}\langle a,a,b,b,c,c,d,d,e,e,f,f,g \rangle]
\end{align*}
\cref{fig:dfg-mismatch} shows the directly follows graphs $G_1, G_2, G_3$ for the event logs $L_1, L_2, L_3$.
\begin{figure}[ht]
	\centering
	\begin{minipage}{0.48\textwidth}
	\centering
	\begin{tikzpicture}[node distance = 1.1cm,>=stealth',bend angle=0,auto]
		\node[transition] (start) {$\triangleright$};
		\node[above of=start] {$G_1$:};
		\node[transition,right of=start] (a) {$a$}
		edge [pre] (start);
		\node[transition,right of=a] (b) {$b$}
		edge [pre] (a);
		\node[transition,below of=b] (c) {$c$}
		edge [pre] (a);
		\node[transition,above of=b] (e) {$e$}
		edge [pre] (start);
		\node[transition,right of=b] (d) {$d$}
		edge [pre] (b)
		edge [pre] (c);
		\node[transition,right of=d] (end) {$\square$}
		edge [pre] (d)
		edge [pre] (e);
	\end{tikzpicture}
	\end{minipage}
	\begin{minipage}{0.48\textwidth}
	\centering
	\begin{tikzpicture}[node distance = 1.1cm,>=stealth',bend angle=0,auto]
		\node[transition] (start) {$\triangleright$};
		\node[above of=start] {$G_2$:};
		\node[transition,right of=start] (a) {$a$}
		edge [pre] (start);
		\node[transition,right of=a] (b) {$b$}
		edge [pre] (a);
		\node[transition,below of=b] (c) {$c$}
		edge [pre] (a)
		edge [pre,bend right=15] (b)
		edge [post,bend left=15] (b);
		\node[transition,above of=b] (e) {$e$}
		edge [pre] (start);
		\node[transition,right of=b] (d) {$d$}
		edge [pre] (b)
		edge [pre] (c)
		edge [post] (e);
		\node[transition,right of=d] (end) {$\square$}
		edge [pre] (d)
		edge [pre] (e);
		\node[transition,right of=e] (f) {$f$}
		edge [pre] (e)
		edge [post] (end);
	\end{tikzpicture}
	\end{minipage}
	
	\medskip
	\hrule
	\medskip
	
	\begin{tikzpicture}[node distance = 1.1cm,>=stealth',bend angle=0,auto]
		\node[transition] (start) {$\triangleright$};
		\node[above of=start] {$G_3$:};
		\node[transition,right of=start] (a) {$a$}
		edge [pre] (start)
		edge [pre,loop,out=300,in=240,looseness=6] (a);
		\node[transition,right of=a] (b) {$b$}
		edge [pre] (a)
		edge [pre,loop,out=60,in=120,looseness=6] (b);
		\node[transition,below of=b] (c) {$c$}
		edge [pre] (a)
		edge [pre,bend right=15] (b)
		edge [post,bend left=15] (b)
		edge [pre,loop,out=300,in=240,looseness=6] (c);
		\node[transition,above of=b] (e) {$e$}
		edge [pre] (start)
		edge [pre,loop,out=60,in=120,looseness=6] (e);
		\node[transition,right of=b] (d) {$d$}
		edge [pre] (b)
		edge [pre] (c)
		edge [pre,loop,out=300,in=240,looseness=6] (d)
		edge [post] (e);
		\node[transition,right of=d] (end) {$\square$}
		edge [pre] (d)
		edge [pre] (e);
		\node[transition,right of=e] (f) {$f$}
		edge [pre] (e)
		edge [pre,loop,out=60,in=120,looseness=6] (f)
		edge [post] (end);
		\node[transition,right of=f] (g) {$g$}
		edge [pre] (f)
		edge [post] (end);
	\end{tikzpicture}
	\caption{The directly follows graphs for the logs $L_1, L_2, L_3$ from the example in \cref{theo:dfg-mismatch-entries}. $G_1$ is the DFG for $L_1$, $G_2$ the one for $L_2$ and $G_3$ the one for $L_3$.}
	\label{fig:dfg-mismatch}
\end{figure}
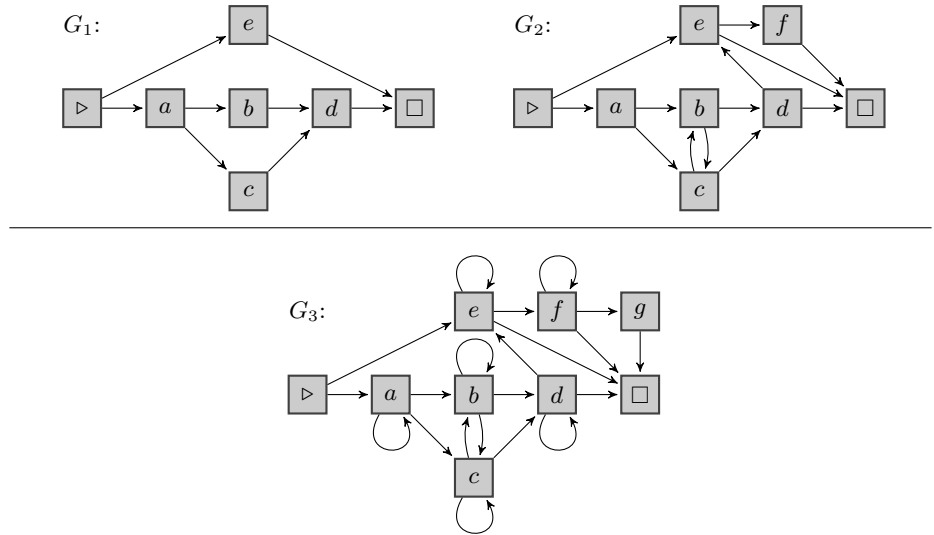
These graphs have the following complexity scores:
\begin{itemize}
	\item[•] $\mismatch(G_1) = 0$,
	\item[•] $\mismatch(G_2) = 1$,
	\item[•] $\mismatch(G_3) = 0$,
\end{itemize}
so $\mismatch(G_1) < \mismatch(G_2)$, $\mismatch(G_2) > \mismatch(G_3)$, and $\mismatch(G_1) = \mismatch(G_3)$.
But the event logs $L_1, L_2, L_3$ have the following log complexity scores:
\begin{center}
	\def\pad{\hspace*{1.5mm}}
	\begin{tabular}{|c|c|c|c|c|c|c|c|c|c|c|}\hline
		 & $\magnitude$ & $\variety$ & $\support$ & $\tlavg$ & $\tlmax$ & $\levelofdetail$ & $\numberofties$ & $\lempelziv$ & $\numberuniquetraces$ & $\percentageuniquetraces$ \\ \hline
		$L_1$ & $\pad 13 \pad$ & $\pad 5 \pad$ & $\pad 5 \pad$ & $\pad 2.6 \pad$ & $\pad 3 \pad$ & $\pad 3 \pad$ & $\pad 4 \pad$ & $\pad 8 \pad$ & $\pad 3 \pad$ & $\pad 0.6 \pad$ \\ \hline
		$L_2$ & $\pad 41 \pad$ & $\pad 6 \pad$ & $\pad 10 \pad$ & $\pad 4.1 \pad$ & $\pad 9 \pad$ & $\pad 14 \pad$ & $\pad 6 \pad$ & $\pad 18 \pad$ & $\pad 8 \pad$ & $\pad 0.8 \pad$ \\ \hline
		$L_3$ & $\pad 83 \pad$ & $\pad 7 \pad$ & $\pad 15 \pad$ & $\pad 5.5333 \pad$ & $\pad 13 \pad$ & $\pad 19 \pad$ & $\pad 7 \pad$ & $\pad 34 \pad$ & $\pad 13 \pad$ & $\pad 0.8667 \pad$ \\ \hline
	\end{tabular}
	
	\medskip
	
	\begin{tabular}{|c|c|c|c|c|c|c|c|c|} \hline
		 & $\structure$ & $\affinity$ & $\deviationfromrandom$ & $\avgdist$ & $\varentropy$ & $\normvarentropy$ & $\seqentropy$ & $\normseqentropy$ \\ \hline
		$L_1$ & $\pad 2.6 \pad$ & $\pad 0.2 \pad$ & $\pad 0.5417 \pad$ & $\pad 2.4 \pad$ & $\pad 6.0684 \pad$ & $\pad 0.5645 \pad$ & $\pad 11.1636 \pad$ & $\pad 0.3348 \pad$ \\ \hline
		$L_2$ & $\pad 3.7 \pad$ & $\pad 0.2316 \pad$ & $\pad 0.6705 \pad$ & $\pad 3.1333 \pad$ & $\pad 32.1247 \pad$ & $\pad 0.5742 \pad$ & $\pad 61.0512 \pad$ & $\pad 0.401 \pad$ \\ \hline
		$L_3$ & $\pad 4.0667 \pad$ & $\pad 0.2384 \pad$ & $\pad 0.6875 \pad$ & $\pad 4.6095 \pad$ & $\pad 91.73 \pad$ & $\pad 0.5843 \pad$ & $\pad 172.88 \pad$ & $\pad 0.4714 \pad$ \\ \hline
	\end{tabular}
\end{center}
Since $\mathcal{C}^L(L_1) < \mathcal{C}^L(L_2) < \mathcal{C}^L(L_3)$ for any log complexity measure $\mathcal{C}^L \in \loc$, we have thus shown that $(\mathcal{C}^L, \mismatch) \in \norel$. \hfill$\square$
\end{proof}

\begin{theorem}
\label{theo:dfg-crossconn-entries}
$(\mathcal{C}^L, \crossconn) \in \norel$ for any log complexity measure $\mathcal{C}^L \in \loc$.
\end{theorem}
\begin{proof}
Consider the following event logs:
\begin{align*}
	L_1 &= [\langle a,b \rangle^{5}, \langle c,d \rangle, \langle e,f \rangle, \langle g \rangle] \\
	L_2 &= L_1 + [\langle a,b,c,d \rangle, \langle s,t,u,v,w,x,y,z \rangle] \\
	L_3 &= L_2 + [\langle h,i,j,k,l,m,n,o,p \rangle]
\end{align*}
\cref{fig:dfg-cross-connectivity} shows the directly follows graphs $G_1, G_2, G_3$ for the event logs $L_1, L_2, L_3$.
\begin{figure}[ht]
	\centering
	\begin{minipage}{0.4\textwidth}
	\centering
	\begin{tikzpicture}[node distance = 1.1cm,>=stealth',bend angle=0,auto]
		\node[transition] (start) {$\triangleright$};
		\node[above of=start] {$G_1$:};
		\node[transition,right of=start,yshift=1.5cm] (a) {$a$}
		edge [pre] (start);
		\node[transition,right of=start,yshift=0.5cm] (c) {$c$}
		edge [pre] (start);
		\node[transition,right of=start,yshift=-0.5cm] (e) {$e$}
		edge [pre] (start);
		\node[transition,right of=a] (b) {$b$}
		edge [pre] (a);
		\node[transition,right of=c] (d) {$d$}
		edge [pre] (c);
		\node[transition,right of=e] (f) {$f$}
		edge [pre] (e);
		\node[transition,yshift=-1cm] at ($0.5*(e) + 0.5*(f)$) (g) {$g$}
		edge [pre] (start);
		\node[transition,right of=d,yshift=-0.5cm] (end) {$\square$}
		edge [pre] (b)
		edge [pre] (d)
		edge [pre] (f)
		edge [pre] (g);
	\end{tikzpicture}
	\end{minipage}
	\begin{minipage}{0.55\textwidth}
	\centering
	\begin{tikzpicture}[node distance = 1.1cm,>=stealth',bend angle=0,auto]
		\node[transition] (start) {$\triangleright$};
		\node[above of=start] {$G_2$:};
		\node[transition,right of=start,yshift=1.5cm] (a) {$a$}
		edge [pre] (start);
		\node[transition,right of=start,yshift=0.5cm] (c) {$c$}
		edge [pre] (start);
		\node[transition,right of=start,yshift=-0.5cm] (e) {$e$}
		edge [pre] (start);
		\node[transition,right of=a] (b) {$b$}
		edge [pre] (a);
		\node[transition,right of=c] (d) {$d$}
		edge [pre] (c);
		\node[transition,right of=e] (f) {$f$}
		edge [pre] (e);
		\node[transition,yshift=-1cm] at ($0.5*(e) + 0.5*(f)$) (g) {$g$}
		edge [pre] (start);
		\node[transition,right of=d,yshift=-0.5cm] (end) {$\square$}
		edge [pre] (b)
		edge [pre] (d)
		edge [pre] (f)
		edge [pre] (g);
		\node[transition,below of=start,xshift=-1.1cm] (s) {$s$}
		edge [pre] (start);
		\node[transition,below of=s] (t) {$t$}
		edge [pre] (s);
		\node[transition,right of=t] (u) {$u$}
		edge [pre] (t);
		\node[transition,right of=u] (v) {$v$}
		edge [pre] (u);
		\node[transition,right of=v] (w) {$w$}
		edge [pre] (v);
		\node[transition,right of=w] (x) {$x$}
		edge [pre] (w);
		\node[transition,right of=x] (y) {$y$}
		edge [pre] (x);
		\node[transition,above of=y] (z) {$z$}
		edge [pre] (y)
		edge [post] (end);
	\end{tikzpicture}
	\end{minipage}
	
	\medskip
	\hrule
	\medskip
	
	\begin{tikzpicture}[node distance = 1.1cm,>=stealth',bend angle=0,auto]
		\node[transition] (start) {$\triangleright$};
		\node[above of=start] {$G_3$:};
		\node[transition,right of=start,yshift=1.5cm] (a) {$a$}
		edge [pre] (start);
		\node[transition,right of=start,yshift=0.5cm] (c) {$c$}
		edge [pre] (start);
		\node[transition,right of=start,yshift=-0.5cm] (e) {$e$}
		edge [pre] (start);
		\node[transition,right of=a] (b) {$b$}
		edge [pre] (a);
		\node[transition,right of=c] (d) {$d$}
		edge [pre] (c);
		\node[transition,right of=e] (f) {$f$}
		edge [pre] (e);
		\node[transition,yshift=-1cm] at ($0.5*(e) + 0.5*(f)$) (g) {$g$}
		edge [pre] (start);
		\node[transition,right of=d,yshift=-0.5cm] (end) {$\square$}
		edge [pre] (b)
		edge [pre] (d)
		edge [pre] (f)
		edge [pre] (g);
		\node[transition,below of=start,xshift=-1.1cm] (s) {$s$}
		edge [pre] (start);
		\node[transition,below of=s] (t) {$t$}
		edge [pre] (s);
		\node[transition,right of=t] (u) {$u$}
		edge [pre] (t);
		\node[transition,right of=u] (v) {$v$}
		edge [pre] (u);
		\node[transition,right of=v] (w) {$w$}
		edge [pre] (v);
		\node[transition,right of=w] (x) {$x$}
		edge [pre] (w);
		\node[transition,right of=x] (y) {$y$}
		edge [pre] (x);
		\node[transition,above of=y] (z) {$z$}
		edge [pre] (y)
		edge [post] (end);
		\node[transition,below of=t,xshift=-1.65cm] (h) {$h$}
		edge [pre,bend left=30] (start);
		\node[transition,right of=h] (i) {$i$}
		edge [pre] (h);
		\node[transition,right of=i] (j) {$j$}
		edge [pre] (i);
		\node[transition,right of=j] (k) {$k$}
		edge [pre] (j);
		\node[transition,right of=k] (l) {$l$}
		edge [pre] (k);
		\node[transition,right of=l] (m) {$m$}
		edge [pre] (l);
		\node[transition,right of=m] (n) {$n$}
		edge [pre] (m);
		\node[transition,right of=n] (o) {$o$}
		edge [pre] (n);
		\node[transition,right of=o] (p) {$p$}
		edge [pre] (o)
		edge [post,bend right=30] (end);
	\end{tikzpicture}
	\caption{The directly follows graphs for the logs $L_1, L_2, L_3$ from the example in \cref{theo:dfg-crossconn-entries}. $G_1$ is the DFG for $L_1$, $G_2$ the one for $L_2$ and $G_3$ the one for $L_3$.}
	\label{fig:dfg-cross-connectivity}
\end{figure}
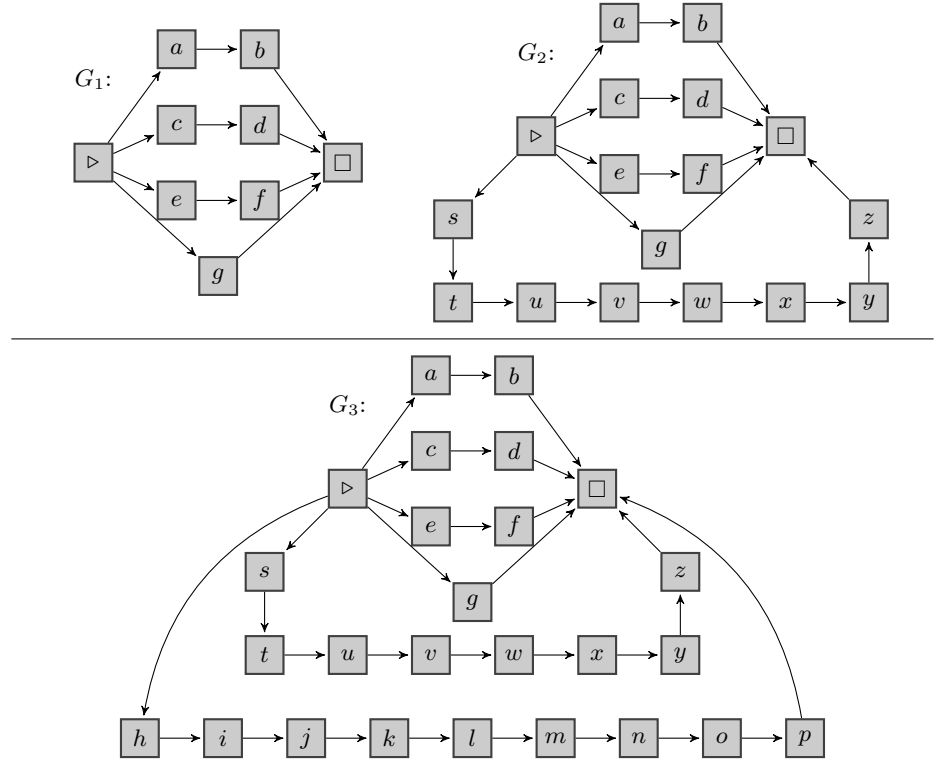

\noindent
These graphs have the following complexity scores:
\begin{itemize}
	\item[•] $\crossconn(M_1) \approx 0.8333$,
	\item[•] $\crossconn(M_2) \approx 0.8245$,
	\item[•] $\crossconn(M_3) \approx 0.8558$,
\end{itemize}
so $\crossconn(M_1) > \crossconn(M_2)$, and $\crossconn(M_2) < \crossconn(M_3)$.
But the logs $L_1, L_2, L_3$ have the following log complexity scores:
\begin{center}
	\def\pad{\hspace*{1.5mm}}
	\begin{tabular}{|c|c|c|c|c|c|c|c|c|c|c|}\hline
		 & $\magnitude$ & $\variety$ & $\support$ & $\tlavg$ & $\tlmax$ & $\levelofdetail$ & $\numberofties$ & $\lempelziv$ & $\numberuniquetraces$ & $\percentageuniquetraces$ \\ \hline
		$L_1$ & $\pad 15 \pad$ & $\pad 7 \pad$ & $\pad 8 \pad$ & $\pad 1.875 \pad$ & $\pad 2 \pad$ & $\pad 4 \pad$ & $\pad 3 \pad$ & $\pad 10 \pad$ & $\pad 4 \pad$ & $\pad 0.5 \pad$ \\ \hline
		$L_2$ & $\pad 27 \pad$ & $\pad 15 \pad$ & $\pad 10 \pad$ & $\pad 2.7 \pad$ & $\pad 8 \pad$ & $\pad 6 \pad$ & $\pad 11 \pad$ & $\pad 19 \pad$ & $\pad 6 \pad$ & $\pad 0.6 \pad$ \\ \hline
		$L_3$ & $\pad 36 \pad$ & $\pad 24 \pad$ & $\pad 11 \pad$ & $\pad 3.2727 \pad$ & $\pad 9 \pad$ & $\pad 7 \pad$ & $\pad 19 \pad$ & $\pad 28 \pad$ & $\pad 7 \pad$ & $\pad 0.6364 \pad$ \\ \hline
	\end{tabular}
	
	\medskip
	
	\begin{tabular}{|c|c|c|c|c|c|c|c|c|} \hline
		 & $\structure$ & $\affinity$ & $\deviationfromrandom$ & $\avgdist$ & $\varentropy$ & $\normvarentropy$ & $\seqentropy$ & $\normseqentropy$ \\ \hline
		$L_1$ & $\pad 1.875 \pad$ & $\pad 0.3571 \pad$ & $\pad 0.2716 \pad$ & $\pad 2.3214 \pad$ & $\pad 9.4625 \pad$ & $\pad 0.6947 \pad$ & $\pad 14.8223 \pad$ & $\pad 0.3649 \pad$ \\ \hline
		$L_2$ & $\pad 2.7 \pad$ & $\pad 0.2667 \pad$ & $\pad 0.5937 \pad$ & $\pad 3.9778 \pad$ & $\pad 23.2113 \pad$ & $\pad 0.4819 \pad$ & $\pad 32.6327 \pad$ & $\pad 0.3667 \pad$ \\ \hline
		$L_3$ & $\pad 3.2727 \pad$ & $\pad 0.2182 \pad$ & $\pad 0.7009 \pad$ & $\pad 5.3818 \pad$ & $\pad 39.9822 \pad$ & $\pad 0.472 \pad$ & $\pad 52.8767 \pad$ & $\pad 0.4099 \pad$ \\ \hline
	\end{tabular}
\end{center}
Therefore, $\mathcal{C}^L(L_1) < \mathcal{C}^L(L_2) < \mathcal{C}^L(L_3)$ for any event log complexity measure $\mathcal{C}^L \in (\loc \setminus \{\affinity, \normvarentropy\}$.
For $\affinity$ and $\normvarentropy$, consider the following event logs:
\begin{align*}
	L_1 &= [\langle a,b,c,d \rangle, \langle c,d,e,f \rangle, \langle e,f,g \rangle, \langle a,b \rangle, \langle c,d \rangle, \langle e,f \rangle, \langle g \rangle] \\
	L_2 &= L_1 + [\langle a,b,c,d \rangle^{2}, \langle q,r,s,t \rangle, \langle u,v,w,x,y,z \rangle] \\
	L_3 &= L_2 + [\langle a,b,c,d \rangle^{3}, \langle h \rangle, \langle i \rangle, \langle j \rangle]
\end{align*}
\cref{fig:dfg-cross-connectivity-affinity-nvare} shows the directly follows graphs $G_1, G_2, G_3$ for the event logs $L_1, L_2, L_3$.
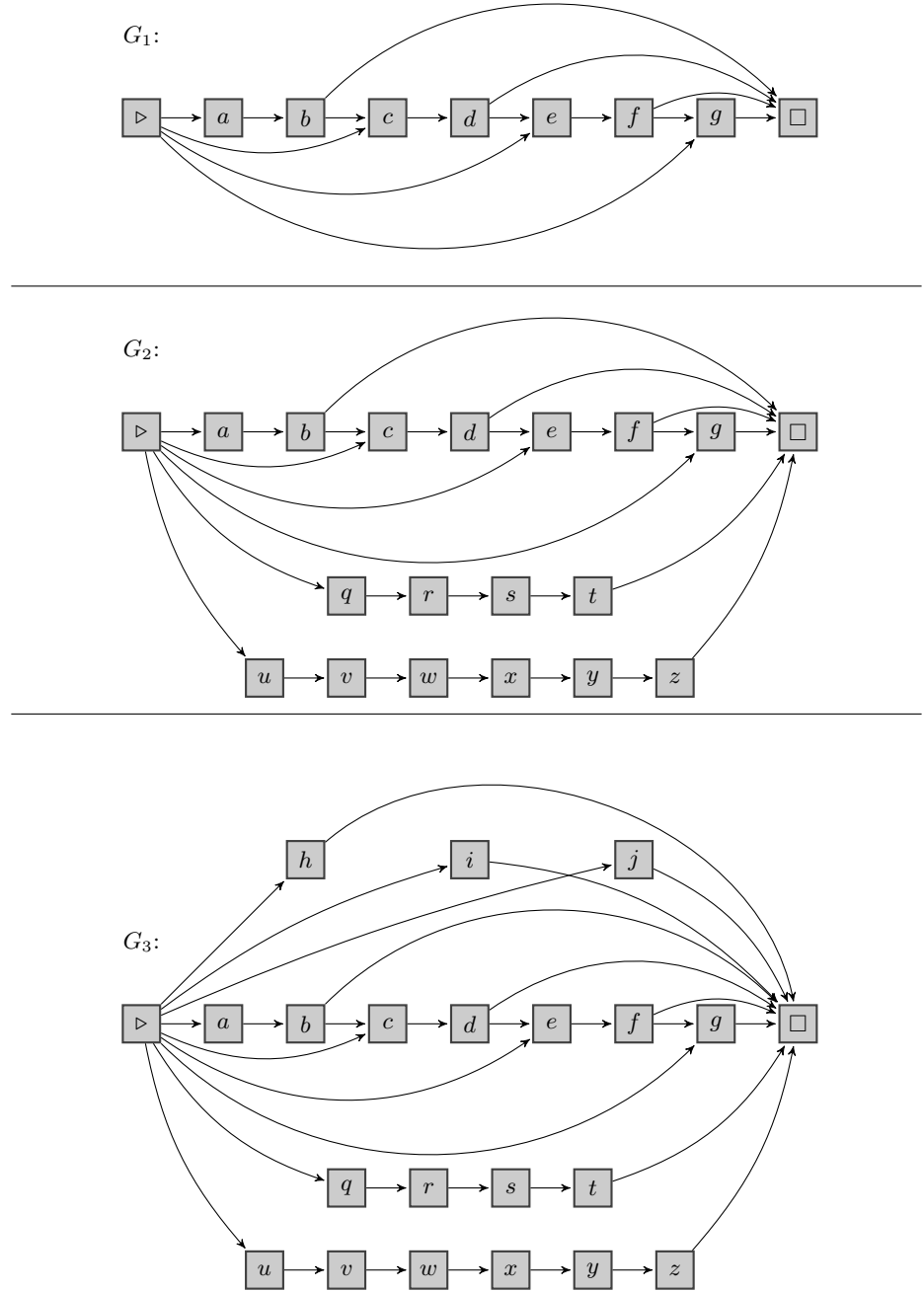
\begin{figure}[htp]
	\centering
	\begin{tikzpicture}[node distance = 1.1cm,>=stealth',bend angle=0,auto]
		\node[transition] (start) {$\triangleright$};
		\node[above of=start] {$G_1$:};
		\node[transition,right of=start] (a) {$a$}
		edge [pre] (start);
		\node[transition,right of=a] (b) {$b$}
		edge [pre] (a);
		\node[transition,right of=b] (c) {$c$}
		edge [pre] (b)
		edge [pre,bend left=25] (start);
		\node[transition,right of=c] (d) {$d$}
		edge [pre] (c);
		\node[transition,right of=d] (e) {$e$}
		edge [pre] (d)
		edge [pre,bend left=35] (start);
		\node[transition,right of=e] (f) {$f$}
		edge [pre] (e);
		\node[transition,right of=f] (g) {$g$}
		edge [pre] (f)
		edge [pre,bend left=45] (start);
		\node[transition,right of=g] (end) {$\square$}
		edge [pre] (g)
		edge [pre,bend right=25] (f)
		edge [pre,bend right=35] (d)
		edge [pre,bend right=45] (b);
	\end{tikzpicture}
	
	\hrule
	
	\begin{tikzpicture}[node distance = 1.1cm,>=stealth',bend angle=0,auto]
		\node[transition] (start) {$\triangleright$};
		\node[above of=start] {$G_2$:};
		\node[transition,right of=start] (a) {$a$}
		edge [pre] (start);
		\node[transition,right of=a] (b) {$b$}
		edge [pre] (a);
		\node[transition,right of=b] (c) {$c$}
		edge [pre] (b)
		edge [pre,bend left=25] (start);
		\node[transition,right of=c] (d) {$d$}
		edge [pre] (c);
		\node[transition,right of=d] (e) {$e$}
		edge [pre] (d)
		edge [pre,bend left=35] (start);
		\node[transition,right of=e] (f) {$f$}
		edge [pre] (e);
		\node[transition,right of=f] (g) {$g$}
		edge [pre] (f)
		edge [pre,bend left=45] (start);
		\node[transition,right of=g] (end) {$\square$}
		edge [pre] (g)
		edge [pre,bend right=25] (f)
		edge [pre,bend right=35] (d)
		edge [pre,bend right=45] (b);
		\node[transition,yshift=-2.2cm] (q) at ($0.5*(b) + 0.5*(c)$) {$q$}
		edge [pre,bend left=20] (start);
		\node[transition,right of=q] (r) {$r$}
		edge [pre] (q);
		\node[transition,right of=r] (s) {$s$}
		edge [pre] (r);
		\node[transition,right of=s] (t) {$t$}
		edge [pre] (s)
		edge [post,bend right=20] (end);
		\node[transition,below of=q] (v) {$v$};
		\node[transition,left of=v] (u) {$u$}
		edge [pre,bend left=15] (start)
		edge [post] (v);
		\node[transition,right of=v] (w) {$w$}
		edge [pre] (v);
		\node[transition,right of=w] (x) {$x$}
		edge [pre] (w);
		\node[transition,right of=x] (y) {$y$}
		edge [pre] (x);
		\node[transition,right of=y] (z) {$z$}
		edge [pre] (y)
		edge [post,bend right=15] (end);
	\end{tikzpicture}
	
	\medskip
	\hrule
	
	\begin{tikzpicture}[node distance = 1.1cm,>=stealth',bend angle=0,auto]
		\node[transition] (start) {$\triangleright$};
		\node[above of=start] {$G_3$:};
		\node[transition,right of=start] (a) {$a$}
		edge [pre] (start);
		\node[transition,right of=a] (b) {$b$}
		edge [pre] (a);
		\node[transition,right of=b] (c) {$c$}
		edge [pre] (b)
		edge [pre,bend left=25] (start);
		\node[transition,right of=c] (d) {$d$}
		edge [pre] (c);
		\node[transition,right of=d] (e) {$e$}
		edge [pre] (d)
		edge [pre,bend left=35] (start);
		\node[transition,right of=e] (f) {$f$}
		edge [pre] (e);
		\node[transition,right of=f] (g) {$g$}
		edge [pre] (f)
		edge [pre,bend left=45] (start);
		\node[transition,right of=g] (end) {$\square$}
		edge [pre] (g)
		edge [pre,bend right=25] (f)
		edge [pre,bend right=35] (d)
		edge [pre,bend right=45] (b);
		\node[transition,yshift=-2.2cm] (q) at ($0.5*(b) + 0.5*(c)$) {$q$}
		edge [pre,bend left=20] (start);
		\node[transition,right of=q] (r) {$r$}
		edge [pre] (q);
		\node[transition,right of=r] (s) {$s$}
		edge [pre] (r);
		\node[transition,right of=s] (t) {$t$}
		edge [pre] (s)
		edge [post,bend right=20] (end);
		\node[transition,below of=q] (v) {$v$};
		\node[transition,left of=v] (u) {$u$}
		edge [pre,bend left=15] (start)
		edge [post] (v);
		\node[transition,right of=v] (w) {$w$}
		edge [pre] (v);
		\node[transition,right of=w] (x) {$x$}
		edge [pre] (w);
		\node[transition,right of=x] (y) {$y$}
		edge [pre] (x);
		\node[transition,right of=y] (z) {$z$}
		edge [pre] (y)
		edge [post,bend right=15] (end);
		\node[transition,yshift=2.2cm] (h) at (b) {$h$}
		edge [pre] (start)
		edge [post,bend left=60] (end);
		\node[transition,yshift=2.2cm] (i) at (d) {$i$}
		edge [pre,bend right=10] (start)
		edge [post,bend left=20] (end);
		\node[transition,yshift=2.2cm] (j) at (f) {$j$}
		edge [pre,bend right=5] (start)
		edge [post,bend left=20] (end);
	\end{tikzpicture}
	\caption{The directly follows graphs for the logs $L_1, L_2, L_3$ from the example in \cref{theo:dfg-crossconn-entries}. $G_1$ is the DFG for $L_1$, $G_2$ the one for $L_2$ and $G_3$ the one for $L_3$.}
	\label{fig:dfg-cross-connectivity-affinity-nvare}
\end{figure}
These graphs have the following complexity scores:
\begin{itemize}
	\item[•] $\crossconn(G_1) \approx 0.9086$,
	\item[•] $\crossconn(G_2) \approx 0.8867$,
	\item[•] $\crossconn(G_3) \approx 0.9108$,
\end{itemize}
so $\crossconn(G_1) > \crossconn(G_2)$, and $\crossconn(G_2) < \crossconn(G_3)$.
But the event logs $L_1, L_2, L_3$ have the following complexity scores:
\begin{center}
	\def\pad{\hspace*{1.5mm}}
	\begin{tabular}{|c|c|c|}\hline
		 & $\affinity$ & $\normvarentropy$ \\ \hline
		$L_1$ & $\pad 0.1087 \pad$ & $\pad 0.5175 \pad$ \\ \hline
		$L_2$ & $\pad 0.1276 \pad$ & $\pad 0.5488 \pad$ \\ \hline
		$L_3$ & $\pad 0.1589 \pad$ & $\pad 0.6187 \pad$ \\ \hline
	\end{tabular}
\end{center}
Thus, in total, we were able to show $(\mathcal{C}^L, \crossconn) \in \norel$ for all $\mathcal{C}^L \in \loc$. \hfill$\square$
\end{proof}

\begin{theorem}
\label{theo:dfg-cfc-leq-entries}
Let $\mathcal{C}^L \in (\loc \setminus \{\variety, \levelofdetail, \numberofties\})$ be an event log complexity measure.
Then, $(\mathcal{C}^L, \controlflow) \in \mleq$.
\end{theorem}
\begin{proof}
By \cref{lemma:dfg-unchanging}, it is possible to increase the log complexity score for $\mathcal{C}^L$ without changing the directly follows graph. 
Thus, we know that there are event logs $L_1, L_2$ with $\mathcal{C}^L(L_1) < \mathcal{C}^L(L_2)$, such that their directly follows graphs $G_1, G_2$ fulfill $\controlflow(G_1) = \controlflow(G_2)$.
To see that $\mathcal{C}^L(L_1) < \mathcal{C}^L(L_2)$ and, at the same time, $\controlflow(G_1) < \controlflow(G_2)$ is also possible, consider the following event logs:
\begin{align*}
	L_1 &= [\langle a,b,c,d \rangle^{2}, \langle a,b,c,d,e \rangle^{2}, \langle d,e,a,b \rangle^{2}] \\
	L_2 &= L_1 + [\langle a,b,c,d,e \rangle^{2}, \langle d,e,a,b,c \rangle, \langle c,d,e,a,b \rangle, \langle e,c,d,a,b,c,f \rangle]
\end{align*}
These two event logs have the following log complexity scores:
\begin{center}
	\begin{tabular}{|c|c|c|c|c|c|c|c|c|c|c|c|c|} \hline
		 & $\magnitude$ & $\variety$ & $\support$ & $\tlavg$ & $\tlmax$ & $\levelofdetail$ & $\numberofties$ & $\lempelziv$ & $\numberuniquetraces$ & $\percentageuniquetraces$ \\ \hline
		$L_1$ & $\pad 26 \pad$ & $\pad 5 \pad$ & $\pad 6 \pad$ & $\pad 4.3333 \pad$ & $\pad 5 \pad$ & $\pad 6 \pad$ & $\pad 5 \pad$ & $\pad 13 \pad$ & $\pad 3 \pad$ & $\pad 0.5 \pad$ \\ \hline
		$L_2$ & $\pad 53 \pad$ & $\pad 6 \pad$ & $\pad 11 \pad$ & $\pad 4.8182 \pad$ & $\pad 7 \pad$ & $\pad 30 \pad$ & $\pad 8 \pad$ & $\pad 22 \pad$ & $\pad 6 \pad$ & $\pad 0.5455 \pad$ \\ \hline
	\end{tabular}
		
	\medskip
		
	\begin{tabular}{|c|c|c|c|c|c|c|c|c|} \hline
		 & $\structure$ & $\affinity$ & $\deviationfromrandom$ & $\avgdist$ & $\varentropy$ & $\normvarentropy$ & $\seqentropy$ & $\normseqentropy$ \\ \hline
		$L_1$ & $\pad 4.3333 \pad$ & $\pad 0.56 \pad$ & $\pad 0.5757 \pad$ & $\pad 2.6667 \pad$ & $\pad 6.1827 \pad$ & $\pad 0.3126 \pad$ & $\pad 16.0483 \pad$ & $\pad 0.1894 \pad$ \\ \hline
		$L_2$ & $\pad 4.7273 \pad$ & $\pad 0.5721 \pad$ & $\pad 0.5995 \pad$ & $\pad 3.0909 \pad$ & $\pad 30.24 \pad$ & $\pad 0.4447 \pad$ & $\pad 62.1108 \pad$ & $\pad 0.2952 \pad$ \\ \hline
	\end{tabular}
\end{center}
Thus, $\mathcal{C}^L(L_1) < \mathcal{C}^L(L_2)$ for any $\mathcal{C}^L \in (\loc \setminus \{\variety, \levelofdetail, \numberofties\})$. 
\cref{fig:dfg-size} shows the directly follows graphs $G_1$ and $G_2$ for $L_1$ and $L_2$.
These models fulfill $\controlflow(G_1) = 8 < 15 = \controlflow(G_2)$, so $\mathcal{C}^L(L_1) < \mathcal{C}^L(L_2)$ and, at the same time, $\controlflow(G_1) < \controlflow(G_2)$ is also possible. \hfill$\square$
\end{proof}

\begin{theorem}
\label{theo:dfg-cfc-less-entries}
Let $\mathcal{C}^L \in \{\variety, \levelofdetail, \numberofties\}$ be an event log complexity measure.
Then, $(\mathcal{C}^L, \controlflow) \in \mless$.
\end{theorem}
\begin{proof}
The control flow complexity $\controlflow$ is the number of arcs that leave split nodes in the directly follows graphs.
We will now show that this amount increases when $\variety$, $\levelofdetail$, or $\numberofties$ increase for the underlying event log.
Let $L_1, L_2$ be event logs with $L_1 \sqsubset L_2$, and $G_1, G_2$ the directly follows graphs for $L_1$ and $L_2$.
\begin{itemize}
	\item \textbf{Variety $\variety$:}
	Suppose $\variety(L_1) < \variety(L_2)$.
	Then, there is a fresh trace $\sigma \in supp(L_2) \setminus supp(L_1)$, containing an activity $a$ that does not occur in $L_1$.
	By construction, all nodes in the directly follows graph lie on a path from $\triangleright$ to $\square$, so there is a path $\triangleright, v_1, \dots, v_k, a$ for some nodes $v_1, \dots, v_k$ in $G_2$ that does not exist in $G_1$.
	But then, either $\triangleright$ or a $v_i$ for some $i \in \{1, \dots, k\}$ must have a new outgoing edge in $G_2$ that does not exist in $G_1$.
	In turn, this node is a split node in $G_2$ and has one more outgoing edge than in $G_1$.
	Since all edges of $G_1$ are also part of $G_2$, this implies $\controlflow(G_1) < \controlflow(G_2)$.
	
	\item \textbf{Level of Detail $\levelofdetail$:}
	Suppose $\levelofdetail(L_1) < \levelofdetail(G_2)$.
	Then, there is a new path $\triangleright, v_1, \dots, v_k, \square$ in $G_2$ that does not exist in $G_1$.
	In turn, there must be an edge $(a,b)$ in $G_2$ that does not exist in $G_1$.
	Because $a$ lies on a path from $\triangleright$ to $\square$ in $G_1$, and all edges of $G_1$ are also edges in $G_2$, this means $\text{outdeg}(a) > 1$.
	Thus, $a$ is a split node in $G_2$ with more than one outgoing edge than in $G_1$.
	Since all edges of $G_1$ are also part of $G_2$, this implies $\controlflow(G_1) < \controlflow(G_2)$.
	
	\item \textbf{Number of Ties $\numberofties$:}
	Suppose $\numberofties(L_1) < \numberofties(L_2)$.
	Then, by definition, there must be a pair $(a,b)$ with $a >_{L_2} b$ and $b \not>_{L_2} a$, but $a \not>_{L_1} b$ or $b >_{L_1} a$.
	Since $L_1 \sqsubset L_2$, of the latter, only $a \not>_{L_1} b$ can be true, so there is no connection between $a$ and $b$ in $G_1$.
	But because $a >_{L_2} b$, we know that $(a,b)$ is an edge in $G_2$, so $a$ has one more outgoing arc in $G_2$ than in $G_1$.
	Because $a$ must lie on a path from $\triangleright$ to $\square$ in $G_1$, this means that $a$ is a connector in $G_2$ with one more outgoing edge than in $G_1$.
	Since all edges of $G_1$ are also part of $G_2$, this implies $\controlflow(G_1) < \controlflow(G_2)$.
\end{itemize}
Thus, $\mathcal{C}^L(L_1) < \mathcal{C}^L(L_2)$ implies $\controlflow(G_1) < \controlflow(G_2)$ for any event log complexity measure $\mathcal{C}^L \in \{\variety, \levelofdetail, \numberofties\}$. \hfill$\square$
\end{proof}

\begin{theorem}
\label{theo:dfg-sep-entries}
$(\mathcal{C}^L, \separability) \in \norel$ for any log complexity measure $\mathcal{C}^L \in \loc$.
\end{theorem}
\begin{proof}
Consider the following event logs:
\begin{align*}
	L_1 &= [\langle a \rangle, \langle a,b,c \rangle] \\
	L_2 &= L_1 + [\langle a,b,c \rangle, \langle i,j,k,l,m \rangle] \\
	L_3 &= L_2 + [\langle a,b,c \rangle^{2}, \langle a,c,d \rangle, \langle a,c,e \rangle, \langle i,j,x,j,k,y,k,l,z,l,m \rangle]
\end{align*}
\cref{fig:dfg-separability} shows the directly follows graphs $G_1, G_2, G_3$ for the event logs $L_1, L_2, L_3$.
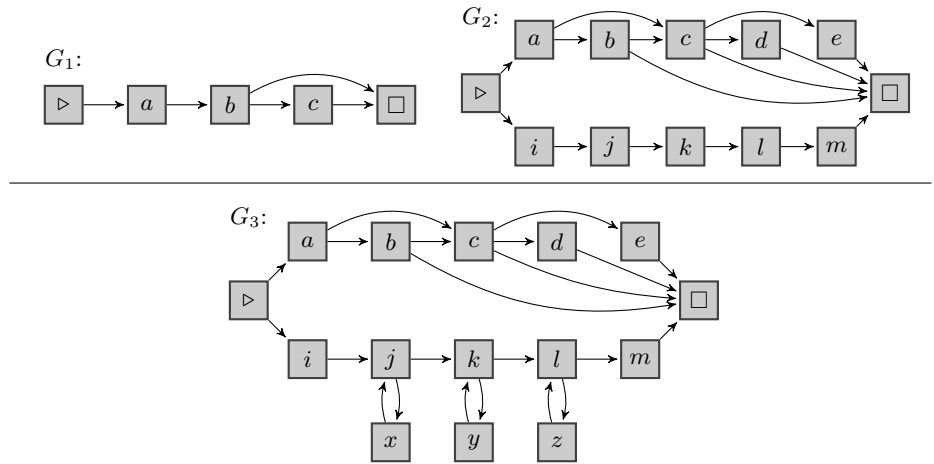
\begin{figure}[ht]
	\centering
	\begin{minipage}{0.45\textwidth}
	\centering
	\begin{tikzpicture}[node distance = 1.1cm,>=stealth',bend angle=0,auto]
		\node[transition] (start) {$\triangleright$};
		\node[above of=start,yshift=-0.5cm] {$G_1$:};
		\node[transition,right of=start] (a) {$a$}
		edge [pre] (start);
		\node[transition,right of=a] (b) {$b$}
		edge [pre] (a);
		\node[transition,right of=b] (c) {$c$}
		edge [pre] (b);
		\node[transition,right of=c] (end) {$\square$}
		edge [pre,bend right=30] (b)
		edge [pre] (c);
	\end{tikzpicture}
	\end{minipage}
	\begin{minipage}{0.52\textwidth}
	\centering
	\begin{tikzpicture}[node distance = 1cm,>=stealth',bend angle=0,auto]
		\node[transition] (start) {$\triangleright$};
		\node[above of=start] {$G_2$:};
		\node[transition,below right of=start] (i) {$i$}
		edge [pre] (start);
		\node[transition,above right of=start] (a) {$a$}
		edge [pre] (start);
		\node[transition,right of=a] (b) {$b$}
		edge [pre] (a);
		\node[transition,right of=i] (j) {$j$}
		edge [pre] (i);
		\node[transition,right of=b] (c) {$c$}
		edge [pre,bend right=30] (a)
		edge [pre] (b);
		\node[transition,right of=j] (k) {$k$}
		edge [pre] (j);
		\node[transition,right of=c] (d) {$d$}
		edge [pre] (c);
		\node[transition,right of=k] (l) {$l$}
		edge [pre] (k);
		\node[transition,right of=d] (e) {$e$}
		edge [pre,bend right=30] (c);
		\node[transition,right of=l] (m) {$m$}
		edge [pre] (l);
		\node[transition,above right of=m] (end) {$\square$}
		edge [pre,bend left=20] (b)
		edge [pre,bend left=10] (c)
		edge [pre] (d)
		edge [pre] (e)
		edge [pre] (m);
	\end{tikzpicture}
	\end{minipage}
	
	\medskip
	\hrule
	\medskip
	
	\begin{tikzpicture}[node distance = 1.1cm,>=stealth',bend angle=0,auto]
		\node[transition] (start) {$\triangleright$};
		\node[above of=start] {$G_3$:};
		\node[transition,below right of=start] (i) {$i$}
		edge [pre] (start);
		\node[transition,above right of=start] (a) {$a$}
		edge [pre] (start);
		\node[transition,right of=a] (b) {$b$}
		edge [pre] (a);
		\node[transition,right of=i] (j) {$j$}
		edge [pre] (i);
		\node[transition,right of=b] (c) {$c$}
		edge [pre,bend right=30] (a)
		edge [pre] (b);
		\node[transition,right of=j] (k) {$k$}
		edge [pre] (j);
		\node[transition,right of=c] (d) {$d$}
		edge [pre] (c);
		\node[transition,right of=k] (l) {$l$}
		edge [pre] (k);
		\node[transition,right of=d] (e) {$e$}
		edge [pre,bend right=30] (c);
		\node[transition,right of=l] (m) {$m$}
		edge [pre] (l);
		\node[transition,above right of=m] (end) {$\square$}
		edge [pre,bend left=20] (b)
		edge [pre,bend left=10] (c)
		edge [pre] (d)
		edge [pre] (e)
		edge [pre] (m);
		\node[transition,below of=j] (x) {$x$}
		edge [pre,bend right=15] (j)
		edge [post,bend left=15] (j);
		\node[transition,below of=k] (y) {$y$}
		edge [pre,bend right=15] (k)
		edge [post,bend left=15] (k);
		\node[transition,below of=l] (z) {$z$}
		edge [pre,bend right=15] (l)
		edge [post,bend left=15] (l);
	\end{tikzpicture}
	\caption{The directly follows graphs for the logs $L_1, L_2, L_3$ from the example in \cref{theo:dfg-sep-entries}. $G_1$ is the DFG for $L_1$, $G_2$ the one for $L_2$ and $G_3$ the one for $L_3$.}
	\label{fig:dfg-separability}
\end{figure}
These graphs have the following complexity scores:
\begin{itemize}
	\item[•] $\separability(G_1) = 2.8$,
	\item[•] $\separability(G_2) = 3$,
	\item[•] $\separability(G_3) = 2.8$,
\end{itemize}
so $\separability(G_1) < \separability(G_2)$, $\separability(G_2) > \separability(G_3)$, and $\separability(G_1) = \separability(G_3)$.
But the logs $L_1, L_2, L_3$ have the following log complexity scores:
\begin{center}
	\def\pad{\hspace*{1.5mm}}
	\begin{tabular}{|c|c|c|c|c|c|c|c|c|c|c|}\hline
		 & $\magnitude$ & $\variety$ & $\support$ & $\tlavg$ & $\tlmax$ & $\levelofdetail$ & $\numberofties$ & $\lempelziv$ & $\numberuniquetraces$ & $\percentageuniquetraces$ \\ \hline
		$L_1$ & $\pad 4 \pad$ & $\pad 3 \pad$ & $\pad 2 \pad$ & $\pad 2 \pad$ & $\pad 3 \pad$ & $\pad 2 \pad$ & $\pad 2 \pad$ & $\pad 3 \pad$ & $\pad 2 \pad$ & $\pad 1 \pad$ \\ \hline
		$L_2$ & $\pad 12 \pad$ & $\pad 8 \pad$ & $\pad 4 \pad$ & $\pad 3 \pad$ & $\pad 5 \pad$ & $\pad 3 \pad$ & $\pad 6 \pad$ & $\pad 9 \pad$ & $\pad 3 \pad$ & $\pad 0.75 \pad$ \\ \hline
		$L_3$ & $\pad 35 \pad$ & $\pad 13 \pad$ & $\pad 9 \pad$ & $\pad 3.8889 \pad$ & $\pad 11 \pad$ & $\pad 8 \pad$ & $\pad 9 \pad$ & $\pad 22 \pad$ & $\pad 6 \pad$ & $\pad 0.6667 \pad$ \\ \hline
	\end{tabular}
	
	\medskip
	
	\begin{tabular}{|c|c|c|c|c|c|c|c|c|} \hline
		 & $\structure$ & $\affinity$ & $\deviationfromrandom$ & $\avgdist$ & $\varentropy$ & $\normvarentropy$ & $\seqentropy$ & $\normseqentropy$ \\ \hline
		$L_1$ & $\pad 2 \pad$ & $\pad 0 \pad$ & $\pad 0.3764 \pad$ & $\pad 2 \pad$ & $\pad 0 \pad$ & $\pad 0 \pad$ & $\pad 0 \pad$ & $\pad 0 \pad$  \\ \hline
		$L_2$ & $\pad 3 \pad$ & $\pad 0.1667 \pad$ & $\pad 0.5854 \pad$ & $\pad 4.3333 \pad$ & $\pad 5.2925 \pad$ & $\pad 0.3181 \pad$ & $\pad 8.1503 \pad$ & $\pad 0.2733 \pad$ \\ \hline
		$L_3$ & $\pad 3.5556 \pad$ & $\pad 0.187 \pad$ & $\pad 0.7122 \pad$ & $\pad 5.1667 \pad$ & $\pad 27.4103 \pad$ & $\pad 0.4575 \pad$ & $\pad 47.1242 \pad$ & $\pad 0.3787 \pad$ \\ \hline
	\end{tabular}
\end{center}
Thus, we have $\mathcal{C}^L(L_1) < \mathcal{C}^L(L_2) < \mathcal{C}^L(L_3)$ for any event log complexity measure $\mathcal{C}^L \in (\loc \setminus \{\percentageuniquetraces\})$.
For $\percentageuniquetraces$, take the following event logs:
\begin{align*}
	L_1 &= [\langle a \rangle, \langle a,b,c \rangle^{3}] \\
	L_2 &= L_1 + [\langle i,j,k,l,m \rangle] \\
	L_3 &= L_2 + [\langle a,c,d \rangle, \langle a,c,e \rangle, \langle i,j,x,j,k,y,k,l,z,l,m \rangle]
\end{align*}
In constrast to the previous ones, only the frequencies changed, so the directly follows graphs $G_1, G_2, G_3$ for these event logs are the same as in \cref{fig:dfg-separability}.
But since $\percentageuniquetraces(L_1) = 0.5 < \percentageuniquetraces(L_2) = 0.6 < \percentageuniquetraces(L_3) = 0.75$, we now know that $(\mathcal{C}^L, \separability) \in \norel$ for any event log complexity measure $\mathcal{C}^L \in \loc$. \hfill$\square$
\end{proof}

\begin{theorem}
\label{theo:dfg-acd-entries}
$(\mathcal{C}^L, \avgconn) \in \norel$ for any event log complexity measure $\mathcal{C}^L \in \loc$.
\end{theorem}
\begin{proof}
Consider the following event logs:
\begin{align*}
	L_1 &= [\langle a,b \rangle^{3}, \langle c \rangle, \langle d \rangle, \langle e \rangle] \\
	L_2 &= L_1 + [\langle a,g,b \rangle] \\
	L_3 &= L_2 + [\langle h,i,j,k \rangle]
\end{align*}
\cref{fig:dfg-avgconn} shows the directly follows graphs $G_1, G_2, G_3$ for the event logs $L_1, L_2, L_3$.
\begin{figure}[ht]
	\centering
	\begin{minipage}{0.45\textwidth}
	\centering
	\begin{tikzpicture}[node distance = 1.1cm,>=stealth',bend angle=0,auto]
		\node[transition] (start) {$\triangleright$};
		\node[above of=start] {$G_1$:};
		\node[transition,right of=start,yshift=1.65cm] (a) {$a$}
		edge [pre] (start);
		\node[transition,right of=a] (b) {$b$}
		edge [pre] (a);
		\node[transition,right of=b,yshift=-1.65cm] (end) {$\square$}
		edge [pre] (b);
		\node[transition,below of=a,xshift=0.55cm] (c) {$c$}
		edge [pre] (start)
		edge [post] (end);
		\node[transition,below of=c] (d) {$d$}
		edge [pre] (start)
		edge [post] (end);
		\node[transition,below of=d] (e) {$e$}
		edge [pre] (start)
		edge [post] (end);
	\end{tikzpicture}
	\end{minipage}
	\begin{minipage}{0.45\textwidth}
	\centering
	\begin{tikzpicture}[node distance = 1.1cm,>=stealth',bend angle=0,auto]
		\node[transition] (start) {$\triangleright$};
		\node[above of=start] {$G_2$:};
		\node[transition,right of=start,yshift=1.65cm] (a) {$a$}
		edge [pre] (start);
		\node[transition,right of=a] (b) {$b$}
		edge [pre] (a);
		\node[transition,right of=b,yshift=-1.65cm] (end) {$\square$}
		edge [pre] (b);
		\node[transition,below of=a,xshift=0.55cm] (c) {$c$}
		edge [pre] (start)
		edge [post] (end);
		\node[transition,below of=c] (d) {$d$}
		edge [pre] (start)
		edge [post] (end);
		\node[transition,below of=d] (e) {$e$}
		edge [pre] (start)
		edge [post] (end);
		\node[transition,yshift=1.1cm] (g) at ($0.5*(a) + 0.5*(b)$) {$g$}
		edge [pre] (a)
		edge [post] (b);
	\end{tikzpicture}
	\end{minipage}
	
	\medskip
	\hrule
	\medskip
	
	\begin{tikzpicture}[node distance = 1.1cm,>=stealth',bend angle=0,auto]
		\node[transition] (start) {$\triangleright$};
		\node[above of=start] {$G_3$:};
		\node[transition,right of=start,yshift=1.65cm] (a) {$a$}
		edge [pre] (start);
		\node[transition,right of=a] (b) {$b$}
		edge [pre] (a);
		\node[transition,right of=b,yshift=-1.65cm] (end) {$\square$}
		edge [pre] (b);
		\node[transition,below of=a,xshift=0.55cm] (c) {$c$}
		edge [pre] (start)
		edge [post] (end);
		\node[transition,below of=c] (d) {$d$}
		edge [pre] (start)
		edge [post] (end);
		\node[transition,below of=d] (e) {$e$}
		edge [pre] (start)
		edge [post] (end);
		\node[transition,yshift=1.1cm] (g) at ($0.5*(a) + 0.5*(b)$) {$g$}
		edge [pre] (a)
		edge [post] (b);
		\node[transition,below of=start,yshift=-1.1cm] (h) {$h$}
		edge [pre] (start);
		\node[transition,right of=h] (i) {$i$}
		edge [pre] (h);
		\node[transition,right of=i] (j) {$j$}
		edge [pre] (i);
		\node[transition,right of=j] (k) {$k$}
		edge [pre] (j)
		edge [post] (end);
	\end{tikzpicture}
	\caption{The directly follows graphs for the logs $L_1, L_2, L_3$ from the example in \cref{theo:dfg-acd-entries}. $G_1$ is the DFG for $L_1$, $G_2$ the one for $L_2$ and $G_3$ the one for $L_3$.}
	\label{fig:dfg-avgconn}
\end{figure}
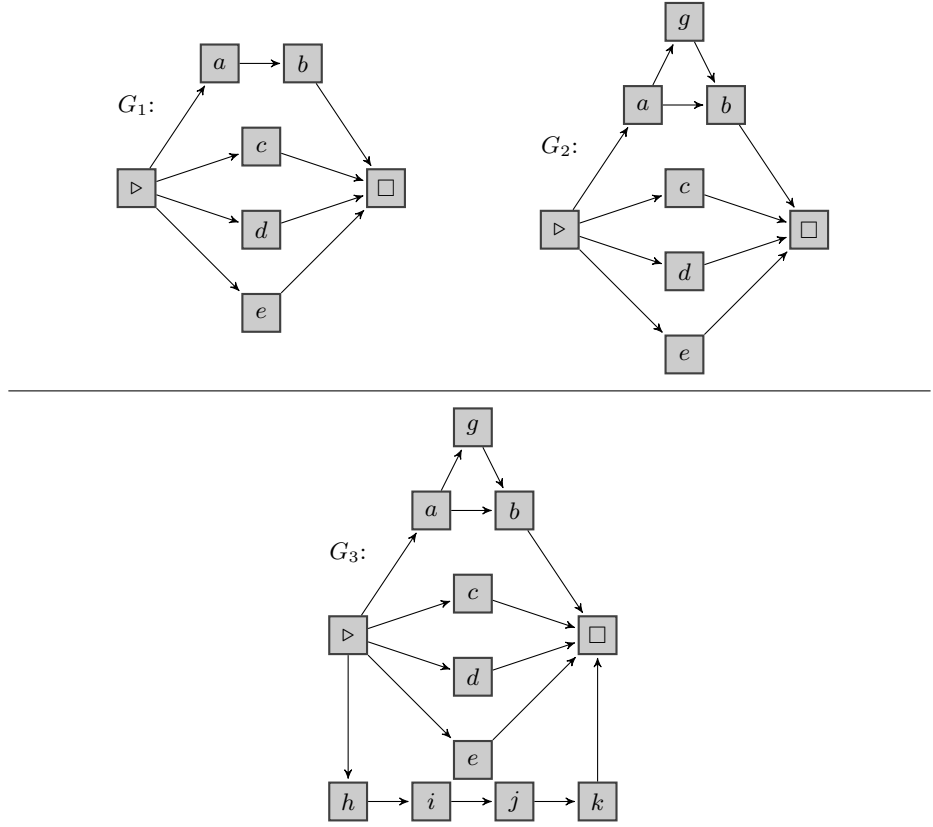
These graphs have the following complexity scores:
\begin{itemize}
	\item[•] $\avgconn(G_1) = 4$,
	\item[•] $\avgconn(G_2) = 3.5$,
	\item[•] $\avgconn(G_3) = 4$,
\end{itemize}
so $\avgconn(G_1) > \avgconn(G_2)$, $\avgconn(G_2) < \avgconn(G_3)$, and $\avgconn(G_1) = \avgconn(G_3)$.
But the logs $L_1, L_2, L_3$ have the following log complexity scores:
\begin{center}
	\def\pad{\hspace*{1.5mm}}
	\begin{tabular}{|c|c|c|c|c|c|c|c|c|c|c|}\hline
		 & $\magnitude$ & $\variety$ & $\support$ & $\tlavg$ & $\tlmax$ & $\levelofdetail$ & $\numberofties$ & $\lempelziv$ & $\numberuniquetraces$ & $\percentageuniquetraces$ \\ \hline
		$L_1$ & $\pad 9 \pad$ & $\pad 5 \pad$ & $\pad 6 \pad$ & $\pad 1.5 \pad$ & $\pad 2 \pad$ & $ \pad 4 \pad$ & $\pad 1 \pad$ & $\pad 6 \pad$ & $\pad 4 \pad$ & $\pad 0.6667 \pad$ \\ \hline
		$L_2$ & $\pad 12 \pad$ & $\pad 6 \pad$ & $\pad 7 \pad$ & $\pad 1.7143 \pad$ & $\pad 3 \pad$ & $\pad 5 \pad$ & $\pad 3 \pad$ & $\pad 7 \pad$ & $\pad 5 \pad$ & $\pad 0.7143 \pad$ \\ \hline
		$L_3$ & $\pad 16 \pad$ & $\pad 10 \pad$ & $\pad 8 \pad$ & $\pad 2 \pad$ & $\pad 4 \pad$ & $\pad 6 \pad$ & $\pad 6 \pad$ & $\pad 11 \pad$ & $\pad 6 \pad$ & $\pad 0.75 \pad$ \\ \hline
	\end{tabular}
	
	\medskip
	
	\begin{tabular}{|c|c|c|c|c|c|c|c|c|} \hline
		 & $\structure$ & $\affinity$ & $\deviationfromrandom$ & $\avgdist$ & $\varentropy$ & $\normvarentropy$ & $\seqentropy$ & $\normseqentropy$ \\ \hline
		$L_1$ & $\pad 1.5 \pad$ & $\pad 0.2 \pad$ & $\pad 0.0202 \pad$ & $\pad 2.2 \pad$ & $\pad 6.6609 \pad$ & $\pad 0.8277 \pad$ & $\pad 9.0245 \pad$ & $\pad 0.4564 \pad$ \\ \hline
		$L_2$ & $\pad 1.7143 \pad$ & $\pad 0.1429 \pad$ & $\pad 0.358 \pad$ & $\pad 2.2857 \pad$ & $\pad 10.8488 \pad$ & $\pad 0.7965 \pad$ & $\pad 14.8112 \pad$ & $\pad 0.4967 \pad$ \\ \hline
		$L_3$ & $\pad 2 \pad$ & $\pad 0.1071 \pad$ & $\pad 0.5431 \pad$ & $\pad 3.1429 \pad$ & $\pad 18.0591 \pad$ & $\pad 0.6847 \pad$ & $\pad 23.8086 \pad$ & $\pad 0.5367 \pad$ \\ \hline
	\end{tabular}
\end{center}
Thus, we have $\mathcal{C}^L(L_1) < \mathcal{C}^L(L_2) < \mathcal{C}^L(L_3)$ for any event log complexity measure $\mathcal{C}^L \in (\loc \setminus \{\affinity, \normvarentropy\})$.
For $\affinity$ and $\normvarentropy$, take the following logs:
\begin{align*}
	L_1 &= [\langle a,b \rangle, \langle c,x \rangle, \langle d,y \rangle, \langle e,z \rangle] \\
	L_2 &= L_1 + [\langle a,b \rangle, \langle a,g,b \rangle] \\
	L_3 &= L_2 + [\langle c,x \rangle, \langle h,i \rangle]
\end{align*}
\cref{fig:dfg-avgconn-affinity-nvare} shows the directly follows graphs $G_1, G_2, G_3$ for these logs $L_1, L_2, L_3$.
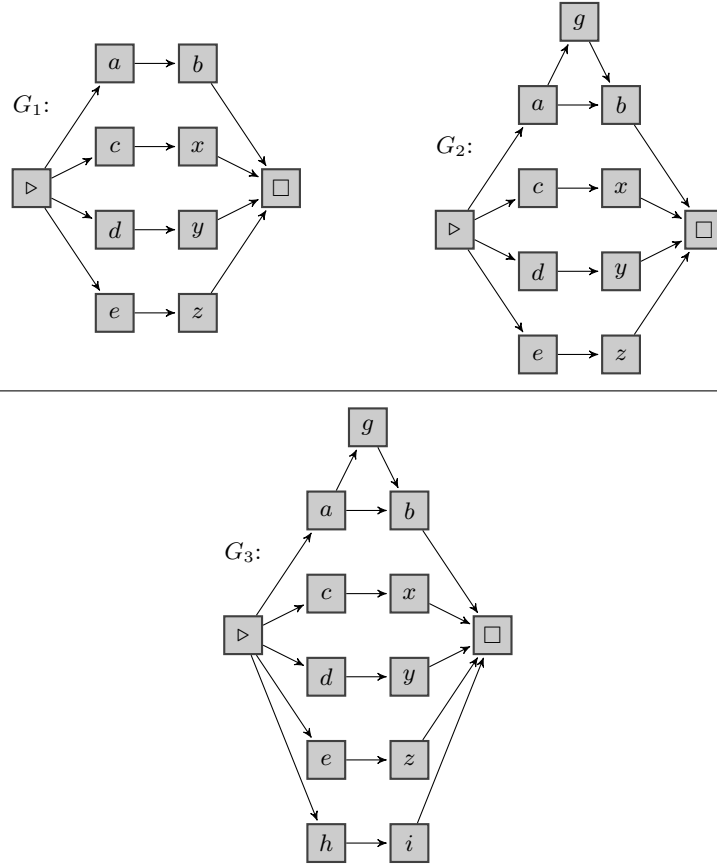
\begin{figure}[ht]
	\centering
	\begin{minipage}{0.45\textwidth}
	\centering
	\begin{tikzpicture}[node distance = 1.1cm,>=stealth',bend angle=0,auto]
		\node[transition] (start) {$\triangleright$};
		\node[above of=start] {$G_1$:};
		\node[transition,right of=start,yshift=1.65cm] (a) {$a$}
		edge [pre] (start);
		\node[transition,right of=a] (b) {$b$}
		edge [pre] (a);
		\node[transition,right of=b,yshift=-1.65cm] (end) {$\square$}
		edge [pre] (b);
		\node[transition,below of=a] (c) {$c$}
		edge [pre] (start);
		\node[transition,right of=c] (x) {$x$}
		edge [pre] (c)
		edge [post] (end);
		\node[transition,below of=c] (d) {$d$}
		edge [pre] (start);
		\node[transition,right of=d] (y) {$y$}
		edge [pre] (d)
		edge [post] (end);
		\node[transition,below of=d] (e) {$e$}
		edge [pre] (start);
		\node[transition,right of=e] (z) {$z$}
		edge [pre] (e)
		edge [post] (end);
	\end{tikzpicture}
	\end{minipage}
	\begin{minipage}{0.45\textwidth}
	\centering
	\begin{tikzpicture}[node distance = 1.1cm,>=stealth',bend angle=0,auto]
		\node[transition] (start) {$\triangleright$};
		\node[above of=start] {$G_2$:};
		\node[transition,right of=start,yshift=1.65cm] (a) {$a$}
		edge [pre] (start);
		\node[transition,right of=a] (b) {$b$}
		edge [pre] (a);
		\node[transition,right of=b,yshift=-1.65cm] (end) {$\square$}
		edge [pre] (b);
		\node[transition,below of=a] (c) {$c$}
		edge [pre] (start);
		\node[transition,right of=c] (x) {$x$}
		edge [pre] (c)
		edge [post] (end);
		\node[transition,below of=c] (d) {$d$}
		edge [pre] (start);
		\node[transition,right of=d] (y) {$y$}
		edge [pre] (d)
		edge [post] (end);
		\node[transition,below of=d] (e) {$e$}
		edge [pre] (start);
		\node[transition,right of=e] (z) {$z$}
		edge [pre] (e)
		edge [post] (end);
		\node[transition,yshift=1.1cm] (g) at ($0.5*(a) + 0.5*(b)$) {$g$}
		edge [pre] (a)
		edge [post] (b);
	\end{tikzpicture}
	\end{minipage}
	
	\medskip
	\hrule
	\medskip
	
	\begin{tikzpicture}[node distance = 1.1cm,>=stealth',bend angle=0,auto]
		\node[transition] (start) {$\triangleright$};
		\node[above of=start] {$G_3$:};
		\node[transition,right of=start,yshift=1.65cm] (a) {$a$}
		edge [pre] (start);
		\node[transition,right of=a] (b) {$b$}
		edge [pre] (a);
		\node[transition,right of=b,yshift=-1.65cm] (end) {$\square$}
		edge [pre] (b);
		\node[transition,below of=a] (c) {$c$}
		edge [pre] (start);
		\node[transition,right of=c] (x) {$x$}
		edge [pre] (c)
		edge [post] (end);
		\node[transition,below of=c] (d) {$d$}
		edge [pre] (start);
		\node[transition,right of=d] (y) {$y$}
		edge [pre] (d)
		edge [post] (end);
		\node[transition,below of=d] (e) {$e$}
		edge [pre] (start);
		\node[transition,right of=e] (z) {$z$}
		edge [pre] (e)
		edge [post] (end);
		\node[transition,yshift=1.1cm] (g) at ($0.5*(a) + 0.5*(b)$) {$g$}
		edge [pre] (a)
		edge [post] (b);
		\node[transition,below of=e] (h) {$h$}
		edge [pre] (start);
		\node[transition,right of=h] (i) {$i$}
		edge [pre] (h)
		edge [post] (end);
	\end{tikzpicture}
	\caption{The directly follows graphs for the logs $L_1, L_2, L_3$ from the example in \cref{theo:dfg-acd-entries}. $G_1$ is the DFG for $L_1$, $G_2$ the one for $L_2$ and $G_3$ the one for $L_3$.}
	\label{fig:dfg-avgconn-affinity-nvare}
\end{figure}
These graphs have the following complexity scores:
\begin{itemize}
	\item[•] $\avgconn(G_1) = 4$,
	\item[•] $\avgconn(G_2) = 3.5$,
	\item[•] $\avgconn(G_3) = 4$,
\end{itemize}
so $\avgconn(G_1) > \avgconn(G_2)$, $\avgconn(G_2) < \avgconn(G_3)$, and $\avgconn(G_1) = \avgconn(G_3)$.
But $\affinity(L_1) = 0 < \affinity(L_2) \approx 0.0667 < \affinity(L_3) \approx 0.0714$, and $\normvarentropy(L_1) \approx 0.6667 < \normvarentropy(L_2) \approx 0.699 < \normvarentropy(L_3) \approx 0.7211$.
Thus, we have shown that $(\mathcal{C}^L, \avgconn) \in \norel$ for all log complexity measures $\mathcal{C}^L \in \loc$. \hfill$\square$
\end{proof}

\begin{theorem}
\label{theo:dfg-mcd-diam-entries}
Let $\mathcal{C}^L \in \loc$ be an arbitrary event log complexity measure and let $\mathcal{C}^M \in \{\maxconn, \diameter\}$ be a model complexity measure.
Then, $(\mathcal{C}^L, \mathcal{C}^M) \in \mleq$.
\end{theorem}
\begin{proof}
Let $L_1, L_2$ be event logs with $L_1 \sqsubset L_2$ and $|supp(L_1)| > 1$, and $G_1, G_2$ be their directly follows graphs.
By \cref{lemma:dfg-monotone-increasing}, we know that $\maxconn(G_1) \leq \maxconn(G_2)$ and $\diameter(G_1) \leq \diameter(G_2)$.
What remains to be shown is that with $\mathcal{C}^L(L_1) < \mathcal{C}^L(L_2)$, both $\mathcal{C}^M(G_1) = \mathcal{C}^M(G_2)$ and $\mathcal{C}^M(G_1) < \mathcal{C}^M(G_2)$ are possible.
For the former, take the following event logs:
\begin{align*}
	L_1 &= [\langle a,b,c,c \rangle, \langle c \rangle^{2}, \langle c,c,d,e \rangle] \\
	L_2 &= L_1 + [\langle a,b,c,d,e \rangle, \langle a,b,f,f,d,e \rangle, \langle a,b,f,f,f,d,e \rangle, \langle a,b,f,f,f,f,d,e \rangle^{2}]
\end{align*}
These two event logs have the following log complexity scores:
\begin{center}
	\def\pad{\hspace*{1.5mm}}
	\begin{tabular}{|c|c|c|c|c|c|c|c|c|c|c|}\hline
		 & $\magnitude$ & $\variety$ & $\support$ & $\tlavg$ & $\tlmax$ & $\levelofdetail$ & $\numberofties$ & $\lempelziv$ & $\numberuniquetraces$ & $\percentageuniquetraces$ \\ \hline
		$L_1$ & $\pad 10 \pad$ & $\pad 5 \pad$ & $\pad 4 \pad$ & $\pad 2.5 \pad$ & $\pad 4 \pad$ & $\pad 4 \pad$ & $\pad 4 \pad$ & $\pad 7 \pad$ & $\pad 3 \pad$ & $\pad 0.75 \pad$ \\ \hline
		$L_2$ & $\pad 44 \pad$ & $\pad 6 \pad$ & $\pad 9 \pad$ & $\pad 4.8889 \pad$ & $\pad 8 \pad$ & $\pad 5 \pad$ & $\pad 6 \pad$ & $\pad 21 \pad$ & $\pad 7 \pad$ & $\pad 0.7778 \pad$ \\ \hline
	\end{tabular}
	
	\medskip
	
	\begin{tabular}{|c|c|c|c|c|c|c|c|c|} \hline
		 & $\structure$ & $\affinity$ & $\deviationfromrandom$ & $\avgdist$ & $\varentropy$ & $\normvarentropy$ & $\seqentropy$ & $\normseqentropy$ \\ \hline
		$L_1$ & $\pad 2 \pad$ & $\pad 0.2 \pad$ & $\pad 0.5731 \pad$ & $\pad 2.6667 \pad$ & $\pad 5.5452 \pad$ & $\pad 0.3333 \pad$ & $\pad 6.7301 \pad$ & $\pad 0.2923 \pad$ \\ \hline
		$L_2$ & $\pad 3.6667 \pad$ & $\pad 0.2857 \pad$ & $\pad 0.6353 \pad$ & $\pad 4.9444 \pad$ & $\pad 35.3011 \pad$ & $\pad 0.5892 \pad$ & $\pad 71.9231 \pad$ & $\pad 0.432 \pad$ \\ \hline
	\end{tabular}
\end{center}
Thus, $\mathcal{C}^L(L_1) < \mathcal{C}^L(L_2)$ for any $\mathcal{C}^L \in \loc$.
\cref{fig:dfg-maxconn-diameter-equal-possible} shows the directly follows graphs $G_1, G_2$ for $L_1$ and $L_2$.
\begin{figure}[ht]
	\centering
	\begin{tikzpicture}[node distance = 1.1cm,>=stealth',bend angle=0,auto]
		\node[transition] (start) {$\triangleright$};
		\node[above of=start,yshift=-0.5cm] {$G_1$:};
		\node[transition,right of=start] (a) {$a$}
		edge [pre] (start);
		\node[transition,right of=a] (b) {$b$}
		edge [pre] (a);
		\node[transition,right of=b] (c) {$c$}
		edge [pre,bend right=30] (start)
		edge [pre] (b)
		edge [pre,loop,out=60,in=120,looseness=6] (c);
		\node[transition,right of=c] (d) {$d$}
		edge [pre] (c);
		\node[transition,right of=d] (e) {$e$}
		edge [pre] (d);
		\node[transition,right of=e] (end) {$\square$}
		edge [pre,bend right=30] (c)
		edge [pre] (e);
	\end{tikzpicture}
	
	\medskip
	\hrule
	\medskip
	
	\begin{tikzpicture}[node distance = 1.1cm,>=stealth',bend angle=0,auto]
		\node[transition] (start) {$\triangleright$};
		\node[above of=start,yshift=-0.5cm] {$G_2$:};
		\node[transition,right of=start] (a) {$a$}
		edge [pre] (start);
		\node[transition,right of=a] (b) {$b$}
		edge [pre] (a);
		\node[transition,right of=b] (c) {$c$}
		edge [pre,bend right=30] (start)
		edge [pre] (b)
		edge [pre,loop,out=60,in=120,looseness=6] (c);
		\node[transition,right of=c] (d) {$d$}
		edge [pre] (c);
		\node[transition,right of=d] (e) {$e$}
		edge [pre] (d);
		\node[transition,right of=e] (end) {$\square$}
		edge [pre,bend right=30] (c)
		edge [pre] (e);
		\node[transition,below of=c] (f) {$f$}
		edge [pre] (b)
		edge [post] (d)
		edge [pre,loop,out=300,in=240,looseness=6] (f);
	\end{tikzpicture}
	\caption{The directly follows graph for the event logs $L_1$ and $L_2$ of \cref{theo:dfg-mcd-diam-entries}.}
	\label{fig:dfg-maxconn-diameter-equal-possible}
\end{figure}
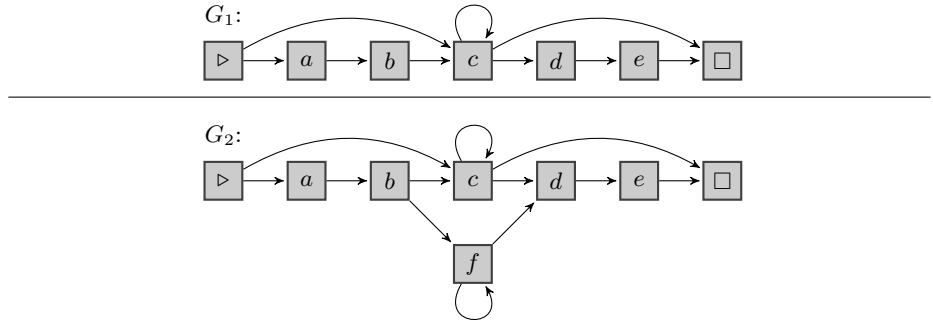
$G_1$ and $G_2$ fulfill $\maxconn(G_1) = 6 = \maxconn(G_2)$ and $\diameter(G_1) = 7 = \diameter(G_2)$, so $\mathcal{C}^L(L_1) < \mathcal{C}^L(L_2)$ and $\mathcal{C}^M(G_1) = \mathcal{C}^M(G_2)$ is possible for any $\mathcal{C}^M \in \{\maxconn, \diameter\}$.
To see that $\maxconn(G_1) < \maxconn(G_2)$ and $\diameter(G_1) < \diameter(G_2)$ is also possible when $\mathcal{C}^L(L_1) < \mathcal{C}^L(L_2)$, consider the following event logs:
\begin{align*}
	L_1 &= [\langle a,b,c,c \rangle, \langle c \rangle^{2}, \langle c,c,d,e \rangle] \\
	L_2 &= L_1 + [\langle a,b,c,d,e \rangle, \langle a,b,f,f,d,e \rangle, \langle a,b,f,f,f,d,e \rangle, \langle a,b,f,f,f,f,d,e \rangle^{2}, \\
	&\phantom{= L_1 + [}\hspace*{1mm} \langle a,c,c,d,e,g \rangle]
\end{align*}
These two event logs have the following log complexity scores:
\begin{center}
	\def\pad{\hspace*{1.5mm}}
	\begin{tabular}{|c|c|c|c|c|c|c|c|c|c|c|}\hline
		 & $\magnitude$ & $\variety$ & $\support$ & $\tlavg$ & $\tlmax$ & $\levelofdetail$ & $\numberofties$ & $\lempelziv$ & $\numberuniquetraces$ & $\percentageuniquetraces$ \\ \hline
		$L_1$ & $\pad 10 \pad$ & $\pad 5 \pad$ & $\pad 4 \pad$ & $\pad 2.5 \pad$ & $\pad 4 \pad$ & $\pad 4 \pad$ & $\pad 4 \pad$ & $\pad 7 \pad$ & $\pad 3 \pad$ & $\pad 0.75 \pad$ \\ \hline
		$L_2$ & $\pad 50 \pad$ & $\pad 7 \pad$ & $\pad 10 \pad$ & $\pad 5 \pad$ & $\pad 8 \pad$ & $\pad 11 \pad$ & $\pad 8 \pad$ & $\pad 24 \pad$ & $\pad 8 \pad$ & $\pad 0.8 \pad$ \\ \hline
	\end{tabular}
	
	\medskip
	
	\begin{tabular}{|c|c|c|c|c|c|c|c|c|} \hline
		 & $\structure$ & $\affinity$ & $\deviationfromrandom$ & $\avgdist$ & $\varentropy$ & $\normvarentropy$ & $\seqentropy$ & $\normseqentropy$ \\ \hline
		$L_1$ & $\pad 2 \pad$ & $\pad 0.2 \pad$ & $\pad 0.5731 \pad$ & $\pad 2.6667 \pad$ & $\pad 5.5452 \pad$ & $\pad 0.3333 \pad$ & $\pad 6.7301 \pad$ & $\pad 0.2923 \pad$ \\ \hline
		$L_2$ & $\pad 3.8 \pad$ & $\pad 0.2613 \pad$ & $\pad 0.656 \pad$ & $\pad 5.0222 \pad$ & $\pad 47.8112 \pad$ & $\pad 0.5941 \pad$ & $\pad 89.2321 \pad$ & $\pad 0.4562 \pad$ \\ \hline
	\end{tabular}
\end{center}
\cref{fig:dfg-maxconn-diameter-less-possible} shows the directly follows graphs $G_1, G_2$ for $L_1$ and $L_2$.
\begin{figure}[ht]
	\centering
	\begin{tikzpicture}[node distance = 1.1cm,>=stealth',bend angle=0,auto]
		\node[transition] (start) {$\triangleright$};
		\node[above of=start,yshift=-0.5cm] {$G_1$:};
		\node[transition,right of=start] (a) {$a$}
		edge [pre] (start);
		\node[transition,right of=a] (b) {$b$}
		edge [pre] (a);
		\node[transition,right of=b] (c) {$c$}
		edge [pre,bend right=30] (start)
		edge [pre] (b)
		edge [pre,loop,out=60,in=120,looseness=6] (c);
		\node[transition,right of=c] (d) {$d$}
		edge [pre] (c);
		\node[transition,right of=d] (e) {$e$}
		edge [pre] (d);
		\node[transition,right of=e] (end) {$\square$}
		edge [pre,bend right=30] (c)
		edge [pre] (e);
	\end{tikzpicture}
	
	\medskip
	\hrule
	\medskip
	
	\begin{tikzpicture}[node distance = 1.1cm,>=stealth',bend angle=0,auto]
		\node[transition] (start) {$\triangleright$};
		\node[above of=start,yshift=-0.5cm] {$G_2$:};
		\node[transition,right of=start] (a) {$a$}
		edge [pre] (start);
		\node[transition,right of=a] (b) {$b$}
		edge [pre] (a);
		\node[transition,right of=b] (c) {$c$}
		edge [pre,bend right=30] (start)
		edge [pre,bend left=40] (a)
		edge [pre] (b)
		edge [pre,loop,out=60,in=120,looseness=6] (c);
		\node[transition,right of=c] (d) {$d$}
		edge [pre] (c);
		\node[transition,right of=d] (e) {$e$}
		edge [pre] (d);
		\node[transition,right of=e] (end) {$\square$}
		edge [pre,bend right=30] (c)
		edge [pre] (e);
		\node[transition,below of=c] (f) {$f$}
		edge [pre] (b)
		edge [post] (d)
		edge [pre,loop,out=300,in=240,looseness=6] (f);
		\node[transition,below of=e] (g) {$g$}
		edge [pre] (e)
		edge [post] (end);
	\end{tikzpicture}
	\caption{The directly follows graph for the event logs $L_1$ and $L_2$ of \cref{theo:dfg-mcd-diam-entries}.}
	\label{fig:dfg-maxconn-diameter-less-possible}
\end{figure}
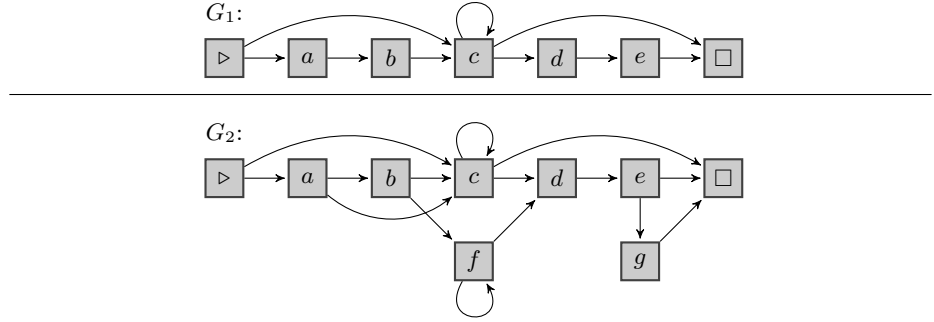
These graphs fulfill $\avgconn(G_1) = 6 < 7 = \avgconn(G_2)$ and $\diameter(G_1) = 7 < 8 = \diameter(G_2)$, which shows that $\mathcal{C}^M(G_1) < \mathcal{C}^M(G_2)$ is also possible for $\mathcal{C}^M \in \{\avgconn, \diameter\}$, when $\mathcal{C}^L(L_1) < \mathcal{C}^L(L_2)$ for any $\mathcal{C}^L \in \loc$. \hfill$\square$
\end{proof}

\begin{theorem}
\label{theo:dfg-seq-entries}
$(\mathcal{C}^L, \sequentiality) \in \norel$ for any event log complexity measure $\mathcal{C}^L \in \loc$.
\end{theorem}
\begin{proof}
Consider the following event logs:
\begin{align*}
	L_1 &= [\langle a,b,d,e \rangle^{2}, \langle a,c,d,e \rangle^{2}, \langle a,b,c,d,e \rangle, \langle e \rangle] \\
	L_2 &= L_1 + [\langle a,b,d,a,c,d \rangle^{2}, \langle a,b,c,d,e,f,g \rangle] \\
	L_3 &= L_2 + [\langle a,b,d,a,b,d,a,c,d \rangle, \langle a,b,c,d,e,f,h \rangle]
\end{align*}
\cref{fig:dfg-sequentiality} shows the directly follows graphs $G_1, G_2, G_3$ for the event logs $L_1, L_2, L_3$.
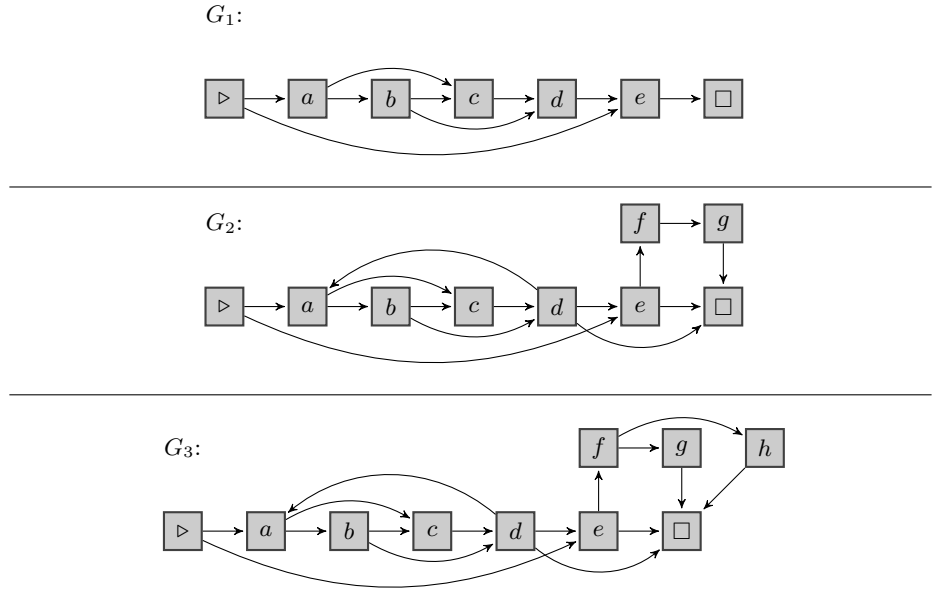
\begin{figure}[ht]
	\centering
	\begin{tikzpicture}[node distance = 1.1cm,>=stealth',bend angle=0,auto]
		\node[transition] (start) {$\triangleright$};
		\node[above of=start] {$G_1$:};
		\node[transition,right of=start] (a) {$a$}
		edge [pre] (start);
		\node[transition,right of=a] (b) {$b$}
		edge [pre] (a);
		\node[transition,right of=b] (c) {$c$}
		edge [pre,bend right=30] (a)
		edge [pre] (b);
		\node[transition,right of=c] (d) {$d$}
		edge [pre,bend left=30] (b)
		edge [pre] (c);
		\node[transition,right of=d] (e) {$e$}
		edge [pre,bend left=25] (start)
		edge [pre] (d);
		\node[transition,right of=e] (end) {$\square$}
		edge [pre] (e);
	\end{tikzpicture}
	
	\medskip
	\hrule
	\medskip
	
	\begin{tikzpicture}[node distance = 1.1cm,>=stealth',bend angle=0,auto]
		\node[transition] (start) {$\triangleright$};
		\node[above of=start] {$G_2$:};
		\node[transition,right of=start] (a) {$a$}
		edge [pre] (start);
		\node[transition,right of=a] (b) {$b$}
		edge [pre] (a);
		\node[transition,right of=b] (c) {$c$}
		edge [pre,bend right=30] (a)
		edge [pre] (b);
		\node[transition,right of=c] (d) {$d$}
		edge [post,bend right=40] (a)
		edge [pre,bend left=30] (b)
		edge [pre] (c);
		\node[transition,right of=d] (e) {$e$}
		edge [pre,bend left=25] (start)
		edge [pre] (d);
		\node[transition,above of=e] (f) {$f$}
		edge [pre] (e);
		\node[transition,right of=f] (g) {$g$}
		edge [pre] (f)
		edge [post] (end);
		\node[transition,right of=e] (end) {$\square$}
		edge [pre,bend left=40] (d)
		edge [pre] (e);
	\end{tikzpicture}
	
	\medskip
	\hrule
	\medskip
	
	\begin{tikzpicture}[node distance = 1.1cm,>=stealth',bend angle=0,auto]
		\node[transition] (start) {$\triangleright$};
		\node[above of=start] {$G_3$:};
		\node[transition,right of=start] (a) {$a$}
		edge [pre] (start);
		\node[transition,right of=a] (b) {$b$}
		edge [pre] (a);
		\node[transition,right of=b] (c) {$c$}
		edge [pre,bend right=30] (a)
		edge [pre] (b);
		\node[transition,right of=c] (d) {$d$}
		edge [post,bend right=40] (a)
		edge [pre,bend left=30] (b)
		edge [pre] (c);
		\node[transition,right of=d] (e) {$e$}
		edge [pre,bend left=25] (start)
		edge [pre] (d);
		\node[transition,above of=e] (f) {$f$}
		edge [pre] (e);
		\node[transition,right of=f] (g) {$g$}
		edge [pre] (f)
		edge [post] (end);
		\node[transition,right of=g] (h) {$h$}
		edge [pre,bend right=30] (f)
		edge [post] (end);
		\node[transition,right of=e] (end) {$\square$}
		edge [pre,bend left=40] (d)
		edge [pre] (e);
	\end{tikzpicture}
	\caption{The directly follows graphs for the logs $L_1, L_2, L_3$ from the example in \cref{theo:dfg-seq-entries}. $G_1$ is the DFG for $L_1$, $G_2$ the one for $L_2$ and $G_3$ the one for $L_3$.}
	\label{fig:dfg-sequentiality}
\end{figure}
These graphs have the following complexity scores:
\begin{itemize}
	\item[•] $\sequentiality(G_1) = 1$,
	\item[•] $\sequentiality(G_2) \approx 0.9286$,
	\item[•] $\sequentiality(G_3) = 1$,
\end{itemize}
so $\sequentiality(G_1) > \sequentiality(G_2)$, $\sequentiality(G_2) < \sequentiality(G_3)$, and $\sequentiality(G_1) = \sequentiality(G_3)$.
But the logs $L_1, L_2, L_3$ have the following log complexity scores:
\begin{center}
	\def\pad{\hspace*{1.5mm}}
	\begin{tabular}{|c|c|c|c|c|c|c|c|c|c|c|}\hline
		 & $\magnitude$ & $\variety$ & $\support$ & $\tlavg$ & $\tlmax$ & $\levelofdetail$ & $\numberofties$ & $\lempelziv$ & $\numberuniquetraces$ & $\percentageuniquetraces$ \\ \hline
		$L_1$ & $\pad 22 \pad$ & $\pad 5 \pad$ & $\pad 6 \pad$ & $\pad 3.6667 \pad$ & $\pad 5 \pad$ & $\pad 4 \pad$ & $\pad 6 \pad$ & $\pad 12 \pad$ & $\pad 4 \pad$ & $\pad 0.6667 \pad$ \\ \hline
		$L_2$ & $\pad 41 \pad$ & $\pad 7 \pad$ & $\pad 9 \pad$ & $\pad 4.5556 \pad$ & $\pad 7 \pad$ & $\pad 11 \pad$ & $\pad 9 \pad$ & $\pad 19 \pad$ & $\pad 6 \pad$ & $\pad 0.6667 \pad$ \\ \hline
		$L_3$ & $\pad 57 \pad$ & $\pad 8 \pad$ & $\pad 11 \pad$ & $\pad 5.1818 \pad$ & $\pad 9 \pad$ & $\pad 15 \pad$ & $\pad 10 \pad$ & $\pad 25 \pad$ & $\pad 8 \pad$ & $\pad 0.7273 \pad$ \\ \hline
	\end{tabular}
	
	\medskip
	
	\begin{tabular}{|c|c|c|c|c|c|c|c|c|} \hline
		 & $\structure$ & $\affinity$ & $\deviationfromrandom$ & $\avgdist$ & $\varentropy$ & $\normvarentropy$ & $\seqentropy$ & $\normseqentropy$ \\ \hline
		$L_1$ & $\pad 3.6667 \pad$ & $\pad 0.2933 \pad$ & $\pad 0.5961 \pad$ & $\pad 1.8667 \pad$ & $\pad 14.24 \pad$ & $\pad 0.5399 \pad$ & $\pad 24.1377 \pad$ & $\pad 0.355 \pad$ \\ \hline
		$L_2$ & $\pad 4.1111 \pad$ & $\pad 0.3026 \pad$ & $\pad 0.6449 \pad$ & $\pad 3.0556 \pad$ & $\pad 24.1774 \pad$ & $\pad 0.545 \pad$ & $\pad 54.2052 \pad$ & $\pad 0.356 \pad$ \\ \hline
		$L_3$ & $\pad 4.3636 \pad$ & $\pad 0.3259 \pad$ & $\pad 0.6543 \pad$ & $\pad 3.7818 \pad$ & $\pad 39.7717 \pad$ & $\pad 0.5849 \pad$ & $\pad 87.744 \pad$ & $\pad 0.3807 \pad$ \\ \hline
	\end{tabular}
\end{center}
Therefore, $\mathcal{C}^L(L_1) < \mathcal{C}^L(L_2) < \mathcal{C}^L(L_3)$ for any event log complexity measure $\mathcal{C}^L \in (\loc \setminus \{\percentageuniquetraces\})$.
For $\percentageuniquetraces$, consider the following event logs:
\begin{align*}
	L_1 &= [\langle a,b,d,e \rangle^{3}, \langle a,c,d,e \rangle^{2}, \langle a,b,c,d,e \rangle, \langle e \rangle] \\
	L_2 &= L_1 + [\langle a,b,d,a,c,d \rangle^{2}, \langle a,b,c,d,e,f,g \rangle] \\
	L_3 &= L_2 + [\langle a,b,d,a,b,d,a,c,d \rangle, \langle a,b,c,d,e,f,h \rangle]
\end{align*}
Note that only the frequency of the trace $\langle a,b,d,e \rangle$ changed compared to the previous event logs.
Thus, the directly follows graphs $G_1, G_2, G_3$ for these new event logs $L_1, L_2, L_3$ are the same as the ones shown in \cref{fig:dfg-sequentiality}.
Since the percentage of unique traces in the event logs $L_1, L_2, L_3$ strictly increase, i.e., $\percentageuniquetraces(G_1) \approx 0.5714 < \percentageuniquetraces(G_2) = 0.6 < \percentageuniquetraces(G_3) \approx 0.6667$, we have thus shown that $(\mathcal{C}^L, \sequentiality) \in \norel$ for any log complexity measure $\mathcal{C}^L \in \loc$. \hfill$\square$
\end{proof}

\begin{theorem}
\label{theo:dfg-depth-entries}
$(\mathcal{C}^L, \depth) \in \norel$ for any log complexity measure $\mathcal{C}^L \in \loc$.
\end{theorem}
\begin{proof}
Consider the following event logs:
\begin{align*}
	L_1 &= [\langle a,b \rangle, \langle c,x \rangle, \langle d,y \rangle, \langle e,z \rangle] \\
	L_2 &= L_1 + [\langle a,b \langle, \langle a,g,b \rangle] \\
	L_3 &= L_2 + [\langle a,b,c,x \rangle, \langle h,i \rangle]
\end{align*}
\cref{fig:dfg-depth} shows the directly follows graphs $G_1, G_2, G_3$ for the event logs $L_1, L_2, L_3$.
\begin{figure}[ht]
	\centering
	\begin{minipage}{0.45\textwidth}
	\centering
	\begin{tikzpicture}[node distance = 1.1cm,>=stealth',bend angle=0,auto]
		\node[transition] (start) {$\triangleright$};
		\node[above of=start] {$G_1$:};
		\node[transition,right of=start,yshift=1.65cm] (a) {$a$}
		edge [pre] (start);
		\node[transition,right of=a] (b) {$b$}
		edge [pre] (a);
		\node[transition,right of=b,yshift=-1.65cm] (end) {$\square$}
		edge [pre] (b);
		\node[transition,below of=a] (c) {$c$}
		edge [pre] (b)
		edge [pre] (start);
		\node[transition,right of=c] (x) {$x$}
		edge [pre] (c)
		edge [post] (end);
		\node[transition,below of=c] (d) {$d$}
		edge [pre] (start);
		\node[transition,right of=d] (y) {$y$}
		edge [pre] (d)
		edge [post] (end);
		\node[transition,below of=d] (e) {$e$}
		edge [pre] (start);
		\node[transition,right of=e] (z) {$z$}
		edge [pre] (e)
		edge [post] (end);
	\end{tikzpicture}
	\end{minipage}
	\begin{minipage}{0.45\textwidth}
	\centering
	\begin{tikzpicture}[node distance = 1.1cm,>=stealth',bend angle=0,auto]
		\node[transition] (start) {$\triangleright$};
		\node[above of=start] {$G_2$:};
		\node[transition,right of=start,yshift=1.65cm] (a) {$a$}
		edge [pre] (start);
		\node[transition,right of=a] (b) {$b$}
		edge [pre] (a);
		\node[transition,right of=b,yshift=-1.65cm] (end) {$\square$}
		edge [pre] (b);
		\node[transition,below of=a] (c) {$c$}
		edge [pre] (start);
		\node[transition,right of=c] (x) {$x$}
		edge [pre] (c)
		edge [post] (end);
		\node[transition,below of=c] (d) {$d$}
		edge [pre] (start);
		\node[transition,right of=d] (y) {$y$}
		edge [pre] (d)
		edge [post] (end);
		\node[transition,below of=d] (e) {$e$}
		edge [pre] (start);
		\node[transition,right of=e] (z) {$z$}
		edge [pre] (e)
		edge [post] (end);
		\node[transition,yshift=1.1cm] (g) at ($0.5*(a) + 0.5*(b)$) {$g$}
		edge [pre] (a)
		edge [post] (b);
	\end{tikzpicture}
	\end{minipage}
	
	\medskip
	\hrule
	\medskip
	
	\begin{tikzpicture}[node distance = 1.1cm,>=stealth',bend angle=0,auto]
		\node[transition] (start) {$\triangleright$};
		\node[above of=start] {$G_3$:};
		\node[transition,right of=start,yshift=1.65cm] (a) {$a$}
		edge [pre] (start);
		\node[transition,right of=a] (b) {$b$}
		edge [pre] (a);
		\node[transition,right of=b,yshift=-1.65cm] (end) {$\square$}
		edge [pre] (b);
		\node[transition,below of=a] (c) {$c$}
		edge [pre] (start);
		\node[transition,right of=c] (x) {$x$}
		edge [pre] (c)
		edge [post] (end);
		\node[transition,below of=c] (d) {$d$}
		edge [pre] (start);
		\node[transition,right of=d] (y) {$y$}
		edge [pre] (d)
		edge [post] (end);
		\node[transition,below of=d] (e) {$e$}
		edge [pre] (start);
		\node[transition,right of=e] (z) {$z$}
		edge [pre] (e)
		edge [post] (end);
		\node[transition,yshift=1.1cm] (g) at ($0.5*(a) + 0.5*(b)$) {$g$}
		edge [pre] (a)
		edge [post] (b);
		\node[transition,below of=e] (h) {$h$}
		edge [pre] (start);
		\node[transition,right of=h] (i) {$i$}
		edge [pre] (h)
		edge [post] (end);
	\end{tikzpicture}
	\caption{The directly follows graphs for the logs $L_1, L_2, L_3$ from the example in \cref{theo:dfg-depth-entries}. $G_1$ is the DFG for $L_1$, $G_2$ the one for $L_2$ and $G_3$ the one for $L_3$.}
	\label{fig:dfg-depth}
\end{figure}
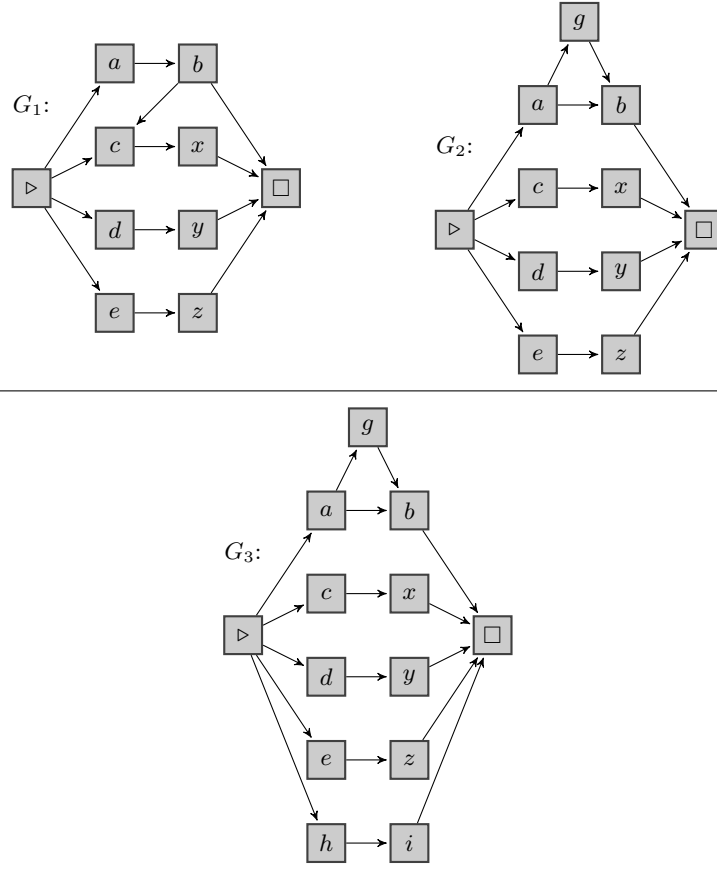
These graphs have the following complexity scores:
\begin{itemize}
	\item[•] $\depth(G_1) = 1$,
	\item[•] $\depth(G_2) = 2$,
	\item[•] $\depth(G_3) = 1$,
\end{itemize}
so these graphs fulfill $\depth(G_1) < \depth(G_2)$, $\depth(G_2) > \depth(G_3)$, and $\depth(G_1) = \depth(G_3)$.
But the event logs $L_1, L_2, L_3$ have the following log complexity scores:

\medskip

\begin{center}
	\def\pad{\hspace*{1.5mm}}
	\begin{tabular}{|c|c|c|c|c|c|c|c|c|c|c|}\hline
		 & $\magnitude$ & $\variety$ & $\support$ & $\tlavg$ & $\tlmax$ & $\levelofdetail$ & $\numberofties$ & $\lempelziv$ & $\numberuniquetraces$ & $\percentageuniquetraces$ \\ \hline
		$L_1$ & $\pad 8 \pad$ & $\pad 8 \pad$ & $\pad 4 \pad$ & $\pad 2 \pad$ & $\pad 2 \pad$ & $\pad 4 \pad$ & $\pad 4 \pad$ & $\pad 8 \pad$ & $\pad 4 \pad$ & $\pad 1 $ \\ \hline
		$L_2$ & $\pad 13 \pad$ & $\pad 9 \pad$ & $\pad 6 \pad$ & $\pad 2.1667 \pad$ & $\pad 3 \pad$ & $\pad 5 \pad$ & $\pad 6 \pad$ & $\pad 10 \pad$ & $\pad 5 \pad$ & $\pad 0.8333 \pad$ \\ \hline
		$L_3$ & $\pad 19 \pad$ & $\pad 11 \pad$ & $\pad 8 \pad$ & $\pad 2.375 \pad$ & $\pad 4 \pad$ & $\pad 8 \pad$ & $\pad 8 \pad$ & $\pad 14 \pad$ & $\pad 7 \pad$ & $\pad 0.875 \pad$ \\ \hline
	\end{tabular}
	
	\medskip
	
	\begin{tabular}{|c|c|c|c|c|c|c|c|c|} \hline
		 & $\structure$ & $\affinity$ & $\deviationfromrandom$ & $\avgdist$ & $\varentropy$ & $\normvarentropy$ & $\seqentropy$ & $\normseqentropy$ \\ \hline
		$L_1$ & $\pad 2 \pad$ & $\pad 0 \pad$ & $\pad 0.5159 \pad$ & $\pad 4 \pad$ & $\pad 11.0904 \pad$ & $\pad 0.6667 \pad$ & $\pad 11.0904 \pad$ & $\pad 0.6667 \pad$ \\ \hline
		$L_2$ & $\pad 2.1667 \pad$ & $\pad 0.0667 \pad$ & $\pad 0.5861 \pad$ & $\pad 3.5333 \pad$ & $\pad 16.0944 \pad$ & $\pad 0.699 \pad$ & $\pad 19.752 \pad$ & $\pad 0.5924 \pad$ \\ \hline
		$L_3$ & $\pad 2.375 \pad$ & $\pad 0.0714 \pad$ & $\pad 0.6143 \pad$ & $\pad 3.75 \pad$ & $\pad 24.4702 \pad$ & $\pad 0.6623 \pad$ & $\pad 29.2378 \pad$ & $\pad 0.5226 \pad$ \\ \hline
	\end{tabular}
\end{center}
Therefore, $\mathcal{C}^L(L_1) < \mathcal{C}^L(L_2) < \mathcal{C}^L(L_3)$ for any event log complexity measure $\mathcal{C}^L \in (\loc \setminus \{\percentageuniquetraces, \avgdist, \normvarentropy, \normseqentropy\})$.
For $\percentageuniquetraces$, $\avgdist$, $\normvarentropy$, and $\normseqentropy$, consider the following event logs that have the same directly follows graphs as the ones shown in \cref{fig:dfg-depth}:
\begin{align*}
	L_1 &= [\langle a,b \rangle^{7}, \langle c,x \rangle, \langle d,y \rangle, \langle e,z \rangle] \\
	L_2 &= L_1 + [\langle a,g,b \rangle] \\
	L_3 &= L_2 + [\langle b,c \rangle, \langle h,i \rangle]
\end{align*}
These event logs fulfill:
\begin{itemize}
	\item[•] $\percentageuniquetraces(L_1) = 0.4 < \percentageuniquetraces(L_2) \approx 0.4545 < \percentageuniquetraces(L_3) \approx 0.5385$,
	\item[•] $\avgdist(L_1) \approx 2.1333 < \avgdist(L_2) \approx 2.1455 < \avgdist(L_3) \approx 2.4872$,
	\item[•] $\normvarentropy(L_1) \approx 0.6667 < \normvarentropy(L_2) \approx 0.699 < \normvarentropy(L_3) \approx 0.7374$, and
	\item[•] $\normseqentropy(L_1) \approx 0.3139 < \normseqentropy(L_2) \approx 0.3598 < \normseqentropy(L_3) \approx 0.4501$.
\end{itemize}
Since the directly follows graphs are the same as in \cref{fig:dfg-depth}, their model complexity scores did not change.
Thus, we were able to show that $(\mathcal{C}^L, \depth) \in \norel$ for any event log complexity measure $\mathcal{C}^L \in \loc$. \hfill$\square$
\end{proof}

\begin{theorem}
\label{theo:dfg-cyc-entries}
$(\mathcal{C}^L, \cyclicity) \in \norel$ for any log complexity measure $\mathcal{C}^L \in \loc$.
\end{theorem}
\begin{proof}
Consider the following event logs:
\begin{align*}
	L_1 &= [\langle a \rangle, \langle a,b,c,c,d \rangle, \langle a,b,b,c,d \rangle] \\
	L_2 &= L_1 + [\langle a,a,b,b,c,c,d,d,e \rangle] \\
	L_3 &= L_2 + [\langle a,b,b,c,c,d \rangle, \langle a,a,a,b,b,b,c,c,c,d,d,d \rangle, \langle v,w,x,x,y,z \rangle]
\end{align*}
\cref{fig:dfg-cyclicity} shows the directly follows graphs $G_1, G_2, G_3$ for the event logs $L_1, L_2, L_3$.
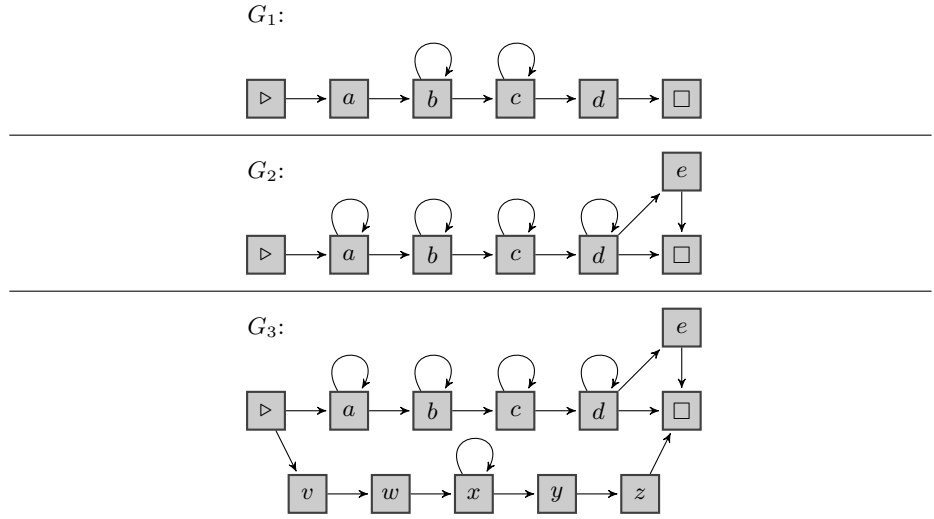
\begin{figure}[ht]
	\centering
	\begin{tikzpicture}[node distance = 1.1cm,>=stealth',bend angle=0,auto]
		\node[transition] (start) {$\triangleright$};
		\node[above of=start] {$G_1$:};
		\node[transition,right of=start] (a) {$a$}
		edge [pre] (start);
		\node[transition,right of=a] (b) {$b$}
		edge [pre] (a)
		edge [pre,loop,out=60,in=120,looseness=6] (b);
		\node[transition,right of=b] (c) {$c$}
		edge [pre] (b)
		edge [pre,loop,out=60,in=120,looseness=6] (c);
		\node[transition,right of=c] (d) {$d$}
		edge [pre] (c);
		\node[transition,right of=d] (end) {$\square$}
		edge [pre] (d);
	\end{tikzpicture}
	
	\medskip
	\hrule
	\medskip
	
	\begin{tikzpicture}[node distance = 1.1cm,>=stealth',bend angle=0,auto]
		\node[transition] (start) {$\triangleright$};
		\node[above of=start] {$G_2$:};
		\node[transition,right of=start] (a) {$a$}
		edge [pre] (start)
		edge [pre,loop,out=60,in=120,looseness=6] (a);
		\node[transition,right of=a] (b) {$b$}
		edge [pre] (a)
		edge [pre,loop,out=60,in=120,looseness=6] (b);
		\node[transition,right of=b] (c) {$c$}
		edge [pre] (b)
		edge [pre,loop,out=60,in=120,looseness=6] (c);
		\node[transition,right of=c] (d) {$d$}
		edge [pre] (c)
		edge [pre,loop,out=60,in=120,looseness=6] (d);
		\node[transition,right of=d] (end) {$\square$}
		edge [pre] (d);
		\node[transition,above of=end] (e) {$e$}
		edge [pre] (d)
		edge [post] (end);
	\end{tikzpicture}
	
	\medskip
	\hrule
	\medskip
	
	\begin{tikzpicture}[node distance = 1.1cm,>=stealth',bend angle=0,auto]
		\node[transition] (start) {$\triangleright$};
		\node[above of=start] {$G_3$:};
		\node[transition,right of=start] (a) {$a$}
		edge [pre] (start)
		edge [pre,loop,out=60,in=120,looseness=6] (a);
		\node[transition,right of=a] (b) {$b$}
		edge [pre] (a)
		edge [pre,loop,out=60,in=120,looseness=6] (b);
		\node[transition,right of=b] (c) {$c$}
		edge [pre] (b)
		edge [pre,loop,out=60,in=120,looseness=6] (c);
		\node[transition,right of=c] (d) {$d$}
		edge [pre] (c)
		edge [pre,loop,out=60,in=120,looseness=6] (d);
		\node[transition,right of=d] (end) {$\square$}
		edge [pre] (d);
		\node[transition,above of=end] (e) {$e$}
		edge [pre] (d)
		edge [post] (end);
		\node (dummy) at ($0.5*(start) + 0.5*(a)$) {};
		\node[transition,below of=dummy] (v) {$v$}
		edge [pre] (start);
		\node[transition,right of=v] (w) {$w$}
		edge [pre] (v);
		\node[transition,right of=w] (x) {$x$}
		edge [pre] (w)
		edge [pre,loop,out=60,in=120,looseness=6] (x);
		\node[transition,right of=x] (y) {$y$}
		edge [pre] (x);
		\node[transition,right of=y] (z) {$z$}
		edge [pre] (y)
		edge [post] (end);
	\end{tikzpicture}
	\caption{The directly follows graphs for the logs $L_1, L_2, L_3$ from the example in \cref{theo:dfg-cyc-entries}. $G_1$ is the DFG for $L_1$, $G_2$ the one for $L_2$ and $G_3$ the one for $L_3$.}
	\label{fig:dfg-cyclicity}
\end{figure}
These graphs have the following complexity scores:
\begin{itemize}
	\item[•] $\cyclicity(G_1) = 0.5$,
	\item[•] $\cyclicity(G_2) = 0.8$,
	\item[•] $\cyclicity(G_3) = 0.5$,
\end{itemize}
so $\cyclicity(G_1) < \cyclicity(G_2)$, $\cyclicity(G_2) > \cyclicity(G_3)$, and $\cyclicity(G_1) = \cyclicity(G_3)$.
But the event logs $L_1, L_2, L_3$ have the following log complexity scores:
\begin{center}
	\def\pad{\hspace*{1.5mm}}
	\begin{tabular}{|c|c|c|c|c|c|c|c|c|c|c|}\hline
		 & $\magnitude$ & $\variety$ & $\support$ & $\tlavg$ & $\tlmax$ & $\levelofdetail$ & $\numberofties$ & $\lempelziv$ & $\numberuniquetraces$ & $\percentageuniquetraces$ \\ \hline
		$L_1$ & $\pad 11 \pad$ & $\pad 4 \pad$ & $\pad 3 \pad$ & $\pad 3.6667 \pad$ & $\pad 5 \pad$ & $\pad 2 \pad$ & $\pad 3 \pad$ & $\pad 5 \pad$ & $\pad 3 \pad$ & $\pad 1 \pad$ \\ \hline
		$L_2$ & $\pad 20 \pad$ & $\pad 5 \pad$ & $\pad 4 \pad$ & $\pad 5 \pad$ & $\pad 9 \pad$ & $\pad 3 \pad$ & $\pad 4 \pad$ & $\pad 9 \pad$ & $\pad 4 \pad$ & $\pad 1 \pad$ \\ \hline
		$L_3$ & $\pad 44 \pad$ & $\pad 10 \pad$ & $\pad 7 \pad$ & $\pad 6.2857 \pad$ & $\pad 12 \pad$ & $\pad 4 \pad$ & $\pad 8 \pad$ & $\pad 20 \pad$ & $\pad 7 \pad$ & $\pad 1 \pad$ \\ \hline
	\end{tabular}
	
	\medskip
	
	\begin{tabular}{|c|c|c|c|c|c|c|c|c|} \hline
		 & $\structure$ & $\affinity$ & $\deviationfromrandom$ & $\avgdist$ & $\varentropy$ & $\normvarentropy$ & $\seqentropy$ & $\normseqentropy$ \\ \hline
		$L_1$ & $\pad 3 \pad$ & $\pad 0.2 \pad$ & $\pad 0.6047 \pad$ & $\pad 3.3333 \pad$ & $\pad 5.2925 \pad$ & $\pad 0.3181 \pad$ & $\pad 6.4455 \pad$ & $\pad 0.2444 \pad$ \\ \hline
		$L_2$ & $\pad 3.5 \pad$ & $\pad 0.2667 \pad$ & $\pad 0.6707 \pad$ & $\pad 4.3333 \pad$ & $\pad 16.3829 \pad$ & $\pad 0.3693 \pad$ & $\pad 20.2083 \pad$ & $\pad 0.3373 \pad$ \\ \hline
		$L_3$ & $\pad 3.8571 \pad$ & $\pad 0.3122 \pad$ & $\pad 0.6856 \pad$ & $\pad 6.9524 \pad$ & $\pad 56.755 \pad$ & $\pad 0.4734 \pad$ & $\pad 73.7006 \pad$ & $\pad 0.4426 \pad$ \\ \hline
	\end{tabular}
\end{center}
Therefore, $\mathcal{C}^L(L_1) < \mathcal{C}^L(L_2) < \mathcal{C}^L(L_3)$ for any event log complexity measure $\mathcal{C}^L \in (\loc \setminus \{\percentageuniquetraces\})$.
For $\percentageuniquetraces$, consider the following event logs.
\begin{align*}
	L_1 &= [\langle a \rangle^{2}, \langle a,b,c,c,d \rangle, \langle a,b,b,c,d \rangle] \\
	L_2 &= [\langle a,a,b,b,c,c,d,d,e \rangle] \\
	L_3 &= [\langle a,b,b,c,c,d \rangle, \langle a,a,a,b,b,b,c,c,c,d,d,d \rangle, \langle v,w,x,x,y,z \rangle]
\end{align*}
Since only the frequency of the trace $\langle a \rangle$ changed in contrast to the previous event logs, the directly follows graphs $G_1, G_2, G_3$ for the new event logs $L_1, L_2, L_3$ are the same as the ones shown in \cref{fig:dfg-cyclicity}.
But since the new event logs fulfill $\percentageuniquetraces(L_1) = 0.75 < \percentageuniquetraces(L_2) = 0.8 < \percentageuniquetraces(L_3) = 0.875$, we have shown that $(\mathcal{C}^L, \cyclicity) \in \norel$ for any event log complexity measure $\mathcal{C}^L \in \loc$. \hfill$\square$
\end{proof}

\begin{theorem}
\label{theo:dfg-cnc-entries}
$(\mathcal{C}^L, \netconn) \in \norel$ for any log complexity measure $\mathcal{C}^L \in \loc$.
\end{theorem}
\begin{proof}
Consider the following event logs:
\begin{align*}
	L_1 &= [\langle a,a,b,b,c,c,d,d \rangle^{2}, \langle b,c,d \rangle^{3}] \\
	L_2 &= L_1 + [\langle b,c,d \rangle, \langle a,a,b,b,c,c,d,d,e,e \rangle, \langle a,b,c,d,e \rangle] \\
	L_3 &= L_2 + [\langle a,a,a,b,b,b,c,c,c,d,d,d,e,e,e \rangle, \langle u,v,x,x,y,z \rangle]
\end{align*}
\cref{fig:dfg-netconn} shows the directly follows graphs $G_1, G_2, G_3$ for the event logs $L_1, L_2, L_3$.
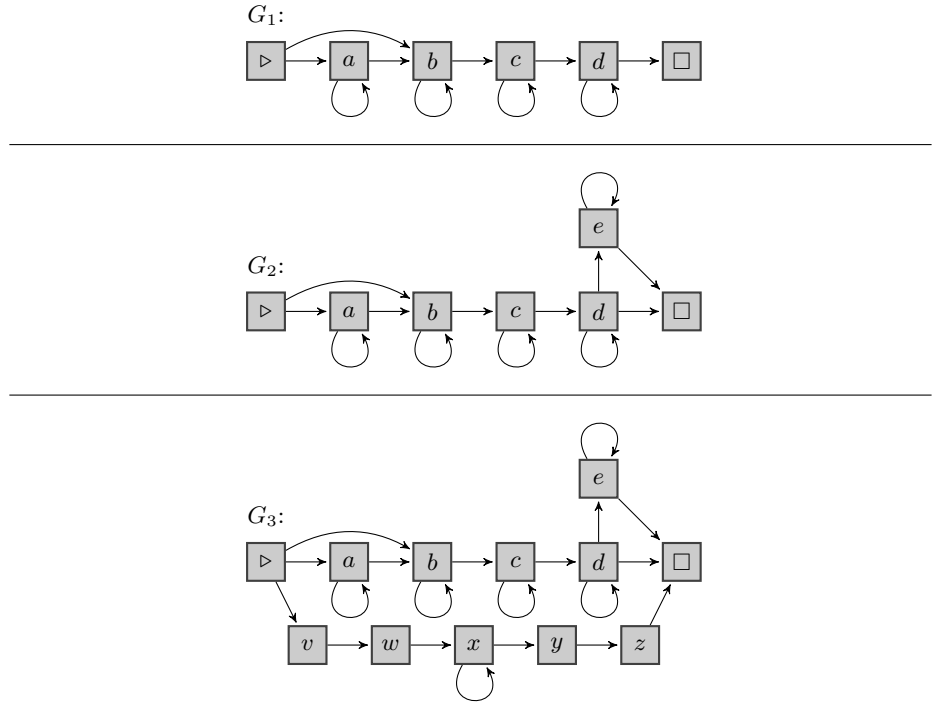
\begin{figure}[ht]
	\centering
	\begin{tikzpicture}[node distance = 1.1cm,>=stealth',bend angle=0,auto]
		\node[transition] (start) {$\triangleright$};
		\node[above of=start,yshift=-0.5cm] {$G_1$:};
		\node[transition,right of=start] (a) {$a$}
		edge [pre] (start)
		edge [pre,loop,out=300,in=240,looseness=6] (a);
		\node[transition,right of=a] (b) {$b$}
		edge [pre,bend right=30] (start)
		edge [pre] (a)
		edge [pre,loop,out=300,in=240,looseness=6] (b);
		\node[transition,right of=b] (c) {$c$}
		edge [pre] (b)
		edge [pre,loop,out=300,in=240,looseness=6] (c);
		\node[transition,right of=c] (d) {$d$}
		edge [pre] (c)
		edge [pre,loop,out=300,in=240,looseness=6] (d);
		\node[transition,right of=d] (end) {$\square$}
		edge [pre] (d);
	\end{tikzpicture}
	
	\medskip
	\hrule
	\medskip
	
	\begin{tikzpicture}[node distance = 1.1cm,>=stealth',bend angle=0,auto]
		\node[transition] (start) {$\triangleright$};
		\node[above of=start,yshift=-0.5cm] {$G_2$:};
		\node[transition,right of=start] (a) {$a$}
		edge [pre] (start)
		edge [pre,loop,out=300,in=240,looseness=6] (a);
		\node[transition,right of=a] (b) {$b$}
		edge [pre,bend right=30] (start)
		edge [pre] (a)
		edge [pre,loop,out=300,in=240,looseness=6] (b);
		\node[transition,right of=b] (c) {$c$}
		edge [pre] (b)
		edge [pre,loop,out=300,in=240,looseness=6] (c);
		\node[transition,right of=c] (d) {$d$}
		edge [pre] (c)
		edge [pre,loop,out=300,in=240,looseness=6] (d);
		\node[transition,right of=d] (end) {$\square$}
		edge [pre] (d);
		\node[transition,above of=d] (e) {$e$}
		edge [pre] (d)
		edge [pre,loop,out=60,in=120,looseness=6] (e)
		edge [post] (end);
	\end{tikzpicture}
	
	\medskip
	\hrule
	\medskip
	
	\begin{tikzpicture}[node distance = 1.1cm,>=stealth',bend angle=0,auto]
		\node[transition] (start) {$\triangleright$};
		\node[above of=start,yshift=-0.5cm] {$G_3$:};
		\node[transition,right of=start] (a) {$a$}
		edge [pre] (start)
		edge [pre,loop,out=300,in=240,looseness=6] (a);
		\node[transition,right of=a] (b) {$b$}
		edge [pre,bend right=30] (start)
		edge [pre] (a)
		edge [pre,loop,out=300,in=240,looseness=6] (b);
		\node[transition,right of=b] (c) {$c$}
		edge [pre] (b)
		edge [pre,loop,out=300,in=240,looseness=6] (c);
		\node[transition,right of=c] (d) {$d$}
		edge [pre] (c)
		edge [pre,loop,out=300,in=240,looseness=6] (d);
		\node[transition,right of=d] (end) {$\square$}
		edge [pre] (d);
		\node[transition,above of=d] (e) {$e$}
		edge [pre] (d)
		edge [pre,loop,out=60,in=120,looseness=6] (e)
		edge [post] (end);
		\node (dummy) at ($0.5*(start) + 0.5*(a)$) {};
		\node[transition,below of=dummy] (v) {$v$}
		edge [pre] (start);
		\node[transition,right of=v] (w) {$w$}
		edge [pre] (v);
		\node[transition,right of=w] (x) {$x$}
		edge [pre] (w)
		edge [pre,loop,out=300,in=240,looseness=6] (x);
		\node[transition,right of=x] (y) {$y$}
		edge [pre] (x);
		\node[transition,right of=y] (z) {$z$}
		edge [pre] (y)
		edge [post] (end);
	\end{tikzpicture}
	\caption{The directly follows graphs for the logs $L_1, L_2, L_3$ from the example in \cref{theo:dfg-cnc-entries}. $G_1$ is the DFG for $L_1$, $G_2$ the one for $L_2$ and $G_3$ the one for $L_3$.}
	\label{fig:dfg-netconn}
\end{figure}
These graphs have the following complexity scores:
\begin{itemize}
	\item[•] $\netconn(G_1) \approx 1.6667$,
	\item[•] $\netconn(G_2) \approx 1.8571$,
	\item[•] $\netconn(G_3) \approx 1.6667$,
\end{itemize}
so these graphs fulfill $\netconn(G_1) < \netconn(G_2)$, $\netconn(G_2) > \netconn(G_3)$, and $\netconn(G_1) = \netconn(G_3)$.
But the event logs $L_1, L_2, L_3$ have the following log complexity scores:
\begin{center}
	\def\pad{\hspace*{1.5mm}}
	\begin{tabular}{|c|c|c|c|c|c|c|c|c|c|c|}\hline
		 & $\magnitude$ & $\variety$ & $\support$ & $\tlavg$ & $\tlmax$ & $\levelofdetail$ & $\numberofties$ & $\lempelziv$ & $\numberuniquetraces$ & $\percentageuniquetraces$ \\ \hline
		$L_1$ & $\pad 25 \pad$ & $\pad 4 \pad$ & $\pad 5 \pad$ & $\pad 5 \pad$ & $\pad 8 \pad$ & $\pad 2 \pad$ & $\pad 3 \pad$ & $\pad 13 \pad$ & $\pad 2 \pad$ & $\pad 0.4 \pad$ \\ \hline
		$L_2$ & $\pad 43 \pad$ & $\pad 5 \pad$ & $\pad 8 \pad$ & $\pad 5.375 \pad$ & $\pad 10 \pad$ & $\pad 4 \pad$ & $\pad 4 \pad$ & $\pad 20 \pad$ & $\pad 4 \pad$ & $\pad 0.5 \pad$ \\ \hline
		$L_3$ & $\pad 64 \pad$ & $\pad 10 \pad$ & $\pad 10 \pad$ & $\pad 6.4 \pad$ & $\pad 15 \pad$ & $\pad 5 \pad$ & $\pad 8 \pad$ & $\pad 30 \pad$ & $\pad 6 \pad$ & $\pad 0.6 \pad$ \\ \hline
	\end{tabular}
	
	\medskip
	
	\begin{tabular}{|c|c|c|c|c|c|c|c|c|} \hline
		 & $\structure$ & $\affinity$ & $\deviationfromrandom$ & $\avgdist$ & $\varentropy$ & $\normvarentropy$ & $\seqentropy$ & $\normseqentropy$ \\ \hline
		$L_1$ & $\pad 3.4 \pad$ & $\pad 0.5714 \pad$ & $\pad 0.6646 \pad$ & $\pad 3 \pad$ & $\pad 6.4455 \pad$ & $\pad 0.2444 \pad$ & $\pad 16.3355 \pad$ & $\pad 0.203 \pad$ \\ \hline
		$L_2$ & $\pad 3.75 \pad$ & $\pad 0.533 \pad$ & $\pad 0.6668 \pad$ & $\pad 3.3929 \pad$ & $\pad 16.2978 \pad$ & $\pad 0.3384 \pad$ & $\pad 37.38 \pad$ & $\pad 0.2311 \pad$ \\ \hline
		$L_3$ & $\pad 4 \pad$ & $\pad 0.4181 \pad$ & $\pad 0.6897 \pad$ & $\pad 6.3111 \pad$ & $\pad 53.0449 \pad$ & $\pad 0.4112 \pad$ & $\pad 89.058 \pad$ & $\pad 0.3346 \pad$ \\ \hline
	\end{tabular}
\end{center}
Therefore, $\mathcal{C}^L(L_1) < \mathcal{C}^L(L_2) < \mathcal{C}^L(L_3)$ for any event log complexity measure $\mathcal{C}^L \in (\loc \setminus \{\affinity\})$.
For $\affinity$, consider the following event logs:
\begin{align*}
	L_1 &= [\langle a,a,b,b,c,c,d,d \rangle, \langle b,c,d \rangle] \\
	L_2 &= L_1 + [\langle a,a,b,b,c,c,d,d,e,e \rangle] \\
	L_3 &= L_2 + [\langle a,a,a,b,b,b,c,c,c,d,d,d,e,e,e \rangle^{3}, \langle u,v,x,x,y,z \rangle]
\end{align*}
Since only the frequencies of traces changed in contrast to the previous event logs, the directly follows graphs $G_1, G_2, G_3$ for the new event logs $L_1, L_2, L_3$ are the same as the ones shown in \cref{fig:dfg-netconn}.
But since the new event logs fulfill $\affinity(L_1) \approx 0.2857 < \affinity(L_2) \approx 0.4286 < \affinity(L_3) \approx 0.4898$, we have shown that $(\mathcal{C}^L, \netconn) \in \norel$ for any log complexity measure $\mathcal{C}^L \in \loc$. \hfill$\square$
\end{proof}

\begin{theorem}
\label{theo:dfg-dens-entries}
$(\mathcal{C}^L, \density) \in \norel$ for any log complexity measure $\mathcal{C}^L \in \loc$.
\end{theorem}
\begin{proof}
Consider the following event logs:
\begin{align*}
	L_1 &= [\langle a \rangle, \langle a,b,c,d \rangle] \\
	L_2 &= L_1 + [\langle a,b,b,c,c,d,d,e,e \rangle] \\
	L_3 &= L_2 + [\langle a,e \rangle, \langle a,b,b,c,b,c,d,d,e,e \rangle^{2}, \langle v,v,x,x,y,x,y,y,z,z \rangle]
\end{align*}
\cref{fig:dfg-density} shows the directly follows graphs $G_1, G_2, G_3$ for the event logs $L_1, L_2, L_3$.
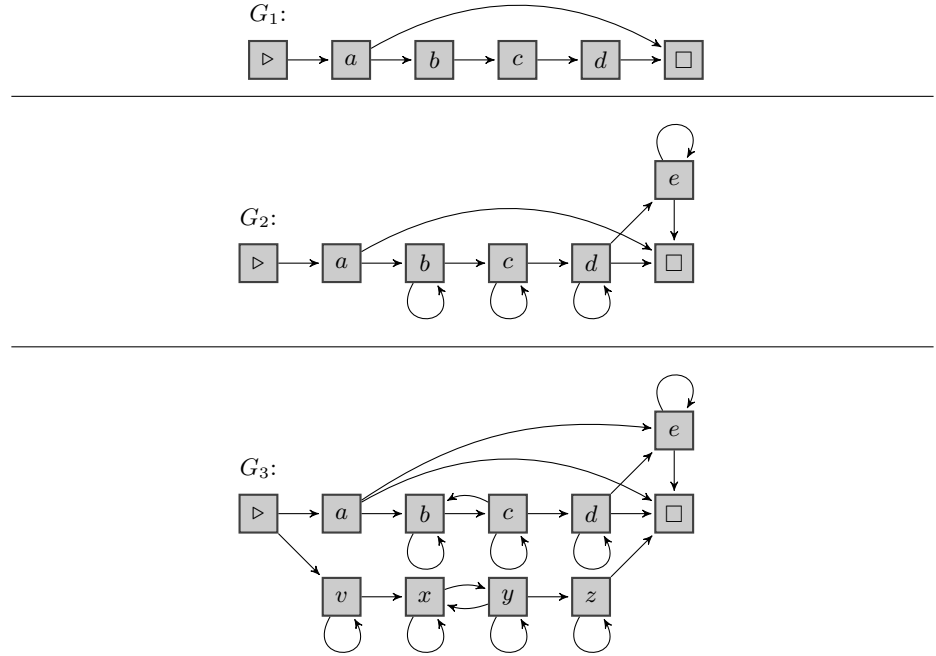
\begin{figure}[ht]
	\centering
	\begin{tikzpicture}[node distance = 1.1cm,>=stealth',bend angle=0,auto]
		\node[transition] (start) {$\triangleright$};
		\node[above of=start,yshift=-0.5cm] {$G_1$:};
		\node[transition,right of=start] (a) {$a$}
		edge [pre] (start);
		\node[transition,right of=a] (b) {$b$}
		edge [pre] (a);
		\node[transition,right of=b] (c) {$c$}
		edge [pre] (b);
		\node[transition,right of=c] (d) {$d$}
		edge [pre] (c);
		\node[transition,right of=d] (end) {$\square$}
		edge [pre,bend right=30] (a)
		edge [pre] (d);
	\end{tikzpicture}
	
	\medskip
	\hrule
	\medskip
	
	\begin{tikzpicture}[node distance = 1.1cm,>=stealth',bend angle=0,auto]
		\node[transition] (start) {$\triangleright$};
		\node[above of=start,yshift=-0.5cm] {$G_2$:};
		\node[transition,right of=start] (a) {$a$}
		edge [pre] (start);
		\node[transition,right of=a] (b) {$b$}
		edge [pre] (a)
		edge [pre,loop,out=300,in=240,looseness=6] (b);
		\node[transition,right of=b] (c) {$c$}
		edge [pre] (b)
		edge [pre,loop,out=300,in=240,looseness=6] (c);
		\node[transition,right of=c] (d) {$d$}
		edge [pre] (c)
		edge [pre,loop,out=300,in=240,looseness=6] (d);
		\node[transition,right of=d] (end) {$\square$}
		edge [pre,bend right=30] (a)
		edge [pre] (d);
		\node[transition,above of=end] (e) {$e$}
		edge [pre] (d)
		edge [pre,loop,out=60,in=120,looseness=6] (e)
		edge [post] (end);
	\end{tikzpicture}
	
	\medskip
	\hrule
	\medskip
	
	\begin{tikzpicture}[node distance = 1.1cm,>=stealth',bend angle=0,auto]
		\node[transition] (start) {$\triangleright$};
		\node[above of=start,yshift=-0.5cm] {$G_3$:};
		\node[transition,right of=start] (a) {$a$}
		edge [pre] (start);
		\node[transition,right of=a] (b) {$b$}
		edge [pre] (a)
		edge [pre,loop,out=300,in=240,looseness=6] (b);
		\node[transition,right of=b] (c) {$c$}
		edge [pre] (b)
		edge [post,bend right=30] (b)
		edge [pre,loop,out=300,in=240,looseness=6] (c);
		\node[transition,right of=c] (d) {$d$}
		edge [pre] (c)
		edge [pre,loop,out=300,in=240,looseness=6] (d);
		\node[transition,right of=d] (end) {$\square$}
		edge [pre,bend right=30] (a)
		edge [pre] (d);
		\node[transition,above of=end] (e) {$e$}
		edge [pre,bend right=20] (a)
		edge [pre] (d)
		edge [pre,loop,out=60,in=120,looseness=6] (e)
		edge [post] (end);
		\node[transition,below of=a] (v) {$v$}
		edge [pre] (start)
		edge [pre,loop,out=300,in=240,looseness=6] (v);
		\node[transition,right of=v] (x) {$x$}
		edge [pre] (v)
		edge [pre,loop,out=300,in=240,looseness=6] (x);
		\node[transition,right of=x] (y) {$y$}
		edge [pre,bend right=20] (x)
		edge [post,bend left=20] (x)
		edge [pre,loop,out=300,in=240,looseness=6] (y);
		\node[transition,right of=y] (z) {$z$}
		edge [pre] (y)
		edge [pre,loop,out=300,in=240,looseness=6] (z)
		edge [post] (end);
	\end{tikzpicture}
	\caption{The directly follows graphs for the logs $L_1, L_2, L_3$ from the example in \cref{theo:dfg-dens-entries}. $G_1$ is the DFG for $L_1$, $G_2$ the one for $L_2$ and $G_3$ the one for $L_3$.}
	\label{fig:dfg-density}
\end{figure}
These graphs have the following complexity scores:
\begin{itemize}
	\item[•] $\density(G_1) = 0.24$,
	\item[•] $\density(G_2) \approx 0.3333$,
	\item[•] $\density(G_3) = 0.24$,
\end{itemize}
so these graphs fulfill $\density(G_1) < \density(G_2)$, $\density(G_2) > \density(G_3)$, and $\density(G_1) = \density(G_3)$.
But the event logs $L_1, L_2, L_3$ have the following log complexity scores:
\begin{center}
	\def\pad{\hspace*{1.5mm}}
	\begin{tabular}{|c|c|c|c|c|c|c|c|c|c|c|}\hline
		 & $\magnitude$ & $\variety$ & $\support$ & $\tlavg$ & $\tlmax$ & $\levelofdetail$ & $\numberofties$ & $\lempelziv$ & $\numberuniquetraces$ & $\percentageuniquetraces$ \\ \hline
		$L_1$ & $\pad 5 \pad$ & $\pad 4 \pad$ & $\pad 2 \pad$ & $\pad 2.5 \pad$ & $\pad 4 \pad$ & $\pad 2 \pad$ & $\pad 3 \pad$ & $\pad 4 \pad$ & $\pad 2 \pad$ & $\pad 1 \pad$ \\ \hline
		$L_2$ & $\pad 14 \pad$ & $\pad 5 \pad$ & $\pad 3 \pad$ & $\pad 4.6667 \pad$ & $\pad 9 \pad$ & $\pad 3 \pad$ & $\pad 4 \pad$ & $\pad 8 \pad$ & $\pad 3 \pad$ & $\pad 1 \pad$ \\ \hline
		$L_3$ & $\pad 46 \pad$ & $\pad 9 \pad$ & $\pad 7 \pad$ & $\pad 6.5714 \pad$ & $\pad 10 \pad$ & $\pad 5 \pad$ & $\pad 6 \pad$ & $\pad 23 \pad$ & $\pad 6 \pad$ & $\pad 0.8571 \pad$ \\ \hline
	\end{tabular}
	
	\medskip
	
	\begin{tabular}{|c|c|c|c|c|c|c|c|c|} \hline
		 & $\structure$ & $\affinity$ & $\deviationfromrandom$ & $\avgdist$ & $\varentropy$ & $\normvarentropy$ & $\seqentropy$ & $\normseqentropy$ \\ \hline
		$L_1$ & $\pad 2.5 \pad$ & $\pad 0 \pad$ & $\pad 0.4796 \pad$ & $\pad 3 \pad$ & $\pad 0 \pad$ & $\pad 0 \pad$ & $\pad 0 \pad$ & $\pad 0 \pad$ \\ \hline
		$L_2$ & $\pad 3.3333 \pad$ & $\pad 0.125 \pad$ & $\pad 0.683 \pad$ & $\pad 5.3333 \pad$ & $\pad 7.2103 \pad$ & $\pad 0.2734 \pad$ & $\pad 9.7041 \pad$ & $\pad 0.2626 \pad$ \\ \hline
		$L_3$ & $\pad 3.7143 \pad$ & $\pad 0.1753 \pad$ & $\pad 0.7408 \pad$ & $\pad 8.1905 \pad$ & $\pad 40.3588 \pad$ & $\pad 0.4326 \pad$ & $\pad 67.077 \pad$ & $\pad 0.3809 \pad$ \\ \hline
	\end{tabular}
\end{center}
Thus, we have $\mathcal{C}^L(L_1) < \mathcal{C}^L(L_2) < \mathcal{C}^L(L_3)$ for any event log complexity measure $\mathcal{C}^L \in (\loc \setminus \{\percentageuniquetraces\})$.
For $\percentageuniquetraces$, consider the following event logs:
\begin{align*}
	L_1 &= [\langle a \rangle^{4}, \langle a,b,c,d \rangle] \\
	L_2 &= L_1 + [\langle a,b,b,c,c,d,d,e,e \rangle] \\
	L_3 &= L_2 + [\langle a,e \rangle, \langle a,b,b,c,b,c,d,d,e,e \rangle^{2}, \langle v,v,x,x,y,x,y,y,z,z \rangle]
\end{align*}
Since only the frequencies of traces changed in contrast to the previous event logs, the directly follows graphs $G_1, G_2, G_3$ for the new event logs $L_1, L_2, L_3$ are the same as the ones shown in \cref{fig:dfg-density}.
But since the new event logs fulfill $\percentageuniquetraces(L_1) = 0.4 < \percentageuniquetraces(L_2) = 0.5 < \percentageuniquetraces(L_3) = 0.6$, we have shown that $(\mathcal{C}^L, \density) \in \norel$ for any log complexity measure $\mathcal{C}^L \in \loc$. \hfill$\square$
\end{proof}

Except for the size and the control flow complexity, none of the existing log complexity measures directly predict the model complexity of the directly follows graph.
The maximum connector degree and the diameter of two directly follows graphs $G_1, G_2$ for event logs $L_1, L_2$ are always increasing or staying unchanged when $L_1 \sqsubset L_2$, so even for these measures, we did not find a direct connection between log and model complexity.
In the following, we will analyze how the model complexity scores of the directly follows graph can be described using properties of the underlying event log.
Thus, let $G = (V,E)$ be the directly follows graph for an event log $L$ over a set of activities $A$.
Since $G$ contains exactly one node for every activity in $L$, as well as the two special nodes $\triangleright$ and $\square$, the amount of nodes in $G$ is $|V| = 2 + \variety(L)$.
Furthermore, by definition of the directly follows graph, we know that $G$ has $|E| = |>_L| + |A_I| + |A_O|$ edges, where $A_I := \{a \in A \mid \exists \sigma \in L: \sigma(1) = a\}$ and $A_O := \{a \in A \mid \exists \sigma \in L: \sigma(|\sigma|) = a\}$.
Using the relation $>_L$, we define the following sets for activities $a, b \in A$:
\begin{align*}
	&\succ_L(a) := \{b \in A \mid a >_L b\} \cup \{\square \mid a \in A_O\} \\
	&\succ_L^{-1}(b) := \{a \in A \mid a >_L b\} \cup \{\triangleright \mid b \in A_I\}
\end{align*}
Furthermore, to keep the formulas as simple as possible, we define
\begin{align*}
	S_{\texttt{xor}}^G &= \{a \in A \mid 1 < |\succ_L(a)|\} \\
	J_{\texttt{xor}}^G &= \{a \in A \mid 1 < |\succ_L^{-1}(a)|\} \\
	C_{\texttt{xor}}^G &= S_{\texttt{xor}}^G \cup J_{\texttt{xor}}^G
\end{align*}

\begin{itemize}
	\item \textbf{Size $\size$:}
	As argued before, the directly follows graph $G$ contains exactly $2 + \variety(L)$ nodes, so by definition $\size(G) = 2 + \variety(L)$.
	
	\item \textbf{Mismatch $\mismatch$:}
	Since the DFG $G$ contains only \texttt{xor}-connectors, the connector mismatch can be described by 
	\[\mismatch(G) = \left| \sum_{a \in S_{\texttt{xor}}^G} |\succ_L(a)| + \sum_{a \in J_{\texttt{xor}}^G} |\succ_L^{-1}(a)| \right|.\] 
	
	\item \textbf{Cross Connectivity $\crossconn$:}
	The cross connectivity depends on all paths through the directly follows graph $G$.
	While it would be possible to describe it formally by using properties of $L$, such a description would be complex and thus of little value.
	We therefore omit this measure.
	
	\item \textbf{Control Flow Complexity $\controlflow$:}
	This measure sums the number of edges exciting split nodes in $G$. 
	Since we have only one type of connectors in $G$, this means $\controlflow(G) = \sum_{a \in S_{\texttt{xor}}^G} |\succ_L(a)|$.
	
	\item \textbf{Separability $\separability$:}
	The separability depends on all paths through the directly follows graph $G$.
	While it would be possible to describe it formally by using properties of $L$, such a description would be complex and thus of little value.
	We therefore omit this measure.
	
	\item \textbf{Average Connector Degree $\avgconn$:}
	With the notions defined above, the average connector degree of $G$ is 
	\[\avgconn(G) = \frac{\sum_{a \in C_{\texttt{xor}}^G} (|\succ_L(a)| + |\succ_L^{-1}(a)|)}{|C_{\texttt{xor}}^G|},\] 
	since the degree of a node $a$ in $G$ is $|\succ_L(a)| + |\succ_L^{-1}(a)|$.
	
	\item \textbf{Maximum Connector Degree $\maxconn$:}
	With the notions defined above, the maximum connector degree of $G$ is 
	\[\avgconn(G) = \max\{|\succ_L(a)| + |\succ_L^{-1}(a)| \mid a \in C_{\texttt{xor}}^G\}.\]
	
	\item \textbf{Sequentiality $\sequentiality$:}
	We will reuse our definition of the set of connectors $C_{\texttt{xor}}^G$ in $G$ and find 
	\[\sequentiality(G) = \frac{|\{(a,b) \in (A \cup \{\triangleright\}) \times (A \cup \{\square\}) \mid a,b \not\in C_{\texttt{xor}}^G\}|}{|>_L| + |A_I| + |A_O|}.\]
	
	\item \textbf{Depth $\depth$:}
	The depth depends on all paths through the directly follows graph $G$.
	While it would be possible to describe it formally by using properties of $L$, such a description would be complex and thus of little value.
	We therefore omit this measure.
	
	\item \textbf{Diameter $\diameter$:}
	The diameter depends on all paths through the directly follows graph $G$.
	While it would be possible to describe it formally by using properties of $L$, such a description would be complex and thus of little value.
	We therefore omit this measure.
	
	\item \textbf{Cyclicity $\cyclicity$:}
	The cyclicity depends on all paths through the directly follows graph $G$.
	While it would be possible to describe it formally by using properties of $L$, such a description would be complex and thus of little value.
	We therefore omit this measure.
	
	\item \textbf{Coefficient of Network Connectivity $\netconn$:}
	Since $|V| = 2 + \variety(L)$ and $|E| = |>_L| + |A_I| + |A_O|$, we get $\netconn(G) = \frac{|>_L| + |A_I| + |A_O|}{2 + \variety(L)}$.
	
	\item \textbf{Density $\density$:}
	With $|V| = 2 + \variety(L)$ and $|E| = |>_L| + |A_I| + |A_O|$, we get $\density(G) = \frac{|>_L| + |A_I| + |A_O|}{(2 + \variety(L)) \cdot (1 + \variety(L))}$.
\end{itemize}

These findings conclude our analysis of the directly follows graph.
\cref{table:dfg-model-complexity} summarizes these findings for quick reference.

\begin{table}[htp]
	\caption{The complexity scores of the DFG $G$ for an event log $L$ over $A$.}
	\label{table:dfg-model-complexity}
	\centering
	\def\pad{\hspace*{1.5mm}}
	\renewcommand{\arraystretch}{2}
	\begin{tabular}{|r|l|} \hline
	$\pad\size(G)\pad$ & $\pad 2 + \variety(L) \pad$ \\ \hline
	$\pad\mismatch(G)\pad$ & $\pad \left| \sum_{a \in S_{\texttt{xor}}^G} |\succ_L(a)| + \sum_{a \in J_{\texttt{xor}}^G} |\succ_L^{-1}(a)| \right| \pad$ \\ \hline
	$\pad\controlflow(G)\pad$ & $\pad \sum_{a \in S_{\texttt{xor}}^G} |\succ_L(a)| \pad$ \\ \hline
	$\pad\avgconn(G)\pad$ & $\pad \frac{\sum_{a \in C_{\texttt{xor}}^G} (|\succ_L(a)| + |\succ_L^{-1}(a)|)}{|C_{\texttt{xor}}^G|} \pad$ \\ \hline
	$\pad\maxconn(G)\pad$ & $\pad \max\{|\succ_L(a)| + |\succ_L^{-1}(a)| \mid a \in C_{\texttt{xor}}^G\} \pad$ \\ \hline
	$\pad\sequentiality(G)\pad$ & $\pad \frac{|\{(a,b) \in (A \cup \{\triangleright\}) \times (A \cup \{\square\}) \mid a,b \not\in C_{\texttt{xor}}^G\}|}{|>_L| + |A_I| + |A_O|} \pad$ \\ \hline
	$\pad\netconn(G)\pad$ & $\pad \frac{|>_L| + |A_I| + |A_O|}{2 + \variety(L)} \pad$ \\ \hline
	$\pad\density(G)\pad$ & $\pad \frac{|>_L| + |A_I| + |A_O|}{(2 + \variety(L)) \cdot (1 + \variety(L))} \pad$ \\ \hline
	\end{tabular}
\end{table}

\newpage
\subsection{Directly Follows Miner}
\label{sec:dfm}
The directly follows miner~\cite{LeePW19} combines the easy readability of directly follows graphs and the expressiveness and theoretical foundation of Petri nets.
For an event log $L$, this discovery technique first creates the directly follows graph $G$ of $L$, including edge weights indicating how often two events follow each other.
In a second step, the traces corresponding to the most infrequent edge weights are filtered from the event log, until a user-chosen maximum number of traces was deleted.
Finally, the algorithm transforms the resulting directly follows graph $G' = (V', E')$ into a sound workflow net by performing the following steps:
\begin{itemize}
	\item[•] Create a place $p_e$ for every node $e \in V'$.
	\item[•] For all edges $(e_1, e_2) \in E'$, add the following construct to the already constructed places of the Petri net:
	\begin{enumerate}
		\item If $e_2 = \square$:
		
		\scalebox{\scalefactor}{
		\begin{tikzpicture}[node distance = 1.5cm,>=stealth',bend angle=0,auto]
			\node[place,label=below:$p_{e_1}$] (pi) {};
			\node[transition,right of=pi] (e2) {$\tau$}
			edge [pre] (pi);
			\node[place,right of=e2,label=below:$p_{\square}$] (pe2) {}
			edge [pre] (e2);
		\end{tikzpicture}}
		
		\item If $e_1 \neq \square$:
		
		\scalebox{\scalefactor}{
		\begin{tikzpicture}[node distance = 1.5cm,>=stealth',bend angle=0,auto]
			\node[place,label=below:$p_{e_1}$] (pi) {};
			\node[transition,right of=pi] (e2) {$e_2$}
			edge [pre] (pi);
			\node[place,right of=e2,label=below:$p_{e_2}$] (pe2) {}
			edge [pre] (e2);
		\end{tikzpicture}}
	\end{enumerate}
\end{itemize}
By setting $p_i := p_{\triangleright}$ and $p_o := p_{\square}$, this construction always results in a sound workflow net~\cite{LeePW19}.
In our analyses, we will skip the filtering step of the directly follows miner, and assume that the event logs are already filtered, as filtering can be performed in a preprocessing step.
\cref{fig:dfm-example} shows the workflow net found for the event log $L$ of \cref{fig:tracenet-example}, whose directly follows graph is shown in \cref{fig:dfg-example}.
\begin{figure}[ht]
	\centering
	\begin{tikzpicture}[node distance = 1.2cm,>=stealth',bend angle=0,auto]
		\node[place,tokens=1,label=below:$p_{\triangleright}$] (pi) {};
		\node[transition,right of=pi] (a) {$a$}
		edge [pre] (pi);
		\node[place,right of=a,label=below:$p_a$] (pa) {}
		edge [pre] (a);
		\node[transition,above right of=pa] (b) {$b$}
		edge [pre] (pa);
		\node[transition,below right of=pa] (c) {$c$}
		edge [pre] (pa);
		\node[place,above right of=b,label=above:$p_b$] (pb) {}
		edge [pre] (b);
		\node[place,below right of=c,label=below:$p_c$] (pc) {}
		edge [pre] (c);
		\node (dummy) at ($0.5*(pb) + 0.5*(pc)$) {};
		\node[transition,xshift=0.5cm] (c1) at (dummy) {$c$}
		edge [pre,bend right=10] (pb)
		edge [post,bend left=10] (pc);
		\node[transition,xshift=-0.5cm] (b1) at (dummy) {$b$}
		edge [post,bend left=10] (pb)
		edge [pre,bend right=10] (pc);
		\node[transition,below right of=pb] (d) {$d$}
		edge [pre] (pb);
		\node[transition,above right of=pc] (d1) {$d$}
		edge [pre] (pc);
		\node[place,below right of=d, label=below:$p_d$] (pd) {}
		edge [pre] (d)
		edge [pre] (d1);
		\node[transition,right of=pd] (tau) {$\tau$}
		edge [pre] (d);
		\node[place,right of=tau,label=below:$p_{\square}$] (po) {}
		edge [pre] (tau);
		\node[transition,below of=pd] (tau1) {$\tau$}
		edge [pre,bend left=5] (pc)
		edge [post,bend right=5] (po);
	\end{tikzpicture}
	\caption{The result of the directly follows miner for the event log $L$ of \cref{fig:tracenet-example}.}
	\label{fig:dfm-example}
\end{figure}
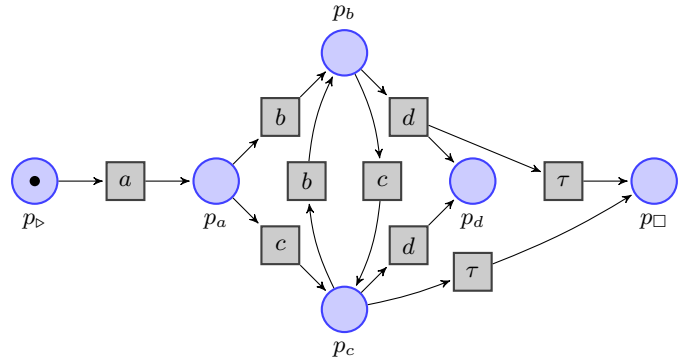
Due to its construction, many complexity scores of a model $M$ found by the directly follows miner for an event log $L$ can be described by the complexity scores of its underlying directly follows graph $G = (V,E)$.
In contrast to the previous sections, we will start by comparing the complexity scores of models found by the directly follows miner to the complexity scores of their underlying directly follows graph, as these findings will render some analyses trivial.
\begin{itemize}
	\item \textbf{Size $\size$:}
	Every node in $G$ becomes a place in $M$, so the number of places in $M$ is $|P| = |V|$.
	Furthermore, every edge in $G$ issues the creation of exactly one transition in $M$, so $|T| = |E|$.
	Therefore, we can describe the size of $M$ as $\size(M) = \size(G) + |E| = 2 + \variety(L) + |A_I| + |A_O|$.
	
	\item \textbf{Connector Mismatch $\mismatch$:}
	By construction, only places in $M$ can be connectors, as all transitions have exactly one incoming and one outgoing edge.
	A place $p_v$ in $M$ has $x$ incoming and $y$ outgoing edges if its corresponding node $v \in V$ has $x$ incoming and $y$ outgoing edges.
	Thus, the set of connectors in $M$ is the same as in $G$, and every connector in $M$ has the same in- and out-degree as its corresponding node in $G$.
	Therefore, $\mismatch(M) = \mismatch(G) = \left| \sum_{a \in S_{\texttt{xor}}^G} |\succ_L(a)| + \sum_{a \in J_{\texttt{xor}}^G} |\succ_L^{-1}(a)| \right|$.
	
	\item \textbf{Connector Heterogeneity $\connhet$:}
	In directly follows graphs, it did not make sense to calculate the entropy of connectors, as this modeling type does not contain semantics for parallelism.
	A model found by the directly follows miner, on the other hand, is a workflow net and thus has the required semantics.
	However, due to its construction, $M$ contains only \texttt{xor}-connectors, so $\connhet(M) = - (1 \cdot \log_2(1) + 0 \cdot \log_2(0)) = 0$.
	
	\item \textbf{Token Split $\tokensplit$:}
	In directly follows graphs, it did not make sense to calculate the entropy of connectors, as this modeling type does not contain semantics for parallelism.
	A model found by the directly follows miner, on the other hand, is a workflow net and thus has the required semantics.
	However, due to its construction, $M$ contains no transitions with more than one outgoing edge, so $\tokensplit(M) = 0$.
	
	\item \textbf{Control Flow Complexity $\controlflow$:}
	Every transition in $M$ has exactly one incoming and one outgoing arc by construction, so there are no \texttt{and}-connectors in $M$.
	But as argued earlier, every \texttt{xor}-connector in $G$ has a corresponding \texttt{xor}-connector in $M$ with the same amount of incoming and outgoing edges.
	Therefore, $\controlflow(M) = \controlflow(G) = \sum_{a \in S_{\texttt{xor}}^G} |\succ_L(a)|$.
	
	\item \textbf{Average Connector Degree $\avgconn$:}
	As argued earlier, every \texttt{xor}-connector in $G$ has a corresponding \texttt{xor}-connector in $M$ with the same amount of incoming and outgoing edges.
	Since there are no other connectors in $M$, we get $\avgconn(M) = \avgconn(G) = \frac{\sum_{a \in C_{\texttt{xor}}^G} (|\succ_L(a)| + |\succ_L^{-1}(a)|)}{|C_{\texttt{xor}}^G|}$.
	
	\item \textbf{Maximum Connector Degree $\avgconn$:}
	As argued before, all \texttt{xor}-connectors in $G$ have a corresponding \texttt{xor}-connector in $M$ with the same amount of incoming and outgoing edges.
	Since there are no other connectors in $M$, we get $\maxconn(M) = \maxconn(G) = \max\{|\succ_L(a)| + |\succ_L^{-1}(a)| \mid a \in C_{\texttt{xor}}^G\}$.
	
	\item \textbf{Sequentiality $\sequentiality$:}
	In $M$, no transition can be a connector of any type.
	Thus, all edges in $M$ have at least one non-connector endpoint.
	Whether the other endpoint $p_v$ of such an edge is also a non-connector depends on whether its corresponding node $v \in V'$ is a connector.
	If $p_v$ is not a connector, then it has exactly one incoming and one outgoing edge when $v \not\in \{\triangleright, \square\}$, and exactly one adjacent edge otherwise.
	Thus, $M$ has a sequentiality score of $\sequentiality(M) = 2 |V' \setminus (C_{\texttt{xor}}^G \cup \{\triangleright, \square\})| + |\{p_i \mid \triangleright \in C_{\texttt{xor}}^G\}| + |\{p_o \mid \square \in C_{\texttt{xor}}^G\}\}|$.
	
	\item \textbf{Diameter $\diameter$:}
	By construction, every path $(\triangleright, v_1, \dots, v_k, \square)$ in $G$ corresponds to a path $(p_{\triangleright}, v_1, p_{v_1}, \dots, v_k, p_{v_k}, \tau, p_{\square})$ in $M$, where $k \in \mathbb{N}_0$.
	Since there are no other paths in $M$, the longest path in $G$ of length $\ell$ corresponds to the longest path in $M$, which has length $2 \ell - 1$.
	Thus, $\diameter(M) = 2 \diameter(G) - 1$.
	
	\item \textbf{Coefficient of Network Connectivity $\netconn$:}
	Since each transition in $M$ has exactly one incoming and one outgoing edge, and contains $|>_L| + |A_I| + |A_O|$ transitions in total, there are $2(|>_L| + |A_I| + |A_O|)$ edges in $M$.
	Thus, $\netconn(M) = \frac{2(|>_L| + |A_I| + |A_O|)}{2 + \variety(L) + |>_L| + |A_I| + |A_O|} = \frac{2|V| \cdot \netconn(G)}{|V| + |E|}$.
	
	\item \textbf{Density $\density$:}
	As argued before, $M$ contains $2(|>_L| + |A_I| + |A_O|)$ edges in total.
	Thus, $\density(M) = \frac{2(|>_L| + |A_I| + |A_O|)}{2(|>_L| + |A_I| + |A_O|) \cdot (1 + \variety(L)} = \frac{1}{1 + \variety(L)}$.
	
	\item \textbf{Number of Empty Sequence Flows $\emptyseq$:}
	Since $M$ does not contain any \texttt{and}-connectors, there cannot be any places in $M$ that have just \texttt{and}-connectors in their pre- and postset. 
	In turn, $\emptyseq(M) = 0$.
\end{itemize}
\cref{table:dfm-model-complexity} summarizes these observations by showing how the complexity scores of the model found by the directly follows miner are defined, base on the notions of the previous subsection for the directly follows graph $G$.

\begin{table}[htp]
    \caption{The complexity scores of the result $M$ of the directly follows miner for an event log $L$ over a set of activities $A$. $G = (V,E)$ is the directly follows graph for $L$.}
    \label{table:dfm-model-complexity}
    \centering
    \renewcommand{\arraystretch}{1.4}
    \begin{tabular}{|c|c|c|} \hline
         $\pad \size(M) \pad$ & $\pad 2 + \variety(L) + |>_L| + |A_I| + |A_O| \pad$ & $\pad \size(G) + |E| \pad$ \\ \hline
         $\pad \mismatch(M) \pad$ & $\pad \left| \sum_{a \in S_{\texttt{xor}}^G} |\succ_L(a)| + \sum_{a \in J_{\texttt{xor}}^G} |\succ_L^{-1}(a)| \right| \pad$ & $\pad \mismatch(G) \pad$ \\ \hline
         $\pad \connhet(M) \pad$ & $\pad 0 \pad$ & $\pad 0 \pad$ \\ \hline
         $\pad \tokensplit(M) \pad$ & $\pad 0 \pad$ & $\pad 0 \pad$ \\ \hline
         $\pad \controlflow(M) \pad$ & $\pad \sum_{a \in S_{\texttt{xor}}^G} |\succ_L(a)| \pad$ & $\pad \controlflow(G) \pad$ \\ \hline
         $\pad \avgconn(M) \pad$ & $\pad \frac{\sum_{a \in C_{\texttt{xor}}^G} (|\succ_L(a)| + |\succ_L^{-1}(a)|)}{|C_{\texttt{xor}}^G|} \pad$ & $\pad \avgconn(G) \pad$ \\ \hline
         $\pad \maxconn(M) \pad$ & $\pad \max\{|\succ_L(a)| + |\succ_L^{-1}(a)| \mid a \in C_{\texttt{xor}}^G\} \pad$ & $\pad \maxconn(G) \pad$ \\ \hline
         $\pad \sequentiality(M) \pad$ & \multicolumn{2}{c|}{$\pad 2 |V' \setminus (C_{\texttt{xor}}^G \cup \{\triangleright, \square\})| + |\{p_i \mid \triangleright \in C_{\texttt{xor}}^G\}| + |\{p_o \mid \square \in C_{\texttt{xor}}^G\}\}| \pad$} \\ \hline
         $\pad \diameter(M) \pad$ & & $\pad 2 \diameter(G) - 1 \pad$ \\ \hline
         $\pad \netconn(M) \pad$ & $\pad \frac{2(|>_L| + |A_I| + |A_O|)}{2 + \variety(L) + |>_L| + |A_I| + |A_O|} \pad$ & $\pad \frac{2|V| \cdot \netconn(G)}{|V| + |E|} \pad$ \\ \hline
         $\pad \density(M) \pad$ & $\pad \frac{1}{\variety(L) + 1} \pad$ & $\pad \frac{1}{|V| - 1} \pad$ \\ \hline
         $\pad \emptyseq(M) \pad$ & $\pad 0 \pad$ & $\pad 0 \pad$ \\ \hline
    \end{tabular}
\end{table}

Next, we will start the analysis of the relations between log- and model complexity.
\cref{table:dfm-findings} shows the relations we found while fixing the directly follows miner.
\begin{table}[ht]
	\caption{The relations between the complexity scores of two nets $M_1$ and $M_2$ found by the directly follows miner for the event logs $L_1$ and $L_2$ as input respectively, where $L_1 \sqsubset L_2$ and the complexity of $L_1$ is lower than the complexity of $L_2$.}
	\label{table:dfm-findings}
	\centering
	\resizebox{\textwidth}{!}{
	\begin{tabular}{|c|c|c|c|c|c|c|c|c|c|c|c|c|c|c|c|c|c|} \hline
		 & $\size$ & $\mismatch$ & $\connhet$ & $\crossconn$ & $\tokensplit$ & $\controlflow$ & $\separability$ & $\avgconn$ & $\maxconn$ & $\sequentiality$ & $\depth$ & $\diameter$ & $\cyclicity$ & $\netconn$ & $\density$ & $\duplicate$ & $\emptyseq$ \\ \hline
		$\magnitude$ & \hyperref[theo:dfm-size-cfc-leq-entries]{$\mleq$} & \hyperref[theo:dfm-mismatch-entries]{$\norel$} & \hyperref[theo:dfm-equals-entries]{$\meq$} & \hyperref[theo:dfm-crossconn-entries]{$\norel^*$} & \hyperref[theo:dfm-equals-entries]{$\meq$} & \hyperref[theo:dfm-size-cfc-leq-entries]{$\mleq$} & \hyperref[theo:dfm-sep-entries]{$\norel$} & \hyperref[theo:dfm-acd-entries]{$\norel$} & \hyperref[theo:dfm-mcd-entries]{$\mleq$} & \hyperref[theo:dfm-seq-entries]{$\norel$} & \hyperref[theo:dfm-depth-entries]{$\norel$} & \hyperref[theo:dfm-diam-entries]{$\mleq$} & \hyperref[theo:dfm-cyc-entries]{$\norel$} & \hyperref[theo:dfm-cnc-entries]{$\norel$} & \hyperref[theo:dfm-dens-geq-entries]{$\mgeq$} & \hyperref[theo:dfm-dup-leq-entries]{$\mleq$} & \hyperref[theo:dfm-equals-entries]{$\meq$} \\ \hline
		
		$\variety$ & \hyperref[theo:dfm-size-cfc-less-entries]{$\mless$} & \hyperref[theo:dfm-mismatch-entries]{$\norel$} & \hyperref[theo:dfm-equals-entries]{$\meq$} & \hyperref[theo:dfm-crossconn-entries]{$\norel^*$} & \hyperref[theo:dfm-equals-entries]{$\meq$} & \hyperref[theo:dfm-size-cfc-less-entries]{$\mless$} & \hyperref[theo:dfm-sep-entries]{$\norel$} & \hyperref[theo:dfm-acd-entries]{$\norel$} & \hyperref[theo:dfm-mcd-entries]{$\mleq$} & \hyperref[theo:dfm-seq-entries]{$\norel$} & \hyperref[theo:dfm-depth-entries]{$\norel$} & \hyperref[theo:dfm-diam-entries]{$\mleq$} & \hyperref[theo:dfm-cyc-entries]{$\norel$} & \hyperref[theo:dfm-cnc-entries]{$\norel$} & \hyperref[theo:dfm-dens-greater-entry]{$\mgreater$} & \hyperref[theo:dfm-duplicate-less-entries]{$\mless$} & \hyperref[theo:dfm-equals-entries]{$\meq$} \\ \hline
		
		$\support$ & \hyperref[theo:dfm-size-cfc-leq-entries]{$\mleq$} & \hyperref[theo:dfm-mismatch-entries]{$\norel$} & \hyperref[theo:dfm-equals-entries]{$\meq$} & \hyperref[theo:dfm-crossconn-entries]{$\norel^*$} & \hyperref[theo:dfm-equals-entries]{$\meq$} & \hyperref[theo:dfm-size-cfc-leq-entries]{$\mleq$} & \hyperref[theo:dfm-sep-entries]{$\norel$} & \hyperref[theo:dfm-acd-entries]{$\norel$} & \hyperref[theo:dfm-mcd-entries]{$\mleq$} & \hyperref[theo:dfm-seq-entries]{$\norel$} & \hyperref[theo:dfm-depth-entries]{$\norel$} & \hyperref[theo:dfm-diam-entries]{$\mleq$} & \hyperref[theo:dfm-cyc-entries]{$\norel$} & \hyperref[theo:dfm-cnc-entries]{$\norel$} & \hyperref[theo:dfm-dens-geq-entries]{$\mgeq$} & \hyperref[theo:dfm-dup-leq-entries]{$\mleq$} & \hyperref[theo:dfm-equals-entries]{$\meq$} \\ \hline
		
		$\tlavg$ & \hyperref[theo:dfm-size-cfc-leq-entries]{$\mleq$} & \hyperref[theo:dfm-mismatch-entries]{$\norel$} & \hyperref[theo:dfm-equals-entries]{$\meq$} & \hyperref[theo:dfm-crossconn-entries]{$\norel^*$} & \hyperref[theo:dfm-equals-entries]{$\meq$} & \hyperref[theo:dfm-size-cfc-leq-entries]{$\mleq$} & \hyperref[theo:dfm-sep-entries]{$\norel$} & \hyperref[theo:dfm-acd-entries]{$\norel$} & \hyperref[theo:dfm-mcd-entries]{$\mleq$} & \hyperref[theo:dfm-seq-entries]{$\norel$} & \hyperref[theo:dfm-depth-entries]{$\norel$} & \hyperref[theo:dfm-diam-entries]{$\mleq$} & \hyperref[theo:dfm-cyc-entries]{$\norel$} & \hyperref[theo:dfm-cnc-entries]{$\norel$} & \hyperref[theo:dfm-dens-geq-entries]{$\mgeq$} & \hyperref[theo:dfm-dup-leq-entries]{$\mleq$} & \hyperref[theo:dfm-equals-entries]{$\meq$} \\ \hline
		
		$\tlmax$ & \hyperref[theo:dfm-size-cfc-leq-entries]{$\mleq$} & \hyperref[theo:dfm-mismatch-entries]{$\norel$} & \hyperref[theo:dfm-equals-entries]{$\meq$} & \hyperref[theo:dfm-crossconn-entries]{$\norel^*$} & \hyperref[theo:dfm-equals-entries]{$\meq$} & \hyperref[theo:dfm-size-cfc-leq-entries]{$\mleq$} & \hyperref[theo:dfm-sep-entries]{$\norel$} & \hyperref[theo:dfm-acd-entries]{$\norel$} & \hyperref[theo:dfm-mcd-entries]{$\mleq$} & \hyperref[theo:dfm-seq-entries]{$\norel$} & \hyperref[theo:dfm-depth-entries]{$\norel$} & \hyperref[theo:dfm-diam-entries]{$\mleq$} & \hyperref[theo:dfm-cyc-entries]{$\norel$} & \hyperref[theo:dfm-cnc-entries]{$\norel$} & \hyperref[theo:dfm-dens-geq-entries]{$\mgeq$} & \hyperref[theo:dfm-dup-leq-entries]{$\mleq$} & \hyperref[theo:dfm-equals-entries]{$\meq$} \\ \hline
		
		$\levelofdetail$ & \hyperref[theo:dfm-size-cfc-less-entries]{$\mless$} & \hyperref[theo:dfm-mismatch-entries]{$\norel$} & \hyperref[theo:dfm-equals-entries]{$\meq$} & \hyperref[theo:dfm-crossconn-entries]{$\norel^*$} & \hyperref[theo:dfm-equals-entries]{$\meq$} & \hyperref[theo:dfm-size-cfc-less-entries]{$\mless$} & \hyperref[theo:dfm-sep-entries]{$\norel$} & \hyperref[theo:dfm-acd-entries]{$\norel$} & \hyperref[theo:dfm-mcd-entries]{$\mleq$} & \hyperref[theo:dfm-seq-entries]{$\norel$} & \hyperref[theo:dfm-depth-entries]{$\norel$} & \hyperref[theo:dfm-diam-entries]{$\mleq$} & \hyperref[theo:dfm-cyc-entries]{$\norel$} & \hyperref[theo:dfm-cnc-entries]{$\norel$} & \hyperref[theo:dfm-dens-geq-entries]{$\mgeq$} & \hyperref[theo:dfm-duplicate-less-entries]{$\mless$} & \hyperref[theo:dfm-equals-entries]{$\meq$} \\ \hline
		
		$\numberofties$ & \hyperref[theo:dfm-size-cfc-less-entries]{$\mless$} & \hyperref[theo:dfm-mismatch-entries]{$\norel$} & \hyperref[theo:dfm-equals-entries]{$\meq$} & \hyperref[theo:dfm-crossconn-entries]{$\norel^*$} & \hyperref[theo:dfm-equals-entries]{$\meq$} & \hyperref[theo:dfm-size-cfc-less-entries]{$\mless$} & \hyperref[theo:dfm-sep-entries]{$\norel$} & \hyperref[theo:dfm-acd-entries]{$\norel$} & \hyperref[theo:dfm-mcd-entries]{$\mleq$} & \hyperref[theo:dfm-seq-entries]{$\norel$} & \hyperref[theo:dfm-depth-entries]{$\norel$} & \hyperref[theo:dfm-diam-entries]{$\mleq$} & \hyperref[theo:dfm-cyc-entries]{$\norel$} & \hyperref[theo:dfm-cnc-entries]{$\norel$} & \hyperref[theo:dfm-dens-geq-entries]{$\mgeq$} & \hyperref[theo:dfm-duplicate-less-entries]{$\mless$} & \hyperref[theo:dfm-equals-entries]{$\meq$} \\ \hline
		
		$\lempelziv$ & \hyperref[theo:dfm-size-cfc-leq-entries]{$\mleq$} & \hyperref[theo:dfm-mismatch-entries]{$\norel$} & \hyperref[theo:dfm-equals-entries]{$\meq$} & \hyperref[theo:dfm-crossconn-entries]{$\norel^*$} & \hyperref[theo:dfm-equals-entries]{$\meq$} & \hyperref[theo:dfm-size-cfc-leq-entries]{$\mleq$} & \hyperref[theo:dfm-sep-entries]{$\norel$} & \hyperref[theo:dfm-acd-entries]{$\norel$} & \hyperref[theo:dfm-mcd-entries]{$\mleq$} & \hyperref[theo:dfm-seq-entries]{$\norel$} & \hyperref[theo:dfm-depth-entries]{$\norel$} & \hyperref[theo:dfm-diam-entries]{$\mleq$} & \hyperref[theo:dfm-cyc-entries]{$\norel$} & \hyperref[theo:dfm-cnc-entries]{$\norel$} & \hyperref[theo:dfm-dens-geq-entries]{$\mgeq$} & \hyperref[theo:dfm-dup-leq-entries]{$\mleq$} & \hyperref[theo:dfm-equals-entries]{$\meq$} \\ \hline
		
		$\numberuniquetraces$ & \hyperref[theo:dfm-size-cfc-leq-entries]{$\mleq$} & \hyperref[theo:dfm-mismatch-entries]{$\norel$} & \hyperref[theo:dfm-equals-entries]{$\meq$} & \hyperref[theo:dfm-crossconn-entries]{$\norel^*$} & \hyperref[theo:dfm-equals-entries]{$\meq$} & \hyperref[theo:dfm-size-cfc-leq-entries]{$\mleq$} & \hyperref[theo:dfm-sep-entries]{$\norel$} & \hyperref[theo:dfm-acd-entries]{$\norel$} & \hyperref[theo:dfm-mcd-entries]{$\mleq$} & \hyperref[theo:dfm-seq-entries]{$\norel$} & \hyperref[theo:dfm-depth-entries]{$\norel$} & \hyperref[theo:dfm-diam-entries]{$\mleq$} & \hyperref[theo:dfm-cyc-entries]{$\norel$} & \hyperref[theo:dfm-cnc-entries]{$\norel$} & \hyperref[theo:dfm-dens-geq-entries]{$\mgeq$} & \hyperref[theo:dfm-dup-leq-entries]{$\mleq$} & \hyperref[theo:dfm-equals-entries]{$\meq$} \\ \hline
		
		$\percentageuniquetraces$ & \hyperref[theo:dfm-size-cfc-leq-entries]{$\mleq$} & \hyperref[theo:dfm-mismatch-entries]{$\norel$} & \hyperref[theo:dfm-equals-entries]{$\meq$} & \hyperref[theo:dfm-crossconn-entries]{$\norel^*$} & \hyperref[theo:dfm-equals-entries]{$\meq$} & \hyperref[theo:dfm-size-cfc-leq-entries]{$\mleq$} & \hyperref[theo:dfm-sep-entries]{$\norel$} & \hyperref[theo:dfm-acd-entries]{$\norel$} & \hyperref[theo:dfm-mcd-entries]{$\mleq$} & \hyperref[theo:dfm-seq-entries]{$\norel$} & \hyperref[theo:dfm-depth-entries]{$\norel$} & \hyperref[theo:dfm-diam-entries]{$\mleq$} & \hyperref[theo:dfm-cyc-entries]{$\norel$} & \hyperref[theo:dfm-cnc-entries]{$\norel$} & \hyperref[theo:dfm-dens-geq-entries]{$\mgeq$} & \hyperref[theo:dfm-dup-leq-entries]{$\mleq$} & \hyperref[theo:dfm-equals-entries]{$\meq$} \\ \hline
		
		$\structure$ & \hyperref[theo:dfm-size-cfc-leq-entries]{$\mleq$} & \hyperref[theo:dfm-mismatch-entries]{$\norel$} & \hyperref[theo:dfm-equals-entries]{$\meq$} & \hyperref[theo:dfm-crossconn-entries]{$\norel^*$} & \hyperref[theo:dfm-equals-entries]{$\meq$} & \hyperref[theo:dfm-size-cfc-leq-entries]{$\mleq$} & \hyperref[theo:dfm-sep-entries]{$\norel$} & \hyperref[theo:dfm-acd-entries]{$\norel$} & \hyperref[theo:dfm-mcd-entries]{$\mleq$} & \hyperref[theo:dfm-seq-entries]{$\norel$} & \hyperref[theo:dfm-depth-entries]{$\norel$} & \hyperref[theo:dfm-diam-entries]{$\mleq$} & \hyperref[theo:dfm-cyc-entries]{$\norel$} & \hyperref[theo:dfm-cnc-entries]{$\norel$} & \hyperref[theo:dfm-dens-geq-entries]{$\mgeq$} & \hyperref[theo:dfm-dup-leq-entries]{$\mleq$} & \hyperref[theo:dfm-equals-entries]{$\meq$} \\ \hline
		
		$\affinity$ & \hyperref[theo:dfm-size-cfc-leq-entries]{$\mleq$} & \hyperref[theo:dfm-mismatch-entries]{$\norel$} & \hyperref[theo:dfm-equals-entries]{$\meq$} & \hyperref[theo:dfm-crossconn-entries]{$\norel^*$} & \hyperref[theo:dfm-equals-entries]{$\meq$} & \hyperref[theo:dfm-size-cfc-leq-entries]{$\mleq$} & \hyperref[theo:dfm-sep-entries]{$\norel$} & \hyperref[theo:dfm-acd-entries]{$\norel$} & \hyperref[theo:dfm-mcd-entries]{$\mleq$} & \hyperref[theo:dfm-seq-entries]{$\norel$} & \hyperref[theo:dfm-depth-entries]{$\norel$} & \hyperref[theo:dfm-diam-entries]{$\mleq$} & \hyperref[theo:dfm-cyc-entries]{$\norel$} & \hyperref[theo:dfm-cnc-entries]{$\norel$} & \hyperref[theo:dfm-dens-geq-entries]{$\mgeq$} & \hyperref[theo:dfm-dup-leq-entries]{$\mleq$} & \hyperref[theo:dfm-equals-entries]{$\meq$} \\ \hline
		
		$\deviationfromrandom$ & \hyperref[theo:dfm-size-cfc-leq-entries]{$\mleq$} & \hyperref[theo:dfm-mismatch-entries]{$\norel$} & \hyperref[theo:dfm-equals-entries]{$\meq$} & \hyperref[theo:dfm-crossconn-entries]{$\norel^*$} & \hyperref[theo:dfm-equals-entries]{$\meq$} & \hyperref[theo:dfm-size-cfc-leq-entries]{$\mleq$} & \hyperref[theo:dfm-sep-entries]{$\norel$} & \hyperref[theo:dfm-acd-entries]{$\norel$} & \hyperref[theo:dfm-mcd-entries]{$\mleq$} & \hyperref[theo:dfm-seq-entries]{$\norel$} & \hyperref[theo:dfm-depth-entries]{$\norel$} & \hyperref[theo:dfm-diam-entries]{$\mleq$} & \hyperref[theo:dfm-cyc-entries]{$\norel$} & \hyperref[theo:dfm-cnc-entries]{$\norel$} & \hyperref[theo:dfm-dens-geq-entries]{$\mgeq$} & \hyperref[theo:dfm-dup-leq-entries]{$\mleq$} & \hyperref[theo:dfm-equals-entries]{$\meq$} \\ \hline
		
		$\avgdist$ & \hyperref[theo:dfm-size-cfc-leq-entries]{$\mleq$} & \hyperref[theo:dfm-mismatch-entries]{$\norel$} & \hyperref[theo:dfm-equals-entries]{$\meq$} & \hyperref[theo:dfm-crossconn-entries]{$\norel^*$} & \hyperref[theo:dfm-equals-entries]{$\meq$} & \hyperref[theo:dfm-size-cfc-leq-entries]{$\mleq$} & \hyperref[theo:dfm-sep-entries]{$\norel$} & \hyperref[theo:dfm-acd-entries]{$\norel$} & \hyperref[theo:dfm-mcd-entries]{$\mleq$} & \hyperref[theo:dfm-seq-entries]{$\norel$} & \hyperref[theo:dfm-depth-entries]{$\norel$} & \hyperref[theo:dfm-diam-entries]{$\mleq$} & \hyperref[theo:dfm-cyc-entries]{$\norel$} & \hyperref[theo:dfm-cnc-entries]{$\norel$} & \hyperref[theo:dfm-dens-geq-entries]{$\mgeq$} & \hyperref[theo:dfm-dup-leq-entries]{$\mleq$} & \hyperref[theo:dfm-equals-entries]{$\meq$} \\ \hline
		
		$\varentropy$ & \hyperref[theo:dfm-size-cfc-leq-entries]{$\mleq$} & \hyperref[theo:dfm-mismatch-entries]{$\norel$} & \hyperref[theo:dfm-equals-entries]{$\meq$} & \hyperref[theo:dfm-crossconn-entries]{$\norel^*$} & \hyperref[theo:dfm-equals-entries]{$\meq$} & \hyperref[theo:dfm-size-cfc-leq-entries]{$\mleq$} & \hyperref[theo:dfm-sep-entries]{$\norel$} & \hyperref[theo:dfm-acd-entries]{$\norel$} & \hyperref[theo:dfm-mcd-entries]{$\mleq$} & \hyperref[theo:dfm-seq-entries]{$\norel$} & \hyperref[theo:dfm-depth-entries]{$\norel$} & \hyperref[theo:dfm-diam-entries]{$\mleq$} & \hyperref[theo:dfm-cyc-entries]{$\norel$} & \hyperref[theo:dfm-cnc-entries]{$\norel$} & \hyperref[theo:dfm-dens-geq-entries]{$\mgeq$} & \hyperref[theo:dfm-dup-leq-entries]{$\mleq$} & \hyperref[theo:dfm-equals-entries]{$\meq$} \\ \hline
		
		$\normvarentropy$ & \hyperref[theo:dfm-size-cfc-leq-entries]{$\mleq$} & \hyperref[theo:dfm-mismatch-entries]{$\norel$} & \hyperref[theo:dfm-equals-entries]{$\meq$} & \hyperref[theo:dfm-crossconn-entries]{$\norel^*$} & \hyperref[theo:dfm-equals-entries]{$\meq$} & \hyperref[theo:dfm-size-cfc-leq-entries]{$\mleq$} & \hyperref[theo:dfm-sep-entries]{$\norel$} & \hyperref[theo:dfm-acd-entries]{$\norel$} & \hyperref[theo:dfm-mcd-entries]{$\mleq$} & \hyperref[theo:dfm-seq-entries]{$\norel$} & \hyperref[theo:dfm-depth-entries]{$\norel$} & \hyperref[theo:dfm-diam-entries]{$\mleq$} & \hyperref[theo:dfm-cyc-entries]{$\norel$} & \hyperref[theo:dfm-cnc-entries]{$\norel$} & \hyperref[theo:dfm-dens-geq-entries]{$\mgeq$} & \hyperref[theo:dfm-dup-leq-entries]{$\mleq$} & \hyperref[theo:dfm-equals-entries]{$\meq$} \\ \hline
		
		$\seqentropy$ & \hyperref[theo:dfm-size-cfc-leq-entries]{$\mleq$} & \hyperref[theo:dfm-mismatch-entries]{$\norel$} & \hyperref[theo:dfm-equals-entries]{$\meq$} & \hyperref[theo:dfm-crossconn-entries]{$\norel^*$} & \hyperref[theo:dfm-equals-entries]{$\meq$} & \hyperref[theo:dfm-size-cfc-leq-entries]{$\mleq$} & \hyperref[theo:dfm-sep-entries]{$\norel$} & \hyperref[theo:dfm-acd-entries]{$\norel$} & \hyperref[theo:dfm-mcd-entries]{$\mleq$} & \hyperref[theo:dfm-seq-entries]{$\norel$} & \hyperref[theo:dfm-depth-entries]{$\norel$} & \hyperref[theo:dfm-diam-entries]{$\mleq$} & \hyperref[theo:dfm-cyc-entries]{$\norel$} & \hyperref[theo:dfm-cnc-entries]{$\norel$} & \hyperref[theo:dfm-dens-geq-entries]{$\mgeq$} & \hyperref[theo:dfm-dup-leq-entries]{$\mleq$} & \hyperref[theo:dfm-equals-entries]{$\meq$} \\ \hline
		
		$\normseqentropy$ & \hyperref[theo:dfm-size-cfc-leq-entries]{$\mleq$} & \hyperref[theo:dfm-mismatch-entries]{$\norel$} & \hyperref[theo:dfm-equals-entries]{$\meq$} & \hyperref[theo:dfm-crossconn-entries]{$\norel^*$} & \hyperref[theo:dfm-equals-entries]{$\meq$} & \hyperref[theo:dfm-size-cfc-leq-entries]{$\mleq$} & \hyperref[theo:dfm-sep-entries]{$\norel$} & \hyperref[theo:dfm-acd-entries]{$\norel$} & \hyperref[theo:dfm-mcd-entries]{$\mleq$} & \hyperref[theo:dfm-seq-entries]{$\norel$} & \hyperref[theo:dfm-depth-entries]{$\norel$} & \hyperref[theo:dfm-diam-entries]{$\mleq$} & \hyperref[theo:dfm-cyc-entries]{$\norel$} & \hyperref[theo:dfm-cnc-entries]{$\norel$} & \hyperref[theo:dfm-dens-geq-entries]{$\mgeq$} & \hyperref[theo:dfm-dup-leq-entries]{$\mleq$} & \hyperref[theo:dfm-equals-entries]{$\meq$} \\ \hline
	\end{tabular}}
	{\scriptsize ${}^*$We did not find examples showing that $\mathcal{C}^L(L_1) < \mathcal{C}^L(L_2)$ and $\crossconn(M_1) = \crossconn(M_2)$ is possible.}
\end{table}
With the observations of \cref{table:dfm-model-complexity}, the analysis of $\mismatch$, $\controlflow$, $\avgconn$, and $\maxconn$ become trivial, since these measures return the exact same score for the directly follows graph and for the model found by the directly follows miner.
Thus, we can reuse our results of \cref{sec:dfg} for these measures.

\begin{theorem}
\label{theo:dfm-size-cfc-leq-entries}
Let $\mathcal{C}^L \in (\loc \setminus \{\variety, \levelofdetail, \numberofties\})$ be a log complexity measure and $\mathcal{C}^M \in \{\size, \controlflow\}$.
Then, $(\mathcal{C}^L, \mathcal{C}^M) \in \mleq$.
\end{theorem}
\begin{proof}
Let $M$ be the model found by the directly follows miner for an event log $L$, and $G$ be the directly follows graph for $L$.
The claim of this theorem is obvious for $\controlflow$, since $\controlflow(M) = \controlflow(G)$, and $(\mathcal{C}^L, \controlflow) \in \mleq$ by \cref{theo:dfg-cfc-leq-entries}.
For $\size$, we can use the same examples as in this theorem.
First, consider the logs:
\begin{align*}
	L_1 &= [\langle a,b,c,c \rangle^2, \langle c,c,d,e \rangle] \\
	L_2 &= L_1 + [\langle a,b,c,d,e \rangle]
\end{align*}
Let $M_1, M_2$ be the models found by the directly follows miner for $L_1, L_2$.
Then, $\size(M_1) = 16 = \size(M_2)$.
As we have seen in \cref{theo:dfg-cfc-leq-entries}, all log complexity scores except $\variety, \levelofdetail, \numberofties$, and $\affinity$ strictly increase for these event logs.
For $\affinity$, we can again use the event logs
\begin{align*}
	L_1 &= [\langle a,b,c,c \rangle, \langle c,c,d,e \rangle] \\
	L_2 &= L_1 + [\langle a,b,c,d,e \rangle]
\end{align*}
in which affinity increases. 
But the directly follows graphs, and therefore the models found by the directly follows miner, are the same for $L_1$ and $L_2$.
Thus, $\mathcal{C}^L(L_1) < \mathcal{C}^L(L_2)$ and $\size(M_1) = \size(M_2)$ is possible.

To see that $\mathcal{C}^L(L_1) < \mathcal{C}^L(L_2)$ and $\size(M_1) < \size(M_2)$ is also possible, consider the following event logs, which were already used and analyzed for their directly follows graph in \cref{theo:dfg-cfc-leq-entries}:
\begin{align*}
	L_1 &= [\langle a,b,c,d \rangle^{2}, \langle a,b,c,d,e \rangle^{2}, \langle d,e,a,b \rangle^{2}] \\
	L_2 &= [\langle a,b,c,d,e \rangle^{2}, \langle d,e,a,b,c \rangle, \langle c,d,e,a,b \rangle, \langle e,c,d,a,b,c,f \rangle]
\end{align*}
The models $M_1, M_2$ found by the directly follows miner for these event logs fulfill $\size(M_1) = 17 < 25 = \size(M_2)$, but \cref{theo:dfg-cfc-leq-entries} shows that all log complexity scores strictly increase for these two event logs.
Thus, $\mathcal{C}^L(L_1) < \mathcal{C}^L(L_2)$ and $\size(M_1) < \size(M_2)$ is also possible.

Finally, it is not possible that $\size$ decreases, as the size of the directly follows model $M$ is exactly the amount of nodes and edges in its underlying directly follows graph $G$. 
The latter can only increase when adding behavior to the underlying event log, as already discussed in \cref{sec:dfg}. \hfill$\square$
\end{proof}

\begin{theorem}
\label{theo:dfm-size-cfc-less-entries}
Let $\mathcal{C}^L \in (\loc \setminus \{\variety, \levelofdetail, \numberofties\})$ be a log complexity measure and $\mathcal{C}^M \in \{\size, \controlflow\}$.
Then, $(\mathcal{C}^L, \mathcal{C}^M) \in \mless$.
\end{theorem}
\begin{proof}
The claim is trivial for $\mathcal{C}^M = \controlflow$, since for any model $M$ found by the directly follows miner for an event log $L$, we have $\controlflow(M) = \controlflow(G)$, where $G = (V,E)$ is the directly follows graph for $L$.
\cref{theo:dfg-cfc-less-entries} shows that the claim is true for $\controlflow(G)$, so we can deduce that it also holds for $\controlflow(M)$.
Furhtermore, \cref{theo:dfg-cfc-less-entries} discusses that an increase in $\mathcal{C}^L$ means that at least one new edge gets introduced to the directly follows graph.
Since $\size(M) = \size(G) + |E|$, we can immediately see that $\size(M_1) < \size(M_2)$ for two models $M_1,M_2$ found by the directly follows miner for event logs $L_1, L_2$, if $\mathcal{C}^L(L_1) < \mathcal{C}^L(L_2)$. \hfill$\square$
\end{proof}

\begin{theorem}
\label{theo:dfm-mismatch-entries}
$(\mathcal{C}^L, \mismatch) \in \norel$ for any log complexity measure $\mathcal{C}^L \in \loc$.
\end{theorem}
\begin{proof}
Let $L_1 \sqsubset L_2$ be event logs, $G_1, G_2$ their directly follows graphs, and $M_1, M_2$ the models found by the directly follows miner for $L_1, L_2$.
Then, we know that $\mismatch(M_1) = \mismatch(G_1)$ and $\mismatch(M_2) = \mismatch(G_2)$.
Furthermore, by \cref{theo:dfg-mismatch-entries}, we know $\mismatch(G_1) < \mismatch(G_2)$, $\mismatch(G_1) > \mismatch(G_2)$, and $\mismatch(G_1) = \mismatch(G_2)$ are possible when event log complexity increases.
Thus, $\mismatch(M_1) < \mismatch(M_2)$, $\mismatch(M_1) > \mismatch(M_2)$, and $\mismatch(M_1) = \mismatch(M_2)$ are all possible as well. \hfill$\square$
\end{proof}

\begin{theorem}
\label{theo:dfm-equals-entries}
$(\mathcal{C}^L, \mathcal{C}^M) \in \meq$ for any log complexity measure $\mathcal{C}^L \in \loc$ and any $\mathcal{C}^M \in \{\connhet, \tokensplit, \emptyseq\}$.
\end{theorem}
\begin{proof}
Let $L$ be an event log and $M$ be the model found by the directly follows miner for $L$.
By construction, all transitions in $M$ have exactly one incoming and one outgoing edge.
Thus, there are no \texttt{and}-connectors in $M$.
In turn, we get $\connhet(M) = 0$, $\tokensplit(M) = 0$, and $\emptyseq(M) = 0$, so for two event logs $L_1, L_2$ and their directly follows models $M_1, M_2$, we always have $\mathcal{C}^M(M_1) = \mathcal{C}^M(M_2)$. \hfill$\square$
\end{proof}

\begin{theorem}
\label{theo:dfm-crossconn-entries}
$(\mathcal{C}^L, \crossconn) \in \norel$ for any log complexity measure $\mathcal{C}^L \in \loc$.
\end{theorem}
\begin{proof}
We can use the same counter examples as those of \cref{theo:dfg-crossconn-entries}.
For models $M_1, M_2, M_3$ found by the directly follows miner for the event logs 
\begin{align*}
	L_1 &= [\langle a,b \rangle^{5}, \langle c,d \rangle, \langle e,f \rangle, \langle g \rangle] \\
	L_2 &= L_1 + [\langle a,b,c,d \rangle, \langle s,t,u,v,w,x,y,z \rangle] \\
	L_3 &= L_2 + [\langle h,i,j,k,l,m,n,o,p \rangle]
\end{align*}
we get $\crossconn(M_1) \approx 0.8893 > \crossconn(M_2) \approx 0.8775 < 0.8911$.
Since $\affinity$ and $\normvarentropy$ do not strictly increase for these event logs, we also use the second counter example of \cref{theo:dfg-crossconn-entries}.
For models $M_1, M_2, M_3$ found by the directly follows miner for hte event logs
\begin{align*}
	L_1 &= [\langle a,b,c,d \rangle, \langle c,d,e,f \rangle, \langle e,f,g \rangle, \langle a,b \rangle, \langle c,d \rangle, \langle e,f \rangle, \langle g \rangle] \\
	L_2 &= L_1 + [\langle a,b,c,d \rangle^{2}, \langle q,r,s,t \rangle, \langle u,v,w,x,y,z \rangle] \\
	L_3 &= L_2 + [\langle a,b,c,d \rangle^{3}, \langle h \rangle, \langle i \rangle, \langle j \rangle]
\end{align*}
we have $\crossconn(M_1) \approx 0.9675 > \crossconn(M_2) \approx 0.931 < \crossconn(M_3) \approx 0.9496$, while the scores of $\affinity$ and $\normvarentropy$ strictly increase.
Thus,in total it is not possible to predict the behaviour of $\crossconn$ when log complexity increases. \hfill$\square$
\end{proof}

\begin{theorem}
\label{theo:dfm-sep-entries}
$(\mathcal{C}^L, \separability) \in \norel$ for any log complexity measure $\mathcal{C}^L \in \loc$.
\end{theorem}
\begin{proof}
Consider the following event logs:
\begin{align*}
	L_1 &= [\langle a \rangle, \langle a,b,c \rangle] \\
	L_2 &= L_1 + [\langle a,b,c \rangle, \langle i,j,j,k \rangle] \\
	L_3 &= L_2 + [\langle a,b,c,d \rangle, \langle a,a,b,b,c,c \rangle, \langle i,i,j,j,k,k \rangle]
\end{align*}
\cref{fig:dfm-separability} shows the models $M_1, M_2, M_3$ found by the directly follows miner for the event logs $L_1, L_2, L_3$.
\begin{figure}[htp]
	\centering
	\scalebox{\scalefactor}{
	\begin{tikzpicture}[node distance = 1.1cm,>=stealth',bend angle=0,auto]
		\node[place,tokens=1] (start) {};
		\node[yshift=1cm] at (start) {$M_1$:};
		\node[transition,right of=start] (a) {$a$}
		edge [pre] (start);
		\node[place,right of=a] (p1) {}
		edge [pre] (a);
		\node[transition,right of=p1] (b) {$b$}
		edge [pre] (p1);
		\node[place,right of=b] (p2) {}
		edge [pre] (b);
		\node[transition,right of=p2] (c) {$c$}
		edge [pre] (p2);
		\node[place,right of=c] (p3) {}
		edge [pre] (c);
		\node[transition,right of=p3] (tau) {$\tau$}
		edge [pre] (p3);
		\node[place,right of=tau] (end) {}
		edge [pre] (tau);
		\node[transition,above of=c] (tau2) {$\tau$}
		edge [pre,bend right=15] (p1)
		edge [post,bend left=15] (end);
	\end{tikzpicture}}
	
	\medskip
	\hrule
	\medskip
	
	\scalebox{\scalefactor}{
	\begin{tikzpicture}[node distance = 1.1cm,>=stealth',bend angle=0,auto]
		\node[place,tokens=1] (start) {};
		\node[yshift=1cm] at (start) {$M_2$:};
		\node[transition,right of=start,yshift=0.5cm] (a) {$a$}
		edge [pre] (start);
		\node[place,right of=a] (p1) {}
		edge [pre] (a);
		\node[transition,right of=p1] (b) {$b$}
		edge [pre] (p1);
		\node[place,right of=b] (p2) {}
		edge [pre] (b);
		\node[transition,right of=p2] (c) {$c$}
		edge [pre] (p2);
		\node[place,right of=c] (p3) {}
		edge [pre] (c);
		\node[transition,right of=p3] (tau) {$\tau$}
		edge [pre] (p3);
		\node[place,right of=tau,yshift=-0.5cm] (end) {}
		edge [pre] (tau);
		\node[transition,above of=c] (tau2) {$\tau$}
		edge [pre,bend right=15] (p1)
		edge [post,bend left=25] (end);
		\node[transition,right of=start,yshift=-0.5cm] (i) {$i$}
		edge [pre] (start);
		\node[place,right of=i] (pi) {}
		edge [pre] (i);
		\node[transition,right of=pi] (j) {$j$}
		edge [pre] (pi);
		\node[place,right of=j] (pj) {}
		edge [pre] (j);
		\node[transition,right of=pj] (k) {$k$}
		edge [pre] (pj);
		\node[place,right of=k] (pk) {}
		edge [pre] (k);
		\node[transition,right of=pk] (tau3) {$\tau$}
		edge [pre] (pk)
		edge [post] (end);
		\node[transition,below of=pj] (j2) {$j$}
		edge [pre,bend right=15] (pj)
		edge [post,bend left=15] (pj);
	\end{tikzpicture}}
	
	\medskip
	\hrule
	\medskip
	
	\scalebox{\scalefactor}{
	\begin{tikzpicture}[node distance = 1cm,>=stealth',bend angle=0,auto]
		\node[place,tokens=1] (start) {};
		\node[yshift=1cm] at (start) {$M_3$:};
		\node[transition,right of=start,yshift=1cm] (a) {$a$}
		edge [pre] (start);
		\node[place,right of=a] (p1) {}
		edge [pre] (a);
		\node[transition,right of=p1] (b) {$b$}
		edge [pre] (p1);
		\node[place,right of=b] (p2) {}
		edge [pre] (b);
		\node[transition,right of=p2] (c) {$c$}
		edge [pre] (p2);
		\node[place,right of=c] (p3) {}
		edge [pre] (c);
		\node[transition,right of=p3] (tau) {$\tau$}
		edge [pre] (p3);
		\node[place,right of=tau,yshift=-1cm] (end) {}
		edge [pre] (tau);
		\node[transition,below of=c] (tau2) {$\tau$}
		edge [pre,bend left=10] (p1)
		edge [post,bend right=5] (end);
		\node[transition,above of=p1] (a2) {$a$}
		edge [pre,bend right=15] (p1)
		edge [post,bend left=15] (p1);
		\node[transition,above of=p2] (b2) {$b$}
		edge [pre,bend right=15] (p2)
		edge [post,bend left=15] (p2);
		\node[transition,above of=p3] (c2) {$c$}
		edge [pre,bend right=15] (p3)
		edge [post,bend left=15] (p3);
		\node[transition,above of=tau] (d) {$d$}
		edge [pre] (p3);
		\node[place,right of=d] (p4) {}
		edge [pre] (d);
		\node[transition,below of=p4] (tau4) {$\tau$}
		edge [pre] (p4)
		edge [post] (end);
		\node[transition,right of=start,yshift=-1cm] (i) {$i$}
		edge [pre] (start);
		\node[place,right of=i] (pi) {}
		edge [pre] (i);
		\node[transition,right of=pi] (j) {$j$}
		edge [pre] (pi);
		\node[place,right of=j] (pj) {}
		edge [pre] (j);
		\node[transition,right of=pj] (k) {$k$}
		edge [pre] (pj);
		\node[place,right of=k] (pk) {}
		edge [pre] (k);
		\node[transition,right of=pk] (tau3) {$\tau$}
		edge [pre] (pk)
		edge [post] (end);
		\node[transition,below of=pi] (i2) {$i$}
		edge [pre,bend right=15] (pi)
		edge [post,bend left=15] (pi);
		\node[transition,below of=pj] (j2) {$j$}
		edge [pre,bend right=15] (pj)
		edge [post,bend left=15] (pj);
		\node[transition,below of=pk] (k2) {$k$}
		edge [pre,bend right=15] (pk)
		edge [post,bend left=15] (pk);
	\end{tikzpicture}}
	\caption{The results of the directly follows miner for the input logs $L_1, L_2, L_3$ from the example in \cref{theo:dfm-sep-entries}. $M_1$ is the model mined from the log $L_1$, $M_2$ the model mined from the log $L_2$, and $M_3$ the model mined from the log $L_3$.}
	\label{fig:dfm-separability}
\end{figure}
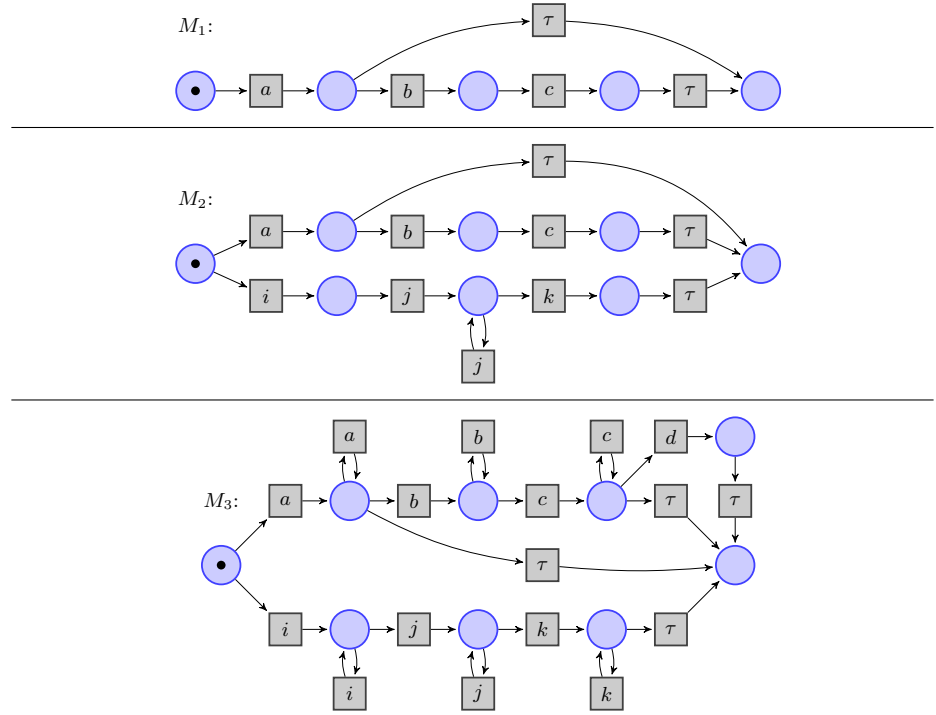
The complexity scores of these models are:
\begin{itemize}
	\item[•] $\separability(M_1) = 0.75$,
	\item[•] $\separability(M_2) \approx 0.9375$,
	\item[•] $\separability(M_3) = 0.75$,
\end{itemize}
so $\separability(M_1) < \separability(M_2)$, $\separability(M_2) > \separability(M_3)$, and $\separability(M_1) = \separability(M_3)$.
But the event logs $L_1, L_2, L_3$ have the following log complexity scores:
\begin{center}
	\begin{tabular}{|c|c|c|c|c|c|c|c|c|c|c|}\hline
		 & $\magnitude$ & $\variety$ & $\support$ & $\tlavg$ & $\tlmax$ & $\levelofdetail$ & $\numberofties$ & $\lempelziv$ & $\numberuniquetraces$ & $\percentageuniquetraces$ \\ \hline
		$L_1$ & $\pad 4 \pad$ & $\pad 3 \pad$ & $\pad 2 \pad$ & $\pad 2 \pad$ & $\pad 3 \pad$ & $\pad 2 \pad$ & $\pad 2 \pad$ & $\pad 3 \pad$ & $\pad 2 \pad$ & $\pad 1 \pad$ \\ \hline
		$L_2$ & $\pad 11 \pad$ & $\pad 6 \pad$ & $ \pad 4 \pad$ & $\pad 2.75 \pad$ & $\pad 4 \pad$ & $\pad 3 \pad$ & $\pad 4 \pad$ & $\pad 7 \pad$ & $\pad 3 \pad$ & $\pad 0.75 \pad$ \\ \hline
		$L_3$ & $\pad 27 \pad$ & $\pad 7 \pad$ & $\pad 7 \pad$ & $\pad 3.8571 \pad$ & $\pad 6 \pad$ & $\pad 4 \pad$ & $\pad 5 \pad$ & $\pad 14 \pad$ & $\pad 6 \pad$ & $\pad 0.8571 \pad$ \\ \hline
	\end{tabular}
	
	\medskip
	
	\begin{tabular}{|c|c|c|c|c|c|c|c|c|} \hline
		 & $\structure$ & $\affinity$ & $\deviationfromrandom$ & $\avgdist$ & $\varentropy$ & $\normvarentropy$ & $\seqentropy$ & $\normseqentropy$ \\ \hline
		$L_1$ & $\pad 2 \pad$ & $\pad 0 \pad$ & $\pad 0.3764 \pad$ & $\pad 2 \pad$ & $\pad 0 \pad$ & $\pad 0 \pad$ & $\pad 0 \pad$ & $\pad 0 \pad$ \\ \hline
		$L_2$ & $\pad 2.5 \pad$ & $\pad 0.1667 \pad$ & $\pad 0.5565 \pad$ & $\pad 3.8333 \pad$ & $\pad 4.7804 \pad$ & $\pad 0.3509 \pad$ & $\pad 7.2103 \pad$ & $\pad 0.2734 \pad$ \\ \hline
		$L_3$ & $\pad 2.8571 \pad$ & $\pad 0.1937 \pad$ & $\pad 0.6766 \pad$ & $\pad 5.2381 \pad$ & $\pad 24.842 \pad$ & $\pad 0.4775 \pad$ & $\pad 35.0271 \pad$ & $\pad 0.3936 \pad$ \\ \hline
	\end{tabular}
\end{center}
Therefore, $\mathcal{C}^L(L_1) < \mathcal{C}^L(L_2) < \mathcal{C}^L(L_3)$ for any event log complexity measure $\mathcal{C}^L \in (\loc \setminus \{\percentageuniquetraces\})$.
For $\percentageuniquetraces$, consider the following event logs:
\begin{align*}
	L_1 &= [\langle a \rangle^{4}, \langle a,b,c \rangle] \\
	L_2 &= L_2 + [\langle a,b,c \rangle, \langle i,j,j,k \rangle] \\
	L_3 &= L_3 + [\langle a,b,c,d \rangle, \langle a,a,b,b,c,c \rangle, \langle i,i,j,j,k,k \rangle]
\end{align*}
These event logs are the same as before, but the frequency of the trace $\langle a \rangle$ increased.
Thus, the directly follows models for these logs are the same as those in \cref{fig:dfm-separability}.
But these logs have an increasing percentage of unique traces, i.e., $\percentageuniquetraces(L_1) = 0.4 < \percentageuniquetraces(L_2) \approx 0.4286 < \percentageuniquetraces(L_3) = 0.6$.
Thus, we have $(\mathcal{C}^L, \separability) \in \norel$ for all $\mathcal{C}^L \in \loc$. \hfill$\square$
\end{proof}

\begin{theorem}
\label{theo:dfm-acd-entries}
$(\mathcal{C}^L, \avgconn) \in \norel$ for any log complexity measure $\mathcal{C}^L \in \loc$.
\end{theorem}
\begin{proof}
Let $L_1 \sqsubset L_2$ be event logs, $G_1, G_2$ be their directly follows graphs, and $M_1, M_2$ the models found by the directly follows miner for $L_1, L_2$.
By previous discussion, we know that $\avgconn(M_1) = \avgconn(G_1)$ and $\avgconn(M_2) = \avgconn(G_2)$.
Furthermore, by \cref{theo:dfg-acd-entries}, we know $\avgconn(G_1) < \avgconn(G_2)$, $\avgconn(G_1) > \avgconn(G_2)$, and $\avgconn(G_1) = \avgconn(G_2)$ are possible when log complexity increases.
Thus, $\avgconn(M_1) < \avgconn(M_2)$, $\avgconn(M_1) > \avgconn(M_2)$, and $\avgconn(M_1) = \avgconn(M_2)$ are all possible as well. \hfill$\square$
\end{proof}

\begin{theorem}
\label{theo:dfm-mcd-entries}
$(\mathcal{C}^L, \maxconn) \in \mleq$ for any log complexity measure $\mathcal{C}^L \in \loc$.
\end{theorem}
\begin{proof}
Let $L_1 \sqsubset L_2$ be event logs, $G_1, G_2$ be their directly follows graphs, and $M_1, M_2$ the models found by the directly follows miner for $L_1, L_2$.
By previous discussion, we know that $\maxconn(M_1) = \maxconn(G_1)$ and $\maxconn(M_2) = \maxconn(G_2)$.
Furthermore, by \cref{theo:dfg-mcd-diam-entries}, we know that $\mathcal{C}^L(L_1) < \mathcal{C}^L(L_2)$ always implies $\maxconn(G_1) \leq \maxconn(G_2)$, and $\maxconn(G_1) < \maxconn(G_2)$ and $\maxconn(G_1) = \maxconn(G_2)$ are both possible outcomes.
Thus, we can deduce that $\mathcal{C}^L(L_1) < \mathcal{C}^L(L_2)$ always implies $\maxconn(M_1) \leq \maxconn(M_2)$, and both $\maxconn(M_1) < \maxconn(M_2)$ and $\maxconn(M_1) = \maxconn(M_2)$ are possible outcomes. \hfill$\square$
\end{proof}

\begin{theorem}
\label{theo:dfm-seq-entries}
$(\mathcal{C}^L, \sequentiality) \in \norel$ for any log complexity measure $\mathcal{C}^L \in \loc$.
\end{theorem}
\begin{proof}
Consider the following event logs:
\begin{align*}
	L_1 &= [\langle a \rangle, \langle a,b,c \rangle] \\
	L_2 &= L_1 + [\langle a,a,b,b,c,c,d,d \rangle] \\
	L_3 &= L_2 + [\langle a,a,b,b,c,c,d,d \rangle, \langle f,g,h,i,j,k,l,m,n,o,p,q \rangle]
\end{align*}
\cref{fig:dfm-sequentiality} shows the models $M_1, M_2, M_3$ found by the directly follows miner for the event logs $L_1, L_2, L_3$.
\begin{figure}[htp]
	\centering
	\scalebox{\scalefactor}{
	\begin{tikzpicture}[node distance = 1.1cm,>=stealth',bend angle=0,auto]
		\node[place,tokens=1] (start) {};
		\node[yshift=1cm] at (start) {$M_1$:};
		\node[transition,right of=start] (a) {$a$}
		edge [pre] (start);
		\node[place,right of=a] (p1) {}
		edge [pre] (a);
		\node[transition,right of=p1] (b) {$b$}
		edge [pre] (p1);
		\node[place,right of=b] (p2) {}
		edge [pre] (b);
		\node[transition,right of=p2] (c) {$c$}
		edge [pre] (p2);
		\node[place,right of=c] (p3) {}
		edge [pre] (c);
		\node[transition,right of=p3] (tau) {$\tau$}
		edge [pre] (p3);
		\node[place,right of=tau] (end) {}
		edge [pre] (tau);
		\node[transition,below of=c] (tau2) {$\tau$}
		edge [pre,bend left=15] (p1)
		edge [post,bend right=15] (end);
	\end{tikzpicture}}
	
	\medskip
	\hrule
	\medskip
	
	\scalebox{\scalefactor}{
	\begin{tikzpicture}[node distance = 1.1cm,>=stealth',bend angle=0,auto]
		\node[place,tokens=1] (start) {};
		\node[yshift=1cm] at (start) {$M_2$:};
		\node[transition,right of=start] (a) {$a$}
		edge [pre] (start);
		\node[place,right of=a] (p1) {}
		edge [pre] (a);
		\node[transition,right of=p1] (b) {$b$}
		edge [pre] (p1);
		\node[place,right of=b] (p2) {}
		edge [pre] (b);
		\node[transition,right of=p2] (c) {$c$}
		edge [pre] (p2);
		\node[place,right of=c] (p3) {}
		edge [pre] (c);
		\node[transition,right of=p3] (d) {$d$}
		edge [pre] (p3);
		\node[place,right of=d] (p4) {}
		edge [pre] (d);
		\node[transition,right of=p4] (tau) {$\tau$}
		edge [pre] (p4);
		\node[place,right of=tau] (end) {}
		edge [pre] (tau);
		\node[transition,above of=p4] (d2) {$d$}
		edge [pre,bend right=15] (p4)
		edge [post,bend left=15] (p4);
		\node[transition,above of=d2] (tau2) {$\tau$}
		edge [pre] (p3)
		edge [post] (end);
		\node[transition,below of=c] (tau3) {$\tau$}
		edge [pre,bend left=15] (p1)
		edge [post,bend right=10] (end);
		\node[transition,above of=p1] (a2) {$a$}
		edge [pre,bend right=15] (p1)
		edge [post,bend left=15] (p1);
		\node[transition,above of=p2] (b2) {$b$}
		edge [pre,bend right=15] (p2)
		edge [post,bend left=15] (p2);
		\node[transition,above of=p3] (c2) {$c$}
		edge [pre,bend right=15] (p3)
		edge [post,bend left=15] (p3);
	\end{tikzpicture}}
	
	\medskip
	\hrule
	\medskip
	
	\scalebox{\scalefactor}{
	\begin{tikzpicture}[node distance = 1cm,>=stealth',bend angle=0,auto]
		\node[place,tokens=1] (start) {};
		\node[yshift=1cm] at (start) {$M_3$:};
		\node[transition,right of=start] (a) {$a$}
		edge [pre] (start);
		\node[place,right of=a] (p1) {}
		edge [pre] (a);
		\node[transition,right of=p1] (b) {$b$}
		edge [pre] (p1);
		\node[place,right of=b] (p2) {}
		edge [pre] (b);
		\node[transition,right of=p2] (c) {$c$}
		edge [pre] (p2);
		\node[place,right of=c] (p3) {}
		edge [pre] (c);
		\node[transition,right of=p3] (d) {$d$}
		edge [pre] (p3);
		\node[place,right of=d] (p4) {}
		edge [pre] (d);
		\node[transition,right of=p4] (tau) {$\tau$}
		edge [pre] (p4);
		\node[place,right of=tau] (end) {}
		edge [pre] (tau);
		\node[transition,above of=p4] (d2) {$d$}
		edge [pre,bend right=15] (p4)
		edge [post,bend left=15] (p4);
		\node[transition,above of=d2] (tau2) {$\tau$}
		edge [pre] (p3)
		edge [post] (end);
		\node[transition,below of=c] (tau3) {$\tau$}
		edge [pre,bend left=15] (p1)
		edge [post,bend right=10] (end);
		\node[transition,above of=p1] (a2) {$a$}
		edge [pre,bend right=15] (p1)
		edge [post,bend left=15] (p1);
		\node[transition,above of=p2] (b2) {$b$}
		edge [pre,bend right=15] (p2)
		edge [post,bend left=15] (p2);
		\node[transition,above of=p3] (c2) {$c$}
		edge [pre,bend right=15] (p3)
		edge [post,bend left=15] (p3);
		\node[below of=start] (dummy) {};
		\node[transition,below of=dummy,xshift=-0.5cm] (f) {$f$}
		edge [pre] (start);
		\node[place,below of=f] (p5) {}
		edge [pre] (f);
		\node[transition,right of=p5] (g) {$g$}
		edge [pre] (p5);
		\node[place,above of=g] (p6) {}
		edge [pre] (g);
		\node[transition,right of=p6] (h) {$h$}
		edge [pre] (p6);
		\node[place,below of=h] (p7) {}
		edge [pre] (h);
		\node[transition,right of=p7] (i) {$i$}
		edge [pre] (p7);
		\node[place,above of=i] (p8) {}
		edge [pre] (i);
		\node[transition,right of=p8] (j) {$j$}
		edge [pre] (p8);
		\node[place,below of=j] (p9) {}
		edge [pre] (j);
		\node[transition,right of=p9] (k) {$k$}
		edge [pre] (p9);
		\node[place,above of=k] (p10) {}
		edge [pre] (k);
		\node[transition,right of=p10] (l) {$l$}
		edge [pre] (p10);
		\node[place,below of=l] (p11) {}
		edge [pre] (l);
		\node[transition,right of=p11] (m) {$m$}
		edge [pre] (p11);
		\node[place,above of=m] (p12) {}
		edge [pre] (m);
		\node[transition,right of=p12] (n) {$n$}
		edge [pre] (p12);
		\node[place,below of=n] (p13) {}
		edge [pre] (n);
		\node[transition,right of=p13] (o) {$o$}
		edge [pre] (p13);
		\node[place,above of=o] (p14) {}
		edge [pre] (o);
		\node[transition,right of=p14] (p) {$p$}
		edge [pre] (p14);
		\node[place,below of=p] (p15) {}
		edge [pre] (p);
		\node[transition,right of=p15] (q) {$q$}
		edge [pre] (p15)
		edge [post] (end);
	\end{tikzpicture}}
	\caption{The results of the directly follows miner for the input logs $L_1, L_2, L_3$ from the example in \cref{theo:dfm-seq-entries}. $M_1$ is the model mined from the log $L_1$, $M_2$ the model mined from the log $L_2$, and $M_3$ the model mined from the log $L_3$.}
	\label{fig:dfm-sequentiality}
\end{figure}
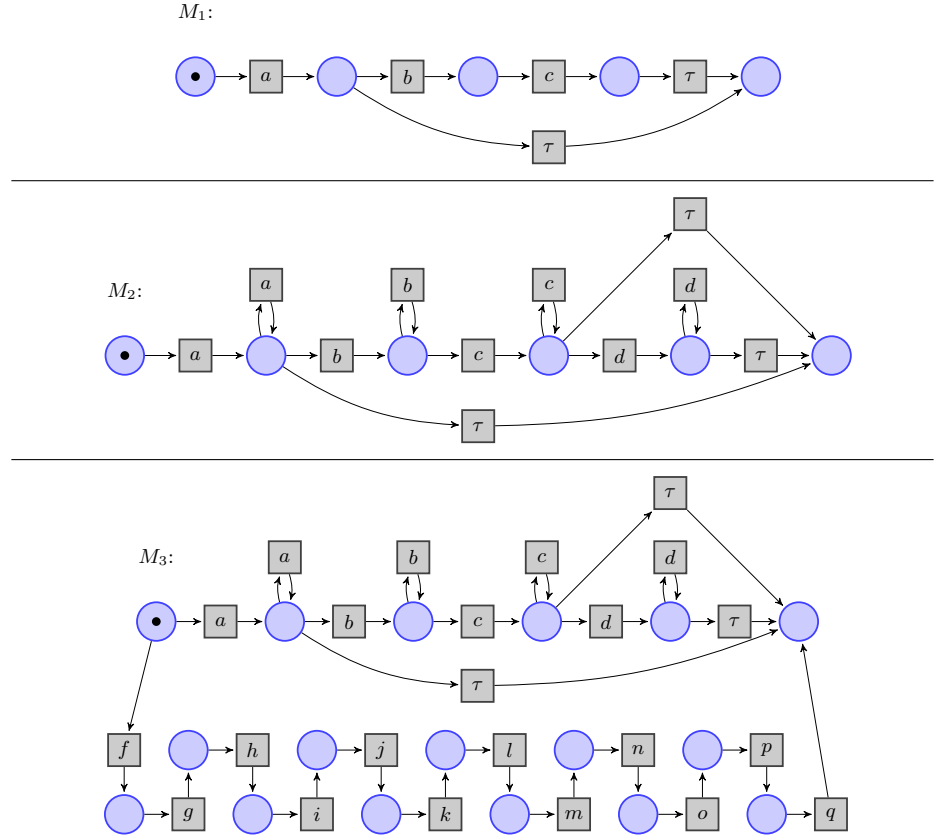
The complexity scores of these models are:
\begin{itemize}
	\item[•] $\sequentiality(M_1) = 0.5$,
	\item[•] $\sequentiality(M_2) = 0.9545$,
	\item[•] $\sequentiality(M_3) = 0.5$,
\end{itemize}
so $\sequentiality(M_1) < \sequentiality(M_2)$, $\sequentiality(M_2) > \sequentiality(M_3)$, and $\sequentiality(M_1) = \sequentiality(M_3)$.
But the event logs $L_1, L_2, L_3$ have the following log complexity scores:
\begin{center}
	\begin{tabular}{|c|c|c|c|c|c|c|c|c|c|c|}\hline
		 & $\magnitude$ & $\variety$ & $\support$ & $\tlavg$ & $\tlmax$ & $\levelofdetail$ & $\numberofties$ & $\lempelziv$ & $\numberuniquetraces$ & $\percentageuniquetraces$ \\ \hline
		$L_1$ & $\pad 4 \pad$ & $\pad 3 \pad$ & $\pad 2 \pad$ & $\pad 2 \pad$ & $\pad 3 \pad$ & $\pad 2 \pad$ & $\pad 2 \pad$ & $\pad 3 \pad$ & $\pad 2 \pad$ & $\pad 1 \pad$ \\ \hline
		$L_2$ & $\pad 12 \pad$ & $\pad 4 \pad$ & $\pad 3 \pad$ & $\pad 4 \pad$ & $\pad 8 \pad$ & $\pad 3 \pad$ & $\pad 3 \pad$ & $\pad 8 \pad$ & $\pad 3 \pad$ & $\pad 1 \pad$ \\ \hline
		$L_3$ & $\pad 32 \pad$ & $\pad 16 \pad$ & $\pad 5 \pad$ & $\pad 6.4 \pad$ & $\pad 12 \pad$ & $\pad 4 \pad$ & $\pad 14 \pad$ & $\pad 23 \pad$ & $\pad 4 \pad$ & $\pad 0.8 \pad$ \\ \hline
	\end{tabular}
	
	\medskip
	
	\begin{tabular}{|c|c|c|c|c|c|c|c|c|} \hline
		 & $\structure$ & $\affinity$ & $\deviationfromrandom$ & $\avgdist$ & $\varentropy$ & $\normvarentropy$ & $\seqentropy$ & $\normseqentropy$ \\ \hline
		$L_1$ & $\pad 2 \pad$ & $\pad 0 \pad$ & $\pad 0.3764 \pad$ & $\pad 2 \pad$ & $\pad 0 \pad$ & $\pad 0 \pad$ & $\pad 0 \pad$ & $\pad 0 \pad$ \\ \hline
		$L_2$ & $\pad 2.6667 \pad$ & $\pad 0.0952 \pad$ & $\pad 0.687 \pad$ & $\pad 4.6667 \pad$ & $\pad 6.1086 \pad$ & $\pad 0.2653 \pad$ & $\pad 8.1503 \pad$ & $\pad 0.2733 \pad$ \\ \hline
		$L_3$ & $\pad 4.8 \pad$ & $\pad 0.1571 \pad$ & $\pad 0.7484 \pad$ & $\pad 9.4 \pad$ & $\pad 21.2668 \pad$ & $\pad 0.3127 \pad$ & $\pad 33.3873 \pad$ & $\pad 0.301 \pad$ \\ \hline
	\end{tabular}
\end{center}
Therefore, $\mathcal{C}^L(L_1) < \mathcal{C}^L(L_2) < \mathcal{C}^L(L_3)$ for any event log complexity measure $\mathcal{C}^L \in (\loc \setminus \{\percentageuniquetraces\})$.
For $\percentageuniquetraces$, consider the following event logs:
\begin{align*}
	L_1 &= [\langle a \rangle, \langle a,b,c \rangle^{5}] \\
	L_2 &= L_1 + [\langle a,a,b,b,c,c,d,d \rangle] \\
	L_3 &= L_2 + [\langle a,a,b,b,c,c,d,d \rangle, \langle f,g,h,i,j,k,l,m,n,o,p,q \rangle]
\end{align*}
These event logs are the same as before, but the frequency of the trace $\langle a,b,c \rangle$ increased.
Thus, the directly follows models for these logs are the same as those in \cref{fig:dfm-sequentiality}. 
But these logs have an increasing percentage of unique traces, i.e., $\percentageuniquetraces(L_1) \approx 0.3333 < \percentageuniquetraces(L_2) \approx 0.4286 < \percentageuniquetraces(L_3) \approx 0.4444$.
Thus, we have $(\mathcal{C}^L, \sequentiality) \in \norel$ for all $\mathcal{C}^L \in \loc$. \hfill$\square$
\end{proof}

\begin{theorem}
\label{theo:dfm-depth-entries}
$(\mathcal{C}^L, \depth) \in \norel$ for any log complexity measure $\mathcal{C}^L \in \loc$.
\end{theorem}
\begin{proof}
Consider the following event logs:
\begin{align*}
	L_1 &= [\langle a,b \rangle^{2}, \langle c,x \rangle^{2}, \langle d,y \rangle^{2}, \langle e,z \rangle] \\
	L_2 &= L_1 + [\langle a,b \rangle, \langle a,g,b \rangle, \langle a,g,g,b \rangle] \\
	L_3 &= L_2 + [\langle a,g,g,g,b \rangle^{2}, \langle b,c \rangle, \langle h,i \rangle]
\end{align*}
\cref{fig:dfm-depth} shows the models $M_1, M_2, M_3$ found by the directly follows miner for the event logs $L_1, L_2, L_3$.
\begin{figure}[htp]
	\centering
	\scalebox{\scalefactor}{
	\begin{tikzpicture}[node distance = 1cm,>=stealth',bend angle=0,auto]
		\node[place,tokens=1] (start) {};
		\node[yshift=1cm] at (start) {$M_1$:};
		\node[transition,right of=start,yshift=1.5cm] (a) {$a$}
		edge [pre] (start);
		\node[transition,right of=start,yshift=0.5cm] (c) {$c$}
		edge [pre] (start);
		\node[transition,right of=start,yshift=-0.5cm] (d) {$d$}
		edge [pre] (start);
		\node[transition,right of=start,yshift=-1.5cm] (e) {$e$}
		edge [pre] (start);
		\node[place,right of=a] (pa) {}
		edge [pre] (a);
		\node[place,right of=c] (pc) {}
		edge [pre] (c);
		\node[place,right of=d] (pd) {}
		edge [pre] (d);
		\node[place,right of=e] (pe) {}
		edge [pre] (e);
		\node[transition,right of=pa] (b) {$b$}
		edge [pre] (pa);
		\node[transition,right of=pc] (x) {$x$}
		edge [pre] (pc);
		\node[transition,right of=pd] (y) {$y$}
		edge [pre] (pd);
		\node[transition,right of=pe] (z) {$z$}
		edge [pre] (pe);
		\node[place,right of=b] (pb) {}
		edge [pre] (b);
		\node[place,right of=x] (px) {}
		edge [pre] (x);
		\node[place,right of=y] (py) {}
		edge [pre] (y);
		\node[place,right of=z] (pz) {}
		edge [pre] (z);
		\node[transition,right of=pb] (tau1) {$\tau$}
		edge [pre] (pb);
		\node[transition,right of=px] (tau2) {$\tau$}
		edge [pre] (px);
		\node[transition,right of=py] (tau3) {$\tau$}
		edge [pre] (py);
		\node[transition,right of=pz] (tau4) {$\tau$}
		edge [pre] (pz);
		\node[place,right of=tau2,yshift=-0.5cm] (end) {}
		edge [pre] (tau1)
		edge [pre] (tau2)
		edge [pre] (tau3)
		edge [pre] (tau4);
	\end{tikzpicture}}
	
	\medskip
	\hrule
	\medskip
	
	\scalebox{\scalefactor}{
	\begin{tikzpicture}[node distance = 1cm,>=stealth',bend angle=0,auto]
		\node[place,tokens=1] (start) {};
		\node[yshift=1cm] at (start) {$M_2$:};
		\node[transition,right of=start,yshift=1.5cm] (a) {$a$}
		edge [pre] (start);
		\node[transition,right of=start,yshift=0.5cm] (c) {$c$}
		edge [pre] (start);
		\node[transition,right of=start,yshift=-0.5cm] (d) {$d$}
		edge [pre] (start);
		\node[transition,right of=start,yshift=-1.5cm] (e) {$e$}
		edge [pre] (start);
		\node[place,right of=a] (pa) {}
		edge [pre] (a);
		\node[place,right of=c] (pc) {}
		edge [pre] (c);
		\node[place,right of=d] (pd) {}
		edge [pre] (d);
		\node[place,right of=e] (pe) {}
		edge [pre] (e);
		\node[transition,right of=pa] (b) {$b$}
		edge [pre] (pa);
		\node[transition,right of=pc] (x) {$x$}
		edge [pre] (pc);
		\node[transition,right of=pd] (y) {$y$}
		edge [pre] (pd);
		\node[transition,right of=pe] (z) {$z$}
		edge [pre] (pe);
		\node[place,right of=b] (pb) {}
		edge [pre] (b);
		\node[place,right of=x] (px) {}
		edge [pre] (x);
		\node[place,right of=y] (py) {}
		edge [pre] (y);
		\node[place,right of=z] (pz) {}
		edge [pre] (z);
		\node[transition,right of=pb] (tau1) {$\tau$}
		edge [pre] (pb);
		\node[transition,right of=px] (tau2) {$\tau$}
		edge [pre] (px);
		\node[transition,right of=py] (tau3) {$\tau$}
		edge [pre] (py);
		\node[transition,right of=pz] (tau4) {$\tau$}
		edge [pre] (pz);
		\node[place,right of=tau2,yshift=-0.5cm] (end) {}
		edge [pre] (tau1)
		edge [pre] (tau2)
		edge [pre] (tau3)
		edge [pre] (tau4);
		\node[transition,above of=pa] (g1) {$g$}
		edge [pre] (pa);
		\node[place,right of=g1] (pg) {}
		edge [pre] (g1);
		\node[transition,above of=pg] (g2) {$g$}
		edge [pre,bend right=15] (pg)
		edge [post,bend left=15] (pg);
		\node[transition,right of=pg] (b2) {$b$}
		edge [pre] (pg)
		edge [post] (pb);
	\end{tikzpicture}}
	
	\medskip
	\hrule
	\medskip
	
	\scalebox{\scalefactor}{
	\begin{tikzpicture}[node distance = 1cm,>=stealth',bend angle=0,auto]
		\node[place,tokens=1] (start) {};
		\node[yshift=1cm] at (start) {$M_3$:};
		\node[transition,right of=start,yshift=2.5cm] (a) {$a$}
		edge [pre] (start);
		\node[transition,right of=start,yshift=1cm] (c) {$c$}
		edge [pre] (start);
		\node[transition,right of=start,yshift=-0.5cm] (d) {$d$}
		edge [pre] (start);
		\node[transition,right of=start,yshift=-1.5cm] (e) {$e$}
		edge [pre] (start);
		\node[place,right of=a] (pa) {}
		edge [pre] (a);
		\node[place,right of=c] (pc) {}
		edge [pre] (c);
		\node[place,right of=d] (pd) {}
		edge [pre] (d);
		\node[place,right of=e] (pe) {}
		edge [pre] (e);
		\node[transition,right of=pa] (b) {$b$}
		edge [pre] (pa);
		\node[transition,right of=pc] (x) {$x$}
		edge [pre] (pc);
		\node[transition,right of=pd] (y) {$y$}
		edge [pre] (pd);
		\node[transition,right of=pe] (z) {$z$}
		edge [pre] (pe);
		\node[place,right of=b] (pb) {}
		edge [pre] (b);
		\node[place,right of=x] (px) {}
		edge [pre] (x);
		\node[place,right of=y] (py) {}
		edge [pre] (y);
		\node[place,right of=z] (pz) {}
		edge [pre] (z);
		\node[transition,right of=pb] (tau1) {$\tau$}
		edge [pre] (pb);
		\node[transition,right of=px] (tau2) {$\tau$}
		edge [pre] (px);
		\node[transition,right of=py] (tau3) {$\tau$}
		edge [pre] (py);
		\node[transition,right of=pz] (tau4) {$\tau$}
		edge [pre] (pz);
		\node[place,right of=tau2,yshift=-0.5cm] (end) {}
		edge [pre] (tau1)
		edge [pre] (tau2)
		edge [pre] (tau3)
		edge [pre] (tau4);
		\node[transition,above of=pa] (g1) {$g$}
		edge [pre] (pa);
		\node[place,right of=g1] (pg) {}
		edge [pre] (g1);
		\node[transition,above of=pg] (g2) {$g$}
		edge [pre,bend right=15] (pg)
		edge [post,bend left=15] (pg);
		\node[transition,right of=pg] (b2) {$b$}
		edge [pre] (pg)
		edge [post] (pb);
		\node[transition] (c2) at ($0.5*(b) + 0.5*(x)$) {$c$}
		edge [pre] (pb);
		edge [post] (pc);
		\node[transition] (tau5) at ($0.5*(px) + 0.5*(y)$) {$\tau$}
		edge [pre] (pc)
		edge [post] (end);
		\node[transition,below of=e] (h) {$h$}
		edge [pre,bend left=15] (start);
		\node[place,right of=h] (ph) {}
		edge [pre] (h);
		\node[transition,right of=ph] (i) {$i$}
		edge [pre] (ph);
		\node[place,right of=i] (pi) {}
		edge [pre] (i);
		\node[transition,right of=pi] (tau6) {$\tau$}
		edge [pre] (pi)
		edge [post,bend right=15] (end);
	\end{tikzpicture}}
	\caption{The results of the directly follows miner for the input logs $L_1, L_2, L_3$ from the example in \cref{theo:dfm-depth-entries}. $M_1$ is the model mined from the log $L_1$, $M_2$ the model mined from the log $L_2$, and $M_3$ the model mined from the log $L_3$.}
	\label{fig:dfm-depth}
\end{figure}
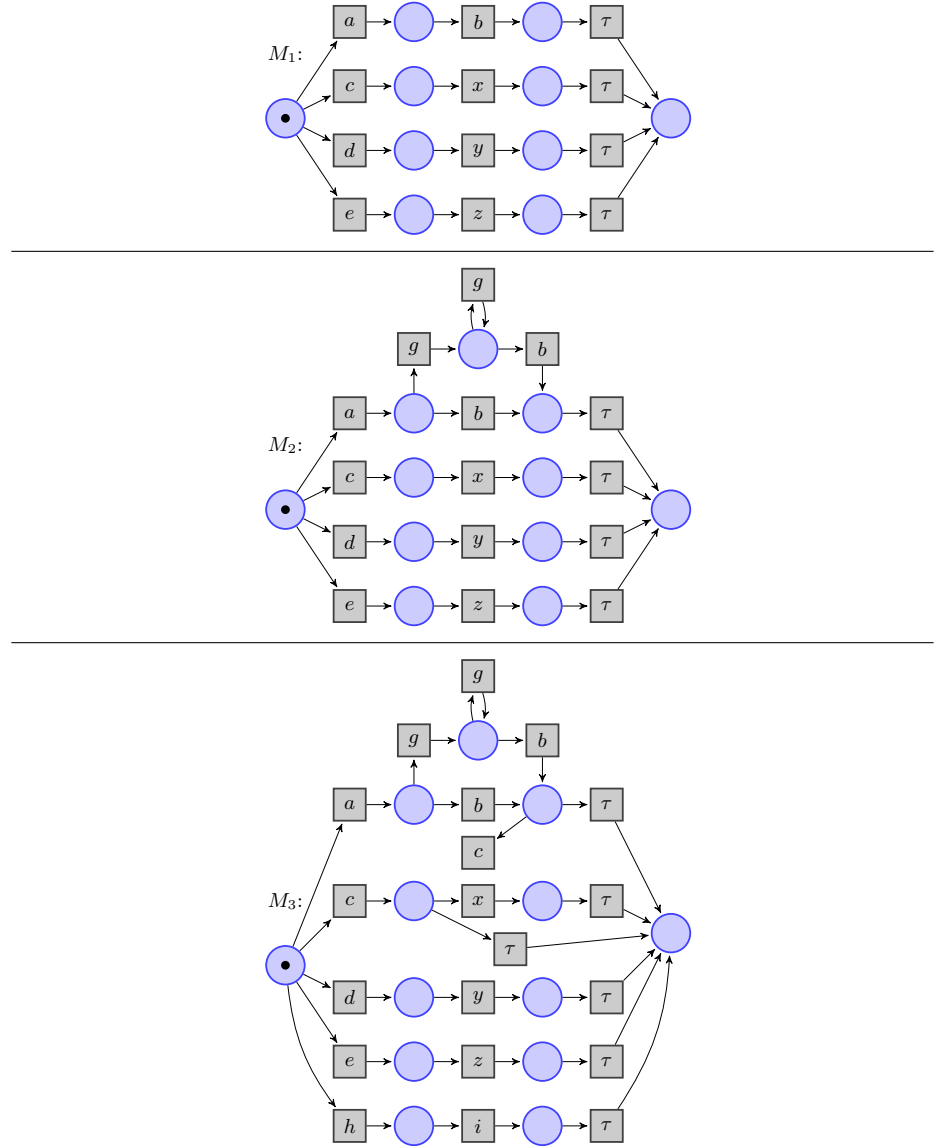
The complexity scores of these models are:
\begin{itemize}
	\item[•] $\depth(M_1) = 1$,
	\item[•] $\depth(M_2) = 2$,
	\item[•] $\depth(M_3) = 1$,
\end{itemize}
thus, we get that $\depth(M_1) < \depth(M_2)$, $\depth(M_2) > \depth(M_3)$, and $\depth(M_1) = \depth(M_3)$.
But the event logs $L_1, L_2, L_3$ have the following log complexity scores:
\begin{center}
	\begin{tabular}{|c|c|c|c|c|c|c|c|c|c|c|}\hline
		 & $\magnitude$ & $\variety$ & $\support$ & $\tlavg$ & $\tlmax$ & $\levelofdetail$ & $\numberofties$ & $\lempelziv$ & $\numberuniquetraces$ & $\percentageuniquetraces$ \\ \hline
		$L_1$ & $\pad 14 \pad$ & $\pad 8 \pad$ & $\pad 7 \pad$ & $\pad 2 \pad$ & $\pad 2 \pad$ & $\pad 4 \pad$ & $\pad 4 \pad$ & $\pad 11 \pad$ & $\pad 4 \pad$ & $\pad 0.5714 \pad$ \\ \hline
		$L_2$ & $\pad 23 \pad$ & $\pad 9 \pad$ & $\pad 10 \pad$ & $\pad 2.3 \pad$ & $\pad 4 \pad$ & $\pad 5 \pad$ & $\pad 6 \pad$ & $\pad 15 \pad$ & $\pad 6 \pad$ & $\pad 0.6 \pad$ \\ \hline
		$L_3$ & $\pad 37 \pad$ & $\pad 11 \pad$ & $\pad 14 \pad$ & $\pad 2.6429 \pad$ & $\pad 5 \pad$ & $\pad 14 \pad$ & $\pad 8 \pad$ & $\pad 21 \pad$ & $\pad 9 \pad$ & $\pad 0.6429 \pad$ \\ \hline
	\end{tabular}
	
	\medskip
	
	\begin{tabular}{|c|c|c|c|c|c|c|c|c|} \hline
		 & $\structure$ & $\affinity$ & $\deviationfromrandom$ & $\avgdist$ & $\varentropy$ & $\normvarentropy$ & $\seqentropy$ & $\normseqentropy$ \\ \hline
		$L_1$ & $\pad 2 \pad$ & $\pad 0.1429 \pad$ & $\pad 0.5003 \pad$ & $\pad 3.4286 \pad$ & $\pad 11.0904 \pad$ & $\pad 0.6667 \pad$ & $\pad 18.925 \pad$ & $\pad 0.5122 \pad$ \\ \hline
		$L_2$ & $\pad 2.2 \pad$ & $\pad 0.1259 \pad$ & $\pad 0.616 \pad$ & $\pad 3.4889 \pad$ & $\pad 21.5011 \pad$ & $\pad 0.7211 \pad$ & $\pad 38.3221 \pad$ & $\pad 0.5314 \pad$ \\ \hline
		$L_3$ & $\pad 2.2857 \pad$ & $\pad 0.1099 \pad$ & $\pad 0.6295 \pad$ & $\pad 3.8571 \pad$ & $\pad 39.55 \pad$ & $\pad 0.7602 \pad$ & $\pad 76.1913 \pad$ & $\pad 0.5703 \pad$ \\ \hline
	\end{tabular}
\end{center}
Therefore, $\mathcal{C}^L(L_1) < \mathcal{C}^L(L_2) < \mathcal{C}^L(L_3)$ for any event log complexity measure $\mathcal{C}^L \in (\loc \setminus \{\affinity\}$. 
For $\affinity$, consider the following event logs:
\begin{align*}
	L_1 &= [\langle a,b \rangle, \langle c,x \rangle, \langle d,y \rangle, \langle e,z \rangle] \\
	L_2 &= L_1 + [\langle a,b \rangle, \langle a,g,g,b \rangle] \\
	L_3 &= L_2 + [\langle a,g,g,g,b \rangle^{3}, \langle b,c \rangle, \langle h,i \rangle]
\end{align*}
The directly follows models for these logs are the same as those in \cref{fig:dfm-depth}.
But, for these logs, $\affinity(L_1) = 0 < \affinity(L_2) \approx 0.0667 < \affinity(L_3) \approx 0.1273$.
Thus, we have $(\mathcal{C}^L, \depth) \in \norel$ for all $\mathcal{C}^L \in \loc$. \hfill$\square$
\end{proof}

\begin{theorem}
\label{theo:dfm-diam-entries}
$(\mathcal{C}^L, \diameter) \in \mleq$ for any log complexity measure $\mathcal{C}^L \in \loc$.
\end{theorem}
\begin{proof}
Let $L_1 \sqsubset L_2$ be event logs, $G_1, G_2$ be their directly follows graphs, and $M_1, M_2$ the models found by the directly follows miner for $L_1, L_2$.
By the introductory discussion of this subsection, we know that $\diameter(M_1) = 2 \diameter(G_1) - 1$ and that $\diameter(M_2) = 2 \diameter(G_2) - 1$.
Furthermore, by \cref{theo:dfg-mcd-diam-entries}, we know that $\mathcal{C}^L(L_1) < \mathcal{C}^L(L_2)$ always implies $\diameter(G_1) \leq \diameter(G_2)$.
Thus, such an increase in log complexity also implies that the diameter scores of $M_1$ and $M_2$ fulfill $\diameter(M_1) = 2 \diameter(G_1) - 1 \leq 2 \diameter(G_2) - 1 = \diameter(M_2)$. \hfill$\square$
\end{proof}

\begin{theorem}
\label{theo:dfm-cyc-entries}
$(\mathcal{C}^L, \cyclicity) \in \norel$ for any log complexity measure $\mathcal{C}^L \in \loc$.
\end{theorem}
\begin{proof}
Consider the following event logs:
\begin{align*}
	L_1 &= [\langle a \rangle, \langle a,b,c,c \rangle] \\
	L_2 &= L_1 + [\langle a,b,b,c,d,e \rangle] \\
	L_3 &= L_2 + [\langle a,b,b,b,c,d,d,e \rangle^{2}, \langle v,w,x,y,z \rangle]
\end{align*}
\cref{fig:dfm-cyclicity} shows the models $M_1, M_2, M_3$ found by the directly follows miner for the event logs $L_1, L_2, L_3$.
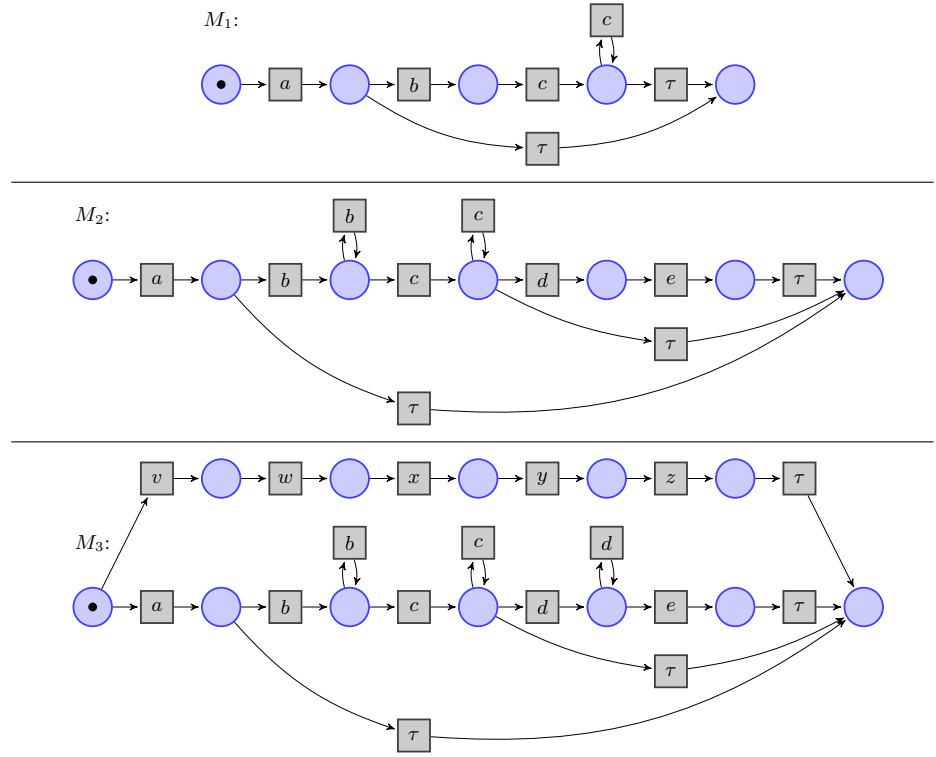
\begin{figure}[htp]
	\centering
	\scalebox{\scalefactor}{
	\begin{tikzpicture}[node distance = 1cm,>=stealth',bend angle=0,auto]
		\node[place,tokens=1] (start) {};
		\node[yshift=1cm] at (start) {$M_1$:};
		\node[transition,right of=start] (a) {$a$}
		edge [pre] (start);
		\node[place,right of=a] (pa) {}
		edge [pre] (a);
		\node[transition,right of=pa] (b) {$b$}
		edge [pre] (pa);
		\node[place,right of=b] (pb) {}
		edge [pre] (b);
		\node[transition,right of=pb] (c) {$c$}
		edge [pre] (pb);
		\node[place,right of=c] (pc) {}
		edge [pre] (c);
		\node[transition,above of=pc] (c2) {$c$}
		edge [pre,bend right=15] (pc)
		edge [post,bend left=15] (pc);
		\node[transition,right of=pc] (tau) {$\tau$}
		edge [pre] (pc);
		\node[place,right of=tau] (end) {}
		edge [pre] (tau);
		\node[transition,below of=c] (tau2) {$\tau$}
		edge [pre,bend left=15] (pa)
		edge [post,bend right=15] (end);
	\end{tikzpicture}}
	
	\medskip
	\hrule
	\medskip
	
	\scalebox{\scalefactor}{
	\begin{tikzpicture}[node distance = 1cm,>=stealth',bend angle=0,auto]
		\node[place,tokens=1] (start) {};
		\node[yshift=1cm] at (start) {$M_2$:};
		\node[transition,right of=start] (a) {$a$}
		edge [pre] (start);
		\node[place,right of=a] (pa) {}
		edge [pre] (a);
		\node[transition,right of=pa] (b) {$b$}
		edge [pre] (pa);
		\node[place,right of=b] (pb) {}
		edge [pre] (b);
		\node[transition,above of=pb] (b2) {$b$}
		edge [pre,bend right=15] (pb)
		edge [post,bend left=15] (pb);
		\node[transition,right of=pb] (c) {$c$}
		edge [pre] (pb);
		\node[place,right of=c] (pc) {}
		edge [pre] (c);
		\node[transition,above of=pc] (c2) {$c$}
		edge [pre,bend right=15] (pc)
		edge [post,bend left=15] (pc);
		\node[transition,right of=pc] (d) {$d$}
		edge [pre] (pc);
		\node[place,right of=d] (pd) {}
		edge [pre] (d);
		\node[transition,right of=pd] (e) {$e$}
		edge [pre] (pd);
		\node[place,right of=e] (pe) {}
		edge [pre] (e);
		\node[transition,right of=pe] (tau) {$\tau$}
		edge [pre] (pe);
		\node[place,right of=tau] (end) {}
		edge [pre] (tau);
		\node[transition,below of=e] (tau2) {$\tau$}
		edge [pre,bend left=10] (pc)
		edge [post,bend right=10] (end);
		\node[transition,below of=c,yshift=-1cm] (tau3) {$\tau$}
		edge [pre,bend left=15] (pa)
		edge [post,bend right=20] (end);
	\end{tikzpicture}}
	
	\medskip
	\hrule
	\medskip
	
	\scalebox{\scalefactor}{
	\begin{tikzpicture}[node distance = 1cm,>=stealth',bend angle=0,auto]
		\node[place,tokens=1] (start) {};
		\node[yshift=1cm] at (start) {$M_3$:};
		\node[transition,right of=start] (a) {$a$}
		edge [pre] (start);
		\node[place,right of=a] (pa) {}
		edge [pre] (a);
		\node[transition,right of=pa] (b) {$b$}
		edge [pre] (pa);
		\node[place,right of=b] (pb) {}
		edge [pre] (b);
		\node[transition,above of=pb] (b2) {$b$}
		edge [pre,bend right=15] (pb)
		edge [post,bend left=15] (pb);
		\node[transition,right of=pb] (c) {$c$}
		edge [pre] (pb);
		\node[place,right of=c] (pc) {}
		edge [pre] (c);
		\node[transition,above of=pc] (c2) {$c$}
		edge [pre,bend right=15] (pc)
		edge [post,bend left=15] (pc);
		\node[transition,right of=pc] (d) {$d$}
		edge [pre] (pc);
		\node[place,right of=d] (pd) {}
		edge [pre] (d);
		\node[transition,above of=pd] (d2) {$d$}
		edge [pre,bend right=15] (pd)
		edge [post,bend left=15] (pd);
		\node[transition,right of=pd] (e) {$e$}
		edge [pre] (pd);
		\node[place,right of=e] (pe) {}
		edge [pre] (e);
		\node[transition,right of=pe] (tau) {$\tau$}
		edge [pre] (pe);
		\node[place,right of=tau] (end) {}
		edge [pre] (tau);
		\node[transition,below of=e] (tau2) {$\tau$}
		edge [pre,bend left=10] (pc)
		edge [post,bend right=10] (end);
		\node[transition,below of=c,yshift=-1cm] (tau3) {$\tau$}
		edge [pre,bend left=15] (pa)
		edge [post,bend right=20] (end);
		\node[above of=a] (dummy) {};
		\node[transition,above of=dummy] (v) {$v$}
		edge [pre] (start);
		\node[place,right of=v] (pv) {}
		edge [pre] (v);
		\node[transition,right of=pv] (w) {$w$}
		edge [pre] (pv);
		\node[place,right of=w] (pw) {}
		edge [pre] (w);
		\node[transition,right of=pw] (x) {$x$}
		edge [pre] (pw);
		\node[place,right of=x] (px) {}
		edge [pre] (x);
		\node[transition,right of=px] (y) {$y$}
		edge [pre] (px);
		\node[place,right of=y] (py) {}
		edge [pre] (y);
		\node[transition,right of=py] (z) {$z$}
		edge [pre] (py);
		\node[place,right of=z] (pz) {}
		edge [pre] (z);
		\node[transition,right of=pz] (tau4) {$\tau$}
		edge [pre] (pz)
		edge [post] (end);
	\end{tikzpicture}}
	\caption{The results of the directly follows miner for the input logs $L_1, L_2, L_3$ from the example in \cref{theo:dfm-cyc-entries}. $M_1$ is the model mined from the log $L_1$, $M_2$ the model mined from the log $L_2$, and $M_3$ the model mined from the log $L_3$.}
	\label{fig:dfm-cyclicity}
\end{figure}
The complexity scores of these models are:
\begin{itemize}
	\item[•] $\cyclicity(M_1) \approx 0.2222$,
	\item[•] $\cyclicity(M_2) \approx 0.2667$,
	\item[•] $\cyclicity(M_3) \approx 0.2222$,
\end{itemize}
so $\cyclicity(M_1) < \cyclicity(M_2)$, $\cyclicity(M_2) > \cyclicity(M_3)$, and $\cyclicity(M_1) = \cyclicity(M_3)$.
But the event logs $L_1, L_2, L_3$ have the following log complexity scores:
\begin{center}
	\begin{tabular}{|c|c|c|c|c|c|c|c|c|c|c|}\hline
		 & $\magnitude$ & $\variety$ & $\support$ & $\tlavg$ & $\tlmax$ & $\levelofdetail$ & $\numberofties$ & $\lempelziv$ & $\numberuniquetraces$ & $\percentageuniquetraces$ \\ \hline
		$L_1$ & $\pad 5 \pad$ & $\pad 3 \pad$ & $\pad 2 \pad$ & $\pad 2.5 \pad$ & $\pad 4 \pad$ & $\pad 2 \pad$ & $\pad 2 \pad$ & $\pad 3 \pad$ & $\pad 2 \pad$ & $\pad 1 \pad$ \\ \hline
		$L_2$ & $\pad 11 \pad$ & $\pad 5 \pad$ & $\pad 3 \pad$ & $\pad 3.6667 \pad$ & $\pad 6 \pad$ & $\pad 3 \pad$ & $\pad 4 \pad$ & $\pad 8 \pad$ & $\pad 3 \pad$ & $\pad 1 \pad$ \\ \hline
		$L_3$ & $\pad 32 \pad$ & $\pad 10 \pad$ & $\pad 6 \pad$ & $\pad 5.3333 \pad$ & $\pad 8 \pad$ & $\pad 4 \pad$ & $\pad 8 \pad$ & $\pad 18 \pad$ & $\pad 5 \pad$ & $\pad 0.8333 \pad$ \\ \hline
	\end{tabular}
	
	\medskip
	
	\begin{tabular}{|c|c|c|c|c|c|c|c|c|} \hline
		 & $\structure$ & $\affinity$ & $\deviationfromrandom$ & $\avgdist$ & $\varentropy$ & $\normvarentropy$ & $\seqentropy$ & $\normseqentropy$ \\ \hline
		$L_1$ & $\pad 2 \pad$ & $\pad 0 \pad$ & $\pad 0.5286 \pad$ & $\pad 3 \pad$ & $\pad 0 \pad$ & $\pad 0 \pad$ & $\pad 0 \pad$ & $\pad 0 \pad$ \\ \hline
		$L_2$ & $\pad 3 \pad$ & $\pad 0.1111 \pad$ & $\pad 0.6159 \pad$ & $\pad 4 \pad$ & $\pad 5.5452 \pad$ & $\pad 0.3333 \pad$ & $\pad 7.2103 \pad$ & $\pad 0.2734 \pad$ \\ \hline
		$L_3$ & $\pad 4 \pad$ & $\pad 0.2381 \pad$ & $\pad 0.662 \pad$ & $\pad 6.2667 \pad$ & $\pad 24.842 \pad$ & $\pad 0.4775 \pad$ & $\pad 42.7031 \pad$ & $\pad 0.385 \pad$ \\ \hline
	\end{tabular}
\end{center}
Therefore, $\mathcal{C}^L(L_1) < \mathcal{C}^L(L_2) < \mathcal{C}^L(L_3)$ for any event log complexity measure $\mathcal{C}^L \in (\loc \setminus \{\percentageuniquetraces\})$.
For $\percentageuniquetraces$, consider the following event logs:
\begin{align*}
	L_1 &= [\langle a \rangle^{3}, \langle a,b,c,c \rangle] \\
	L_2 &= L_1 + [\langle a,b,b,c,d,e \rangle] \\
	L_3 &= L_2 + [\langle a,b,b,b,c,d,d,e \rangle^{2}, \langle v,w,x,y,z \rangle]
\end{align*}
These event logs are the same as before, but the frequency of the trace $\langle a \rangle$ increased.
Thus, the directly follows models for these logs are the same as those in \cref{fig:dfm-cyclicity}.
But these logs have an increasing percentage of unique traces, i.e., $\percentageuniquetraces(L_1) = 0.5 < \percentageuniquetraces(L_2) = 0.6 < \percentageuniquetraces(L_3) = 0.625$.
Thus, we have $(\mathcal{C}^L, \cyclicity) \in \norel$ for all $\mathcal{C}^L \in \loc$. \hfill$\square$
\end{proof}

\begin{theorem}
\label{theo:dfm-cnc-entries}
$(\mathcal{C}^L, \netconn) \in \norel$ for any log complexity measure $\mathcal{C}^L \in \loc$.
\end{theorem}
\begin{proof}
Consider the following event logs:
\begin{align*}
	L_1 &= [\langle a,a,b,b,c,c,d,d \rangle, \langle b,c,d \rangle^{3}] \\
	L_2 &= L_1 + [\langle b,c,d \rangle, \langle a,a,b,b,c,c,d,d,e,e \rangle, \langle a,b,c,d,e \rangle] \\
	L_3 &= L_2 + [\langle a,a,a,b,b,b,c,c,c,d,d,d,e,e,e \rangle, \langle u,v,x,x,y,z \rangle]
\end{align*}
\cref{fig:dfm-netconn} shows the models $M_1, M_2, M_3$ found by the directly follows miner for the event logs $L_1, L_2, L_3$.
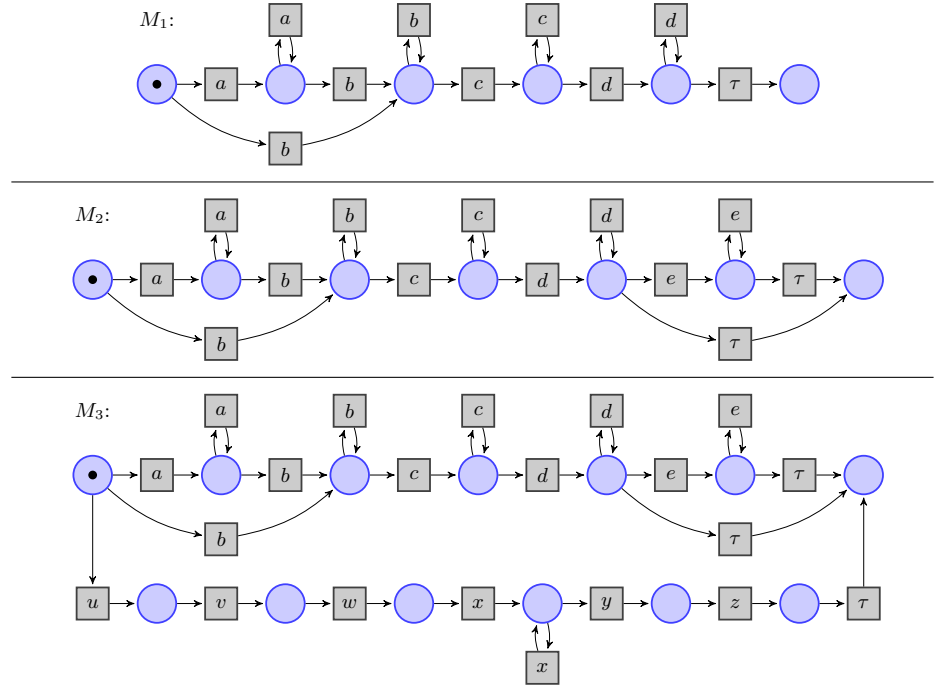
\begin{figure}[htp]
	\centering
	\scalebox{\scalefactor}{
	\begin{tikzpicture}[node distance = 1cm,>=stealth',bend angle=0,auto]
		\node[place,tokens=1] (start) {};
		\node[yshift=1cm] at (start) {$M_1$:};
		\node[transition,right of=start] (a) {$a$}
		edge [pre] (start);
		\node[place,right of=a] (pa) {}
		edge [pre] (a);
		\node[transition,above of=pa] (a2) {$a$}
		edge [pre,bend right=15] (pa)
		edge [post,bend left=15] (pa);
		\node[transition,right of=pa] (b) {$b$}
		edge [pre] (pa);
		\node[place,right of=b] (pb) {}
		edge [pre] (b);
		\node[transition,above of=pb] (b2) {$b$}
		edge [pre,bend right=15] (pb)
		edge [post,bend left=15] (pb);
		\node[transition,below of=pa] (b3) {$b$}
		edge [pre,bend left=15] (start)
		edge [post,bend right=15] (pb);
		\node[transition,right of=pb] (c) {$c$}
		edge [pre] (pb);
		\node[place,right of=c] (pc) {}
		edge [pre] (c);
		\node[transition,above of=pc] (c2) {$c$}
		edge [pre,bend right=15] (pc)
		edge [post,bend left=15] (pc);
		\node[transition,right of=pc] (d) {$d$}
		edge [pre] (pc);
		\node[place,right of=d] (pd) {}
		edge [pre] (d);
		\node[transition,above of=pd] (d2) {$d$}
		edge [pre,bend right=15] (pd)
		edge [post,bend left=15] (pd);
		\node[transition,right of=pd] (tau) {$\tau$}
		edge [pre] (pd);
		\node[place,right of=tau] (end) {}
		edge [pre] (tau);
	\end{tikzpicture}}
	
	\medskip
	\hrule
	\medskip
	
	\scalebox{\scalefactor}{
	\begin{tikzpicture}[node distance = 1cm,>=stealth',bend angle=0,auto]
		\node[place,tokens=1] (start) {};
		\node[yshift=1cm] at (start) {$M_2$:};
		\node[transition,right of=start] (a) {$a$}
		edge [pre] (start);
		\node[place,right of=a] (pa) {}
		edge [pre] (a);
		\node[transition,above of=pa] (a2) {$a$}
		edge [pre,bend right=15] (pa)
		edge [post,bend left=15] (pa);
		\node[transition,right of=pa] (b) {$b$}
		edge [pre] (pa);
		\node[place,right of=b] (pb) {}
		edge [pre] (b);
		\node[transition,above of=pb] (b2) {$b$}
		edge [pre,bend right=15] (pb)
		edge [post,bend left=15] (pb);
		\node[transition,below of=pa] (b3) {$b$}
		edge [pre,bend left=15] (start)
		edge [post,bend right=15] (pb);
		\node[transition,right of=pb] (c) {$c$}
		edge [pre] (pb);
		\node[place,right of=c] (pc) {}
		edge [pre] (c);
		\node[transition,above of=pc] (c2) {$c$}
		edge [pre,bend right=15] (pc)
		edge [post,bend left=15] (pc);
		\node[transition,right of=pc] (d) {$d$}
		edge [pre] (pc);
		\node[place,right of=d] (pd) {}
		edge [pre] (d);
		\node[transition,above of=pd] (d2) {$d$}
		edge [pre,bend right=15] (pd)
		edge [post,bend left=15] (pd);
		\node[transition,right of=pd] (e) {$e$}
		edge [pre] (pd);
		\node[place,right of=e] (pe) {}
		edge [pre] (e);
		\node[transition,above of=pe] (e2) {$e$}
		edge [pre,bend right=15] (pe)
		edge [post,bend left=15] (pe);
		\node[transition,right of=pe] (tau) {$\tau$}
		edge [pre] (pe);
		\node[place,right of=tau] (end) {}
		edge [pre] (tau);
		\node[transition,below of=pe] (tau2) {$\tau$}
		edge [pre,bend left=15] (pd)
		edge [post,bend right=15] (end);
	\end{tikzpicture}}
	
	\medskip
	\hrule
	\medskip
	
	\scalebox{\scalefactor}{
	\begin{tikzpicture}[node distance = 1cm,>=stealth',bend angle=0,auto]
		\node[place,tokens=1] (start) {};
		\node[yshift=1cm] at (start) {$M_3$:};
		\node[transition,right of=start] (a) {$a$}
		edge [pre] (start);
		\node[place,right of=a] (pa) {}
		edge [pre] (a);
		\node[transition,above of=pa] (a2) {$a$}
		edge [pre,bend right=15] (pa)
		edge [post,bend left=15] (pa);
		\node[transition,right of=pa] (b) {$b$}
		edge [pre] (pa);
		\node[place,right of=b] (pb) {}
		edge [pre] (b);
		\node[transition,above of=pb] (b2) {$b$}
		edge [pre,bend right=15] (pb)
		edge [post,bend left=15] (pb);
		\node[transition,below of=pa] (b3) {$b$}
		edge [pre,bend left=15] (start)
		edge [post,bend right=15] (pb);
		\node[transition,right of=pb] (c) {$c$}
		edge [pre] (pb);
		\node[place,right of=c] (pc) {}
		edge [pre] (c);
		\node[transition,above of=pc] (c2) {$c$}
		edge [pre,bend right=15] (pc)
		edge [post,bend left=15] (pc);
		\node[transition,right of=pc] (d) {$d$}
		edge [pre] (pc);
		\node[place,right of=d] (pd) {}
		edge [pre] (d);
		\node[transition,above of=pd] (d2) {$d$}
		edge [pre,bend right=15] (pd)
		edge [post,bend left=15] (pd);
		\node[transition,right of=pd] (e) {$e$}
		edge [pre] (pd);
		\node[place,right of=e] (pe) {}
		edge [pre] (e);
		\node[transition,above of=pe] (e2) {$e$}
		edge [pre,bend right=15] (pe)
		edge [post,bend left=15] (pe);
		\node[transition,right of=pe] (tau) {$\tau$}
		edge [pre] (pe);
		\node[place,right of=tau] (end) {}
		edge [pre] (tau);
		\node[transition,below of=pe] (tau2) {$\tau$}
		edge [pre,bend left=15] (pd)
		edge [post,bend right=15] (end);
		\node[transition,below of=b3] (v) {$v$};
		\node[place,left of=v] (pu) {}
		edge [post] (v);
		\node[transition,left of=pu] (u) {$u$}
		edge [pre] (start)
		edge [post] (pu);
		\node[place,right of=v] (pv) {}
		edge [pre] (v);
		\node[transition,right of=pv] (w) {$w$}
		edge [pre] (pv);
		\node[place,right of=w] (pw) {}
		edge [pre] (w);
		\node[transition,right of=pw] (x) {$x$}
		edge [pre] (pw);
		\node[place,right of=x] (px) {}
		edge [pre] (x);
		\node[transition,below of=px] (x2) {$x$}
		edge [pre,bend right=15] (px)
		edge [post,bend left=15] (px);
		\node[transition,right of=px] (y) {$y$}
		edge [pre] (px);
		\node[place,right of=y] (py) {}
		edge [pre] (y);
		\node[transition,right of=py] (z) {$z$}
		edge [pre] (py);
		\node[place,right of=z] (pz) {}
		edge [pre] (z);
		\node[transition,right of=pz] (tau3) {$\tau$}
		edge [pre] (pz)
		edge [post] (end);
	\end{tikzpicture}}
	\caption{The results of the directly follows miner for the input logs $L_1, L_2, L_3$ from the example in \cref{theo:dfm-cnc-entries}. $M_1$ is the model mined from the log $L_1$, $M_2$ the model mined from the log $L_2$, and $M_3$ the model mined from the log $L_3$.}
	\label{fig:dfm-netconn}
\end{figure}
The complexity scores of these models are:
\begin{itemize}
	\item[•] $\netconn(M_1) = 1.25$,
	\item[•] $\netconn(M_2) = 1.3$,
	\item[•] $\netconn(M_3) = 1.25$,
\end{itemize}
thus, we get the inequalities $\netconn(M_1) < \netconn(M_2)$, $\netconn(M_2) > \netconn(M_3)$, and $\netconn(M_1) = \netconn(M_3)$.
But the event logs $L_1, L_2, L_3$ have the following log complexity scores:
\begin{center}
	\begin{tabular}{|c|c|c|c|c|c|c|c|c|c|c|}\hline
		 & $\magnitude$ & $\variety$ & $\support$ & $\tlavg$ & $\tlmax$ & $\levelofdetail$ & $\numberofties$ & $\lempelziv$ & $\numberuniquetraces$ & $\percentageuniquetraces$ \\ \hline
		$L_1$ & $\pad 17 \pad$ & $\pad 4 \pad$ & $\pad 4 \pad$ & $\pad 4.25 \pad$ & $\pad 8 \pad$ & $\pad 2 \pad$ & $\pad 3 \pad$ & $\pad 9 \pad$ & $\pad 2 \pad$ & $\pad 0.5 \pad$ \\ \hline
		$L_2$ & $\pad 35 \pad$ & $\pad 5 \pad$ & $\pad 7 \pad$ & $\pad 5 \pad$ & $\pad 10 \pad$ & $\pad 4 \pad$ & $\pad 4 \pad$ & $\pad 17 \pad$ & $\pad 4 \pad$ & $\pad 0.5714 \pad$ \\ \hline
		$L_3$ & $\pad 56 \pad$ & $\pad 10 \pad$ & $\pad 9 \pad$ & $\pad 6.2222 \pad$ & $\pad 15 \pad$ & $\pad 5 \pad$ & $\pad 8 \pad$ & $\pad 27 \pad$ & $\pad 6 \pad$ & $\pad 0.6667 \pad$ \\ \hline
	\end{tabular}
	
	\medskip
	
	\begin{tabular}{|c|c|c|c|c|c|c|c|c|} \hline
		 & $\structure$ & $\affinity$ & $\deviationfromrandom$ & $\avgdist$ & $\varentropy$ & $\normvarentropy$ & $\seqentropy$ & $\normseqentropy$ \\ \hline
		$L_1$ & $\pad 3.25 \pad$ & $\pad 0.6429 \pad$ & $\pad 0.6045 \pad$ & $\pad 2.5 \pad$ & $\pad 6.4455 \pad$ & $\pad 0.2444 \pad$ & $\pad 11.7541 \pad$ & $\pad 0.244 \pad$ \\ \hline
		$L_2$ & $\pad 3.7143 \pad$ & $\pad 0.5538 \pad$ & $\pad 0.6489 \pad$ & $\pad 3.2381 \pad$ & $\pad 16.2978 \pad$ & $\pad 0.3384 \pad$ & $\pad 33.1288 \pad$ & $\pad 0.2662 \pad$ \\ \hline
		$L_3$ & $\pad 4 \pad$ & $\pad 0.4094 \pad$ & $\pad 0.6925 \pad$ & $\pad 6.5556 \pad$ & $\pad 53.0449 \pad$ & $\pad 0.4112 \pad$ & $\pad 82.0258 \pad$ & $\pad 0.3639 \pad$ \\ \hline
	\end{tabular}
\end{center}
Therefore, $\mathcal{C}^L(L_1) < \mathcal{C}^L(L_2) < \mathcal{C}^L(L_3)$ for any event log complexity measure $\mathcal{C}^L \in (\loc \setminus \{\affinity\})$.
For $\affinity$, consider the following event logs:
\begin{align*}
	L_1 &= [\langle a,a,b,b,c,c,d,d \rangle, \langle b,c,d \rangle] \\
	L_2 &= L_1 + [\langle a,a,b,b,c,c,d,d,e,e \rangle, \langle a,b,c,d,e \rangle] \\
	L_3 &= L_2 + [\langle a,a,a,b,b,b,c,c,c,d,d,d,e,e,e \rangle^{3}, \langle u,v,x,x,y,z \rangle]
\end{align*}
These event logs differ from those from before only in their frequencies.
Thus, the directly follows models for these logs are the same as those in \cref{fig:dfm-netconn}.
But these event logs have increasing affinity scores, since we can calculate that $\affinity(L_1) \approx 0.2857 < \affinity(L_2) \approx 0.4342 < \affinity(L_3) \approx 0.4621$.
Thus, we have $(\mathcal{C}^L, \netconn) \in \norel$ for all $\mathcal{C}^L \in \loc$. \hfill$\square$
\end{proof}

\begin{theorem}
\label{theo:dfm-dens-geq-entries}
Let $\mathcal{C}^L \in (\loc \setminus \{\variety\})$ be any log complexity measure.
Then, $(\mathcal{C}^L, \density) \in \mgeq$.
\end{theorem}
\begin{proof}
Let $L_1 \sqsubset L_2$ be event logs and $M_1, M_2$ the models found by the directly follows miner for $L_1, L_2$.
By the introductory discussion at the start of this subsection, we know that $\density(M_1) = \frac{1}{\variety(L_1) + 1}$ and $\density(M_2) = \frac{1}{\variety(L_2) + 1}$.
By \ref{lemma:not-variety-dependent}, we know that $\mathcal{C}^L(L_1) < \mathcal{C}^L(L_2)$ and $\density(M_1) = \density(M_2)$ is possible, since we can increase $\mathcal{C}^L$ without changing variety, and thus not changing density.
To see that $\mathcal{C}^L(L_1) < \mathcal{C}^L(L_2)$ and $\density(M_1) > \density(M_2)$ is also possible, consider the following event logs:
\begin{align*}
	L_1 &= [\langle a,b,c,d \rangle^{2}, \langle a,b,c,d,e \rangle^{2}, \langle d,e,a,b \rangle^{2}] \\
	L_2 &= L_1 + [\langle a,b,c,d,e \rangle^{2}, \langle d,e,a,b,c \rangle, \langle c,d,e,a,b \rangle, \langle e,c,d,a,b,c,f \rangle]
\end{align*}
Then, for the models $M_1, M_2$ found by the directly follows miner for $L_1, L_2$, we have $\density(M_1) = \frac{1}{6} > \frac{1}{7} = \density(M_2)$, because $\variety(L_1) = 5$ and $\variety(L_2) = 6$.
However, all log complexity scores increase between these event logs:
\begin{center}
	\begin{tabular}{|c|c|c|c|c|c|c|c|c|c|c|}\hline
		 & $\magnitude$ & $\variety$ & $\support$ & $\tlavg$ & $\tlmax$ & $\levelofdetail$ & $\numberofties$ & $\lempelziv$ & $\numberuniquetraces$ & $\percentageuniquetraces$ \\ \hline
		$L_1$ & $\pad 26 \pad$ & $\pad 5 \pad$ & $\pad 6 \pad$ & $\pad 4.3333 \pad$ & $\pad 5 \pad$ & $\pad 6 \pad$ & $\pad 5 \pad$ & $\pad 13 \pad$ & $\pad 3 \pad$ & $\pad 0.5 \pad$ \\ \hline
		$L_2$ & $\pad 53 \pad$ & $\pad 6 \pad$ & $\pad 11 \pad$ & $\pad 4.8182 \pad$ & $\pad 7 \pad$ & $\pad 30 \pad$ & $\pad 8 \pad$ & $\pad 22 \pad$ & $\pad 6 \pad$ & $\pad 0.5455 \pad$ \\ \hline
	\end{tabular}
	
	\medskip
	
	\begin{tabular}{|c|c|c|c|c|c|c|c|c|} \hline
		 & $\structure$ & $\affinity$ & $\deviationfromrandom$ & $\avgdist$ & $\varentropy$ & $\normvarentropy$ & $\seqentropy$ & $\normseqentropy$ \\ \hline
		$L_1$ & $\pad 4.3333 \pad$ & $\pad 0.56 \pad$ & $\pad 0.5757 \pad$ & $\pad 2.6667 \pad$ & $\pad 6.1827 \pad$ & $\pad 0.3126 \pad$ & $\pad 16.0483 \pad$ & $\pad 0.1894 \pad$ \\ \hline
		$L_2$ & $\pad 4.7273 \pad$ & $\pad 0.5721 \pad$ & $\pad 0.5995 \pad$ & $\pad 3.0909 \pad$ & $\pad 30.24 \pad$ & $\pad 0.4447 \pad$ & $\pad 62.1108 \pad$ & $\pad 0.2952 \pad$ \\ \hline
	\end{tabular}
\end{center}
Therefore, $\mathcal{C}^L(L_1) < \mathcal{C}^L(L_2)$ for all $\mathcal{C}^L \in \loc$, and $\density(M_1) > \density(M_2)$. \hfill$\square$
\end{proof}

\begin{theorem}
\label{theo:dfm-dens-greater-entry}
$(\variety, \density) \in \mgreater$.
\end{theorem}
\begin{proof}
Let $L_1 \sqsubset L_2$ be event logs and $M_1, M_2$ be the models found by the directly follows miner for $L_1, L_2$.
Suppose $\variety(L_1) < \variety(L_2)$.
Then, by the results of the introductory discussion at the start of this subsection, we get $\density(M_1) = \frac{1}{\variety(L_1) + 1} > \frac{1}{\variety(L_2) + 1} = \density(M_2)$. \hfill$\square$
\end{proof}

\begin{theorem}
\label{theo:dfm-dup-leq-entries}
Let $\mathcal{C}^L \in (\loc \setminus \{\variety, \levelofdetail, \numberofties\}$ be a log complexity measure.
Then, $(\mathcal{C}^L, \duplicate) \in \mleq$.
\end{theorem}
\begin{proof}
Let $L_1 \sqsubset L_2$ be event logs, $G_1, G_2$ their directly follows graphs, and $M_1, M_2$ be the models found by the directly follows miner for $L_1, L_2$.
We first observe that duplicate labels in the directly follows models appear whenever a node $v$ in the directly follows graph has multiple incoming edges.
Suppose $\mathcal{C}^L(L_1) \leq \mathcal{C}^L(L_2)$.
Then, every edge of $G_1$ is also part of $G_2$.
In turn, every node in $G_2$ has at least as many incoming edges as the same node in $G_1$.
Since we cannot delete any edges in the directly follows graph by adding behavior to an event log, this means $\duplicate(M_1) \leq \duplicate(M_2)$.
What remains to be shown is that both $\duplicate(M_1) = \duplicate(M_2)$ and $\duplicate(M_1) < \duplicate(M_2)$ are possible when $\mathcal{C}^L(L_1) < \mathcal{C}^L(L_2)$.

For the former, we have seen in \cref{lemma:dfg-unchanging} that it is possible to increase the log complexity scores for $\mathcal{C}^L$ without changing the directly follows graph.
By construction of the directly follows miner, then $M_1$ and $M_2$ also don't change, and thus $\duplicate(M_1) = \duplicate(M_2)$.
To see that $\duplicate(M_1) < \duplicate(M_2)$ is also possible, consider the following event logs:
\begin{align*}
	L_1 &= [\langle a,b,d \rangle^{2}, \langle a,c,d \rangle^{2}, \langle e \rangle] \\
	L_2 &= [\langle a,b,d,e \rangle, \langle a,c,d,e \rangle, \langle a,b,c,d \rangle, \langle a,b,c,b,d,e,f \rangle, \langle a,b,c,b,c,b,d,e,f \rangle]
\end{align*}
These event logs have the following log complexity scores:
\begin{center}
	\def\pad{\hspace*{1.5mm}}
	\begin{tabular}{|c|c|c|c|c|c|c|c|c|c|c|}\hline
		 & $\magnitude$ & $\variety$ & $\support$ & $\tlavg$ & $\tlmax$ & $\levelofdetail$ & $\numberofties$ & $\lempelziv$ & $\numberuniquetraces$ & $\percentageuniquetraces$ \\ \hline
		$L_1$ & $\pad 13 \pad$ & $\pad 5 \pad$ & $\pad 5 \pad$ & $\pad 2.6 \pad$ & $\pad 3 \pad$ & $\pad 3 \pad$ & $\pad 4 \pad$ & $\pad 8 \pad$ & $\pad 3 \pad$ & $\pad 0.6 \pad$ \\ \hline
		$L_2$ & $\pad 41 \pad$ & $\pad 6 \pad$ & $\pad 10 \pad$ & $\pad 4.1 \pad$ & $\pad 9 \pad$ & $\pad 14 \pad$ & $\pad 6 \pad$ & $\pad 18 \pad$ & $\pad 8 \pad$ & $\pad 0.8 \pad$ \\ \hline
	\end{tabular}
	
	\medskip
	
	\begin{tabular}{|c|c|c|c|c|c|c|c|c|} \hline
		 & $\structure$ & $\affinity$ & $\deviationfromrandom$ & $\avgdist$ & $\varentropy$ & $\normvarentropy$ & $\seqentropy$ & $\normseqentropy$ \\ \hline
		$L_1$ & $\pad 2.6 \pad$ & $\pad 0.2 \pad$ & $\pad 0.5417 \pad$ & $\pad 2.4 \pad$ & $\pad 6.0684 \pad$ & $\pad 0.5645 \pad$ & $\pad 11.1636 \pad$ & $\pad 0.3348 \pad$ \\ \hline
		$L_2$ & $\pad 3.7 \pad$ & $\pad 0.2316 \pad$ & $\pad 0.6705 \pad$ & $\pad 3.1333 \pad$ & $\pad 32.1247 \pad$ & $\pad 0.5742 \pad$ & $\pad 61.0512 \pad$ & $\pad 0.401 \pad$ \\ \hline
	\end{tabular}
\end{center}
\cref{fig:dfm-duplicate} shows the models $M_1, M_2$ found by the directly follows miner for the event logs $L_1, L_2$.
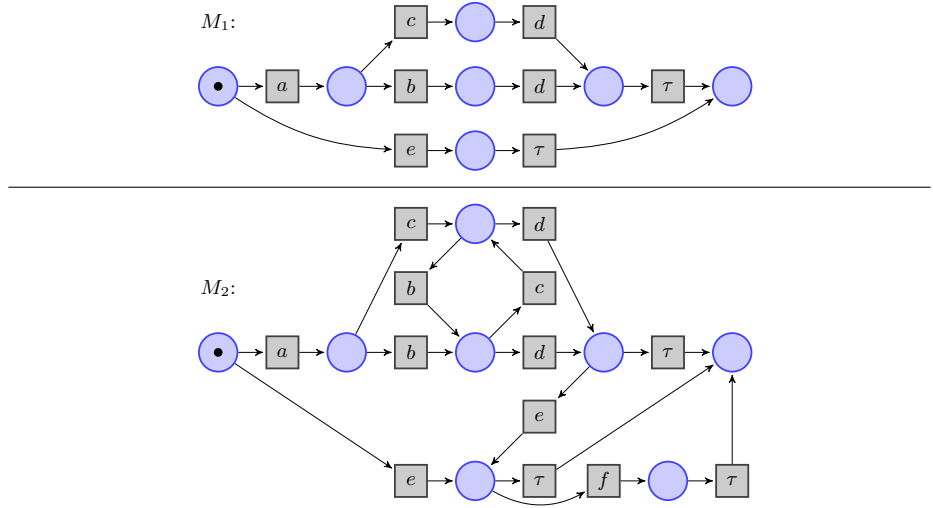
\begin{figure}[htp]
	\centering
	\scalebox{\scalefactor}{
	\begin{tikzpicture}[node distance = 1cm,>=stealth',bend angle=0,auto]
		\node[place,tokens=1] (start) {};
		\node[yshift=1cm] at (start) {$M_1$:};
		\node[transition,right of=start] (a) {$a$}
		edge [pre] (start);
		\node[place,right of=a] (pa) {}
		edge [pre] (a);
		\node[transition,right of=pa] (b) {$b$}
		edge [pre] (pa);
		\node[place,right of=b] (pb) {}
		edge [pre] (b);
		\node[transition,above of=b] (c) {$c$}
		edge [pre] (pa);
		\node[place,right of=c] (pc) {}
		edge [pre] (c);
		\node[transition,right of=pb] (d1) {$d$}
		edge [pre] (pb);
		\node[transition,right of=pc] (d2) {$d$}
		edge [pre] (pc);
		\node[place,right of=d1] (pd) {}
		edge [pre] (d1)
		edge [pre] (d2);
		\node[transition,right of=pd] (tau) {$\tau$}
		edge [pre] (pd);
		\node[place,right of=tau] (end) {}
		edge [pre] (tau);
		\node[transition,below of=b] (e) {$e$}
		edge [pre,bend left=15] (start);
		\node[place,right of=e] (pe) {}
		edge [pre] (e);
		\node[transition,right of=pe] (tau2) {$\tau$}
		edge [pre] (pe)
		edge [post,bend right=15] (end);
	\end{tikzpicture}}
	
	\medskip
	\hrule
	\medskip
	
	\scalebox{\scalefactor}{
	\begin{tikzpicture}[node distance = 1cm,>=stealth',bend angle=0,auto]
		\node[place,tokens=1] (start) {};
		\node[yshift=1cm] at (start) {$M_2$:};
		\node[transition,right of=start] (a) {$a$}
		edge [pre] (start);
		\node[place,right of=a] (pa) {}
		edge [pre] (a);
		\node[transition,right of=pa] (b) {$b$}
		edge [pre] (pa);
		\node[place,right of=b] (pb) {}
		edge [pre] (b);
		\node[transition,above of=b,yshift=1cm] (c) {$c$}
		edge [pre] (pa);
		\node[place,right of=c] (pc) {}
		edge [pre] (c);
		\node[transition,right of=pb] (d1) {$d$}
		edge [pre] (pb);
		\node[transition,right of=pc] (d2) {$d$}
		edge [pre] (pc);
		\node[place,right of=d1] (pd) {}
		edge [pre] (d1)
		edge [pre] (d2);
		\node[transition,right of=pd] (tau) {$\tau$}
		edge [pre] (pd);
		\node[place,right of=tau] (end) {}
		edge [pre] (tau);
		\node[transition,below of=b,yshift=-1cm] (e) {$e$}
		edge [pre] (start);
		\node[place,right of=e] (pe) {}
		edge [pre] (e);
		\node[transition,right of=pe] (tau2) {$\tau$}
		edge [pre] (pe)
		edge [post] (end);
		\node (dummy) at ($0.5*(pc) + 0.5*(pb)$) {};
		\node[transition,left of=dummy] (b2) {$b$}
		edge [pre] (pc)
		edge [post] (pb);
		\node[transition,right of=dummy] (c2) {$c$}
		edge [pre] (pb)
		edge [post] (pc);
		\node[transition,below of=d1] (e2) {$e$}
		edge [pre] (pd)
		edge [post] (pe);
		\node[transition,right of=tau2] (f) {$f$}
		edge [pre,bend left=30] (pe);
		\node[place,right of=f] (pf) {}
		edge [pre] (f);
		\node[transition,right of=pf] (tau3) {$\tau$}
		edge [pre] (pf)
		edge [post] (end);
	\end{tikzpicture}}
	\caption{The results of the directly follows miner for the input logs $L_1, L_2$ from the example in \cref{theo:dfm-dup-leq-entries}. $M_1$ is the model mined from the log $L_1$ and $M_2$ the model mined from the log $L_2$.}
	\label{fig:dfm-duplicate}
\end{figure}
For these models, we have $\duplicate(M_1) = 2 < 6 = \duplicate(M_2)$. \hfill$\square$
\end{proof}

\begin{theorem}
\label{theo:dfm-duplicate-less-entries}
Let $\mathcal{C}^L \in \{\variety, \levelofdetail, \numberofties\}$.
Then, $(\mathcal{C}^L, \duplicate) \in \mless$.
\end{theorem}
\begin{proof}
Let $L_1 \sqsubset L_2$ be event logs, $G_1, G_2$ their directly follows graphs, and $M_1, M_2$ be the models found by the directly follows miner for $L_1, L_2$.
In the proof of \cref{theo:dfm-dup-leq-entries}, we already argued that $\mathcal{C}^L(L_1) < \mathcal{C}^L(L_2)$, since duplicate labels in $M$ come from multiple edges entering a node in $G$.
Therefore, we get $\duplicate(M_1) \leq \duplicate(M_2)$.
Suppose $\mathcal{C}^L(L_1) < \mathcal{C}^L(L_2$.
In the proof of \cref{theo:dfg-cfc-less-entries}, we argued that this means $G_2$ contains a new path starting in $\triangleright$ and ending in $\square$ that is not part of $G_1$.
But then, there must be a node $v$ in $G_1$ whose number of incoming edges increased in $G_2$.
The directly follows miner creates transitions with the same labels for all of these edges, so the number of duplicate labels increases, i.e., $\duplicate(M_1) < \duplicate(M_2)$. \hfill$\square$
\end{proof}

\subsection{Other Discovery Algorithms}
Apart from the ones in the previous sections, we also investigated the relations between event log and model complexity for models found by the Inductive Miner~\cite{LeeFA14}, the Heuristics Miner~\cite{WeiAM06}, the Hybrid ILP Miner~\cite{ZelDAV18}, and several versions of the Alpha Miner~\cite{MedDAW04,WenAWS07,GuoWWYY15}.
The results of these analyses are reported in \cref{table:other-findings}.
\begin{table}[ht]
	\caption{The relations between the complexity scores of two nets $M_1$ and $M_2$ found by Inductive Miner, Heuristics Miner, Hybrid ILP Miner, and several versions of the Alpha Miner for the event logs $L_1$ and $L_2$ as input respectively, where $L_1 \sqsubset L_2$ and the complexity of $L_1$ is lower than the complexity of $L_2$.}
	\label{table:other-findings}
	\centering
	\resizebox{\textwidth}{!}{
	\begin{tabular}{|c|c|c|c|c|c|c|c|c|c|c|c|c|c|c|c|c|c|} \hline
		 & $\size$ & $\mismatch$ & $\connhet$ & $\crossconn$ & $\tokensplit$ & $\controlflow$ & $\separability$ & $\avgconn$ & $\maxconn$ & $\sequentiality$ & $\depth$ & $\diameter$ & $\cyclicity$ & $\netconn$ & $\density$ & $\duplicate$ & $\emptyseq$ \\ \hline
		$\magnitude$ & {$\norel$} & {$\norel$} & {$\norel$} & {$\norel$} & {$\norel$} & {$\norel$} & {$\norel$} & {$\norel$} & {$\norel$} & {$\norel$} & {$\norel$} & {$\norel$} & {$\norel$} & {$\norel$} & {$\norel$} & {$\norel$}/{$\meq$} & {$\norel$}/{$\meq$} \\ \hline
		
		$\variety$ & {$\norel$} & {$\norel$} & {$\norel$} & {$\norel$} & {$\norel$} & {$\norel$} & {$\norel$} & {$\norel$} & {$\norel$} & {$\norel$} & {$\norel$} & {$\norel$} & {$\norel$} & {$\norel$} & {$\norel$} & {$\norel$}/{$\meq$} & {$\norel$}/{$\meq$} \\ \hline
		
		$\support$ & {$\norel$} & {$\norel$} & {$\norel$} & {$\norel$} & {$\norel$} & {$\norel$} & {$\norel$} & {$\norel$} & {$\norel$} & {$\norel$} & {$\norel$} & {$\norel$} & {$\norel$} & {$\norel$} & {$\norel$} & {$\norel$}/{$\meq$} & {$\norel$}/{$\meq$} \\ \hline
		
		$\tlavg$ & {$\norel$} & {$\norel$} & {$\norel$} & {$\norel$} & {$\norel$} & {$\norel$} & {$\norel$} & {$\norel$} & {$\norel$} & {$\norel$} & {$\norel$} & {$\norel$} & {$\norel$} & {$\norel$} & {$\norel$} & {$\norel$}/{$\meq$} & {$\norel$}/{$\meq$} \\ \hline
		
		$\tlmax$ & {$\norel$} & {$\norel$} & {$\norel$} & {$\norel$} & {$\norel$} & {$\norel$} & {$\norel$} & {$\norel$} & {$\norel$} & {$\norel$} & {$\norel$} & {$\norel$} & {$\norel$} & {$\norel$} & {$\norel$} & {$\norel$}/{$\meq$} & {$\norel$}/{$\meq$} \\ \hline
		
		$\levelofdetail$ & {$\norel$} & {$\norel$} & {$\norel$} & {$\norel$} & {$\norel$} & {$\norel$} & {$\norel$} & {$\norel$} & {$\norel$} & {$\norel$} & {$\norel$} & {$\norel$} & {$\norel$} & {$\norel$} & {$\norel$} & {$\norel$}/{$\meq$} & {$\norel$}/{$\meq$} \\ \hline
		
		$\numberofties$ & {$\norel$} & {$\norel$} & {$\norel$} & {$\norel$} & {$\norel$} & {$\norel$} & {$\norel$} & {$\norel$} & {$\norel$} & {$\norel$} & {$\norel$} & {$\norel$} & {$\norel$} & {$\norel$} & {$\norel$} & {$\norel$}/{$\meq$} & {$\norel$}/{$\meq$} \\ \hline
		
		$\lempelziv$ & {$\norel$} & {$\norel$} & {$\norel$} & {$\norel$} & {$\norel$} & {$\norel$} & {$\norel$} & {$\norel$} & {$\norel$} & {$\norel$} & {$\norel$} & {$\norel$} & {$\norel$} & {$\norel$} & {$\norel$} & {$\norel$}/{$\meq$} & {$\norel$}/{$\meq$} \\ \hline
		
		$\numberuniquetraces$ & {$\norel$} & {$\norel$} & {$\norel$} & {$\norel$} & {$\norel$} & {$\norel$} & {$\norel$} & {$\norel$} & {$\norel$} & {$\norel$} & {$\norel$} & {$\norel$} & {$\norel$} & {$\norel$} & {$\norel$} & {$\norel$}/{$\meq$} & {$\norel$}/{$\meq$} \\ \hline
		
		$\percentageuniquetraces$ & {$\norel$} & {$\norel$} & {$\norel$} & {$\norel$} & {$\norel$} & {$\norel$} & {$\norel$} & {$\norel$} & {$\norel$} & {$\norel$} & {$\norel$} & {$\norel$} & {$\norel$} & {$\norel$} & {$\norel$} & {$\norel$}/{$\meq$} & {$\norel$}/{$\meq$} \\ \hline
		
		$\structure$ & {$\norel$} & {$\norel$} & {$\norel$} & {$\norel$} & {$\norel$} & {$\norel$} & {$\norel$} & {$\norel$} & {$\norel$} & {$\norel$} & {$\norel$} & {$\norel$} & {$\norel$} & {$\norel$} & {$\norel$} & {$\norel$}/{$\meq$} & {$\norel$}/{$\meq$} \\ \hline
		
		$\affinity$ & {$\norel$} & {$\norel$} & {$\norel$} & {$\norel$} & {$\norel$} & {$\norel$} & {$\norel$} & {$\norel$} & {$\norel$} & {$\norel$} & {$\norel$} & {$\norel$} & {$\norel$} & {$\norel$} & {$\norel$} & {$\norel$}/{$\meq$} & {$\norel$}/{$\meq$} \\ \hline
		
		$\deviationfromrandom$ & {$\norel$} & {$\norel$} & {$\norel$} & {$\norel$} & {$\norel$} & {$\norel$} & {$\norel$} & {$\norel$} & {$\norel$} & {$\norel$} & {$\norel$} & {$\norel$} & {$\norel$} & {$\norel$} & {$\norel$} & {$\norel$}/{$\meq$} & {$\norel$}/{$\meq$} \\ \hline
		
		$\avgdist$ & {$\norel$} & {$\norel$} & {$\norel$} & {$\norel$} & {$\norel$} & {$\norel$} & {$\norel$} & {$\norel$} & {$\norel$} & {$\norel$} & {$\norel$} & {$\norel$} & {$\norel$} & {$\norel$} & {$\norel$} & {$\norel$}/{$\meq$} & {$\norel$}/{$\meq$} \\ \hline
		
		$\varentropy$ & {$\norel$} & {$\norel$} & {$\norel$} & {$\norel$} & {$\norel$} & {$\norel$} & {$\norel$} & {$\norel$} & {$\norel$} & {$\norel$} & {$\norel$} & {$\norel$} & {$\norel$} & {$\norel$} & {$\norel$} & {$\norel$}/{$\meq$} & {$\norel$}/{$\meq$} \\ \hline
		
		$\normvarentropy$ & {$\norel$} & {$\norel$} & {$\norel$} & {$\norel$} & {$\norel$} & {$\norel$} & {$\norel$} & {$\norel$} & {$\norel$} & {$\norel$} & {$\norel$} & {$\norel$} & {$\norel$} & {$\norel$} & {$\norel$} & {$\norel$}/{$\meq$} & {$\norel$}/{$\meq$} \\ \hline
		
		$\seqentropy$ & {$\norel$} & {$\norel$} & {$\norel$} & {$\norel$} & {$\norel$} & {$\norel$} & {$\norel$} & {$\norel$} & {$\norel$} & {$\norel$} & {$\norel$} & {$\norel$} & {$\norel$} & {$\norel$} & {$\norel$} & {$\norel$}/{$\meq$} & {$\norel$}/{$\meq$} \\ \hline
		
		$\normseqentropy$ & {$\norel$} & {$\norel$} & {$\norel$} & {$\norel$} & {$\norel$} & {$\norel$} & {$\norel$} & {$\norel$} & {$\norel$} & {$\norel$} & {$\norel$} & {$\norel$} & {$\norel$} & {$\norel$} & {$\norel$} & {$\norel$}/{$\meq$} & {$\norel$}/{$\meq$} \\ \hline
	\end{tabular}
	}
\end{table}
Similar to the alpha miner, we were able to show that none of the log complexity measures are able to predict the complexity of discovered models, as almost all table entries contain an $\norel$.
The only exceptions are the complexity measure $\duplicate$ and $\emptyseq$ that contain the entry $\meq$ for some discovery algorithms. 
As such, the Inductive Miner and the Heuristics Miner never create empty sequence flows, so $\emptyseq(M) = 0$ for any model $M$ discovered by one of these algorithms.
The different versions of the Alpha Miner, on the other hand, create exactly one transition for each event name, and no other transitions, so $\duplicate(M) = 0$ for any model $M$ found by these algorithms. 
We can find a similar argument for the Hybrid ILP Miner, which creates exactly one transition for each event name, alongside two transitions labeled $\tau$. 
Therefore, a model $M$ found by the Hybrid ILP Miner always contains exactly one label repetition, and hence $\duplicate(M) = 1$.
The counter-examples that show the $\norel$-entries in \cref{table:other-findings} can be repoduced by the \texttt{anaLOG} tool available at \url{https://github.com/Pati-nets/anaLOG}.

\section{Conclusion}
\label{sec:conclusion}
Mature process discovery algorithms must give their users formal guarantees on the returned results~\cite{WerPWBR23}.
Such formal guarantees may predict what happens to discovered models when the complexity of the underlying event log increases.
Multiple authors define log complexity measures to use as a predictor for model complexity~\cite{Aal16,Guen09}.
But so far, no formal guarantees exist on whether these measures actually predict the complexity of discovered models.
In this paper, we thus investigated $18$ log complexity measures and $17$ model complexity measures that found recent interest from researchers, across a total of 10 discovery algorithms.
We found that even some complexity scores of the trace net could not be predicted by the complexity of the underlying event log.
For the alpha algorithm and other more intricate discovery techniques´, we found no connections between log- and model complexity at all.
Across the complexity scores of the directly follows miner and the directly follows graph, we found that only the size, control flow complexity, density, and the number of duplicate tasks can be described by current log complexity measures.
Our analyses showed that especially the variety (number of distinct activity names), the level of detail (number of distinct, simple paths in the directly follows graph), and the number of directly follows relations have the highest influence on the investigated discovery algorithms.
We further deepened our analysis by describing the model complexity scores of models found by the investigated discovery algorithms using only properties of the underlying event log.
We invite inventors of future discovery algorithms to perform these analyses as well, to provide insights into which log complexity measures predict the complexity of their results.
To help with this endeavor, we provided a publicly available command-line tool\footnote{Tool available at: \url{https://github.com/Pati-nets/anaLOG}} that can also be used to reproduce the results presented in this paper.

\FloatBarrier
\

\end{document}